\documentclass[reqno]{amsart}
\usepackage{fullpage,amsfonts,amssymb,amsthm,mathrsfs,graphicx,color,framed,esint,stackrel}
\usepackage{cancel,bm}
\usepackage{tikz}
\usetikzlibrary{decorations.markings}
\usepackage{booktabs,longtable}
\usepackage[initials]{amsrefs}

\DeclareMathAlphabet\mathbfcal{OMS}{cmsy}{b}{n}

\usepackage[colorlinks=true, pdfstartview=FitV, linkcolor=blue, citecolor=blue, urlcolor=blue]{hyperref}

\renewcommand{\d}{{\mathrm d}}

\newcommand{\im}{\mathrm{i}}
\newcommand{\e}{\mathrm{e}}

\def\tr{\mathop{\mathrm{tr}}\limits}

\newtheorem{theo}{Theorem}[section]

\newtheorem{lem}[theo]{Lemma}
\newtheorem{rem}[theo]{Remark}
\newtheorem{problem}[theo]{Riemann-Hilbert Problem}
\newtheorem{remark}[theo]{Remark}
\newtheorem{prop}[theo]{Proposition} 
\newtheorem{cor}[theo]{Corollary} 
\newtheorem{definition}[theo]{Definition}

\newtheorem{assum}[theo]{Assumption}

\usepackage{mdframed}

\begin{document}

\title[Kernel atlas]{A Riemann-Hilbert approach to Fredholm determinants of Hankel composition operators: scalar-valued kernels}

\author{Thomas Bothner}
\address{School of Mathematics, University of Bristol, Fry Building, Woodland Road, Bristol, BS8 1UG, United Kingdom}
\email{thomas.bothner@bristol.ac.uk}
\date{\today}


\thanks{The author is grateful to J. Baik, M. Bertola, A. Its and A. Krajenbrink for stimulating discussions. This work is supported by the Engineering and Physical Sciences Research Council through grant EP/T013893/2 and we dedicate the paper to Harold Widom (1932-2021) and his many path-breaking works. The author would also like to thank the anonymous referees for their valuable suggestions which improved the paper in a variety of ways, in particular with regards to Remark \ref{PM}.}

\begin{abstract}
We characterize Fredholm determinants of a class of Hankel composition operators via matrix-valued Riemann-Hilbert problems, for additive and multiplicative compositions. The scalar-valued kernels of the underlying integral operators are not assumed to display the integrable structure known from the seminal work of Its, Izergin, Korepin and Slavnov \cite{IIKS}. Yet we are able to describe the corresponding Fredholm determinants through a naturally associated Riemann-Hilbert problem of Zakharov-Shabat type by solely exploiting the kernels' Hankel composition structures. We showcase the efficiency of this approach through a series of examples, we then compute several rank one perturbed determinants in terms of Riemann-Hilbert data and finally derive Akhiezer-Kac asymptotic theorems for suitable kernel classes.
\end{abstract}

\dedicatory{Dedicated to the memory of Harold Widom}
\maketitle

\section{Introduction and motivation}\label{sec1}
This paper is concerned with a family of trace class integral operators acting on $L^2(0,\infty)$ or $L^2(0,1)$ whose Fredholm determinants can be characterized in terms of a canonical, auxiliary Riemann-Hilbert problem.\smallskip\\

To motivate our work we consider two examples from random matrix theory, cf. \cite{F1}.
\begin{definition}[Hsu \cite{H}, 1939; Wigner \cite{Wig}, 1957] A random real symmetric matrix ${\bf X}\in\mathbb{R}^{n\times n}$ belongs to the Gaussian orthogonal ensemble \textnormal{(GOE)} if its diagonal and upper triangular entries are independently chosen with pdfs
\begin{equation*}
	\frac{1}{\sqrt{2\pi}}\e^{-\frac{1}{2}x_{jj}^2}\ \ \ \textnormal{and}\ \ \ \frac{1}{\sqrt{\pi}}\e^{-x_{jk}^2},
\end{equation*}
respectively.
\end{definition}
\begin{definition}[Ginibre \cite{Gin}, 1965] A random real matrix ${\bf X}\in\mathbb{R}^{n\times n}$ belongs to the real Ginibre ensemble \textnormal{(GinOE)} if its entries are independently chosen with pdfs
\begin{equation*}
	\frac{1}{\sqrt{2\pi}}\e^{-\frac{1}{2}x_{jk}^2}.
\end{equation*}
\end{definition}
The GOE and GinOE constitute two of the most well-studied random matrix ensembles and many of their properties are textbook knowledge. For instance, the statistical behavior of extreme eigenvalues in both ensembles is classical: first, we have in distribution for any ${\bf X}\in\textnormal{GOE}$, see \cite{Bai,F2},
\begin{equation*}
	\sqrt{2}n^{\frac{1}{6}}\Big(\max_{i=1,\ldots,n}\lambda_i({\bf X})-\sqrt{2n}\,\Big)\Rightarrow F_1,\ \ \ \textnormal{as}\ n\rightarrow\infty,
\end{equation*}
where the random variable $F_1$, the ubiquitous Tracy-Widom GOE variable, obeys the law
\begin{equation}\label{i:1}
	\mathbb{P}\left(F_1\leq t\right)=\exp\left[-\frac{1}{2}\int_t^{\infty}(s-t)\big(q(s)\big)^2\,\d s-\frac{1}{2}\int_t^{\infty}q(s)\,\d s\right],\ \ t\in\mathbb{R},
\end{equation}
determined through a particular solution $q=q(s)$ of a second order nonlinear ODE boundary value problem, see \cite[$(53)$]{TW3}. Alternatively, and equivalently, the special function $q$ is determined through the solution of a distinguished Riemann-Hilbert problem, see Section \ref{seci23}. Second, for any ${\bf X}\in\textnormal{GinOE}$, it is known from \cite{RS,PTZ} that in distribution
\begin{equation*}
	\max_{\substack{i=1,\ldots,n\\ \lambda_i\in\mathbb{R}}}\lambda_i({\bf X})-\sqrt{n}\Rightarrow G_1,\ \ \ \textnormal{as}\ n\rightarrow\infty,
\end{equation*}
where the Rider-Sinclair GinOE variable $G_1$ obeys the law, cf. \cite[Theorem $1.4$]{BB0},
\begin{equation}\label{i:2}
	\mathbb{P}\left(G_1\leq t\right)=\exp\left[-\frac{1}{2}\int_t^{\infty}(s-t)\big(p(s)\big)^2\,\d s-\frac{1}{2}\int_t^{\infty}p(s)\,\d s\right],\ \ t\in\mathbb{R},
\end{equation}
determined through another special function $p=p(s)$ which can also be characterized via a Riemann-Hilbert problem, see Section \ref{seci23}. Thus, comparing \eqref{i:1} and \eqref{i:2}, we shall ask, and answer, the following question:\smallskip
\begin{quote}
	\hspace{2cm}Why are the right hand sides in formul\ae\,\eqref{i:1} and \eqref{i:2} so similar? 
\end{quote}\smallskip
After all, identity \eqref{i:1} describes the fluctuations of the largest real eigenvalue in a Hermitian ensemble, while identity \eqref{i:2} describes the fluctuations of the same extreme eigenvalue in a non-Hermitian ensemble. Yet, the expressions for the two soft edge limit laws look very similar, so perhaps this is not a coincidence!
\begin{rem} Next to any potential structural universality, there is the famous probabilistic universality: the limit law \eqref{i:1} captures the extreme eigenvalue fluctuations not only in the GOE, but in a larger class of real Wigner random matrices, see \cite{Sos}. Likewise, \eqref{i:2} holds true in a class of real non-Hermitian random matrices with independent, identically distributed entries \cite{Ces}. And these are just appearances of \eqref{i:1} and \eqref{i:2} random matrix theory, but they appear in other areas of mathematical physics, too, cf. \cite{Cor,KDO}.
\end{rem}
In order to answer our question, we will now recall how one typically proves \eqref{i:1} and \eqref{i:2}: Given that both, GOE and GinOE, are exactly solvable matrix models, their correlation functions are computable in terms of skew-orthogonal polynomials. More to the point, the random processes formed by the eigenvalues are Pfaffian point processes, so one can express gap functions in both ensembles as Fredholm Pfaffians for any finite $n$. Once simplified, along the lines of \cite[Section II]{TW3}, the Fredholm Pfaffians with matrix-valued kernels become Fredholm determinants with scalar-valued kernels, and are then amenable to the large $n$ limit. See \cite[Chapters $7,9$]{F1} for the details of this approach in case of the GOE and \cite[Section $2.1$]{RS} in case of the GinOE.
\begin{rem} Instead of converting Fredholm Pfaffians to Fredholm determinants, one can study the scaling limit without the conversion step and derive limiting Fredholm Pfaffian formul\ae\, for the distribution functions, cf. \cite{TW4} for the GOE. However it is not clear how the limiting Fredholm Pfaffians, without converting them to  Fredholm determinants, can directly yield the compact expressions in the right hand side of \eqref{i:1} or \eqref{i:2}.
\end{rem}
Executing the above approach, the left hand sides in \eqref{i:1} and \eqref{i:2} become square roots of Fredholm determinants of an operator $K_t+\alpha_t\otimes\beta_t$ acting on $L^2(0,\infty)$, where $K_t:L^2(0,\infty)\rightarrow L^2(0,\infty)$ is an integral operator with \textit{additive} Hankel composition kernel
\begin{equation}\label{i:4}
	K_t(x,y):=\int_0^{\infty}\phi(x+z+t)\psi(z+y+t)\,\d z,\ \ \ \ \ t\in\mathbb{R}.
\end{equation}
Here, $\phi,\psi,\alpha_t,\beta_t$ are certain scalar-valued functions and $\alpha\otimes\beta$ denotes a general rank one integral operator on $L^2(0,\infty)$. Next, in order to simplify the Fredholm determinant of $K_t+\alpha_t\otimes\beta_t$ and subsequently identify the special functions $q$ and $p$ in \eqref{i:1} and \eqref{i:2}, the traditional approach would ask one to convert the Hankel composition kernel \eqref{i:4} into an \textit{integrable} kernel, i.e. one would try to find functions $f_j,g_j\in L^{\infty}(0,\infty)$ such that
\begin{equation}\label{i:5}
	K_t(x,y)=\frac{\sum_{j=1}^Nf_j(x)g_j(y)}{x-y},\ \ \ \ \ \ \ \sum_{j=1}^Nf_j(x)g_j(x)=0,
\end{equation}
holds true on $\mathbb{R}^2$ for all $t\in\mathbb{R}$ with $N\in\mathbb{Z}_{\geq 1}$\footnote{The Fredholm determinant of $K_t+\alpha_t\otimes\beta_t$ involves a rank one perturbation, however after factoring out $I-K_t$, the major obstacle in simplifying it originates from the Fredholm determinant of $K_t$. The remaining finite-dimensional determinant is easier to handle, see Theorem \ref{theo2}.}. This step works out well for the GOE where $\phi$ and $\psi$ in \eqref{i:4} are Airy functions, but it does not work out for the GinOE where $\phi$ and $\psi$ are Gaussian functions. Still, achieving equality \eqref{i:5} has been considered desirable for a long time as integral operators with integrable kernels have many remarkable properties, see \cite{DInt}: linear combinations and compositions of integrable operators are integrable, so are resolvents of integrable operators, and most remarkably, the resolvent of an integrable operator can be computed in terms of the solution of a naturally associated Riemann-Hilbert problem. Examples of integrable operators have been studied for many years and their full theory was developed by Its, Izergin, Korepin and Slavnov \cite{IIKS} in the early 1990s. The same theory allows for a systematic analysis of the Fredholm determinant of an integrable operator from the viewpoint of integrable systems theory, as done in \cite{TW,TW4} for \eqref{i:1}, and from the viewpoint of asymptotic analysis, see \cite{DIZ}.
\begin{rem} A Hankel composition kernel \eqref{i:4} is frequently easier to work with when estimating and answering trace class kind of questions for the corresponding integral operator. Thus, effort has been spent on finding conditions so that integrable kernels are Hankel composition kernels, cf. \cite{Blow1,Blow2}. These conditions are mostly formulated in terms of differential equations satisfied by $f_j$ and $g_j$ in \eqref{i:5}, and thus more or less in terms of regularity properties of the same functions. As we will see, the methods developed in this paper do not rely on differential equations but on the algebraic composition structure of the kernel \eqref{i:4}.
\end{rem}

What happens if \eqref{i:5} cannot be achieved? In that case, one could recall that a Fredholm determinant, for a reasonable trace class operator, is defined in terms of the operator's traces and thus the determinant enjoys some conjugation invariance by Lidskii's theorem \cite[Corollary $3.8$]{S}. The very same gauge freedom was exploited in the works of Bertola and Cafasso \cite{BerCa0,BerCa1} on multi-time processes and it allowed them to convert some additive Hankel composition operators to integrable operators in Fourier space. In more detail, assuming that $\phi$ and $\psi$ in \eqref{i:4} admit contour integral representations of the form
\begin{equation}\label{i:6}
	\phi(x)=\int_{\gamma_1}\hat{\phi}(\lambda)\e^{-\im\lambda x}\,\d\lambda,\ \ \ \ \ \psi(y)=\int_{\gamma_2}\hat{\psi}(\mu)\e^{\im\mu y}\,\d\mu,\ \ \ \ \ (x,y)\in\mathbb{R}^2,
\end{equation}
for suitably chosen functions $\hat{\phi},\hat{\psi}$ and contours $\gamma_j\subset\mathbb{C}$, we obtain in \eqref{i:4},
\begin{equation*}
	K_t(x,y)=-\im\int_{\gamma_1}\int_{\gamma_2}\hat{\phi}(\lambda)\hat{\psi}(\mu)\e^{-\im(\lambda x-\mu y)}\e^{-\im t(\lambda-\mu)}\frac{\d\mu\,\d\lambda}{\lambda-\mu}\ \ \ \ \ \textnormal{provided}\ \ \ \Im(\lambda-\mu)<0,
\end{equation*}
and so with an integral formula for the characteristic function $\chi_{+}$ on $(0,\infty)$, such as, cf. \cite[Lemma $2.2$]{BB0},
\begin{equation*}
	\chi_+(y)=\frac{1}{2\pi\im}\int_{-\infty}^{\infty}\e^{\im y(\xi-\mu)}\frac{\d\xi}{\xi-\mu},\ \ \ \ y\in\mathbb{R}\setminus\{0\},\ \ \Im\mu>0,
\end{equation*}
all together,
\begin{equation}\label{i:7}
	K_t(x,y)\chi_+(y)=\int_{\gamma_1}\int_{-\infty}^{\infty}\frac{\e^{-\im\lambda x}}{\sqrt{2\pi}}\underbrace{\left[\hat{\phi}(\lambda)\e^{-\im t\lambda}\int_{\gamma_2}\frac{\hat{\psi}(\mu)\e^{\im t\mu}}{(\lambda-\mu)(\mu-\xi)}\,\d\mu\right]}_{=:N_t(\lambda,\xi)}\frac{\e^{\im\xi y}}{\sqrt{2\pi}}\,\d\xi\,\d\lambda.
\end{equation}
By identity \eqref{i:7}, $K_t\chi_+$ on $L^2(\mathbb{R})$, the operator with kernel $K_t(x,y)\chi_+(y)$, is essentially Fourier equivalent to the integral transformation $N_t:L^2(\mathbb{R})\rightarrow L^2(\gamma_1)$. More importantly, the kernel $N_t(\lambda,\xi)$ in \eqref{i:7} has the variable difference $\lambda-\xi$ in its denominator after using partial fractions, and it is of integrable type up to conjugation with a suitable multiplication operator. Hence, assuming \eqref{i:6}, a Hankel composition kernel of type \eqref{i:4} can sometimes 
be converted to an integrable operator in Fourier variables and so the integrable operator techniques of \cite{IIKS} are readily available for such kernels. In fact, the very same Fourier conjugation lies at the heart of the workings in \cite{BB0} and thus the derivation of \eqref{i:2} in the GinOE.\bigskip

However, working with the contour integral formul\ae\,\eqref{i:6} requires some degree of analyticity from either $\phi$ or $\psi$ since we cannot choose $\gamma_1=\gamma_2=\mathbb{R}$. As we will see, the method presented in this paper for additive Hankel composition operators does not rely on contour integral representations and is in that sense more general and comprehensive than the approach of Bertola and Cafasso. The fact that one can go without Fourier conjugation was first worked out in theoretical physics by Krajenbrink \cite{K}, itself inspired by \cite{DMS,DMSa}. The current paper puts the ideas in \cite{K} on firm mathematical grounds while extending them in several directions, both algebraic and asymptotic. For instance, \cite{K} is concerned with additive Hankel composition kernels only, but there are other integral operators in random matrix theory and/or mathematical physics. Here is one source of examples:
\begin{definition}[Wishart \cite{Wis}, 1928] A random real symmetric matrix ${\bf X}={\bf Y}^{\dagger}{\bf Y}\in\mathbb{R}^{n\times n}$ belongs to the Laguerre orthogonal ensemble \textnormal{(LOE)} if the entries of ${\bf Y}\in\mathbb{R}^{m\times n},m\geq n$ are independently chosen with pdfs
\begin{equation*}
	\frac{1}{\sqrt{2\pi}}\e^{-\frac{1}{2}y_{jk}^2}.
\end{equation*}
\end{definition}
Although somewhat less celebrated than the GOE or GinOE, the statistical behavior of the LOE's extreme eigenvalues is well understood: in distribution, for any ${\bf X}\in\textnormal{LOE}$, see \cite{Bro,F2,MP},
\begin{equation*}
	4n\Big(\min_{i=1,\ldots,n}\lambda_i({\bf X})\Big)\Rightarrow W_1,\ \ \ \ \textnormal{as}\ \ n,m\rightarrow\infty:\ \frac{n}{m}\rightarrow 1,
\end{equation*}
where the random variable $W_1$ obeys the law
\begin{equation}\label{i:8}
	\mathbb{P}\left(W_1\geq t\right)=\exp\left[-\frac{1}{8}\int_0^t\ln\Big(\frac{t}{s}\Big)\big(r(s)\big)^2\,\d s-\frac{1}{4}\int_0^tr(s)\frac{\d s}{\sqrt{s}}\right],\ \ \ t\in(0,\infty),
\end{equation}
determined through a particular solution $r=r(s)$ of a second order nonlinear ODE boundary value problem, see \cite[$(1.31)$]{F0}. Alternatively, and equivalently, the special function $r$ is determined through the solution of a distinguished Riemann-Hilbert problem, see Section \ref{seci26}. The point we are making is that the structure in the right hand side of \eqref{i:8} is common to other hard edge ensembles, cf. \cite{CGS}, and we will explain why this is the case: For one, following the same approach as in the computation of the soft edge limits \eqref{i:1},\eqref{i:2}, the left hand side in \eqref{i:8} is again a square root of a Fredholm determinant, this time of $K_t+\alpha_t\otimes\beta_t$ acting on $L^2(0,1)$ where $K_t:L^2(0,1)\rightarrow L^2(0,1)$ is an integral operator with \textit{multiplicative} Hankel composition kernel
\begin{equation}\label{i:10}
	K_t(x,y):=t\int_0^1\phi(xzt)\psi(zyt)\,\d z,\ \ \ \ \ \ t\in(0,\infty),
\end{equation}
with certain scalar-valued functions $\phi,\psi,\alpha_t,\beta_t$. Hence, to simplify the determinant of $K_t+\alpha_t\otimes\beta_t$, we are again tempted to bring the Hankel composition kernel in integrable shape, i.e. we try to achieve equality \eqref{i:5} but with \eqref{i:10} in its left hand side instead of \eqref{i:4}. This is surely possible for the LOE, where $\phi$ and $\psi$ are Bessel functions of square root argument, see \cite{TW2,F0}. However it is impossible for other hard-edge ensembles, e.g. for Muttalib-Borodin ensembles, Meijer-G product ensembles or chain matrix ensembles, see \cite{CGS,BerBo}. In those cases where \eqref{i:5} cannot be achieved for \eqref{i:10} the next best idea is to adapt \eqref{i:6} to the needs of \eqref{i:10}, so instead of Fourier representations, one uses Mellin integral representations for kernels with multiplicative composition structure. In detail, assuming
\begin{equation}\label{i:11}
	\phi(x)=\int_{\gamma_1}\hat{\phi}(\lambda)x^{\lambda-1}\,\d\lambda,\ \ \ \ \ \ \psi(y)=\int_{\gamma_2}\hat{\psi}(\mu)y^{-\mu}\,\d\mu,\ \ \ \ \ \ (x,y)\in(0,\infty)\times(0,\infty),
\end{equation}
hold true with some $\hat{\phi},\hat{\psi}$ and contours $\gamma_j\subset\mathbb{C}$, one transforms \eqref{i:10} into
\begin{equation*}
	K_t(x,y)=\int_{\gamma_1}\int_{\gamma_2}\hat{\phi}(\lambda)\hat{\psi}(\mu)x^{\lambda-1}y^{-\mu}t^{\lambda-\mu}\frac{\d\mu\,\d\lambda}{\lambda-\mu}\ \ \ \ \ \textnormal{provided}\ \ \ \ \Re(\lambda-\mu)>0,
\end{equation*}
and afterwards uses an integral formula for the characteristic function $\chi_{(0,1)}$ on $(0,1)\subset\mathbb{R}$. In turn, $K_t\chi_{(0,1)}$ on $L^2(0,\infty)$, the operator with kernel $K_t(x,y)\chi_{(0,1)}(y)$, is Mellin equivalent to an integrable operator, modulo a multiplication operator. The very same Mellin conjugation was first used in the works of Girotti \cite{Gir1,Gir2} on multi-time Bessel processes and subsequently in \cite{CGS} related to other hard-edge ensembles. Still, just as in \eqref{i:6}, the representation \eqref{i:11} implicitly requires some degree of regularity from $\phi$ and $\psi$, so we are not inclined to follow this route. Instead, without assuming contour integral formul\ae\, or differential equations for $\phi$ and $\psi$, the multiplicative Hankel composition structure itself will prove sufficient in the derivation of the right hand side in \eqref{i:8}. We now proceed with the statement and discussion of our results.
\section{Statement of results}\label{sec2}
Throughout, we abbreviate $\mathbb{R}_+:=(0,\infty)\subset\mathbb{R}$, we take $D$ as weak differentiation with respect to the independent variable, $M$ as multiplication by the independent variable and we require the following function spaces: with $N\in\mathbb{Z}_{\geq 1}$ and $1\leq p<\infty$ and some open subset $\Omega\subset\mathbb{R}$, the Sobolev space $W^{N,p}(\Omega):=\big\{f\in L^p(\Omega):\,D^kf\in L^p(\Omega),\ k=1,\ldots,N\big\}$, the space $H^{N,p}(\Omega):=\big\{f\in L^p(\Omega):\,(MD)^kf\in L^p(\Omega),\ k=1,\ldots,N\big\}$, and, in terms of 
\begin{equation*}
	L_{\circ}^1(\mathbb{R}_+):=\left\{f:\mathbb{R}_+\rightarrow\mathbb{C}\ \textnormal{measurable such that}\ f(x)/\sqrt{x}\in L^1(\mathbb{R}_+)\right\},\ \ \ \ \ \|f\|_{L_{\circ}^1(\mathbb{R}_+)}:=\int_0^{\infty}|f(x)|\frac{\d x}{\sqrt{x}},
\end{equation*}
the function space $H_{\circ}^{N,1}(\mathbb{R}_+):=\big\{f\in L_{\circ}^1(\mathbb{R}_+):\,(MD)^kf\in L_{\circ}^1(\mathbb{R}_+),\ k=1,\ldots,N\big\}$. Lastly, with $X$ and $Y$ Banach spaces and $\mathcal{H}$ a separable Hilbert space, we denote by $\mathcal{L}(X,Y)$ the linear space of all bounded and linear transformations from $X$ to $Y$ (abbreviating $\mathcal{L}(X)=\mathcal{L}(X,X)$), by $\mathcal{C}_0(X)\subset\mathcal{L}(X)$ the set of all finite-rank operators on $X$, and by $\mathcal{C}_p(\mathcal{H}),1\leq p<\infty$ the Schatten class of operators on $\mathcal{H}$.


\subsection{Additive composition, part 1}\label{seci21} Let $t\in J$ be chosen from some open subset $J\subseteq\mathbb{R}$ and consider two integral operators $M_t,N_t\in\mathcal{C}_2(L^2(\mathbb{R}_+))$ of \textit{additive} Hankel type. Precisely,
\begin{equation}\label{n0}
	(M_tf)(x):=\int_0^{\infty}\phi(x+y+t)f(y)\,\d y,\ \ \ \ \ (N_tf)(x):=\int_0^{\infty}\psi(x+y+t)f(y)\,\d y,
\end{equation}
where $\phi,\psi:\mathbb{R}\rightarrow\mathbb{C}$ are such that for all $t\in J$,
\begin{equation}\label{n00}
	\int_0^{\infty}x|\phi(x+t)|^2\d x<\infty\ \ \ \ \textnormal{and}\ \ \ \ \int_0^{\infty}x|\psi(x+t)|^2\d x<\infty.
\end{equation}
Now let $K_t\in\mathcal{L}(L^2(\mathbb{R}_+))$ denote the composition $K_t:=M_tN_t$ on $L^2(\mathbb{R}_+)$, equivalently,
\begin{equation}\label{n1}
	(K_tf)(x)=\int_0^{\infty}K_t(x,y)f(y)\,\d y,\ \ \ \ \ \ \ K_t(x,y):=\int_0^{\infty}\phi(x+z+t)\psi(z+y+t)\,\d z,
\end{equation}
and note that, although both, $M_t$ and $N_t$, are compact symmetric operators on $L^2(\mathbb{R}_+)$, their composition \eqref{n1} is in general not symmetric. Moving ahead, we need the following dominance notion, cf. \cite[$(4.6.9)$]{S2}.
\begin{definition}\label{domplus}
A family of functions, $\mathcal{F}:=\{f_{\alpha}\}_{\alpha\in S}$, in $L^p(\mathbb{R}_+),p\in\{1,2\}$, indexed by a set $S\subseteq\mathbb{R}$, is called $L^p(\mathbb{R}_+)$ dominated if for some $g\in L^p(\mathbb{R}_+)$, we have for all $\alpha\in S$ and for a.e. $x\in\mathbb{R}_+$ that 
\begin{equation*}
	|f_{\alpha}(x)|\leq g(x).
\end{equation*}
\end{definition}
Moreover, we use the shift of $f$ by $-t$, $(\tau_tf)(x):=f(x+t)$, and now introduce our main quantity of interest: since $K_t,t\in J$ equals the composition of two Hilbert-Schmidt operators, compare \eqref{n00}, its Fredholm determinant $F(t)$ on $L^2(\mathbb{R}_+)$,
\begin{equation}\label{i:12}
	F(t):=\prod_{k=1}^{\infty}\big(1-\lambda_k(t)\big),\ \ \ t\in J,
\end{equation}
with $\lambda_k(t)$ as the nonzero eigenvalues of $K_t$, counting multiplicity, is well-defined for all $t\in J$, see \cite[$(3.11)$]{S}. In turn, we record our first result, a very natural generalization of \cite[$(2.15)$]{K}.
\begin{lem}\label{lem1} Suppose $\phi,\psi:\mathbb{R}\rightarrow\mathbb{C}$, besides satisfying \eqref{n00}, are continuously differentiable on $\mathbb{R}$, obey
\begin{equation}\label{infbeh}
	\lim_{x\rightarrow+\infty}\phi(x)=\lim_{x\rightarrow+\infty}\psi(x)=0, 
\end{equation}
and for every $t\in J$, 
\begin{equation}\label{intbeh}
	\int_0^{\infty}x|(D\tau_t\phi)(x)|^2\,\d x<\infty,\ \ \ \ \ \ \ \int_0^{\infty}x|(D\tau_t\psi)(x)|^2\,\d x<\infty.
\end{equation}
Then we have, provided $I-K_t,t\in J$ is invertible on $L^2(\mathbb{R}_+)$ and the families $\{\tau_tf\}_{t\in\mathbb{R}}$, resp. $\{(I-K_t)^{-1}\tau_tg\}_{t\in J}$, with $f\in\{\phi,\psi,D\phi,D\psi\}$, resp. $g\in\{\phi,D\phi\}$, are $L^2(\mathbb{R}_+)$ dominated,
\begin{equation}\label{r1}
	\frac{\d^2}{\d t^2}\ln F(t)=-q_0(t)q_0^{\ast}(t),\ \ \ t\in J,\ \ \ \ \ \ \ \ \ \begin{cases}q_0(t):=\big((I-K_t)^{-1}\tau_t\phi\big)(0)&\smallskip\\ 
	q_0^{\ast}(t):=\big((I-K_t^{\ast})^{-1}\tau_t\psi\big)(0)&\end{cases}.
\end{equation}
Here, $K_t^{\ast}:=N_tM_t$ denotes the real adjoint of $K_t$.\footnote{As $J\ni t\mapsto F(t)$ might have non-zero winding, the second logarithmic derivative in the left hand side of \eqref{r1} is to be seen as shorthand for $\big(\frac{F'(t)}{F(t)}\big)'$ with $(')=\frac{\d}{\d t}$.}
\end{lem}
\begin{rem} Any smooth functions $\phi,\psi:\mathbb{R}\rightarrow\mathbb{R}$ of rapid decay at $+\infty$ such that $\|K_t\|<1,t\in J$ in operator norm on $L^2(\mathbb{R}_+)$ will work in Lemma \ref{lem1}. This is the setup of \cite{K}, here we will assume less.
\end{rem}

Identity \eqref{r1} is implicitly contained in the works of Tracy and Widom, cf. \cite[$(1.14)$]{TW}, on the Airy kernel where $\phi=\psi$ in \eqref{n0} are Airy functions. However, \cite{TW} relies on the integrable structure of the same kernel and so on the differential equation of the Airy function. Here we show that \eqref{r1} is a simple consequence of the kernel's algebraic structure \eqref{n1} and thus holds for \textit{any} additive Hankel composition operator in the class of Lemma \ref{lem1}. Observe that \eqref{r1} was first derived in \cite{K} for real-valued symmetric Hankel composition operators $K_t$ for which $q_0^{\ast}=q_0$. Identity \eqref{r1} is remarkable since it shows that the second logarithmic derivative of the Fredholm determinant \eqref{i:12} localizes into a product of two functions evaluated at $t\in J$. There is no apparent reason why this should happen for an arbitrary trace class integral operator.\bigskip

Once $J=\mathbb{R}$ in Lemma \ref{lem1} and $t\mapsto q_0(t),t\mapsto q_0^{\ast}(t)$ decay sufficiently fast at $t=+\infty$, formula \eqref{r1} implies a Tracy-Widom representation of $F(t)$, see \cite[$(1.17)$]{TW} for the original example of such a formula. The very same formula explains the common structure of the first factors in \eqref{i:1} and \eqref{i:2}.
\begin{cor}\label{gencoo} Let $\epsilon>0$. Assume that $\phi,\psi:\mathbb{R}\rightarrow\mathbb{C}$, besides satisfying \eqref{n00}, are continuously differentiable on $\mathbb{R}$, obey \eqref{infbeh},\eqref{intbeh} for all $t\in\mathbb{R}$ and the families $\{\sqrt{(\cdot)}\tau_tf\}_{t\in\mathbb{R}}$, resp. $\{\tau_tg\}_{t\in\mathbb{R}}$ and $\{(I-K_t)^{-1}\tau_th\}_{t\in\mathbb{R}}$, with $f\in\{\phi,\psi\}$, resp. $g\in\{\phi,\psi,D\phi,D\psi\}$ and $h\in\{\phi,D\phi\}$, are $L^2(\mathbb{R}_+)$ dominated. Then 
\begin{equation}\label{bet77}
	\ln F(t)=-\int_t^{\infty}(s-t)q_0(s)q_0^{\ast}(s)\,\d s,\ \ t\in\mathbb{R},
\end{equation}
provided $I-K_t$ is invertible on $L^2(\mathbb{R}_+)$ for all $t\in\mathbb{R}$ and provided there exist $c,t_0>0$ such that 
\begin{equation*}
	|q_0(t)q_0^{\ast}(t)|\leq ct^{-2-\epsilon}\ \ \  \textnormal{for all}\ t\geq t_0.
\end{equation*}
\end{cor}
Up to this point, our analysis has not yet unveiled a convenient way for the characterization of the functions $q_0$ and $q_0^{\ast}$ beyond the explicit formul\ae\,stated in Lemma \ref{lem1}. This will change, starting with the next algebraic result which provides a dynamical system for $\{q_0,q_0^{\ast}\}$ and certain generalizations thereof.
\begin{cor}\label{highercor} Let $N\in\mathbb{Z}_{\geq 1}$. Assume that $\phi,\psi:\mathbb{R}\rightarrow\mathbb{C}$, besides satisfying \eqref{n00}, are $N$-times continuously differentiable such that for every $k\in\{0,1,\ldots,N-1\}$,
\begin{equation}\label{n7}
	\lim_{x\rightarrow+\infty}(D^k\phi)(x)=\lim_{x\rightarrow+\infty}(D^k\psi)(x)=0,
\end{equation}
and for every $t\in J$,
\begin{equation*}
	\int_0^{\infty}x|(D\tau_t\phi)(x)|^2\d x<\infty,\ \ \ \ \ \int_0^{\infty}x|(D\tau_t\psi)(x)|^2\d x<\infty.
\end{equation*}
Then, provided $I-K_t,t\in J$ is invertible on $L^2(\mathbb{R}_+)$ and $\{D^n\tau_tf\}_{t\in\mathbb{R}}$, resp. $\{(I-K_t)^{-1}D^n\tau_tg\}_{t\in J}, \{(I-K_t^{\ast})^{-1}D^nh\}_{t\in J}$, with $f\in\{\phi,\psi\}$, resp. $g\in\{\phi\},h\in\{\psi\}$, are $L^2(\mathbb{R}_+)$ dominated for all $n\in\{0,1,\ldots,N\}$, the functions
\begin{equation}\label{n8}
	q_n(t):=\big((I-K_t)^{-1}D^n\tau_t\phi\big)(0),\ \ \ \ \ \ p_n(t):=\tr_{L^2(\mathbb{R}_+)}\big((I-K_t)^{-1}D^n\tau_t\phi\otimes\tau_t\psi\big),
\end{equation}
\begin{equation}\label{n9}
	\,q_n^{\ast}(t):=\big((I-K_t^{\ast})^{-1}D^n\tau_t\psi\big)(0),\ \ \ \ \ p_n^{\ast}(t):=\tr_{L^2(\mathbb{R}_+)}\big((I-K_t^{\ast})^{-1}D^n\tau_t\psi\otimes\tau_t\phi\big),
\end{equation}
defined for $t\in J$ and $n=0,1,\ldots,N$, satisfy the coupled system
\begin{equation}\label{r2}
	\begin{cases}\displaystyle\frac{\d q_n}{\d t}(t)=q_{n+1}(t)-q_0(t)p_n(t),\ \ \ \ \ \ \ \ \frac{\d p_n}{\d t}(t)=-q_0^{\ast}(t)q_n(t),&\bigskip\\
	\displaystyle \frac{\d q_n^{\ast}}{\d t}(t)=q_{n+1}^{\ast}(t)-q_0^{\ast}(t)p_n^{\ast}(t),\ \ \ \ \ \ \ \ \frac{\d p_n^{\ast}}{\d t}(t)=-q_0(t)q_n^{\ast}(t),&\end{cases}\ \ n=0,1,\ldots,N-1,\ \ \ \ \ \ t\in J.
\end{equation}
\end{cor}
System \eqref{r2} constitutes a chain of differential recurrence relations which, in general, does not close on its own, see below.
The same system was first derived for Fredholm determinants of additive real-valued symmetric Hankel composition operators in \cite[$(2.20)$]{K} and \cite[$(109)$]{DMSa}. In the symmetric setup $\phi=\psi$ and so one has $q_n=q_n^{\ast}$ and $p_n=p_n^{\ast}$. System \eqref{r2} admits several conserved quantities since
\begin{equation}\label{i:13}
	I_n:=p_{n+1}(t)+(-1)^np_{n+1}^{\ast}(t)+\sum_{k=0}^n(-1)^k\big(q_k^{\ast}(t)q_{n-k}(t)-p_k^{\ast}(t)p_{n-k}(t)\big)
\end{equation}
is $J\ni t$-independent for any $n\in\{0,1,\ldots,N-1\}$ by \eqref{r2} (see \cite[$(2.24)$]{K} for a special case of \eqref{i:13}). However, without placing further constraints on $\phi$ and $\psi$, system \eqref{r2} seemingly cannot yield a closed differential equation for $q_0$ or $q_0^{\ast}$, say, nor does it constitute a Hamiltonian dynamical system, compare the discussion in \cite[Section $4$]{K}. Still, we will show that system \eqref{r2} can be encoded in the solution of a canonical, auxiliary Riemann-Hilbert problem (RHP), and thus using \eqref{r1}, the same RHP allows us to access and efficiently analyze the Fredholm determinant $F(t)$. The details are as follows: consider the below problem, formulated for two arbitrary functions $\phi,\psi\in L^1(\mathbb{R})$. This problem relates to the Zakharov-Shabat systems \cite{ZS}, well-known in integrable systems and scattering theory, cf. \cite{AC,BDT,BC}, via a dressing transformation.
\begin{problem}\label{master2} Fix $t\in\mathbb{R}$ and $\phi,\psi\in L^1(\mathbb{R})$. Now determine ${\bf X}(z)={\bf X}(z;t,\phi,\psi)\in\mathbb{C}^{2\times 2}$ such that
\begin{enumerate}
	\item[(1)] ${\bf X}(z)$ is analytic for $z\in\mathbb{C}\setminus\mathbb{R}$.
	\item[(2)] ${\bf X}(z)$ admits continuous pointwise limits ${\bf X}_{\pm}(z):=\lim_{\epsilon\downarrow 0}{\bf X}(z\pm\im\epsilon),z\in\mathbb{R}$ which obey
	\begin{equation}\label{s8}
		{\bf X}_+(z)={\bf X}_-(z)\begin{bmatrix}1-r_1(z)r_2(z) & -r_2(z)\e^{-\im tz}\smallskip\\
		r_1(z)\e^{\im tz} & 1\end{bmatrix},\ \ z\in\mathbb{R},
	\end{equation}
	with
	\begin{equation}\label{s9}
		r_1(z):=-\im\int_{-\infty}^{\infty}\phi(y)\e^{-\im zy}\,\d y,\ \ \ \ \ \ r_2(z):=\im\int_{-\infty}^{\infty}\psi(y)\e^{\im zy}\,\d y,\ \ z\in\mathbb{R}.
	\end{equation}
	\item[(3)] As $z\rightarrow\infty$, 
	\begin{equation}\label{s10}
		{\bf X}(z)=\mathbb{I}+o(1).
	\end{equation}
\end{enumerate}
\end{problem}
We will prove in Lemma \ref{unique2} that conditions $(1)-(3)$ in RHP \ref{master2} determine ${\bf X}(z)$ uniquely. Moreover, subject to some additional assumptions placed on $\phi$ and $\psi$, the same RHP is solvable and its solution connects to $q_0$ and $q_0^{\ast}$ in the following way:
\begin{theo}\label{theo1} Suppose $\phi,\psi\in W^{1,1}(\mathbb{R})\cap L^{\infty}(\mathbb{R})$, besides obeying \eqref{n00}, are continuously differentiable on $\mathbb{R}$ with $D\phi,D\psi\in L^{\infty}(\mathbb{R}_+)$, satisfy
\begin{equation}\label{i:14}
	\int_0^{\infty}x|(D\tau_tf)(x)|^2\,\d x<\infty,\ \ \ \ \int_0^{\infty}\sqrt{\int_0^{\infty}|(\tau_tf)(x+y)|^2\,\d y}\,\d x<\infty,\ \ \ \ f\in\{\phi,\psi\},
\end{equation}
and $I-K_t$ is invertible on $L^2(\mathbb{R}_+)$. Then RHP \ref{master2} is uniquely solvable and its solution satisfies
\begin{equation}\label{i:14a}
	\lim_{\substack{z\rightarrow\infty\\ \Im z\not\equiv\textnormal{const.}}}z\big({\bf X}(z)-\mathbb{I}\big)=\begin{bmatrix}-\im p_0 & q_0^{\ast}\\ q_0 & \im p_0^{\ast}\end{bmatrix}=:{\bf X}_1={\bf X}_1(t,\phi,\psi)=\big[X_1^{mn}(t,\phi,\psi)\big]_{m,n=1}^2,
\end{equation}
with $\{q_0,q_0^{\ast},p_0=p_0^{\ast}\}$ as in \eqref{n8} and \eqref{n9}.
\end{theo}
The proof of Theorem \ref{theo1} is constructive and inspired by the non-rigorous workings in \cite[Section $5$]{K}, see Section \ref{sec4}.
\begin{rem}\label{IIKScon1} We now compare \eqref{s8} to the jump condition in the RHP for an integrable operator, cf. \cite[$(5.19)$]{IIKS}: if either $z\mapsto r_1(z)$ in \eqref{s9} admits analytic continuation to $H_+:=\{z\in\mathbb{C}:\ 0<\Im z<\epsilon\}$ such that $|r_1(z)|\rightarrow 0$ as $|z|\rightarrow\infty$ in $H_+$, or if $z\mapsto r_2(z)$ admits analytic continuation to $H_-:=\{z\in\mathbb{C}:\ 0<-\Im z<\epsilon\}$ such that $|r_2(z)|\rightarrow 0$ as $|z|\rightarrow\infty$ in $H_-$, then the factorization
\begin{equation*}
	\begin{bmatrix}1-r_1(z)r_2(z) & -r_2(z)\e^{-\im tz}\smallskip\\
		r_1(z)\e^{\im tz} & 1\end{bmatrix}=\begin{bmatrix}1 & -r_2(z)\e^{-\im tz}\\ 0 & 1\end{bmatrix}\begin{bmatrix}1 & 0\\ r_1(z)\e^{\im tz} & 1\end{bmatrix},\ \ \ z\in\mathbb{R},
\end{equation*}
readily shows that RHP \ref{master2} can be transformed to a RHP for an integrable operator by moving one of the matrix factors off the axis.
However, without placing explicit decay assumptions on $\phi$ and $\psi$, we cannot guarantee that their $L^1$-Fourier transforms \eqref{s9} admit analytic extensions in the first place. Still, a suitable approximation of $r_j$  and/or $\phi,\psi$, with passage to a limit after completing the corresponding Riemann-Hilbert analysis, might circumvent the extension problem. However, even then, there does not seem to exist an obvious relation between RHP \ref{master2} and the RHP for an integrable operator.
\end{rem}
Combining Lemma \ref{lem1} and Theorem \ref{theo1} we arrive at the following Zakharov-Shabat Riemann-Hilbert characterization of the Fredholm determinant \eqref{i:12}.
\begin{cor}\label{impcor} Under the assumptions of Lemma \ref{lem1} and Theorem \ref{theo1}, with \eqref{i:14} valid for all $t\in J$ and $I-K_t$ invertible on $L^2(\mathbb{R}_+)$ for the same $t$,
\begin{equation*}
	\frac{\d}{\d t}\ln F(t)=\im X_1^{11}(t,\phi,\psi),\ \ \ \ \frac{\d^2}{\d t^2}\ln F(t)=-X_1^{12}(t,\phi,\psi)X_1^{21}(t,\phi,\psi),\ \ t\in J,
\end{equation*}
in terms of the entries of the matrix coefficient ${\bf X}_1$ in \eqref{i:14a}.
\end{cor}
What's more, RHP \ref{master2} also allows us to access the higher functions \eqref{n8},\eqref{n9} for $n\in\mathbb{Z}_{\geq 1}$ provided we impose additional restrictions on $\phi$ and $\psi$, somewhat reminiscent of those in Corollary \ref{highercor}.
\begin{cor}\label{deeper} Fix $N\in\mathbb{Z}_{\geq 1}$. Provided $\phi,\psi\in W^{N,1}(\mathbb{R})$, besides obeying \eqref{n00}, are $N$-times continuously differentiable on $\mathbb{R}$, satisfy for every $k\in\{0,1,\ldots,N-1\}$,
\begin{equation*}
	D^kf\in L^{\infty}(\mathbb{R}),\ \ \ \ \int_0^{\infty}x|(D\tau_tf)(x)|^2\,\d x<\infty,\ \ \ \ \ \int_0^{\infty}\sqrt{\int_0^{\infty}|(\tau_tf)(x+y)|^2\,\d y}\,\d x<\infty,\ \ f\in\{\phi,\psi\},
\end{equation*}
as well as $D^Nf\in L^{\infty}(\mathbb{R}_+),f\in\{\phi,\psi\}$ and $I-K_t$ is invertible on $L^2(\mathbb{R}_+)$, then RHP \ref{master2} is uniquely solvable and its solution satisfies
\begin{equation}\label{s17}
	{\bf X}(z)=\mathbb{I}+\sum_{k=1}^N\begin{bmatrix}(-\im)^kp_{k-1} & \im^{k-1}q_{k-1}^{\ast}\smallskip\\
	(-\im)^{k-1}q_{k-1} & \im^kp_{k-1}^{\ast}\end{bmatrix}z^{-k}+o\big(z^{-N}\big),\ \ \ |z|\rightarrow\infty,\ \ \ \Im z\not\equiv\textnormal{const.}
\end{equation}
with $\{q_k,p_k,q_k^{\ast},p_k^{\ast}\}$ as in \eqref{n8} and \eqref{n9}.
\end{cor}
The last corollary concludes our first set of results for Fredholm determinants of additive Hankel composition operators. We now proceed to discuss some of their other properties, both algebraic and asymptotic, but within a narrower class of kernels.


\subsection{Additive composition, part 2}\label{seci22} As discussed in Section \ref{sec1}, certain determinants of additive Hankel composition operators appear in a perturbed fashion, especially in random matrix theory when concerned with extreme value statistics in real or symplectic ensembles. We now show how some of these perturbed determinants can be explicitly computed in terms of the solution of RHP \ref{master2}. The origin of this computation can be traced back to the works of Ferrari, Spohn \cite{FS}, Desrosiers, Forrester \cite{DF} and Krajenbrink \cite[Section $2.1$]{K} in theoretical physics. Consider the Hankel integral operator $H_t\in\mathcal{L}(L^2(\mathbb{R}_+))$ given by
\begin{equation}\label{sp:0}
	(H_tf)(x):=\int_0^{\infty}\phi(x+y+t)f(y)\,\d y,\ \ \ \ t\in\mathbb{R},
\end{equation} 
where $\phi:\mathbb{R}\rightarrow\mathbb{C}$ is continuously differentiable.
\begin{assum}\label{ass1} Assume $H_t\in\mathcal{C}_1(L^2(\mathbb{R}_+))$ and $\frac{\d}{\d t}H_t\in\mathcal{C}_1(L^2(\mathbb{R}_+))$ for all $t\in\mathbb{R}$, the latter with kernel $(D\tau_t\phi)(x+y)$, such that $\|H_t\|_1\rightarrow 0$ as $t\rightarrow+\infty$, for the trace norm $\|H_t\|_1$ of $H_t$. Moreover assume $I-\gamma H_t^2$ is invertible for all $(t,\gamma)\in\mathbb{R}\times[0,1]$.
\end{assum}
If $K_t$ denotes the composition operator $K_t:=H_t^2$ on $L^2(\mathbb{R}_+)$ with symmetric kernel
\begin{equation*}
	K_t(x,y)=\int_0^{\infty}\phi(x+z+t)\phi(z+y+t)\,\d z,
\end{equation*}
then the Fredholm determinant	$F(t,\gamma):=\prod_{k=1}^{\infty}\big(1-\gamma\lambda_k(t)\big),(t,\gamma)\in\mathbb{R}\times[0,1]$, where $\lambda_k(t)$ are the non-zero eigenvalues of $K_t$, satisfies by Corollary \ref{gencoo}, subject to the therein listed assumptions,
\begin{equation}\label{sp:-1}
	\ln F(t,\gamma)=-\int_t^{\infty}(s-t)\big(q(s,\gamma)\big)^2\d s,\ \ \ (t,\gamma)\in\mathbb{R}\times[0,1],
\end{equation}
with $q(t,\gamma):=\sqrt{\gamma}((I-\gamma K_t)^{-1}\tau_t\phi)(0)$. Now define the Fredholm determinant of $\gamma H_t$ on $L^2(\mathbb{R}_+)$, i.e.
\begin{equation}\label{i:15}
	G(t,\gamma):=\prod_{k=1}^{\infty}\big(1-\gamma\mu_k(t)\big),\ \ (t,\gamma)\in\mathbb{R}\times[-1,1],
\end{equation}
where $\mu_k(t)$ are the non-zero eigenvalues of $H_t$, counting multiplicity, and consider the following three functions, indexed according to their typical appearances in real or real-quaternion matrix ensembles:
\begin{equation}\label{i:16}
	F_1^{[1]}(t,\gamma):=F(t,\gamma)\left\{1-\gamma\int_0^{\infty}\big((I-\gamma K_t)^{-1}\tau_t\phi\big)(x)\left[1-\int_x^{\infty}(\tau_t\phi)(y)\,\d y\right]\d x\right\},
\end{equation}
followed by, where $\gamma_{\circ}:=\gamma(2-\gamma)$,
\begin{equation}\label{i:17}
	F_1^{[2]}(t,\gamma):=F(t,\gamma_{\circ})\left\{1-\gamma\int_0^{\infty}\big((I-\gamma_{\circ}K_t)^{-1}\tau_t\phi\big)(x)\left[1-\int_x^{\infty}(\tau_t\phi)(y)\,\d y\right]\d x\right\},
\end{equation}
and concluding with
\begin{equation}\label{i:18}
	F_4(t,\gamma):=F(t,\gamma)\left\{1+\gamma\int_0^{\infty}\big((I-\gamma K_t)^{-1}\tau_t\phi\big)(x)\left[\frac{1}{2}\int_x^{\infty}(\tau_t\phi)(y)\,\d y\right]\d x\right\}.
\end{equation}
Indeed, special cases of \eqref{i:16} appear in the GOE and GinOE, specifically in the superimposed orthogonal ensemble, cf. \cite{DF,F0}. Likewise, \eqref{i:17} appears in the analysis of the thinned GOE and GinOE \cite{BB,BBu} and \eqref{i:18}, for instance, in the thinned Gaussian symplectic ensemble, cf. \cite{BBu}. Thanks to the additive Hankel structure in \eqref{sp:0}, we can evaluate all three functions \eqref{i:16},\eqref{i:17},\eqref{i:18} either in terms of Riemann-Hilbert data, i.e. in terms of $q(t,\gamma)$, see \eqref{sp:-1} and Theorem \ref{theo1}, or, equivalently, in terms of $G(t,\gamma)$, compare \eqref{i:15}. The details are as follows and they explain in particular why the right-hand sides in \eqref{i:1} and \eqref{i:2} are so similar, see also Section \ref{seci23}.
\begin{theo}\label{theo2} Suppose $H_t\in\mathcal{L}(L^2(\mathbb{R}_+))$ defined in \eqref{sp:0} satisfies the conditions in Assumption \ref{ass1} with $\phi:\mathbb{R}\rightarrow\mathbb{C}$ continuously differentiable on $\mathbb{R}$. Assume further $\phi\in L^1(\mathbb{R}_+)$, that
\begin{equation*}
	\lim_{x\rightarrow+\infty}\phi(x)=0,
\end{equation*}
and that for every $t\in\mathbb{R}$,
\begin{equation}\label{i:18a}
	\int_0^{\infty}x|(\tau_tf)(x)|^2\,\d x<\infty,\ \ \ \int_0^{\infty}\sqrt{\int_0^{\infty}|(\tau_tf)(x+y)|^2\,\d y}\,\d x<\infty,\ \ \ f\in\{\phi,D\phi\}.
\end{equation}
Then we have for any $(t,\gamma)\in\mathbb{R}\times[0,1]$, denoting $\omega(t,\gamma):=\int_t^{\infty}q(s,\gamma)\,\d s$,
\begin{align}\label{i:19}
	F_1^{[1]}(t,\gamma)=F(t,\gamma)\Big\{\cosh\omega(t,\gamma)&\,-\sqrt{\gamma}\sinh\omega(t,\gamma)\Big\}\nonumber\\
	&\,=\frac{1}{2}(1-\sqrt{\gamma})\big(G(t,-\sqrt{\gamma})\big)^2+\frac{1}{2}(1+\sqrt{\gamma})\big(G(t,\sqrt{\gamma})\big)^2,
\end{align}
followed by
\begin{align}\label{i:20}
	F_1^{[2]}(t,\gamma)=F(t,\gamma_{\circ})&\,\left\{\frac{1-\gamma+\cosh\omega(t,\gamma_{\circ})-\sqrt{\gamma_{\circ}}\sinh\omega(t,\gamma_{\circ})}{2-\gamma}\right\}\nonumber\\
	&\hspace{1.65cm}=\left[\sqrt{\frac{1-\sqrt{\gamma_{\circ}}}{2(2-\gamma)}}G(t,-\sqrt{\gamma_{\circ}})+\sqrt{\frac{1+\sqrt{\gamma_{\circ}}}{2(2-\gamma)}}G(t,\sqrt{\gamma_{\circ}})\right]^2,
\end{align}
and concluding with
\begin{equation}\label{i:21}
	F_4(t,\gamma)=F(t,\gamma)\left\{\cosh\left(\frac{1}{2}\omega(t,\gamma)\right)\right\}^2=\left[\frac{1}{2}G(t,-\sqrt{\gamma})+\frac{1}{2}G(t,\sqrt{\gamma})\right]^2,
\end{equation}
provided $\{(I-\gamma K_t)^{-1}\tau_t\phi\}_{t\in\mathbb{R}},\{(I\pm\sqrt{\gamma}H_t)^{-1}\tau_tf\}_{t\in\mathbb{R}}$ with $f\in\{\phi,D\phi\}$ and $\{(I-\gamma K_t)^{-1}H_tD\tau_t\phi\}_{t\in\mathbb{R}}$ are $L^1(\mathbb{R}_+)$ dominated and $\{\tau_t\phi\}_{t\in\mathbb{R}},\{D\tau_t\phi\}_{t\in\mathbb{R}}$ are $L^2(\mathbb{R}_+)$ dominated, and provided there exist $c,t_0>0$ such that $|q(t,\gamma)|\leq ct^{-1-\epsilon}$ for all $t\geq t_0,\gamma\in[0,1]$ with some $\epsilon>0$,
\end{theo}
\begin{rem} Our focus in this paper lies on Fredholm determinants of Hankel composition operators. However, the proof of Theorem \ref{theo2} will use the identity
\begin{equation*}
	2\ln G(t,\sqrt{\gamma})=\ln F(t,\gamma)-\omega(t,\gamma),\ \ (t,\gamma)\in\mathbb{R}\times[0,1],
\end{equation*}
see \eqref{c3}, so our results, both analytic and asymptotic, allow us to study Fredholm determinants of ordinary trace class Hankel integral operators, too - subject to the necessary assumptions placed on their kernels.
\end{rem}
\begin{rem} Observe that \eqref{i:20} and \eqref{i:21} expresses the functions $F_1^{[2]}$ and $F_4$ as squares of convex combinations of the simpler determinants $G(t,\pm\sqrt{\gamma})$. This is particularly useful in numerical simulations.
\end{rem}
At this point of our analysis of Fredholm determinants of additive Hankel composition operators, we turn to their asymptotic behavior as $t\rightarrow\pm\infty$. Firstly, within the class of $K_t$ as in \eqref{n1},\eqref{n0}, by \cite[Theorem $3.4$]{S},
\begin{equation}\label{i:22}
	|F(t)-1|\leq\|K_t\|_1\e^{1+\|K_t\|_1},\ \ t\in\mathbb{R},
\end{equation}
where the trace norm $\|K_t\|_1$ of $K_t$ obeys
\begin{equation*}
	\|K_t\|_1\leq\sqrt{\int_0^{\infty}x|(\tau_t\phi)(x)|^2\,\d x}\,\sqrt{\int_0^{\infty}x|(\tau_t\psi)(x)|^2\,\d x}\rightarrow 0\ \ \ \textnormal{as}\ \ \ t\rightarrow+\infty,
\end{equation*}
by \eqref{infbeh}, provided the family $\{\sqrt{(\cdot)}\,\tau_tf\}_{t\geq t_0}$ for $f\in\{\phi,\psi\}$ is $L^2(\mathbb{R}_+)$ dominated for some $t_0>0$. Thus \eqref{i:22} translates into the following easy, albeit very crude, estimate.
\begin{cor} Suppose $\phi,\psi:\mathbb{R}\rightarrow\mathbb{C}$, besides satisfying \eqref{n00}, obey \eqref{infbeh} and the families $\{\sqrt{(\cdot)}\tau_tf\}_{t\geq t_0}$ with $f\in\{\phi,\psi\}$ are $L^2(\mathbb{R}_+)$ dominated for some $t_0>0$. Then, as $t\rightarrow+\infty$,
\begin{equation}\label{i:23}
	F(t)=1+o(1).
\end{equation}
\end{cor}
One can clearly improve \eqref{i:23} by imposing further regularity and/or integrability constraints on $\phi,\psi$, but we shall focus on the behavior of $F(t)$ as $t\rightarrow-\infty$ instead. As it happens, using RHP \ref{master2} and the Deift-Zhou nonlinear steepest descent method \cite{DZ}, we are able to prove the following Akhiezer-Kac theorem for \eqref{i:12} in a suitable class of additive Hankel composition kernels.
\begin{theo}\label{theo2a} Let $\epsilon>0$. Suppose $\phi,\psi:\mathbb{R}\rightarrow\mathbb{C}$ are continuously differentiable such that
\begin{equation}\label{i:24}
	|f(x)|\leq \e^{-a|x|},\ \ x\in\mathbb{R}
\end{equation}
holds true for $f\in\{\phi,\psi,D\phi,D\psi\}$ with $a\geq 2+\epsilon$ and $I-\gamma K_t$ is invertible for all $(t,\gamma)\in\mathbb{R}\times[0,1]$. Then there exist $c=c(\epsilon),t_0=t_0(\epsilon)>0$ so that
\begin{equation}\label{i:25}
	\ln F(t)=s(0)t+\int_0^{\infty}s(x)s(-x)x\,\d x-\frac{1}{4\pi\im}\int_{-\infty}^{\infty}\left\{\frac{r_1'(\lambda)}{r_1(\lambda)}-\frac{r_2'(\lambda)}{r_2(\lambda)}\right\}\ln\big(1-r_1(\lambda)r_2(\lambda)\big)\,\d\lambda+r(t)
\end{equation}
for $(-t)\geq t_0$ with the function, in terms of the principal branch for the logarithm $\ln:\mathbb{C}\setminus(-\infty,0]\rightarrow\mathbb{C}$,
\begin{equation*}
	s(x):=-\frac{1}{2\pi}\int_{-\infty}^{\infty}\ln\big(1-r_1(y)r_2(y)\big)\e^{\im xy}\,\d y,\ \ \ x\in\mathbb{R}_{\geq 0}.
\end{equation*} 
The error term $r(t)$ in \eqref{i:25} is $t$-differentiable and satisfies
\begin{equation}\label{i:26}
	|r(t)|\leq c\e^{-\epsilon|t|}\ \ \ \forall\,(-t)\geq t_0.
\end{equation}
\end{theo}
\begin{rem} Assumption \eqref{i:24} is natural from the viewpoint of Paley-Wiener theory since our proof of Theorem \ref{theo2a} requires analytic extensions of $z\mapsto r_j(z)$ off the real line. However, the same assumption is by no means optimal, compare the recent work \cite{FTZ} further discussed in Remark \ref{Kacrem} below.
\end{rem}
The first two terms in estimate \eqref{i:25} bear a striking resemblance to the leading order terms in the classical Akhiezer-Kac theorem for truncated Wiener-Hopf operators \cite{A,Ka}. Indeed, the same theorem states that under appropriate conditions on the scalar function $k(x)$ the operator $T_{\alpha}$ on $L^2(0,\alpha)$ defined by
\begin{equation*}
	(T_{\alpha}f)(x):=f(x)+\int_0^{\alpha}k(x-y)f(y)\,\d y,
\end{equation*}
has its Fredholm determinant given asymptotically as $\alpha\rightarrow+\infty$ by
\begin{equation}\label{i:27}
	\ln\det T_{\alpha}=r(0)\alpha+\int_0^{\infty}r(x)r(-x)x\,\d x+o(1),
\end{equation}
where
\begin{equation*}
	r(x)=\frac{1}{2\pi}\int_{-\infty}^{\infty}\e^{\im xy}\ln\left\{1+\int_{-\infty}^{\infty}\e^{-\im yz}k(z)\,\d z\right\}\d y.
\end{equation*}
The difference between our $s(x)$ and $r(x)$ in \eqref{i:27}, and the appearance of the third term in \eqref{i:25}, is a manifestation of the fact that we are analyzing determinants of \textit{composition} operators, rather than determinants of ordinary integral operators. In particular, when $\phi=\psi$ and both functions are even, then the third term in \eqref{i:25} is identically zero since then $r_1=-r_2$. This choice of symmetric composition kernels thus matches precisely the classical form \eqref{i:27} of the Akhiezer-Kac theorem.
\begin{rem}\label{Kacrem} Noting that both, $r_1$ and $r_2$ are Fourier integrals, see \eqref{s9}, we can rewrite \eqref{i:25} as follows: let $f^{\ast n}$ denote the $n$-fold convolution of a function $f\in L^1(\mathbb{R})$ with itself, i.e.
\begin{equation*}
	f^{\ast 1}(x)=f(x),\ \ \ f^{\ast 2}(x)=\int_{-\infty}^{\infty}f(x-y)f(y)\,\d y,\ \ \ f^{\ast 3}(x)=\int_{-\infty}^{\infty}f(x-y)f^{\ast 2}(y)\,\d y,\ \ \ \ldots
\end{equation*}
Then
\begin{equation}\label{i:28}
	s(0)=\sum_{n=1}^{\infty}\frac{1}{n}\omega^{\ast n}(0),\ \ \ \ \ \ \omega(x):=\int_{-\infty}^{\infty}\phi(x+y)\psi(y)\,\d y,
\end{equation}
followed by
\begin{equation}\label{i:29}
	\int_0^{\infty}s(x)s(-x)x\,\d x=\sum_{n,m=1}^{\infty}\frac{1}{nm}\int_0^{\infty}\omega^{\ast n}(x)\omega^{\ast m}(-x)x\,\d x,
\end{equation}
and concluding with
\begin{equation}\label{i:30}
	-\frac{1}{4\pi\im}\int_{-\infty}^{\infty}\left\{\frac{r_1'(\lambda)}{r_1(\lambda)}-\frac{r_2'(\lambda)}{r_2(\lambda)}\right\}\ln\big(1-r_1(\lambda)r_2(\lambda)\big)\,\d\lambda=\sum_{n=1}^{\infty}\frac{1}{n^2}\int_{-\infty}^{\infty}x\phi^{\ast n}(x)\psi^{\ast n}(x)\,\d x.
\end{equation}
Estimate \eqref{i:25}, with \eqref{i:28},\eqref{i:29},\eqref{i:30} in place, for $\phi=\psi$ when $\omega^{\ast n}(x)=\omega^{\ast n}(-x)$, with $\phi=\psi$ real-valued, non-negative, continuous and in $L^2(\mathbb{R})\cap L^{\infty}(\mathbb{R})$ such that $M\phi\in L^1(\mathbb{R})$ was recently derived in \cite[Theorem $2$]{FTZ} through probabilistic reasoning. Here, based on nonlinear steepest descent complex analytic techniques, we prove \eqref{i:25} in the non-symmetric complex-valued setup but with stronger asymptotic constraints placed on $\phi,\psi$, see \eqref{i:24}.
\end{rem}
\begin{rem} Our proof workings for Theorem \ref{theo2a} allow us to compute the $t\rightarrow-\infty$ asymptotics of the specialized determinant $F(t,\gamma),\gamma\in[0,1]$ that appears in \eqref{i:16},\eqref{i:17},\eqref{i:18}, subject to the constraint \eqref{i:24}. Moreover, using the approach of \cite[Section $7.2.1$]{BB}, we can compute the $t\rightarrow-\infty$ asymptotics of $\omega(t,\gamma)$ appearing in Theorem \ref{theo2} and so the $t\rightarrow-\infty$ asymptotics of \eqref{i:19},\eqref{i:20} and \eqref{i:21} become accessible. We leave the details of this calculation to the dedicated reader.
\end{rem}
The last remark concludes this subsection.

\subsection{Additive composition, part 3}\label{seci23} We now apply the results of Section \ref{seci21} and \ref{seci22} to two composition operators occurring in random matrix theory, namely those that underlie \eqref{i:1} and \eqref{i:2}.

\subsubsection{The Gaussian kernel} Let $J=\mathbb{R}$ and choose
\begin{equation*}
	\phi(x)=\psi(x)=\frac{1}{\sqrt{\pi}}\e^{-x^2},\ \ \ x\in\mathbb{R}.
\end{equation*}
The Gaussian is a Schwartz function on $\mathbb{R}$, so all dominance and integrability assumptions in Section \ref{seci21} and \ref{seci22} are trivially satisfied. Moreover, the symmetric integral operator $K_t$ in \eqref{n1} with kernel
\begin{equation}\label{examp1}
	K_t(x,y)=\frac{1}{\pi}\int_0^{\infty}\e^{-(x+z+t)^2}\e^{-(z+y+t)^2}\,\d z,
\end{equation}
corresponding to the Gaussian choice, is trace class on $L^2(\mathbb{R}_+)$ and $I-\gamma K_t$ is invertible on the same space for all $(t,\gamma)\in\mathbb{R}\times[0,1]$, see \cite[Lemma $2.1$]{BB0}. Lastly, $H_t\in\mathcal{C}_1(L^2(\mathbb{R}_+))$, $\frac{\d}{\d t}H_t\in\mathcal{C}_1(L^2(\mathbb{R}_+))$ and $\|H_t\|_1\rightarrow 0$ as $t\rightarrow+\infty$, see \cite{BB}. Hence, all results in Section \ref{seci21} and \ref{seci22} apply to the Gaussian kernel and so almost all workings in \cite{BB0,BB} boil down to the below Corollary.
\begin{cor}\label{Gaucor} The Fredholm determinant $F(t)$ of the Gaussian kernel \eqref{examp1} equals
\begin{equation*}
	F(t)=\exp\left[-\int_t^{\infty}(s-t)\big(q_0(s)\big)^2\,\d s\right],\ \ \ t\in\mathbb{R},
\end{equation*}
where
\begin{equation*}
	q_0(t)=\big((I-K_t)^{-1}\tau_t\phi\big)(0)=\lim_{\substack{z\rightarrow\infty\\ z\notin\mathbb{R}}}\big(zX^{21}(z;t)\big),\ \ t\in\mathbb{R},
\end{equation*}
is expressed in terms of the solution ${\bf X}(z)={\bf X}(z;t)=[X^{mn}(z;t)]_{m,n=1}^2\in\mathbb{C}^{2\times 2}$ of RHP \ref{master2} with
\begin{equation*}
	r_1(z)=-\im\e^{-\frac{1}{4}z^2},\ \ \ \ \ \ \ r_2(z)=\im\e^{-\frac{1}{4}z^2},\ \ z\in\mathbb{R},
\end{equation*}
in the same problem, see \eqref{s8}. Moreover, \eqref{i:19} matches \cite[$(1.9)$]{BB0} and \eqref{i:20} is \cite[$(1.13),(1.14)$]{BB} after taking a square root. Also, \eqref{i:2} is the special case of \eqref{i:20} when $\gamma=1$ and $p\equiv q_0$, again after taking a square root. Lastly, \eqref{i:25} reduces to \cite[$(1.20),(1.21)$]{BB} with the indicated choices for $\phi$ and $\psi$.
\end{cor}
\begin{proof} For every $z\in\mathbb{R}$,
\begin{equation*}
	\frac{1}{\sqrt{\pi}}\int_{-\infty}^{\infty}\e^{-y^2}\e^{-izy}\,\d y=\e^{-\frac{1}{4}z^2},
\end{equation*}
so RHP \ref{master2} with the indicated choices for $r_j(z)$ matches exactly \cite[Riemann-Hilbert Problem $1.5$]{BB0}, modulo the labelling $2x\mapsto t$ in \cite[$(1.8)$]{BB0}. Consequently, the above formula for $F(t)$ is \cite[$(3.33)$]{BB0}, compare \cite[Proposition $3.10$]{BB0}.
\end{proof}
\subsubsection{The Airy kernel} Let $J=\mathbb{R}$ and choose
\begin{equation*}
	\phi(x)=\psi(x)=\textnormal{Ai}(x),\ \ \ x\in\mathbb{R}.
\end{equation*}
The Airy function is a Schwartz function on $\mathbb{R}_+$, so all dominance and integrability assumptions at $+\infty$ in Section \ref{seci21} and \ref{seci22} are satisfied. Moreover, the symmetric integral operator $K_t$ in \eqref{n1} with kernel
\begin{equation}\label{examp3}
	K_t(x,y)=\int_0^{\infty}\textnormal{Ai}(x+z+t)\textnormal{Ai}(z+y+t)\,\d z,
\end{equation}
corresponding to the Airy choice, is trace class on $L^2(\mathbb{R}_+)$ and $I-\gamma K_t$ invertible on the same space for all $(t,\gamma)\in\mathbb{R}\times[0,1]$, see for instance \cite[Lemma $6.15$]{BDS}. But the Airy function is not in $W^{1,1}(\mathbb{R})$, hence we cannot use Theorem \ref{theo1}. However, relying on the below ad hoc arguments the outcome of Theorem \ref{theo1} for the Airy kernel determinant is still valid, see Corollary \ref{AiCor}. On the other hand, Theorem \ref{theo2a} on the asymptotic behavior of $F(t)$ as $t\rightarrow-\infty$ famously fails for the Airy kernel, compare \cite[$(1.19)$]{TW}.
\begin{cor}\label{AiCor} The Fredholm determinant $F(t)$ of the Airy kernel \eqref{examp3} equals
\begin{equation*}
	F(t)=\exp\left[-\int_t^{\infty}(s-t)\big(q_0(s)\big)^2\,\d s\right],\ \ \ t\in\mathbb{R},
\end{equation*}
where
\begin{equation*}
	q_0(t)=\big((I-K_t)^{-1}\tau_t\phi\big)(0)=\lim_{\substack{z\rightarrow\infty\\ z\notin\mathbb{R}}}\big(zX^{21}(z;t)\big),\ \ t\in\mathbb{R}
\end{equation*}
is expressed in terms of the solution ${\bf X}(z)={\bf X}(z;t)=[X^{mn}(z;t)]_{m,n=1}^2\in\mathbb{C}^{2\times 2}$ of RHP \ref{master2} with
\begin{equation*}
	r_1(z)=-\im\e^{\frac{\im}{3}z^3},\ \ \ r_2(z)=\im\e^{-\frac{\im}{3}z^3},\ \ \ z\in\mathbb{R},
\end{equation*}
in the same problem, see \eqref{s8}, and with the modification that \eqref{s10} is valid only in certain directions, with ${\bf X}(z)$ having at worst polynomial growth at $z=\infty$. Moreover, $q_0(t)$ solves the Painlev\'e-II boundary value problem
\begin{equation}\label{co0}
	\frac{\d^2q_0}{\d t^2}=tq_0+2q_0^3,\ \ \ q_0(t)\sim\textnormal{Ai}(t)\ \textnormal{as}\ t\rightarrow+\infty,
\end{equation}
identity \eqref{i:19} matches exactly \cite[$(9.150)$]{F1}, \eqref{i:20} is \cite[$(2.1)$]{Di} and \eqref{i:21} is \cite[$(2.2)$]{Di}. In addition, \eqref{i:1} is the special case of \eqref{i:20} with $\gamma=1$ and $q\equiv q_0$, after taking a square root.
\end{cor}
\begin{proof} Given the contour integral representation \cite[$9.5.1$]{NIST} for the Airy function, the formul\ae\,for $r_j(z)$ can be obtained via Fourier inversion. Alternatively, and based on first principles, we derive the following integral identity
\begin{equation}\label{co1}
	\int_{-\infty}^{\infty}\textnormal{Ai}(y)\e^{-\im zy}\,\d y=\e^{\frac{\im}{3}z^3},\ \ z\in\mathbb{R}.
\end{equation}
Indeed, for any $z\in\mathbb{R}$,
\begin{align}
	&\int_{-\infty}^{\infty}\!\e^{-\im zy}\textnormal{Ai}(y)\,\d y=\int_0^{\infty}\!\e^{-\im zy}\textnormal{Ai}(y)\,\d y+\int_0^{\infty}\!\e^{\im zy}\textnormal{Ai}(-y)\,\d y
	=\int_0^{\infty}\e^{-\im zy}\textnormal{Ai}(y)\,\d y
	+\int_0^{\infty}\e^{\im zy}\e^{\im\frac{\pi}{3}}\textnormal{Ai}\big(y\e^{\im\frac{\pi}{3}}\big)\,\d y\nonumber\\
	&+\int_0^{\infty}\e^{\im zy}\e^{-\im\frac{\pi}{3}}\textnormal{Ai}\big(y\e^{-\im\frac{\pi}{3}}\big)\,\d y
	=\int_0^{\infty}\e^{-\im zy}\textnormal{Ai}(y)\,\d y+\int_0^{\infty\,\e^{\im\frac{\pi}{3}}}\e^{\omega zy}\textnormal{Ai}(y)\,\d y+\int_0^{\infty\,\e^{-\im\frac{\pi}{3}}}\e^{-\bar{\omega}zy}\textnormal{Ai}(y)\,\d y\label{ex:2}
\end{align}
with $\omega:=\e^{\im\frac{\pi}{6}}$ and where we used the cyclic constraint of the Airy function \cite[$9.2.14$]{NIST}. But from \cite[$9.7.5$]{NIST},
\begin{equation*}
	\textnormal{Ai}(z)=\mathcal{O}\left(z^{-\frac{1}{4}}\exp\left[-\frac{2}{3}z^{\frac{3}{2}}\right]\right),\ \ \ \ \ |\textnormal{arg}\,z|\leq\frac{\pi}{3},
\end{equation*}
so we can rotate both integration paths in \eqref{ex:2} back to the half ray $\mathbb{R}_+$ by Cauchy's theorem,
\begin{equation*}
	\int_{-\infty}^{\infty}\e^{-\im zy}\textnormal{Ai}(y)\,\d y=\int_0^{\infty}\e^{-\im zy}\textnormal{Ai}(y)\,\d y+\int_0^{\infty}\e^{\omega zy}\textnormal{Ai}(y)\,\d y+\int_0^{\infty}\e^{-\bar{\omega}zy}\textnormal{Ai}(y)\,\d y.
\end{equation*}
Each of the remaining three integrals can be evaluated in terms of confluent hypergeometric functions, see \cite[$9.10.14$]{NIST}, and their sum immediately yields the claimed result \eqref{co1}. Next, moving to RHP \ref{master2} with the aforementioned choices for $r_j(z)$ in \eqref{s8}, we replace \eqref{s10} by the weaker requirement that $X(z)$ has at worst polynomial growth at $z\rightarrow\infty$ but $X(z)\rightarrow \mathbb{I}$ as $z\rightarrow\infty$ in certain directions. This is necessary due to the fact that $r_j(z)\nrightarrow 0$ as $z\rightarrow\pm\infty$ on $\mathbb{R}$. Still, the so modified RHP \ref{master2} admits at most one solution, see \cite[Proposition $3.1$]{FZ}, and the problem itself relates to Painlev\'e-II by the following simple argument: factorize the relevant jump matrix \eqref{s8},
\begin{align*}
	\begin{bmatrix}
	1-r_1(z)r_2(z)&-r_2(z)\e^{-\im tz}\smallskip\\
	r_1(z)\e^{\im tz} & 1\end{bmatrix}=\begin{bmatrix}1 & -r_2(z)\e^{-\im tz}\smallskip\\
	0 & 1\end{bmatrix}&\,\begin{bmatrix}1 & 0\\
	r_1(z)\e^{\im tz} & 1\end{bmatrix}\\
	&\,\equiv\begin{bmatrix}1 & -\im\e^{-\im(\frac{1}{3}z^3+tz)}\smallskip\\
	0 & 1\end{bmatrix}\begin{bmatrix}1 & 0\smallskip\\
	-\im\e^{\im(\frac{1}{3}z^3+tz)} & 1\end{bmatrix},
\end{align*}
and define
\begin{equation}\label{hi1}
	{\bf Y}(z;t):={\bf X}(2z;t)\begin{cases}\begin{bmatrix}1 & 0\\
	\im\e^{2\im(\frac{4}{3}z^3+tz)} & 1\end{bmatrix},&z\in\Omega_+\smallskip\\
	\begin{bmatrix}1 & -\im\e^{-2\im(\frac{4}{3}z^3+tz)}\smallskip\\
	0 & 1\end{bmatrix},&z\in\Omega_-\smallskip\\
	\mathbb{I},&\textnormal{else}\end{cases}
\end{equation}
with the domains shown in Figure \ref{fig1}. 
\begin{figure}[tbh]
\begin{tikzpicture}[xscale=0.7,yscale=0.7]
\draw [->] (-6,0) -- (6,0) node[below]{{\small $\Re z$}};
\draw [->] (0,-3) -- (0,3) node[left]{{\small $\Im z$}};
\draw [very thin, dashed, color=darkgray,-] (0,0) -- (4.33012701892,2.5) node[right]{$\frac{\pi}{6}$};
\draw [very thin, dashed, color=darkgray,-] (0,0) -- (-4.33012701892,2.5) node[left]{$\frac{5\pi}{6}$};
\draw [very thin, dashed, color=darkgray,-] (0,0) -- (4.33012701892,-2.5) node[right]{$\frac{11\pi}{6}$};
\draw [very thin, dashed, color=darkgray,-] (0,0) -- (-4.33012701892,-2.5) node[left]{$\frac{7\pi}{6}$};
\draw [fill=red, dashed,opacity=0.4] (0,0) -- (3.46410161514,2) arc (30:0:4cm) -- (0,0);
\draw [fill=red, dashed,opacity=0.4] (0,0) -- (-3.46410161514,2) arc (150:180:4cm) -- (0,0);
\draw [fill=blue, dashed,opacity=0.4] (0,0) -- (3.46410161514,-2) arc (-30:0:4cm) -- (0,0);
\draw [fill=blue, dashed,opacity=0.4] (0,0) -- (-3.46410161514,-2) arc (210:180:4cm) -- (0,0);

\node [red] at (3.1,2.3) {{\small$\Gamma_1$}};
\node [red] at (-3.1,2.3) {{\small$\Gamma_3$}};
\node [red] at (-3.1,-2.3) {{\small$\Gamma_4$}};
\node [red] at (3.1,-2.3) {{\small$\Gamma_6$}};
\node [red,opacity=0.4] at (4.6,0.75) {{\small$\Omega_+$}};
\node [red,opacity=0.4] at (-4.6,0.75) {{\small$\Omega_+$}};
\node [blue,opacity=0.4] at (4.6,-0.75) {{\small$\Omega_-$}};
\node [blue,opacity=0.4] at (-4.6,-0.75) {{\small$\Omega_-$}};
\draw [thick, color=red, decoration={markings, mark=at position 0.5 with {\arrow{>}}}, postaction={decorate}] (0,0) -- (3.89711431703,2.25);
\draw [thick, color=red, decoration={markings, mark=at position 0.5 with {\arrow{>}}}, postaction={decorate}] (0,0) -- (3.89711431703,-2.25);
\draw [thick, color=red, decoration={markings, mark=at position 0.5 with {\arrow{>}}}, postaction={decorate}] (0,0) -- (-3.89711431703,-2.25);
\draw [thick, color=red, decoration={markings, mark=at position 0.5 with {\arrow{>}}}, postaction={decorate}] (0,0) -- (-3.89711431703,2.25);

\end{tikzpicture}
\caption{The oriented jump contour $\Sigma_{\bf Y}:=\Gamma_1\cup\Gamma_3\cup\Gamma_4\cup\Gamma_6$ in RHP \ref{PII}.}
\label{fig1}
\end{figure}
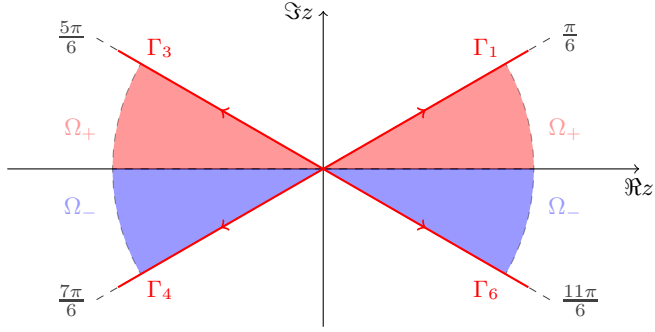
The analytic and asymptotic properties of ${\bf Y}(z)$ are summarized below.
\begin{problem}\label{PII} For every $t\in\mathbb{R}$, the $2\times2$ matrix-valued function ${\bf Y}(z)={\bf Y}(z;t)$ defined in \eqref{hi1} is such that
\begin{enumerate}
	\item[(1)] ${\bf Y}(z)$ is analytic for $z\in\mathbb{C}\setminus\Sigma_{\bf Y}$ and extends continuously from all sides to the oriented jump contour $\Sigma_{\bf Y}=\Gamma_1\cup\Gamma_3\cup\Gamma_4\cup\Gamma_6$ shown in Figure \ref{fig1}.
	\item[(2)] On $\Sigma_{\bf Y}$ the jump behavior of ${\bf Y}(z)$ reads as
	\begin{equation*}
		{\bf Y}_+(z)={\bf Y}_-(z)\e^{-\theta(z,t)\sigma_3}{\bf S}_k\e^{\theta(z,t)\sigma_3},\ \ \ z\in\Gamma_k,\ \ k=1,3,4,6
	\end{equation*}
	with $\theta(z,t):=\im(\frac{4}{3}z^3+tz)$ and the matrices
	\begin{equation*}
		{\bf S}_k:=\begin{bmatrix}1 & 0\\
		s_k & 1\end{bmatrix}, \ \ k=1,3;\hspace{1cm}{\bf S}_k:=\begin{bmatrix}1 & s_k\\ 0 & 1\end{bmatrix},\ \ k=4,6
	\end{equation*}
	defined with the multipliers $s_1:=-\im,s_3:=\im,s_4:=-s_1,s_6:=-s_3$.
	\item[(3)] As $z\rightarrow\infty$ with $z\notin\Sigma_{\bf Y}\cup\mathbb{R}$,
	\begin{equation*}
		{\bf Y}(z)=\mathbb{I}+o(1).
	\end{equation*}
\end{enumerate} 
\end{problem}
Observe that RHP \ref{PII} is the standard RHP used in the analysis of the homogeneous Painlev\'e-II equation, cf. \cite[Theorem $3.4$]{FIKN}, and by \cite[$(4.2.28),(4.2.29),(4.2.40),(4.2.42)$, Theorem $10.1$]{FIKN} we thus conclude that
\begin{equation}\label{ex:3}
	q(t):=2Y_1^{12}(t):=2\lim_{\substack{z\rightarrow\infty\\ z\notin\Sigma_{\bf Y}\cup\mathbb{R}}}\big(zY^{12}(z)\big)
\end{equation}
solves the boundary value problem \eqref{co0}. It now remains to establish the second equality in
\begin{equation}\label{ex:4}
	q_0(t)\stackrel{\eqref{n8}}{=}((I-K_t)^{-1}\tau_t\phi)(0)=q(t),\ \ \ \ t\in\mathbb{R},
\end{equation}
i.e. the relation between the Fredholm determinant $F(t)$ of the Airy kernel \eqref{examp3} and RHP \ref{master2} with coefficients $r_j(z)$ as above and the aforementioned modification at $z=\infty$ (we cannot use Theorem \ref{theo1} because of the modification and since the Airy function is \textit{not} in $W^{1,1}(\mathbb{R})$). Fortunately for us, the second equality in \eqref{ex:4} was derived in \cite[Section $3$]{BerCa0}: our RHP \ref{PII} is the AiO-RHP in \cite[Definition $3.1$]{BerCa0} rotated by ninety degrees, thus the outstanding claim follows from \cite[Theorem $3.1$]{BerCa0}.
\end{proof}
The last Corollary concludes the content of this short subsection and our list of results for \textit{additive} Hankel composition operators on $L^2(\mathbb{R}_+)$. We now discuss \textit{multiplicative} Hankel composition operators on $L^2(0,1)$, reusing the symbols and abbreviations of Section \ref{seci21} and \ref{seci22}. The reader will notice that most of our results for multiplicative composition operators bear a striking resemblance to the corresponding results for additive composition operators.
\subsection{Multiplicative composition, part 1}\label{seci24}
Let $t\in J$ be chosen from some open subset $J\subseteq\mathbb{R}_+$ and consider two integral operators $M_t,N_t\in\mathcal{C}_2(L^2(0,1))$ of \textit{multiplicative} Hankel type. Precisely,
\begin{equation*}
	(M_tf)(x):=\sqrt{t}\int_0^1\phi(xyt)f(y)\,\d y,\ \ \ \ \ \ (N_tf)(x):=\sqrt{t}\int_0^1\psi(xyt)f(y)\,\d y,
\end{equation*}
where $\phi,\psi:\mathbb{R}_+\rightarrow\mathbb{C}$ are such that for all $t\in J$,
\begin{equation}\label{b:1}
	\int_0^1|\phi(xt)|^2\ln\Big(\frac{1}{x}\Big)\,\d x<\infty\ \ \ \ \textnormal{and}\ \ \ \ \int_0^1|\psi(xt)|^2\ln\Big(\frac{1}{x}\Big)\,\d x<\infty.
\end{equation}
Now let $K_t\in\mathcal{L}(L^2(0,1))$ denote the composition $K_t:=M_tN_t$ on $L^2(0,1)$, equivalently,
\begin{equation}\label{b:2}
	(K_tf)(x)=\int_0^1K_t(x,y)f(y)\,\d y,\ \ \ \ \ \ K_t(x,y):=t\int_0^1\phi(xzt)\psi(zyt)\,\d z.
\end{equation}
Just as in Section \ref{seci21}, $M_t$ and $N_t$ are compact symmetric operators on $L^2(0,1)$, however their composition \eqref{b:2} is in general not symmetric. We now state the analogue of Definition \ref{domplus}
\begin{definition}
A family of functions, $\mathcal{F}:=\{f_{\alpha}\}_{\alpha\in S}$, in $L^p(0,1),p\in\{1,2\}$, indexed by a set $S\subseteq\mathbb{R}_+$, is called $L^p(0,1)$ dominated if for some $g\in L^p(0,1)$, we have for all $\alpha\in S$ and for a.e. $x\in(0,1)$ that
\begin{equation*}
	|f_{\alpha}(x)|\leq g(x).
\end{equation*}
\end{definition}
Next, we use the dilation of $f$ by $t$, $(\tau_tf)(x):=f(xt)$, and let $F(t)$ denote the Fredholm determinant of $K_t\in\mathcal{C}_1(L^2(0,1))$,
\begin{equation}\label{b:3}
	F(t):=\prod_{k=1}^{\infty}\big(1-\lambda_k(t)\big),\ \ \ t\in J,
\end{equation}
with $\lambda_k(t)$ as the nonzero eigenvalues of $K_t$, counting multiplicity. The analogue of Lemma \ref{lem1} for \eqref{b:3},\eqref{b:2} reads as follows.
\begin{lem}\label{lem2} Suppose $\phi,\psi:\mathbb{R}_+\rightarrow\mathbb{C}$, besides satisfying \eqref{b:1}, are continuously differentiable on $\mathbb{R}_+$, obey
\begin{equation}\label{b:4}
	\lim_{x\downarrow 0}\sqrt{x}\,\phi(x)=\lim_{x\downarrow 0}\sqrt{x}\,\psi(x)=0
\end{equation}
and for every $t\in J$,
\begin{equation}\label{b:5}
	\int_0^1|(MD\tau_t\phi)(x)|^2\ln\Big(\frac{1}{x}\Big)\,\d x<\infty,\ \ \ \ \int_0^1|(MD\tau_t\psi)(x)|^2\ln\Big(\frac{1}{x}\Big)\,\d x<\infty.
\end{equation}
Then we have, provided $I-K_t,t\in J$ is invertible on $L^2(0,1)$ and the families $\{\tau_tf\}_{t\in\mathbb{R}_+},\{MD\tau_tf\}_{t\in\mathbb{R}_+}$, resp. $\{(I-K_t)^{-1}\tau_tg\}_{t\in J},\{(I-K_t)^{-1}MD\tau_tg\}_{t\in J}$, with $f\in\{\phi,\psi\}$, resp. $g\in\{\phi\}$, are $L^2(0,1)$ dominated,
\begin{equation}\label{b:6}
	t\frac{\d}{\d t}\left\{t\frac{\d}{\d t}\ln F(t)\right\}=-q_0(t)q_0^{\ast}(t),\ \ \ t\in J,\ \ \ \ \ \ \ \ \begin{cases}q_0(t):=\big((I-K_t)^{-1}\tau_t\phi\big)(1)&\\ q_0^{\ast}(t):=t\big((I-K_t^{\ast})^{-1}\tau_t\psi\big)(1)&\end{cases}.
\end{equation}
Here, $K_t^{\ast}:=N_tM_t$ denotes the real adjoint of $K_t$.
\end{lem}
Identity \eqref{b:6} is implicit in the works of Tracy and Widom \cite[$(2.26)$]{TW2} on the Bessel kernel where $\phi=\psi$ are Bessel functions with square root variables. But again, \cite{TW2} relies on the integrable structure of the same kernel, whereas we only exploit the kernel's algebraic structure \eqref{b:2} and so \eqref{b:6} holds for \textit{any} multiplicative Hankel composition operator in the class of Lemma \ref{lem2}. Just as in \eqref{r1}, we notice that the second modified $(D\mapsto MD)$ logarithmic derivative of \eqref{b:3} localizes again into a product of two functions evaluated at $t\in J$.\bigskip

Once $J=\mathbb{R}_+$ in Lemma \ref{lem2} and $t\mapsto q_0(t),t\mapsto q_0^{\ast}(t)$ are integrable at $t=0$, formula \eqref{b:6} implies a Tracy-Widom representation of $F(t)$, see \cite[$(1.19)$]{TW} for the original example of such a formula. The very same formula explains the structure of the first factor in \eqref{i:8} and in other hard-edge ensembles.
\begin{cor}\label{corb1} Let $\epsilon>0$. Assume that $\phi,\psi:\mathbb{R}_+\rightarrow\mathbb{C}$, besides satisfying \eqref{b:1}, are continuously differentiable on $\mathbb{R}_+$, obey \eqref{b:4},\eqref{b:5} for $t\in\mathbb{R}_+$ and the families $\{\sqrt{-\ln(\cdot)}\,\tau_tf\}_{t\in\mathbb{R}_+}$, resp. $\{\tau_tg\}_{t\in\mathbb{R}_+},\{MD\tau_tg\}_{t\in\mathbb{R}_+}$ and $\{(I-K_t)^{-1}\tau_th\}_{t\in\mathbb{R}_+},\{(I-K_t)^{-1}MD\tau_th\}_{t\in\mathbb{R}_+}$, with $f\in\{\phi,\psi\}$, resp. $g\in\{\phi,\psi\}$ and $h\in\{\phi\}$, are $L^2(0,1)$ dominated. Then
\begin{equation}\label{b:7}
	\ln F(t)=-\int_0^t\ln\Big(\frac{t}{s}\Big)q_0(s)q_0^{\ast}(s)\frac{\d s}{s},\ \ \ \ t\in\mathbb{R}_+,
\end{equation}
provided $I-K_t$ is invertible on $L^2(0,1)$ for all $t\in\mathbb{R}_+$ and provided there exist $c,t_0>0$ such that
\begin{equation*}
	|q_0(t)q_0^{\ast}(t)|\leq ct^{\epsilon}\ \ \ \ \textnormal{for all}\ 0<t\leq t_0^{-1}.
\end{equation*}
\end{cor}
Similarly to our approach in Section \ref{seci21}, we now aim to characterize $q_0$ and $q_0^{\ast}$ beyond the formul\ae\,stated in \eqref{b:6}. We begin with the following algebraic result, which the reader should contrast to Corollary \ref{highercor}.
\begin{cor}\label{corb2} Let $N\in\mathbb{Z}_{\geq 1}$. Assume $\phi,\psi:\mathbb{R}_+\rightarrow\mathbb{C}$, besides satisfying \eqref{b:1}, are $N$-times continuously differentiable such that for every $k\in\{0,1,\ldots,N-1\}$,
\begin{equation}\label{b:8}
	\lim_{x\downarrow 0}\sqrt{x}\big((MD)^k\phi\big)(x)=\lim_{x\downarrow 0}\sqrt{x}\big((DM)^k\psi\big)(x)=0,
\end{equation}
and for every $t\in J$,
\begin{equation*}
	\int_0^1|(MD\tau_t\phi)(x)|^2\ln\Big(\frac{1}{x}\Big)\,\d x<\infty,\ \ \ \ \ \ \int_0^1|(MD\tau_t\psi)(x)|^2\ln\Big(\frac{1}{x}\Big)\,\d x<\infty.
\end{equation*}
Then, provided $I-K_t,t\in J$ is invertible on $L^2(0,1)$ and $\{(MD)^n\tau_tf\}_{t\in\mathbb{R}_+},\{(I-K_t)^{-1}(MD)^n\tau_tf\}_{t\in J}$, resp. $\{(DM)^n\tau_tg\}_{t\in\mathbb{R}_+},\{(I-K_t^{\ast})^{-1}(DM)^n\tau_tg\}_{t\in J}$, with $f\in\{\phi\}$, resp. $g\in\{\psi\}$, are $L^2(0,1)$ dominated for all $n\in\{0,1,\ldots,N\}$, the functions
\begin{equation}\label{b:9}
	q_n(t):=\big((I-K_t)^{-1}(MD)^n\tau_t\phi\big)(1),\ \ \ \ \ \ \ \ \ \,p_n(t):=t\tr_{L^2(0,1)}\big((I-K_t)^{-1}(MD)^n\tau_t\phi\otimes\tau_t\psi\big),
\end{equation}
\begin{equation}\label{b:10}
	\,\ \ \ \ \,\,q_n^{\ast}(t):=t\big((I-K_t^{\ast})^{-1}(DM)^n\tau_t\psi\big)(1),\ \ \ \ \ \ \ \ p_n^{\ast}(t):=t\tr_{L^2(0,1)}\big((I-K_t^{\ast})^{-1}(DM)^n\tau_t\psi\otimes\tau_t\phi\big),
\end{equation}
defined for $t\in J$ and $n=0,1,\ldots,N$, satisfy the coupled system
\begin{equation}\label{b:11}
	\begin{cases}\displaystyle t\frac{\d q_n}{\d t}(t)=q_{n+1}(t)+q_0(t)p_n(t),\ \ \ \ \ \ \ t\frac{\d p_n}{\d t}(t)=q_0^{\ast}(t)q_n(t),&\bigskip\\
	\displaystyle t\frac{\d q_n^{\ast}}{\d t}(t)=q_{n+1}^{\ast}(t)+q_0^{\ast}(t)p_n^{\ast}(t),\ \ \ \ \ \ \ t\frac{\d p_n^{\ast}}{\d t}(t)=q_0(t)q_n^{\ast}(t),
	\end{cases}\ \ \ \  n=0,1,\ldots,N-1,\ \ \ \ t\in J.
\end{equation}
\end{cor}
System \eqref{b:11} also constitutes a chain of differential relations 
with $t\frac{\d}{\d t}$ replacing $\frac{\d}{\d t}$ in \eqref{r2}, due to the multiplicative structure in \eqref{b:2}. The same system is seemingly unknown in the literature and we point out that it admits several conserved quantities. Indeed, exactly as in \eqref{i:13},
\begin{equation*}
	I_n:=p_{n+1}(t)+(-1)^np_{n+1}^{\ast}(t)-\sum_{k=0}^n(-1)^k(q_k^{\ast}(t)q_{n-k}(t)-p_k^{\ast}(t)p_{n-k}(t))
\end{equation*}	
is $J\ni t$-independent for any $n\in\{0,1,\ldots,N-1\}$ by \eqref{b:11}. However, without placing further constraints on $\phi$ and $\psi$, \eqref{b:11} cannot yield a closed differential equation for $q_0$ or $q_0^{\ast}$. Still, just as it had been the case for \eqref{r2}, we now show that system \eqref{b:11} can be encoded in the solution of a canonical, auxiliary RHP, and thus using \eqref{b:6}, the same RHP allows us to access and efficiently analyze the Fredholm determinant $F(t)$. The details are as follows: consider the below problem, formulated for two arbitrary functions $\phi,\psi\in L_{\circ}^1(\mathbb{R}_+)$.
\begin{problem}\label{masterb} Fix $t\in\mathbb{R}_+$ and $\phi,\psi\in L_{\circ}^1(\mathbb{R}_+)$. Now determine ${\bf X}(z)={\bf X}(z;t,\phi,\psi)\in\mathbb{C}^{2\times 2}$ such that
\begin{enumerate}
	\item[(1)] ${\bf X}(z)$ is analytic for $z\in\mathbb{C}\setminus(\frac{1}{2}+\im\mathbb{R})$.
	\item[(2)] ${\bf X}(z)$ admits continuous pointwise limits ${\bf X}_{\pm}(z):=\lim_{\epsilon\downarrow 0}{\bf X}(z\mp\epsilon),z\in\frac{1}{2}+\im\mathbb{R}$ which obey
	\begin{equation}\label{b:12}
		{\bf X}_+(z)={\bf X}_-(z)\begin{bmatrix}1-r_1(z)r_2(z) & -r_2(z)t^z\smallskip\\ r_1(z)t^{-z} & 1\end{bmatrix},\ \ z\in\frac{1}{2}+\im\mathbb{R},
	\end{equation}
	with
	\begin{equation}\label{b:13}
		r_1(z):=\int_0^{\infty}y^{z-1}\phi(y)\,\d y,\ \ \ \ \ \ \ r_2(z):=\int_0^{\infty}y^{-z}\psi(y)\,\d y,\ \ \ z\in\frac{1}{2}+\im\mathbb{R}.
	\end{equation}
	\item[(3)] As $z\rightarrow\infty$,
	\begin{equation}\label{b:14}
		{\bf X}(z)=\mathbb{I}+o(1).
	\end{equation}
\end{enumerate}
\end{problem}
We establish in Lemma \ref{unique3} that conditions $(1)-(3)$ in RHP \ref{masterb} determine ${\bf X}(z)$ uniquely. Moreover, subject to some additional assumptions placed on $\phi$ and $\psi$, the same RHP is solvable and its solution connects to $q_0$ and $q_0^{\ast}$ in the following way:
\begin{theo}\label{theo3} Suppose $\phi,\psi\in H_{\circ}^{1,1}(\mathbb{R}_+)$, besides obeying \eqref{b:1}, are continuously differentiable on $\mathbb{R}_+$ with $\sqrt{(\cdot)}\phi,\sqrt{(\cdot)}\psi\in L^{\infty}(\mathbb{R}_+)$ and $\sqrt{(\cdot)}MD\phi,\sqrt{(\cdot)}MD\psi\in L^{\infty}(0,1)$, satisfy
\begin{equation}\label{b:15}
	\int_0^1|(MD\tau_tf)(x)|^2\ln\Big(\frac{1}{x}\Big)\,\d x<\infty,\ \ \ \ \int_0^1\sqrt{\int_0^1|(\tau_tf)(xy)|^2\,\d y}\,\frac{\d x}{\sqrt{x}}<\infty,\ \ \ \ f\in\{\phi,\psi\},
\end{equation}
and $I-K_t$ is invertible on $L^2(0,1)$. Then RHP \ref{masterb} is uniquely solvable and its solution satisfies
\begin{equation}\label{b:16}
	\lim_{\substack{z\rightarrow\infty\\ \Re z\not\equiv\textnormal{const.}}}z\big({\bf X}(z)-\mathbb{I}\big)=\begin{bmatrix}p_0 & q_0^{\ast}\\-q_0 & -p_0^{\ast}\end{bmatrix} =:{\bf X}_1={\bf X}_1(t,\phi,\psi)=\big[X_1^{mn}(t,\phi,\psi)\big]_{m,n=1}^2,
\end{equation}
with $\{q_0,q_0^{\ast},p_0=p_0^{\ast}\}$ as in \eqref{b:9} and \eqref{b:10}.
\end{theo}
Our proof of Theorem \ref{theo3} is constructive, see Section \ref{sec7}, and there is no equivalent of it in the literature, to the best of our knowledge.
\begin{rem}\label{IIKScon2} As done in Remark \ref{IIKScon1} we compare \eqref{b:12} to the jump constraint in the RHP for an integral operator: if either $z\mapsto r_1(z)$ in \eqref{b:13} admits analytic continuation to $H_+:=\{z\in\mathbb{C}: \frac{1}{2}-\epsilon<\Re z<\frac{1}{2}\}$ such that $|r_1(z)|\rightarrow 0$ as $|z|\rightarrow\infty$ in $H_+$, or if $z\mapsto r_2(z)$ admits analytic continuation to $H_-:=\{z\in\mathbb{C}:\,\frac{1}{2}<\Re z<\frac{1}{2}+\epsilon\}$ such that $|r_2(z)|\rightarrow 0$ as $|z|\rightarrow\infty$ in $H_-$, then
\begin{equation*}
	\begin{bmatrix} 1-r_1(z)r_2(z) & -r_2(z)t^z\\ r_1(z)t^{-z} & 1\end{bmatrix}=\begin{bmatrix}1 & -r_2(z)t^z\\ 0 & 1\end{bmatrix}\begin{bmatrix}1 & 0\\ r_1(z)t^{-z} & 1\end{bmatrix},\ \ \ z\in\frac{1}{2}+\im\mathbb{R},
\end{equation*}
readily shows that RHP \ref{masterb} can be transformed to a RHP for an integrable integral operator by moving one of the matrix factors off the critical line. Still, without explicit decay assumptions placed on $\phi,\psi$ their Mellin transforms \eqref{b:13} won't admit such extensions. Thus, in general, there is again no obvious relation between RHP \ref{masterb} and the RHP for an integrable operator, although some types of \eqref{b:13} might be approximable by suitable analytic functions, compare the pointers in Remark \ref{IIKScon1}.
\end{rem}
Combining Lemma \ref{lem2} with Theorem \ref{theo3}, we arrive at the following 
Riemann-Hilbert characterization of the Fredholm determinant \eqref{b:3}.
\begin{cor}\label{impcor2} Under the assumptions of Lemma \ref{lem2} and Theorem \ref{theo3}, with \eqref{b:15} valid for all $t\in J$ and $I-K_t$ invertible on $L^2(0,1)$ for the same $t$,
\begin{equation*}
	t\frac{\d}{\d t}\ln F(t)=-X_1^{11}(t,\phi,\psi),\ \ \ \ \ t\frac{\d}{\d t}\left\{t\frac{\d}{\d t}\ln F(t)\right\}=X_1^{12}(t,\phi,\psi)X_1^{21}(t,\phi,\psi),\ \ \ t\in J,
\end{equation*}
in terms of the entries of the matrix coefficient ${\bf X}_1$ in \eqref{b:16}.
\end{cor}
What's more, RHP \ref{masterb} also lets us access the higher functions \eqref{b:9},\eqref{b:10} for $n\in\mathbb{Z}_{\geq 1}$, provided we impose additional restrictions on $\phi$ and $\psi$, somewhat reminiscent of those in Corollary \ref{corb2}.
\begin{cor}\label{deeper2} Fix $N\in\mathbb{Z}_{\geq 1}$. Provided $\phi,\psi\in H_{\circ}^{N,1}(\mathbb{R}_+)$, besides obeying \eqref{b:1}, are $N$-times continuously differentiable on $\mathbb{R}_+$, satisfy for every $k\in\{0,1,\ldots,N-1\}$ and $f\in\{\phi,\psi\}$,
\begin{equation*}
	\sqrt{(\cdot)}(MD)^kf\in L^{\infty}(\mathbb{R}_+),\ \ \int_0^1|(MD\tau_t f)(x)|^2\ln\Big(\frac{1}{x}\Big)\,\d x<\infty,\ \ \int_0^1\sqrt{\int_0^1|(\tau_tf)(xy)|^2\,\d y}\frac{\d x}{\sqrt{x}}<\infty,
\end{equation*}
as well as $\sqrt{(\cdot)}(MD)^Nf\in L^{\infty}(0,1),f\in\{\phi,\psi\}$ and $I-K_t$ is invertible on $L^2(0,1)$, then RHP \ref{masterb} is uniquely solvable and its solution satisfies
\begin{equation}\label{b:17}
	{\bf X}(z)=\mathbb{I}+\sum_{k=1}^N\begin{bmatrix}(-1)^{k-1}p_{k-1} & q_{k-1}^{\ast}\\ (-1)^kq_{k-1} & -p_{k-1}^{\ast}\end{bmatrix}z^{-k}+o\big(z^{-N}\big),\ \ \ \ |z|\rightarrow\infty,\ \ \ \Re z\not\equiv\textnormal{const.}
\end{equation}
with $\{q_k,p_k,q_k^{\ast},p_k^{\ast}\}$ as in \eqref{b:9} and \eqref{b:10}.
\end{cor}
The above concludes our first set of results for Fredholm determinants of multiplicative Hankel composition operators. We now continue to discuss some of their other properties, but within a narrower class of kernels.
\subsection{Multiplicative composition, part $2$}\label{seci25} As discussed earlier on, certain determinants of multiplicative Hankel composition operators appear in a perturbed fashion, especially in random matrix theory when concerned with extreme values in real or symplectic hard edge ensembles, compare \eqref{i:8}. We now show how some of those perturbed determinants can be computed in terms of the solution of RHP \ref{masterb}, thus perfectly mirroring the results of Section \ref{seci22} for additive compositions. Consider the Hankel integral operator $H_t\in\mathcal{L}(L^2(0,1))$ given by
\begin{equation}\label{b:18}
	(H_tf)(x):=\sqrt{t}\int_0^1\phi(xyt)f(y)\,\d y,\ \ \ t\in\mathbb{R}_+,
\end{equation}
where $\phi:\mathbb{R}_+\rightarrow\mathbb{C}$ is continuously differentiable.
\begin{assum}\label{ass2} Assume $H_t\in\mathcal{C}_1(L^2(0,1))$ and $t\frac{\d}{\d t}H_t\in\mathcal{C}_1(L^2(0,1))$ for all $t\in\mathbb{R}_+$, the latter with kernel $\frac{1}{2}\sqrt{t}(\tau_t\phi)(xy)+\sqrt{t}(MD\tau_t\phi)(xy)$, such that $\|H_t\|_1\rightarrow 0$ as $t\downarrow 0$, for the trace norm $\|H_t\|_1$ of $H_t$. Moreover assume $I-\gamma H_t^2$ is invertible for all $(t,\gamma)\in\mathbb{R}_+\times[0,1]$.
\end{assum}
If $K_t$ denotes the composition operator $K_t:=H_t^2$ on $L^2(0,1)$ with symmetric kernel
\begin{equation*}
	K_t(x,y)=t\int_0^1\phi(xzt)\phi(zyt)\,\d z,
\end{equation*}
then the Fredholm determinant $F(t,\gamma):=\prod_{k=1}^{\infty}(1-\gamma\lambda_k(t)),(t,\gamma)\in\mathbb{R}_+\times[0,1]$, where $\lambda_k(t)$ are the non-zero eigenvalues of $K_t$, satisfies by Corollary \ref{corb1}, subject to the therein listed assumptions,
\begin{equation}\label{b:19}
	\ln F(t,\gamma)=-\int_0^t\ln\Big(\frac{t}{s}\Big)\big(q(s,\gamma)\big)^2\,\d s
\end{equation}
with $q(t,\gamma):=\sqrt{\gamma}((I-K_t)^{-1}\tau_t\phi)(1)$. Now define the Fredholm determinant of $\gamma H_t$ on $L^2(0,1)$, i.e.
\begin{equation}\label{b:20}
	G(t,\gamma):=\prod_{k=1}^{\infty}\big(1-\gamma\mu_k(t)\big),\ \ \ (t,\gamma)\in\mathbb{R}_+\times[-1,1],
\end{equation}
where $\mu_k(t)$ are the non-zero eigenvalues of $H_t$, counting multiplicity, and consider the following three functions:
\begin{equation}\label{b:21}
	F_1^{[1]}(t,\gamma):=F(t,\gamma)\left\{1-\gamma\sqrt{t}\int_0^1\big((I-\gamma K_t)^{-1}\tau_t\phi\big)(x)\left[1-\sqrt{t}\int_0^x(\tau_t\phi)(y)\frac{\d y}{\sqrt{y}}\right]\frac{\d x}{\sqrt{x}}\right\},
\end{equation}
followed by, where $\gamma_{\circ}:=\gamma(2-\gamma)$,
\begin{equation}\label{b:22}
	F_1^{[2]}(t,\gamma):=F(t,\gamma_{\circ})\left\{1-\gamma\sqrt{t}\int_0^1\big((I-\gamma_{\circ} K_t)^{-1}\tau_t\phi\big)(x)\left[1-\sqrt{t}\int_0^x(\tau_t\phi)(y)\frac{\d y}{\sqrt{y}}\right]\frac{\d x}{\sqrt{x}}\right\},
\end{equation}
and concluding with
\begin{equation}\label{b:23}
	F_4(t,\gamma):=F(t,\gamma)\left\{1+\gamma\sqrt{t}\int_0^1\big((I-\gamma K_t)^{-1}\tau_t\phi\big)(x)\left[\frac{1}{2}\sqrt{t}\int_0^x(\tau_t\phi)(y)\frac{\d y}{\sqrt{y}}\right]\frac{\d x}{\sqrt{x}}\right\}.
\end{equation}
A concrete example of \eqref{b:21}, resp. \eqref{b:22}, appears in the hard edge analysis of the LOE, compare \cite[$(3.14)$]{F0}, resp. \cite[$(1.31)$]{F0}. Likewise, a special case of \eqref{b:23} occurs in the Laguerre symplectic ensemble LSE, see \cite[$(1.33)$]{F0}, again at the hard edge. Using the multiplicative Hankel structure only, we now evaluate all three functions \eqref{b:21},\eqref{b:22},\eqref{b:23} in terms of Riemann-Hilbert data, i.e. in terms of $q(t,\gamma)$, see \eqref{b:19} and Theorem \ref{theo3}, or, equivalently in terms of $G(t,\gamma)$ in \eqref{b:20}. The details are as follows and they explain in particular why the structure in the right hand side of \eqref{i:8} is common to hard edge ensembles.
\begin{theo}\label{theo4} Suppose $H_t\in\mathcal{C}_1(L^2(0,1))$ defined in \eqref{b:18} satisfies the conditions in Assumption \ref{ass2} with $\phi:\mathbb{R}_+\rightarrow\mathbb{C}$ continuously differentiable on $\mathbb{R}_+$. Assume further $\phi\in L_{\circ}^1(0,1)$, that
\begin{equation*}
	\lim_{x\downarrow 0}\sqrt{x}\,\phi(x)=0
\end{equation*}
and that for every $t\in\mathbb{R}_+$,
\begin{equation*}
		\int_0^1|(\tau_t\phi)|^2\ln\Big(\frac{1}{x}\Big)\,\d x<\infty,\ \ \ \int_0^1|(MD\tau_t\phi)(x)|^2\ln\Big(\frac{1}{x}\Big)\,\d x<\infty,\ \ \ \int_0^1\sqrt{\int_0^1|(\tau_t\phi)(xy)|^2\,\d y}\frac{\d x}{\sqrt{x}}<\infty,
\end{equation*}	
\begin{equation}\label{b:24}
	\int_0^1\sqrt{\int_0^1|(MD\tau_t\phi)(xy)|^2\,\d y}\frac{\d x}{\sqrt{x}}<\infty.
\end{equation}
Then we have for any $(t,\gamma)\in\mathbb{R}_+\times[0,1]$, denoting $\omega(t,\gamma):=\int_0^tq(s,\gamma)\frac{\d s}{\sqrt{s}}$,
\begin{align}
	F_1^{[1]}(t,\gamma)=F(t,\gamma)\Big\{\cosh\omega(t,\gamma)&\,-\sqrt{\gamma}\sinh\omega(t,\gamma)\Big\}\nonumber\\
	&\,=\frac{1}{2}(1-\sqrt{\gamma})\big(G(t,-\sqrt{\gamma})\big)^2+\frac{1}{2}(1+\sqrt{\gamma})\big(G(t,\sqrt{\gamma})\big)^2,\label{b:25}
\end{align}
followed by
\begin{align}
	F_1^{[2]}(t,\gamma)=F(t,\gamma_{\circ})&\,\left\{\frac{1-\gamma+\cosh\omega(t,\gamma_{\circ})-\sqrt{\gamma_{\circ}}\sinh\omega(t,\gamma)}{2-\gamma}\right\}\nonumber\\
	&\hspace{1.65cm}=\left[\sqrt{\frac{1-\sqrt{\gamma_{\circ}}}{2(2-\gamma)}}G(t,-\sqrt{\gamma_{\circ}})+\sqrt{\frac{1+\sqrt{\gamma_{\circ}}}{2(2-\gamma)}}G(t,\sqrt{\gamma_{\circ}})\right]^2,\label{b:26}
\end{align}
and concluding with
\begin{equation}\label{b:27}
	F_4(t,\gamma)=F(t,\gamma)\bigg\{\cosh\bigg(\frac{1}{2}\omega(t,\gamma)\bigg)\bigg\}^2=\left[\frac{1}{2}G(t,-\sqrt{\gamma})+\frac{1}{2}G(t,\sqrt{\gamma})\right]^2,
\end{equation}
provided $\{(\cdot)^{-1/2}(I-\gamma K_t)^{-1}\tau_t\phi\}_{t\in\mathbb{R}_+},\{(\cdot)^{-1/2}(I-\gamma K_t)^{-1}H_tMD\tau_t\phi\}_{t\in\mathbb{R}_+},\{(\cdot)^{-1/2}(I\pm\sqrt{\gamma}H_t)^{-1}\tau_t\phi\}_{t\in\mathbb{R}_+}$, $\{(\cdot)^{-1/2}(I\pm\sqrt{\gamma}H_t)^{-1}MD\tau_t\phi\}_{t\in\mathbb{R}_+}$ are $L^1(0,1)$ dominated and $\{\tau_t\phi\}_{t\in\mathbb{R}_+},\{MD\tau_t\phi\}_{t\in\mathbb{R}_+}$ are $L^2(0,1)$ dominated, and provided there exist $c,t_0>0$ so $|q(t,\gamma)|\leq ct^{-\frac{1}{2}+\epsilon}$ for all $0<t\leq t_0^{-1},\gamma\in[0,1]$ with $\epsilon>0$.
\end{theo}
\begin{rem} Although our workings on multiplicative Hankel operators focus on composition operators, the proof of Theorem \ref{theo4} uses the identity
\begin{equation*}
	2\ln G(t,\sqrt{\gamma})=\ln F(t,\gamma)-\omega(t,\gamma),\ \ \ \ \ (t,\gamma)\in\mathbb{R}_+\times[0,1],
\end{equation*}
see \eqref{bp41a}, so we could use our results for the study of ordinary multiplicative trace class Hankel operators  \eqref{b:18}, subject to the necessary assumptions placed on their kernels.
\end{rem}
\begin{rem} Just as in \eqref{i:20} and \eqref{i:21}, the functions $F_1^{[2]}$ and $F_4$ in \eqref{b:22} and \eqref{b:23} are again squares of convex combinations of the simpler determinants $G(t,\pm\sqrt{\gamma})$, see \eqref{b:26} and \eqref{b:27}.
\end{rem}
We now turn to the asymptotic analysis of $F(t)$ in \eqref{b:3} as $t\downarrow 0$ and $t\rightarrow+\infty$. First, with \eqref{i:22} in mind, we note that the trace norm of $K_t$ in \eqref{b:2} satisfies
\begin{equation*}
	\|K_t\|_1\leq\sqrt{t\int_0^1|(\tau_t\phi)(x)|^2\ln\Big(\frac{1}{x}\Big)\,\d x}\,\sqrt{t\int_0^1|(\tau_t\psi)(x)|^2\ln\Big(\frac{1}{x}\Big)\,\d x}\rightarrow 0\ \ \ \ \textnormal{as}\ \ t\downarrow 0
\end{equation*}
by \eqref{b:4}, provided the family $\{\sqrt{-\ln(\cdot)}\tau_tf\}_{0<t\leq t_0^{-1}}$ for $f\in\{\phi,\psi\}$ is $L^2(0,1)$ dominated for some $t_0>0$. We thus arrive at the below crude estimate.
\begin{cor} Suppose $\phi,\psi:\mathbb{R}_+\rightarrow\mathbb{C}$, besides satisfying \eqref{b:1}, obey \eqref{b:4} and $\{\sqrt{-\ln(\cdot)}\tau_tf\}_{0<t\leq t_0^{-1}}$ with $f\in\{\phi,\psi\}$ is $L^2(0,1)$ dominated for some $t_0>0$. Then, as $t\downarrow 0$,
\begin{equation}\label{b:28}
	F(t)=1+o(1).
\end{equation}
\end{cor}
To improve \eqref{b:28} we could impose further regularity and/or integrability constraints on $\phi,\psi$, but we shall concern ourselves with the more challenging behavior of $F(t)$ as $t\rightarrow+\infty$ instead. Indeed, using RHP \ref{masterb} and the Deift-Zhou nonlinear steepest descent method, we will prove the following Akhiezer-Kac theorem for \eqref{b:3} in a suitable class of kernels.
\begin{theo}\label{theo4a} Let $\epsilon>0$. Suppose $\phi,\psi:\mathbb{R}_+\rightarrow\mathbb{C}$ are continuously differentiable such that
\begin{equation}\label{b:28a}
	|f(x)|\leq\frac{1}{\sqrt{x}}\e^{-a|\ln x|},\ \ x\in\mathbb{R}_+
\end{equation}
for $f\in\{\phi,\psi,MD\phi,MD\psi\}$ with $a\geq 2+\epsilon$ and $I-\gamma K_t$ is invertible for all $(t,\gamma)\in\mathbb{R}_+\times[0,1]$. Then there exist $c=c(\epsilon),t_0=t_0(\epsilon)>0$ so that
\begin{equation}\label{b:29a}
	\ln F(t)=-s(1)\ln t+\int_1^{\infty}s(x)\hat{s}(x)\ln x\,\d x-\frac{1}{4\pi\im}\int_{\frac{1}{2}-\im\infty}^{\frac{1}{2}+\im\infty}\left\{\frac{r_1'(\lambda)}{r_1(\lambda)}-\frac{r_2'(\lambda)}{r_2(\lambda)}\right\}\ln\big(1-r_1(\lambda)r_2(\lambda)\big)\,\d\lambda+r(t)
\end{equation}
for $t\geq t_0$ with the functions, in terms of the principal branch for the logarithm $\ln:\mathbb{C}\setminus(-\infty,0]\rightarrow\mathbb{C}$,
\begin{equation*}
	s(x):=-\frac{1}{2\pi\im}\int_{\frac{1}{2}-\im\infty}^{\frac{1}{2}+\im\infty}\ln\big(1-r_1(\lambda)r_2(\lambda)\big)x^{-\lambda}\,\d\lambda,\ \ \hat{s}(x):=-\frac{1}{2\pi\im}\int_{\frac{1}{2}-\im\infty}^{\frac{1}{2}+\im\infty}\ln\big(1-r_1(\lambda)r_2(\lambda)\big)x^{\lambda-1}\,\d\lambda,
\end{equation*}
that are defined for $x\geq 1$. The error $r(t)$ in \eqref{b:29a} is $t$-differentiable and satisfies
\begin{equation}\label{b:29b}
	|r(t)|\leq ct^{-\frac{1}{2}\epsilon}\ \ \ \forall\,t\geq t_0.
\end{equation}
\end{theo}
To the best of our knowledge there is no comparable result available in the literature on the $t\rightarrow+\infty$ asymptotic behavior of Fredholm determinants of multiplicative composition operators, i.e. \eqref{b:29a} seemingly constitutes the first general Akhiezer-Kac theorem for determinants with kernels of type \eqref{b:2}.

\begin{rem} Our proof workings for Theorem \ref{theo4a} allow us to compute the $t\rightarrow+\infty$ asymptotics of the specialized determinant $F(t,\gamma),\gamma\in[0,1]$ that appears in \eqref{b:21},\eqref{b:22},\eqref{b:23}, subject to the constraint \eqref{b:28a}. Moreover, using the approach of \cite[Section $7.2.1$]{BB}, we can compute the $t\rightarrow+\infty$ asymptotics of $\omega(t,\gamma)$ appearing in Theorem \ref{theo4} and so the $t\rightarrow+\infty$ asymptotics of \eqref{b:25},\eqref{b:26} and \eqref{b:27} become accessible. We leave all relevant details once more to the dedicated reader.
\end{rem}

\subsection{Multiplicative composition, part 3}\label{seci26} We now apply the results of Section \ref{seci24} and \ref{seci25} to one composition operator occurring in random matrix theory, namely the one that underlies \eqref{i:8}.

\subsubsection{The Bessel kernel} Let $J=\mathbb{R}_+$ and choose
\begin{equation*}
	\phi(x)=\psi(x)=\frac{1}{2}J_{\alpha}(\sqrt{x}),\ \ \ x\in\mathbb{R}_+,\ \ \ \alpha>-1
\end{equation*}
The Bessel function of the first kind with square root argument and index $\alpha>-1$ is of order $\mathcal{O}(x^{\alpha/2})$ at zero and smooth for $x>0$, so all dominance and integrability assumptions at zero in Section \ref{seci24} and \ref{seci25} are satisfied. Moreover, the symmetric integral operator $K_t$ in \eqref{b:2} with kernel
\begin{equation}\label{examp4}
	K_t(x,y)=\frac{t}{4}\int_0^1J_{\alpha}(\sqrt{xzt})J_{\alpha}(\sqrt{zyt})\,\d z,
\end{equation}
corresponding to the Bessel choice, is trace class on $L^2(0,1)$ and $I-\gamma K_t$ is invertible on the same space for all $(t,\gamma)\in\mathbb{R}_+\times[0,1]$, see for instance \cite{TW2}. However, the Bessel function with square root argument is not in $H_{\circ}^{1,1}(\mathbb{R}_+)$, hence we cannot use Theorem \ref{theo3}. But, similar to the earlier Airy kernel example, the below ad hoc arguments show that the outcome of Theorem \ref{theo3} for the Bessel kernel remains valid, see Corollary \ref{BeCor}. On the other hand, Theorem \ref{theo4a} on the asymptotic behavior of $F(t)$ as $t\rightarrow+\infty$ famously fails for the Bessel kernel, compare \cite[$(5)$]{E}.
\begin{cor}\label{BeCor} The Fredholm determinant of the Bessel kernel \eqref{examp4} equals
\begin{equation*}
	F(t)=\exp\left[-\int_0^t\ln\Big(\frac{t}{s}\Big)\big(q_0(s)\big)^2\,\d s\right],\ \ \ t\in\mathbb{R}_+,
\end{equation*}
where
\begin{equation*}
	q_0(t)=\big((I-K_t)^{-1}\tau_t\phi\big)(1)=-\lim_{\substack{z\rightarrow\infty\\ z\notin\frac{1}{2}+\im\mathbb{R}}}\big(zX^{21}(z;t)\big),\ \ \ t\in\mathbb{R}_+
\end{equation*}
is expressed in terms of the solution ${\bf X}(z)={\bf X}(z;t)=[X^{mn}(z;t)]_{m,n=1}^2\in\mathbb{C}^{2\times 2}$ of RHP \ref{masterb} subject to
\begin{equation*}
	r_1(z)=2^{2z-1}\frac{\Gamma(\frac{\alpha}{2}+z)}{\Gamma(\frac{\alpha}{2}-z+1)},\ \ \ \ \ \ r_2(z)=2^{1-2z}\frac{\Gamma(\frac{\alpha}{2}-z+1)}{\Gamma(\frac{\alpha}{2}+z)},\ \ z\in\frac{1}{2}+\im\mathbb{R},
\end{equation*}
in terms of Euler's gamma function $w=\Gamma(z)$, and where \eqref{b:14} is valid only in certain directions, with ${\bf X}(z)$ having at worst polynomial growth at $z=\infty$. Moreover, $q_0(t)$ solves the ODE boundary value problem
\begin{equation}\label{b:29}
	t\left(q_0^2-\frac{1}{4}\right)\frac{\d}{\d t}\left(t\frac{\d q_0}{\d t}\right)=q_0\left(t\frac{\d q_0}{\d t}\right)^2+\frac{1}{16}\big(t-\alpha^2\big)q_0+tq_0^3\left(q_0^2-\frac{1}{2}\right),\ \ \ \ \ q_0(t)\sim\frac{t^{\frac{1}{2}\alpha}}{2^{1+\alpha}\Gamma(1+\alpha)},\ \ t\downarrow 0,
\end{equation}
identity \eqref{b:25} matches exactly \cite[$(3.14)$]{F0}, \eqref{b:26} is \cite[$(1.31)$]{F0} and \eqref{b:27} is \cite[$(1.33)$]{F0} (with $\tilde{q}=2q_0$ in \cite{F0}). In addition, \eqref{i:8} is the special case of \eqref{b:26} when $\gamma=1$ and $r\equiv 2q_0$, after taking a square root.
\end{cor}
\begin{proof} By \cite[$10.22.43$]{NIST}, for any $z\in\frac{1}{2}+\im\mathbb{R}$ and $\alpha>-1$,
\begin{equation*}
	\int_0^{\infty}y^{z-1}J_{\alpha}(\sqrt{y})\,\d y=4^z\frac{\Gamma(\frac{\alpha}{2}+z)}{\Gamma(\frac{\alpha}{2}-z+1)},\ \ \ \ \ \int_0^{\infty}y^{-z}J_{\alpha}(\sqrt{y})\,\d y=4^{1-z}\frac{\Gamma(\frac{\alpha}{2}-z+1)}{\Gamma(\frac{\alpha}{2}+z)},
\end{equation*}
so the coefficients $r_j(z)$ in \eqref{b:13} are indeed as listed above. Next, moving to RHP \ref{masterb} with the same choices for $r_j(z)$, but \eqref{b:14} valid only in certain directions, otherwise enforcing at most polynomial growth at $z=\infty$, we note that the same problem admits at most one solution by \cite[Proposition $3.1$]{FZ} and the same problem relates to \eqref{b:29} in the following way: factorize the underlying jump matrix \eqref{b:12}, with $\beta:=\frac{\alpha}{2}>-\frac{1}{2}$,
\begin{align*}
	\begin{bmatrix}1-r_1(z)r_2(z) & -r_2(z)t^z\\
	r_1(z)t^{-z} & 1\end{bmatrix}=\begin{bmatrix}1 & -r_2(z)t^z\\
	0 & 1\end{bmatrix}&\begin{bmatrix}1 & 0\\ r_1(z)t^{-z} & 1\end{bmatrix}\\
	&\equiv\begin{bmatrix}1 & -2^{1-2z}\frac{\Gamma(\beta-z+1)}{\Gamma(\beta+z)}t^z\\ 0 & 1\end{bmatrix}\begin{bmatrix}1 & 0\\ 2^{2z-1}\frac{\Gamma(\beta+z)}{\Gamma(\beta-z+1)}t^{-z} & 1\end{bmatrix},
\end{align*}
and define, noting the location of the poles of $\Gamma(\beta+z)$ in $(-\infty,\frac{1}{2})$ and the poles of $\Gamma(\beta-z+1)$ in $(\frac{1}{2},\infty)$,
\begin{equation}\label{b:30}
	{\bf Y}(z;t):={\bf X}(z;t)\begin{cases}\begin{bmatrix}1 & 0\\ -2^{2z-1}\frac{\Gamma(\beta+z)}{\Gamma(\beta-z+1)}t^{-z} & 1\end{bmatrix}, &z\in\Omega_+\smallskip\\
	\begin{bmatrix}1 & -2^{1-2z}\frac{\Gamma(\beta-z+1)}{\Gamma(\beta+z)}t^z\\ 0 & 1\end{bmatrix},&z\in\Omega_-\smallskip\\
	\mathbb{I},&\textnormal{else}
	\end{cases}
\end{equation}
with the domains shown in Figure \ref{fig2}. Observe that we do not explicitly indicate the $\beta$-dependency in our notations in \eqref{b:30}. The analytic and asymptotic properties of ${\bf Y}(z)$ are summarized below.
\begin{figure}[tbh]
\begin{tikzpicture}[xscale=0.7,yscale=0.7]
\draw [->] (-4,0) -- (4,0) node[below]{{\small $\Re z$}};
\draw [->] (0,-4.5) -- (0,4.5) node[left]{{\small $\Im z$}};
\draw [very thin, dashed, color=darkgray,-] (0.5,0) -- (2.5,3.46410161) node[right]{$\frac{\pi}{3}$};
\draw [very thin, dashed, color=darkgray,-] (0.5,0) -- (2.5,-3.46410161) node[right]{$\frac{5\pi}{3}$};
\draw [very thin, dashed, color=darkgray,-] (0.5,0) -- (-1.5,3.46410161) node[left]{$\frac{2\pi}{3}$};
\draw [very thin, dashed, color=darkgray,-] (0.5,0) -- (-1.5,-3.46410161) node[left]{$\frac{4\pi}{3}$};
\draw [fill=red, dashed,opacity=0.4] (0.5,0) -- (-1,2.598076211) arc (120:90:3cm) -- (0.5,0);
\draw [fill=red, dashed,opacity=0.4] (0.5,0) -- (-1,-2.598076211) arc (240:270:3cm) -- (0.5,0);
\draw [fill=blue, dashed,opacity=0.4] (0.5,0) -- (2,2.598076211) arc (60:90:3cm) -- (0.5,0);
\draw [fill=blue, dashed,opacity=0.4] (0.5,0) -- (2,-2.598076211) arc (300:270:3cm) -- (0.5,0);
\node [red,opacity=0.4] at (-0.5,3.5) {{\small$\Omega_+$}};
\node [red,opacity=0.4] at (-0.5,-3.5) {{\small$\Omega_+$}};
\node [blue,opacity=0.4] at (1.5,3.5) {{\small$\Omega_-$}};
\node [blue,opacity=0.4] at (1.5,-3.5) {{\small$\Omega_-$}};
\draw [thick, color=red, decoration={markings, mark=at position 0.5 with {\arrow{>}}}, postaction={decorate}] (0.5,0) -- (2.25,3.03108891);
\draw [thick, color=red, decoration={markings, mark=at position 0.5 with {\arrow{>}}}, postaction={decorate}] (0.5,0) -- (-1.25,3.03108891);
\draw [thick, color=red, decoration={markings, mark=at position 0.5 with {\arrow{>}}}, postaction={decorate}] (0.5,0) -- (2.25,-3.03108891);
\draw [thick, color=red, decoration={markings, mark=at position 0.5 with {\arrow{>}}}, postaction={decorate}] (0.5,0) -- (-1.25,-3.03108891);
\node [red] at (-1.3,2.3) {{\small$\Gamma_1$}};
\node [red] at (-1.4,-2.3) {{\small$\Gamma_2$}};
\node [red] at (2.4,2.3) {{\small$\Gamma_4$}};
\node [red] at (2.4,-2.3) {{\small$\Gamma_3$}};
\end{tikzpicture}
\caption{The oriented jump contour $\Sigma_{\bf Y}:=\Gamma_1\cup\Gamma_2\cup\Gamma_3\cup\Gamma_4$ in RHP \ref{PIII}. The four red lines intersect at $z=\frac{1}{2}$.}
\label{fig2}
\end{figure}
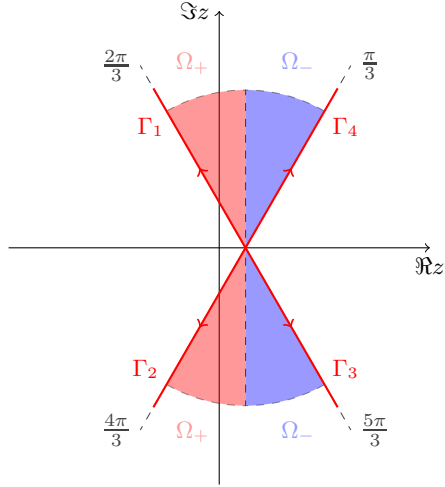

\begin{problem}\label{PIII} For every $t\in\mathbb{R}_+$ and $\alpha>-1$, the matrix-valued function ${\bf Y}(z)={\bf Y}(z;t)\in\mathbb{C}^{2\times 2}$ defined in \eqref{b:30} is such that
\begin{enumerate}
	\item[(1)] ${\bf Y}(z)$ is analytic for $z\in\mathbb{C}\setminus\Sigma_{\bf Y}$ and extends continuously from all sides to the oriented jump contour $\Sigma_{\bf Y}=\Gamma_1\cup\Gamma_2\cup\Gamma_3\cup\Gamma_4$ shown in Figure \ref{fig2}.
	\item[(2)] On $\Sigma_{\bf Y}$ the jump behavior of ${\bf Y}(z)$ reads as
	\begin{equation*}
		{\bf Y}_+(z)={\bf Y}_-(z)\e^{-\theta(z;t)\sigma_3}{\bf S}_k(z)\e^{\theta(z;t)\sigma_3},\ \ \ \ z\in\Gamma_k,\ \ k=1,2,3,4,
	\end{equation*}
	with $\theta(z;t):=(z-\frac{1}{2})\ln 2-\frac{z}{2}\ln t$ and the matrices
	\begin{equation*}
		{\bf S}_k(z):=\begin{bmatrix}1 & 0\\ s_k(z) & 1\end{bmatrix},\ \ k=1,2;\ \ \ \ \ \ \ {\bf S}_k(z):=\begin{bmatrix}1 & s_k(z)\\ 0 & 1\end{bmatrix},\ \ k=3,4,
	\end{equation*}
	defined with the multipliers $s_1(z):=\frac{\Gamma(\beta+z)}{\Gamma(\beta-z+1)},s_2(z):=-\frac{\Gamma(\beta+z)}{\Gamma(\beta-z+1)},s_3(z):=\frac{1}{s_1(z)},s_4(z):=\frac{1}{s_2(z)}$.
	\item[(3)] As $z\rightarrow\infty$ with $z\notin\Sigma_{\bf Y}\cup(\frac{1}{2}+\im\mathbb{R})$, using Stirling's formula for the ratio of gamma functions in \eqref{b:30},
	\begin{equation*}
		{\bf Y}(z)=\mathbb{I}+o(1).
	\end{equation*}
\end{enumerate}
\end{problem}
Observe that RHP \ref{PIII} is the RHP of \cite[page $2806,2807$]{CGS} with $j=3$ modulo a conjugation of ${\bf Y}(z)$ in \eqref{b:30} by $\sigma_2=\bigl[\begin{smallmatrix}0 & -\im\\ \im & 0\end{smallmatrix}\bigr]$. Thus the unique solvability of RHP \ref{PIII} for all $t\in\mathbb{R}_+$ and $\alpha>-1$ follows from \cite[Proposition $2.1$]{CGS} and from the aforementioned invertibility of $I-K_t$ on $L^2(0,1)$. Furthermore, using \cite[$(1.22),(2.9)$]{CGS} and \cite[$(1.5),(2.26)$]{TW2} we conclude that
\begin{equation*}
	q_0(t)=\big((I-K_t)^{-1}\tau_t\phi\big)(1)=-\lim_{\substack{z\rightarrow\infty\\ z\notin\Sigma_{\bf Y}\cup\frac{1}{2}+\im\mathbb{R}}}\big(zY^{21}(z;t)\big),\ \ \ t\in\mathbb{R}_+
\end{equation*}
solves \eqref{b:29}. This concludes our proof.
\end{proof}
\subsection{Outline of paper} The remaining sections of this paper are organized as follows. Because of our analysis of both, additive and multiplicative composition operators, it is natural to derive all algebraic, analytic and asymptotic results first for one class and afterwards address the remaining one. To this end, we devote Sections \ref{sec3}, \ref{sec4} and \ref{sec5} to operators with additive composition kernels of type \eqref{n1}. Precisely, in Section \ref{sec3} we prove the algebraic identities summarized in Lemma \ref{lem1}, Corollary \ref{gencoo} and Corollary \ref{highercor}. These follow by carefully exploiting the additive Hankel composition structure when $t$-differentiating first the Fredholm determinant \eqref{i:12} and afterwards the higher functions \eqref{n8},\eqref{n9} which arise quite naturally in this first step. Once the differential system \eqref{r2} has been worked out we then show how it can be encoded in the solution of RHP \ref{master2}. To the point, we address the unique solvability of the same problem in Section \ref{sec4} and afterwards construct its solution in \eqref{j14}. Our solution crucially depends on two distinguished Green's functions of a first order ODE, whose properties are derived with basic tools of Fourier analysis. Notably, our solution \eqref{j14} does not depend on Cauchy-type integrals commonly encountered in Riemann-Hilbert analysis. The simple proofs of Corollary \ref{impcor} and \ref{deeper} are given at the end of Section \ref{sec4}. Afterwards, in Section \ref{sec5}, we specialize the general setup \eqref{n1} to the class of symmetric Hankel composition kernels for which $\phi=\psi$. In this setup our determinant \eqref{i:12} factorizes with the help of \eqref{i:15} and we can subsequently evaluate \eqref{i:16},\eqref{i:17} and \eqref{i:18} in closed form. While proving Theorem \ref{theo2} we rely on the auxiliary identities \eqref{c3} and \eqref{c10} which are implicit in the non-rigorous theoretical physics works \cite{FS,DF,K}. Lastly, we address the asymptotic behavior of $F(t)$ in \eqref{i:12} as $t\rightarrow\pm\infty$ in Section \ref{sec5}. The crude right tail estimate \eqref{i:23} requires no further explanation whereas the left tail expansion in \eqref{i:25} follows from a nonlinear steepest descent analysis of RHP \ref{master2}. We impose the explicit, and easy to work with, constraint \eqref{i:24} which yields \eqref{i:25} in a rather straightforward fashion, working only with a $g$-function transformation and analytic continuation in the asymptotic analysis, see Section \ref{sec5}. We note in particular that \eqref{i:25} remains valid with properly adjusted error estimate, and follows from our analysis in Section \ref{sec5}, as soon as $r_j$ in \eqref{s9} admit analytic extensions to a horizontal strip of the form $\{z\in\mathbb{C}:\,\Im z\in(-\delta,\delta)\}$ with $\delta>0$ such that $|r_1(z)r_2(z)|<1$ and $r_j(z)=o(z^{-1})$ as $|z|\rightarrow\infty$ in the same strip. With the conclusion of Section \ref{sec5} we then address operators with multiplicative composition kernels of the form \eqref{b:2}. Unlike for \eqref{n1}, there are no comparable algebraic, analytic or asymptotic results available, even in theoretical physics, so the results of Sections \ref{seci24} and \ref{seci25} are certainly novel. In Section \ref{sec6} we derive the analogues of Lemma \ref{lem1}, Corollary \ref{gencoo} and Corollary \ref{highercor} in the form of Lemma \ref{lem2}, Corollary \ref{corb1} and Corollary \ref{corb2}. Due to the multiplicative composition structure, our proofs of the same results are somewhat more complicated, in particular defining the higher functions \eqref{b:9},\eqref{b:10} is less canonical than in \eqref{n8} and \eqref{n9}. However, system \eqref{b:11} is seemingly the only differential system that can be characterized through a simple RHP of type \ref{masterb}, other choices in \eqref{b:9},\eqref{b:10} would lead to more complicated Riemann-Hilbert problems. The unique solvability of RHP \ref{masterb} is established in Section \ref{sec7}, again relying on two Green's functions, see \eqref{bp19}. Afterwards, in Section \ref{sec8}, we address the symmetric version of \eqref{b:2} and compute the specialized functions \eqref{b:21},\eqref{b:22} and \eqref{b:23} in terms of Riemann-Hilbert data and, equivalently, in terms of the simpler determinant \eqref{b:20}. The second part of Section \ref{sec8} is devoted to the asymptotic analysis of \eqref{b:3} as $t\downarrow 0$ and $t\rightarrow+\infty$. Note that the nonlinear steepest descent analysis, as $t\rightarrow+\infty$, is similar to the previous $t\rightarrow-\infty$ analysis of \eqref{i:12} in Section \ref{sec5}. And, once more, \eqref{b:29a} holds with properly adjusted error estimate, by our methods, as soon as $r_j$ in \eqref{b:13} admit analytic extensions to a vertical strip of the form $\{z\in\mathbb{C}:\,\Re z\in(\frac{1}{2}-\delta,\frac{1}{2}+\delta)\}$ with $\delta>0$ such that $|r_1(z)r_2(z)|<1$ and $r_j(z)=o(z^{-1})$ as $|z|\rightarrow\infty$ in the same strip. Our paper concludes with Appendix \ref{appA} which summarizes a series of algebraic identities used in the proofs of Theorem \ref{theo1} and \ref{theo3}.

\begin{rem} As mentioned in Section \ref{sec1}, integrable operators have been celebrated in the past decades, but they were not the first integral operators whose Fredholm determinants can be analyzed via Riemann-Hilbert methods! In fact, the first occurrence of a Fredholm determinant accessible via Riemann-Hilbert techniques, see \cite[Chapter $2$, $\S 4$]{FT}, dates back to the works of Gelfand, Levitan, and Marchenko in the early 1950s on the inverse scattering problem, a topic closely related to the Zakharov-Shabat problems in this paper. Namely, in the inverse scattering problem one needs to solve an integral equation with kernel of additive Hankel type. Moreover, by Dyson's formula \cite{Dy}, the solution of the inverse problem constitutes in evaluating the second logarithmic derivative of the Fredholm determinant of the same integral operator. Thus the main focus of our paper relates back to the pioneering works of Gelfand-Levitan-Marchenko in integrable systems and scattering theory.
\end{rem}

\begin{rem}\label{PM} The additive chain \eqref{r2} of differential recurrence relations bears a striking resemblance to the multiplicative one in \eqref{b:11}. Here is one way to think about this feature: if $S\in\mathcal{L}(L^2(0,1),L^2(\mathbb{R}_+))$ is given by $(Sf)(x):=\e^{-\frac{1}{2}x}f(\e^{-x}),x\in\mathbb{R}_+$, then $S$ is invertible and with $K_t$ as in \eqref{b:2},
\begin{equation*}
	\big(SK_{\e^{-t}}S^{-1}f\big)(x)=\int_0^{\infty}\e^{-\frac{1}{2}x}K_{\e^{-t}}(\e^{-x},\e^{-y})\e^{-\frac{1}{2}y}f(y)\,\d y,\ \ \ f\in L^2(\mathbb{R}_+),\ \ \ t\in\mathbb{R},
\end{equation*}
where
\begin{equation*}
	\e^{-\frac{1}{2}x}K_{\e^{-t}}(\e^{-x},\e^{-y})\e^{-\frac{1}{2}y}\stackrel{\eqref{b:2}}{=}\int_0^{\infty}\Phi(x+z+t)\Psi(z+y+t)\,\d z,\ \ \ \ \ (x,y)\in\mathbb{R}_+^2,
\end{equation*}
given in terms of $\Phi(x):=\phi(\e^{-x})\e^{-\frac{1}{2}x}$ and $\Psi(x):=\psi(\e^{-x})\e^{-\frac{1}{2}x}$. Thus, on $L^2(\mathbb{R}_+)$,
\begin{equation}\label{PM1}
	SK_{\e^{-t}}S^{-1}=K_t^a,\ \ \ \ \ t\in\mathbb{R},
\end{equation}
where $K_t^a\in\mathcal{L}(L^2(\mathbb{R}_+))$ stands for $K_t$ in \eqref{n1} but with $\phi\mapsto\Phi$ and $\psi\mapsto\Psi$ therein. Consequently, by Lidskii's theorem, cf. \cite[Corollary $3.8$]{S}, the Fredholm determinant of $K_{\e^{-t}}$ on $L^2(0,1)$, being of multiplicative Hankel composition type, is equal to the Fredholm determinant of $K_t^a$ on $L^2(\mathbb{R}_+)$, the latter being of additive Hankel composition type. Equipped with \eqref{PM1} and \eqref{r2}, say, one could now derive \eqref{b:11} directly from \eqref{r2} - and vice versa. However, we prefer to work with \eqref{b:9} and \eqref{b:10}, rather than with the through \eqref{PM1} transformed quantities \eqref{n8},\eqref{n9}, because \eqref{b:9} and \eqref{b:10} are more natural from an applied viewpoint, for instance in the context of the classical Bessel kernel. Nevertheless, \eqref{PM1} has the potential to derive most, if not all, results for multiplicative Hankel composition operators in this paper from their additive counterparts - and vice versa.

\end{rem}

\section{Proof of Lemma \ref{lem1}, Corollary \ref{gencoo} and Corollary \ref{highercor}}\label{sec3}
In this section we prove the above mentioned three algebraic results about Fredholm determinants \eqref{i:12} of \textit{additive} Hankel composition operators with kernels \eqref{n1}.
\begin{proof}[Proof of Lemma \ref{lem1}] $I-K_t,t\in J$ is invertible on $L^2(\mathbb{R}_+)$ and $K_t,t\in J$ trace class on the same space, so $F(t)\neq 0$ on $J$ and thus $J\ni t\mapsto\ln F(t)$ is well-defined, with a branch for the logarithm that is analytically continued along the orbit of $F(t)$. Next, by the $L^2(\mathbb{R}_+)$ dominance of $\tau_t\phi,\tau_t\psi,D\tau_t\phi$ and $D\tau_t\psi$ and the continuous differentiability of $\phi,\psi$, $J\ni t\mapsto K_t$ is differentiable with $\frac{\d}{\d t}K_t=-\tau_t\phi\otimes\tau_t\psi,t\in J$, using \eqref{infbeh}. In fact, $\frac{\d}{\d t}K_t\in\mathcal{C}_0(L^2(\mathbb{R}_+))\subset\mathcal{C}_1(L^2(\mathbb{R}_+))$ and so by Jacobi's formula, cf. \cite[$(3.12)$]{Wid},
\begin{equation}\label{n2}
	\frac{\d}{\d t}\ln F(t)=-\tr_{L^2(\mathbb{R}_+)}\left((I-K_t)^{-1}\frac{\d K_t}{\d t}\right)=\tr_{L^2(\mathbb{R}_+)}\big((I-K_t)^{-1}\tau_t\phi\otimes\tau_t\psi\big),\ t\in J.
\end{equation}
Moving ahead, $\tau_t\phi,\tau_t\psi\in W^{1,2}(\mathbb{R}_+)$ by dominance, so $J\ni t\mapsto\tau_t\phi\otimes\tau_t\psi\in\mathcal{C}_1(L^2(\mathbb{R}_+))$ is differentiable\footnote{By the dominance assumption, differentiation with respect to the dominated variable commutes with integration with respect to the same. We shall use the same fact numerous times in the following.},
\begin{equation}\label{j1}
	\frac{\d}{\d t}(\tau_t\phi\otimes\tau_t\psi)=D\tau_t\phi\otimes\tau_t\psi+\tau_t\phi\otimes D\tau_t\psi\in\mathcal{C}_0(L^2(\mathbb{R}_+))\subset\mathcal{C}_1(L^2(\mathbb{R}_+)).
\end{equation}
 Likewise, $J\ni t\mapsto (I-K_t)^{-1}\in\mathcal{L}(L^2(\mathbb{R}_+))$ is differentiable,
\begin{equation}\label{j2}
	\frac{\d}{\d t}(I-K_t)^{-1}=(I-K_t)^{-1}\frac{\d K_t}{\d t}(I-K_t)^{-1}=-(I-K_t)^{-1}(\tau_t\phi\otimes\tau_t\psi)(I-K_t)^{-1}\in\mathcal{C}_0(L^2(\mathbb{R}_+)),
\end{equation} 
since $I-K_t,t\in J$ is invertible and $J\ni t\mapsto K_t$ differentiable. So, \eqref{j1} and \eqref{j2} together imply that $J\ni t\mapsto(I-K_t)^{-1}\tau_t\phi\otimes\tau_t\psi\in\mathcal{C}_1(L^2(\mathbb{R}_+))$ is differentiable and
\begin{align}
	\frac{\d}{\d t}\big((I-K_t)^{-1}&\,\tau_t\phi\otimes\tau_t\psi\big)=(I-K_t)^{-1}\tau_t\phi\otimes D\tau_t\psi\nonumber\\
	&\hspace{0.5cm}+\Big[(I-K_t)^{-1}D\tau_t\phi-(I-K_t)^{-1}(\tau_t\phi\otimes\tau_t\psi)(I-K_t)^{-1}\tau_t\phi\Big]\otimes\tau_t\psi\in\mathcal{C}_1(L^2(\mathbb{R}_+)).\label{j3}
\end{align}
Consequently, using that by \cite[$(3.4)$]{Bris}
\begin{equation}\label{j4}
	\tr_{L^2(\mathbb{R}_+)}\big((I-K_t)^{-1}\tau_t\phi\otimes\tau_t\psi\big)=\int_0^{\infty}\big((I-K_t)^{-1}\tau_t\phi\big)(x)(\tau_t\psi)(x)\,\d x,
\end{equation}
since $\mathbb{R}_{\geq 0}^2\ni(x,y)\mapsto\big((I-K_t)^{-1}\tau_t\phi\big)(x)(\tau_t\psi)(y)$ is continuous\footnote{The map $\mathbb{R}_{\geq 0}\ni x\mapsto((I-K_t)^{-1}\tau_t\phi)(x)$ is locally $\frac{1}{2}$-H\"older continuous for $t\in J$ by invertibility of $I-K_t,t\in J$ and by \eqref{n00},\eqref{intbeh}. This shows, in particular, that $J\ni t\mapsto q_0(t)$ is well-defined.}, the dominated convergence theorem yields
\begin{align}
	\frac{\d^2}{\d t^2}\ln &\,F(t)\stackrel[\eqref{n2}]{\eqref{j4}}{=}\frac{\d}{\d t}\int_0^{\infty}\big((I-K_t)^{-1}\tau_t\phi\big)(x)(\tau_t\psi)(x)\,\d x=\int_0^{\infty}\frac{\d}{\d t}\Big[\big((I-K_t)^{-1}\tau_t\phi\big)(x)(\tau_t\psi)(x)\Big]\,\d x\nonumber\\
	&\hspace{0.72cm}=\tr_{L^2(\mathbb{R}_+)}\Big(\frac{\d}{\d t}\big((I-K_t)^{-1}\tau_t\phi\otimes\tau_t\psi\big)\Big)\stackrel{\eqref{j3}}{=}\tr_{L^2(\mathbb{R}_+)}\big((I-K_t)^{-1}\tau_t\phi\otimes D\tau_t\psi\big)\nonumber\\
	&\hspace{1.5cm}+\tr_{L^2(\mathbb{R}_+)}\Big(\Big[(I-K_t)^{-1}D\tau_t\phi-(I-K_t)^{-1}(\tau_t\phi\otimes\tau_t\psi)(I-K_t)^{-1}\tau_t\phi\Big]\otimes\tau_t\psi\Big),\ t\in J.\label{j4a}
\end{align}
This shows that $J\ni t\mapsto\ln F(t)$ is twice-differentiable. In order to simplify the last expression, we use that all involved integral operators have continuous kernels, so in particular, integrating by parts,
\begin{align}
	\tr_{L^2(\mathbb{R}_+)}\big((I-K_t)^{-1}\tau_t\phi\otimes D\tau_t\psi\big)=&\,\int_0^{\infty}\big((I-K_t)^{-1}\tau_t\phi)(x)(D\tau_t\psi)(x)\,\d x\nonumber\\
	=&\,-q_0(t)(\tau_t\psi)(0)-\int_0^{\infty}\big(D(I-K_t)^{-1}\tau_t\phi\big)(x)(\tau_t\psi)(x)\,\d x.\label{j5}
\end{align}
Note that $\mathbb{R}_+\ni x\mapsto((I-K_t)^{-1}\tau_t\phi)(x)$ is continuously differentiable by dominance, \eqref{n00} and by invertibility of $I-K_t,t\in J$ on $L^2(\mathbb{R}_+)$. Moreover, by \eqref{n00},\eqref{intbeh} and dominance, $(I-K_t)^{-1}\tau_t\phi\in W^{1,2}(\mathbb{R}_+)$ and $((I-K_t)^{-1}\tau_t\phi)(x)\rightarrow 0$ as $x\rightarrow+\infty$ by \eqref{infbeh}. Thus, \eqref{j5} is justified and we can rewrite it as, cf. \cite[$(3.4)$]{Bris},
\begin{equation*}
	\tr_{L^2(\mathbb{R}_+)}\big((I-K_t)^{-1}\tau_t\phi\otimes D\tau_t\psi\big)=-q_0(t)(\tau_t\psi)(0)-\tr_{L^2(\mathbb{R}_+)}\big(D(I-K_t)^{-1}\tau_t\phi\otimes\tau_t\psi\big),\ \ t\in J.
\end{equation*}
We now return to \eqref{j4a}, insert the last identity,
\begin{align}
	\frac{\d^2}{\d t^2}\ln F(t)=-q_0(t)(\tau_t\psi)(0)-\tr_{L^2(\mathbb{R}_+)}\bigg(&\,\Big[D(I-K_t)^{-1}\tau_t\phi-(I-K_t)^{-1}D\tau_t\phi\nonumber\\
	&\,+(I-K_t)^{-1}(\tau_t\phi\otimes\tau_t\psi)(I-K_t)^{-1}\tau_t\phi\Big]\otimes\tau_t\psi\bigg),\ \ t\in J,\label{j6}
\end{align}
afterwards perform algebra,
\begin{equation}\label{bet5}
	D(I-K_t)^{-1}\tau_t\phi-(I-K_t)^{-1}D\tau_t\phi=(I-K_t)^{-1}(DK_t-K_tD)(I-K_t)^{-1}\tau_t\phi\in L^2(\mathbb{R}_+),
\end{equation}
and then simplify with the help of the identities
\begin{equation}\label{bet5a}
	\frac{\partial}{\partial x}K_t(x,y)+\frac{\partial}{\partial y}K_t(x,y)\stackrel{\eqref{infbeh}}{=}-(\tau_t\phi\otimes\tau_t\psi)(x,y),\ \ \ \ \ \ \lim_{y\rightarrow+\infty}K_t(x,y)=0\ \ \ \ \forall\,x\in\mathbb{R}_+.
\end{equation}
The intermediate result equals
\begin{equation}\label{bet6}
	\big((DK_t-K_tD)(I-K_t)^{-1}\tau_t\phi\big)(x)=K_t(x,0)q_0(t)-(\tau_t\phi)(x)\int_0^{\infty}(\tau_t\psi)(y)\big((I-K_t)^{-1}\tau_t\phi\big)(y)\,\d y,
\end{equation}
where $K_t(x,0)$ is shorthand for $K_t(x,0):=\lim_{y\downarrow 0}K_t(x,y)$, and so with \eqref{bet6} back in \eqref{j6}, using that all involved trace class operators have continuous kernels,
\begin{equation*}
	\frac{\d^2}{\d t^2}\ln F(t)=-q_0(t)\left[(\tau_t\psi)(0)+\int_0^{\infty}\big((I-K_t)^{-1}K_t\big)(x,0)(\tau_t\psi)(x)\,\d x\right],\ \ t\in J.
\end{equation*}
Lastly, it remains to use $(I-K_t)^{-1}=I+(I-K_t)^{-1}K_t$ on $L^2(\mathbb{R}_+)$ and $(I-K_t)^{-1}(x,0)=(I-K_t^{\ast})^{-1}(0,x)$ for $x>0$, with $(I-K_t)^{-1}(x,y)$ denoting the kernel of $(I-K_t)^{-1}$,
\begin{equation*}
	\frac{\d^2}{\d t^2}\ln F(t)=-q_0(t)\int_0^{\infty}(I-K_t)^{-1}(x,0)(\tau_t\psi)(x)\,\d x=-q_0(t)q_0^{\ast}(t),\ \ t\in J,
\end{equation*}
as claimed in \eqref{r1}. Note that $q_0^{\ast}$ is well-defined by \eqref{n00},\eqref{intbeh} and by invertibility of $I-K_t$. This completes our proof.
\end{proof}
\begin{remark} The reader can notice that the additive Hankel composition structure is used precisely in the identity $\frac{\d}{\d t}\tau_tf=D\tau_tf$. As such we use the additive composition structure to get $\frac{\d}{\d t}K_t=-\tau_t\phi\otimes\tau_t\psi$ and then again in the derivation of \eqref{j3} and \eqref{bet5a}.
\end{remark}
\begin{proof}[Proof of Corollary \ref{gencoo}] We have, see Lemma \ref{lem1},
\begin{equation*}
	\frac{\d^2}{\d t^2}\ln F(t)=-q_0(t)q_0^{\ast}(t),\ \ \ t\in J.
\end{equation*}
By our workings, $J\ni t\mapsto(I-K_t)^{-1}\in\mathcal{L}(W^{1,2}(\mathbb{R}_+))$ is differentiable with $\frac{\d}{\d t}(I-K_t)^{-1}\in\mathcal{C}_0(L^2(\mathbb{R}_+))$, compare \eqref{j2}. But $t\mapsto\tau_t\in\mathcal{L}(W^{1,2}(\mathbb{R}_+))$ is also differentiable with 
\begin{equation*}
	\frac{\d}{\d t}\tau_t=D\tau_t=\tau_tD\in\mathcal{L}\big(W^{1,2}(\mathbb{R}_+),L^2(\mathbb{R}_+)\big).
\end{equation*}
Hence, $J\ni t\mapsto(I-K_t)^{-1}\tau_t\in\mathcal{L}(W^{1,2}(\mathbb{R}_+))$ is differentiable and so $J\ni t\mapsto q_0(t)=((I-K_t)^{-1}\tau_t\phi)(0)$ in particular continuous. Repeating the same argument with $\phi\leftrightarrow\psi$ yields the corresponding result for $J\ni t\mapsto q_0^{\ast}(t)$. Now let $J=\mathbb{R}$, then $q_0q_0^{\ast}\in L^1(\mathbb{R}_+)$ by assumption, and so
\begin{equation}\label{j7}
	\lim_{s\rightarrow+\infty}\frac{F'(s)}{F(s)}-\frac{\d}{\d t}\ln F(t)=\int_t^{\infty}\frac{\d^2}{\d s^2}\ln F(s)\,\d s\stackrel{\eqref{r1}}{=}-\int_t^{\infty}q_0(s)q_0^{\ast}(s)\,\d s,\ \ \ t\in\mathbb{R}.
\end{equation}
Here, $\lim_{s\rightarrow+\infty}F(s)=1$ by \eqref{i:23} and the right hand side of \eqref{j7}, viewed as function of $t\in\mathbb{R}$, lies in $L^1(\mathbb{R}_+)$ by assumption. Consequently we must have $\lim_{s\rightarrow+\infty}F'(s)=0$ and so
\begin{equation*}
	\ln F(t)=-\int_t^{\infty}\frac{\d}{\d u}\ln F(u)\,\d u\stackrel{\eqref{j7}}{=}-\int_t^{\infty}\left[\int_u^{\infty}q_0(s)q_0^{\ast}(s)\,\d s\right]\d u=-\int_t^{\infty}(s-t)q_0(s)q_0^{\ast}(s)\,\d s,\ \ t\in\mathbb{R}.
\end{equation*}
This concludes our proof of Corollary \ref{gencoo}.
\end{proof}

\begin{proof}[Proof of Corollary \ref{highercor}] The symmetry $\phi\leftrightarrow\psi$ exchanges $K_t\leftrightarrow K_t^{\ast},q_n\leftrightarrow q_n^{\ast}$ and $p_n\leftrightarrow p_n^{\ast}$, so it is sufficient to derive the stated equations in \eqref{r2} for $\{q_n,p_n\}_{n=0}^{N-1}$. First we note that $J\ni t\mapsto q_n(t)$ is differentiable for $n=0,1,\ldots,N-1$ since $D^n\tau_t\phi\in W^{1,2}(\mathbb{R}_+)$ for the same $n$ and since $t\mapsto(I-K_t)^{-1}\tau_t\in\mathcal{L}(W^{1,2}(\mathbb{R}_+))$ is differentiable with
\begin{equation}\label{j8}
	\frac{\d}{\d t}(I-K_t)^{-1}\tau_t=(I-K_t)^{-1}D\tau_t-(I-K_t)^{-1}(\tau_t\phi\otimes\tau_t\psi)(I-K_t)^{-1}\tau_t\in\mathcal{L}\big(W^{1,2}(\mathbb{R}_+),L^2(\mathbb{R}_+)\big).
\end{equation}
Second, $J\ni t\mapsto q_N(t)$ is well-defined, and by the workings in Lemma \ref{lem1}, $J\ni t\mapsto p_n(t)$ is also differentiable for $n=0,1,\ldots,N-1$. This is because $J\ni t\mapsto (I-K_t)^{-1}D^n\tau_t\phi\otimes\tau_t\psi\in\mathcal{C}_1(L^2(\mathbb{R}_+))$ is differentiable,
\begin{align}
	\frac{\d}{\d t}\big(&(I-K_t)^{-1}D^n\tau_t\phi\otimes\tau_t\psi\big)=(I-K_t)^{-1}D^n\tau_t\phi\otimes D\tau_t\psi\nonumber\\
	&+\Big[(I-K_t)^{-1}D^{n+1}\tau_t\phi-(I-K_t)^{-1}(\tau_t\phi\otimes\tau_t\psi)(I-K_t)^{-1}D^n\tau_t\phi\Big]\otimes\tau_t\psi\in\mathcal{C}_1(L^2(\mathbb{R}_+)),\label{j9}
\end{align}
for $n=0,1,\ldots,N-1$, and since by the dominated convergence theorem and \cite[$(3.4)$]{Bris},
\begin{equation}\label{j10}
	\frac{\d p_n}{\d t}(t)=\frac{\d}{\d t}\int_0^{\infty}\big((I-K_t)^{-1}D^n\tau_t\phi\big)(x)(\tau_t\psi)(x)\,\d x=\tr_{L^2(\mathbb{R}_+)}\Big(\frac{\d}{\d t}\big((I-K_t)^{-1}\tau_t\phi\otimes\tau_t\psi\big)\Big).
\end{equation}
At this point we compute the derivative of $q_n$ and $p_n$. By \eqref{j8} and \cite[$(3.4)$]{Bris}, for $n=0,1,\ldots,N-1$,
\begin{align*}
	\frac{\d q_n}{\d t}(t)=\frac{\d}{\d t}\big((I-K_t)^{-1}D^n\tau_t\phi\big)(0)\stackrel{\eqref{j8}}{=}q_{n+1}(t)-q_0(t)p_n(t),\ \ t\in J.
\end{align*}
And by \eqref{j10} and \eqref{j9}, for $n=0,1,\ldots,N-1$,
\begin{align}
	\frac{\d p_n}{\d t}(t)&\,=\tr_{L^2(\mathbb{R}_+)}\big((I-K_t)^{-1}D^n\tau_t\phi\otimes D\tau_t\psi\big)\nonumber\\
	&+\tr_{L^2(\mathbb{R}_+)}\Big(\Big[(I-K_t)^{-1}D^{n+1}\tau_t\phi-(I-K_t)^{-1}(\tau_t\phi\otimes\tau_t\psi)(I-K_t)^{-1}D^n\tau_t\phi\Big]\otimes\tau_t\psi\Big),\ \ t\in J,\label{j11}
\end{align}
where $\mathbb{R}_+\ni x\mapsto((I-K_t)^{-1}D^n\tau_t\phi)(x)$ is continuously differentiable, $(I-K_t)^{-1}D^n\tau_t\phi\in W^{1,2}(\mathbb{R}_+)$ for all $t\in J$ and $((I-K_t)^{-1}D^n\tau_t\phi)(x)\rightarrow 0$ as $x\rightarrow+\infty$ by \eqref{n7}. These facts allow us to simplify the first trace in \eqref{j11} via integration by parts,
\begin{equation*}
	\tr_{L^2(\mathbb{R}_+)}\big((I-K_t)^{-1}D^n\tau_t\phi\otimes D\tau_t\psi\big)=-q_n(t)(\tau_t\psi)(0)-\tr_{L^2(\mathbb{R}_+)}\big(D(I-K_t)^{-1}D^n\tau_t\phi\otimes\tau_t\psi\big),
\end{equation*}
and so \eqref{j11} becomes
\begin{align}
	\frac{\d p_n}{\d t}(t)=-q_n(t)(\tau_t\psi)(0)-\tr_{L^2(\mathbb{R}_+)}\bigg(&\,\Big[D(I-K_t)^{-1}D^n\tau_t\phi-(I-K_t)^{-1}D^{n+1}\tau_t\phi\nonumber\\
	&\,+(I-K_t)^{-1}(\tau_t\phi\otimes\tau_t\psi)(I-K_t)^{-1}D^n\tau_t\phi\Big]\otimes\tau_t\psi\bigg),\ \ t\in J,\label{j12}
\end{align}
which generalizes \eqref{j6}. We then simplify \eqref{j12} just as we did it for \eqref{j6} in the proof of Lemma \ref{lem1},
\begin{equation*}
	\big((DK_t-K_tD)(I-K_t)^{-1}D^n\tau_t\phi\big)(x)=K_t(x,0)q_n(t)-(\tau_t\phi)(x)\int_0^{\infty}(\tau_t\psi)(y)\big((I-K_t)^{-1}D^n\tau_t\phi)(y)\,\d y
\end{equation*}
and obtain in \eqref{j12},
\begin{equation*}
	\frac{\d p_n}{\d t}=-q_n(t)\left[(\tau_t\psi)(0)+\int_0^{\infty}\big((I-K_t)^{-1}K_t\big)(x,0)(\tau_t\psi)(x)\,\d x\right],\ \ t\in J,\ \ n=0,1,\ldots,N-1.
\end{equation*}
It remains to use $(I-K_t)^{-1}=I+(I-K_t)^{-1}K_t$ and $(I-K_t)^{-1}(x,0)=(I-K_t^{\ast})^{-1}(0,x)$, i.e. we have
\begin{equation*}
	\frac{\d p_n}{\d t}=-q_n(t)\int_0^{\infty}(I-K_t)^{-1}(x,0)(\tau_t\psi)(x)\,\d x=-q_n(t)q_0^{\ast}(t),\ \ t\in J,
\end{equation*}
as claimed in \eqref{r2}.
\end{proof}
The last proof concludes the content of this section, we now move to RHP \ref{master2}.


\section{Proof of Theorem \ref{theo1} and Corollaries \ref{impcor} and \ref{deeper}}\label{sec4}

In order to arrive at Theorem \ref{theo1} we begin with the following Lemma that asserts the unique solvability of RHP \ref{master2}, provided the same problem is solvable.

\begin{lem}\label{unique2} The RHP \ref{master2}, if solvable, is uniquely solvable.
\end{lem}
\begin{proof} The pointwise convergence requirement in condition (2) of RHP \ref{master2} implies ${\bf X}(z\pm\im\epsilon)\rightarrow{\bf X}_{\pm}(z)$ for $z\in\mathbb{R}$ locally in $L^2(\mathbb{R})$ as $\epsilon\downarrow 0$. Precisely, with $z\in\mathbb{R}$,
\begin{equation*}
	\lim_{\epsilon\downarrow 0}\int_{z-\delta}^{z+\delta}\|{\bf X}(\lambda\pm\im\epsilon)-{\bf X}_{\pm}(\lambda)\|^2\,\d\lambda=0\ \ \ \textnormal{for some}\ \delta>0.
\end{equation*}
In turn $\det{\bf X}(z\pm\im\epsilon)\rightarrow(\det{\bf X}_{\pm})(z)$ locally in $L^1(\mathbb{R})$ as $\epsilon\downarrow 0$. Since also by \eqref{s8}, for every $z\in\mathbb{R}$,
\begin{equation*}
	(\det{\bf X})_+(z)=\det{\bf X}_+(z)=\det\left({\bf X}_-(z)\begin{bmatrix}
	1-r_1(z)r_2(z) & -r_2(z)\e^{-\im tz}\smallskip\\
		r_1(z)\e^{\im tz} & 1\end{bmatrix}\right)=\det{\bf X}_-(z)=(\det{\bf X})_-(z),
\end{equation*}
we can now copy and paste the proof workings of \cite[Theorem $7.18$]{D} in order to conclude that $\det{\bf X}(z)$ is analytic across $\mathbb{R}$, thus - with condition (1) - an entire function. On the other hand, by condition (3) of RHP \ref{master2} we have $\det{\bf X}(z)=1+\mathcal{O}(z^{-1})$ uniformly as $z\rightarrow\infty$ in $\mathbb{C}\setminus\mathbb{R}$. Hence, by Liouville's theorem,
\begin{equation*}
	\det{\bf X}(z)\equiv 1\ \ \ \forall\,z.
\end{equation*}
Consequently, each solution of RHP \ref{master2} is invertible. Now suppose ${\bf X}_1(z)$ and ${\bf X}_2(z)$ are two solutions of RHP \ref{master2}. Set ${\bf R}(z):={\bf X}_1(z){\bf X}_2(z)^{-1},z\in\mathbb{C}\setminus\mathbb{R}$ and note that ${\bf R}(z)$ has boundary values on $\mathbb{R}$ locally in $L^1(\mathbb{R})$ sense, using again the pointwise convergence requirement (2), in fact ${\bf R}_+(z)={\bf R}_-(z),z\in\mathbb{R}$ by \eqref{s8}. The same argument used to prove that $\det{\bf X}(z)\equiv 1$ now shows that ${\bf R}(z)\equiv1$, i.e. ${\bf X}_1(z)={\bf X}_2(z)$, implying uniqueness.
\end{proof}
At this point we need to introduce two auxiliary functions, cf. \cite[$(5.20)$]{K}, namely specific solutions of 
\begin{equation*}
	-\im \frac{\d y}{\d x}=zy-\im\phi\ \ \ \ \ \textnormal{and}\ \ \ \ -\im\frac{\d y}{\d x}=-zy-\im\psi,\ \ \ \ \ \ \ \ y=y(x,z):\mathbb{R}\times\mathbb{C}\rightarrow\mathbb{C}.
\end{equation*}
These will play an important role in the proof of Theorem \ref{theo1}. 
\begin{prop}\label{crucial1} Suppose $\phi,\psi\in L^{\infty}(\mathbb{R})\cap W^{1,1}(\mathbb{R})$ are continuously differentiable on $\mathbb{R}$ with $D\phi,D\psi\in L^{\infty}(\mathbb{R}_+)$. Define, for all $x\in\mathbb{R}$,
\begin{equation}\label{sq11}
	y_1(x,z):=\begin{cases}\displaystyle\int_{-\infty}^x\e^{\im z(x-s)}\phi(s)\,\d s,&\Im z>0\bigskip\\
	\displaystyle-\int_x^{\infty}\e^{\im z(x-s)}\phi(s)\,\d s,&\Im z<0\end{cases}\,,\ \ \ \ \ \ \ y_2(x,z):=\begin{cases}\displaystyle-\int_x^{\infty}\e^{-\im z(x-s)}\psi(s)\,\d s,&\Im z>0\bigskip\\
	\displaystyle\int_{-\infty}^x\e^{-\im z(x-s)}\psi(s)\,\d s,&\Im z<0\end{cases}\,,
\end{equation}
and let $y_k^z$ denote the function $\mathbb{R}\ni x\mapsto y^z_k(x):=y_k(x,z),k=1,2$. Then 
\begin{enumerate}
	\item[(a)] $y_k^z\in L^2(\mathbb{R}),k=1,2$ for every $z\in\mathbb{C}\setminus\mathbb{R}$. 
	\item[(b)] $\mathbb{C}\setminus\mathbb{R}\ni z\mapsto y_k^z(x),k=1,2$ are analytic for all $x\in\mathbb{R}$ and extend continuously to the closed upper and lower half-planes. For $z\in\mathbb{R}$ we record the, pointwise in $x\in\mathbb{R}$, limits
	\begin{equation}\label{sq12}
		\lim_{\epsilon\downarrow 0}y_1^{z+\im\epsilon}(x)-\lim_{\epsilon\downarrow 0}y_1^{z-\im\epsilon}(x)=r_1(z)\im\e^{\im zx},\ \ \ 
		\lim_{\epsilon\downarrow 0}y_2^{z+\im\epsilon}(x)-\lim_{\epsilon\downarrow 0}y_2^{z-\im\epsilon}(x)=r_2(z)\im\e^{-\im zx},
	\end{equation}
	with $r_k(z),k=1,2$ defined in \eqref{s9}.
	\item[(c)] We have
\begin{equation}\label{j13}
	y_1^z(x)=\frac{\im}{z}\phi(x)+m_1^z(x),\ \ \ \ \ y_2^z(x)=-\frac{\im}{z}\psi(x)+m_2^z(x),\ \ \ \ \ \ \ \ (x,z)\in\mathbb{R}\times(\mathbb{C}\setminus\mathbb{R}),
\end{equation}
where $m_k^z\in L^2(\mathbb{R}_+),k=1,2$ for every $z\in\mathbb{C}\setminus\mathbb{R}$. In fact, there exist $c_k>0$ so that for all $z\in\mathbb{C}\setminus\mathbb{R}$,
\begin{equation*}
	\|m_1^z\|_{L^2(\mathbb{R}_+)}\leq\frac{c_1}{|z|\sqrt{|\Im z|}}\|D\phi\|_{L^1(\mathbb{R})},\ \ \ \ \ \ \ \|m_2^z\|_{L^2(\mathbb{R}_+)}\leq\frac{c_2}{|z|\sqrt{|\Im z|}}\|D\psi\|_{L^1(\mathbb{R})}.
\end{equation*}
\end{enumerate}
\end{prop}
\begin{proof} By construction \eqref{sq11}, for all $(x,z)\in\mathbb{R}\times(\mathbb{C}\setminus\mathbb{R})$,
\begin{equation*}
	|y_1(x,z)|\leq\|\phi\|_{L^1(\mathbb{R})}\ \ \ \ \ \ \textnormal{and}\ \ \ \ \ \ |y_2(x,z)|\leq\|\psi\|_{L^1(\mathbb{R})},
\end{equation*}
so the functions $y_k^z$ are well-defined for $z\in\mathbb{C}\setminus\mathbb{R}$. More is true, for fixed $z\in\mathbb{C}\setminus\mathbb{R}$ with $\delta:=|\Im z|>0$,
\begin{equation*}
	|y_1^z(x)|\leq\int_{-\infty}^{\infty}\e^{-\delta|x-s|}|\phi(s)|\,\d s=(E\ast|\phi|)(x),\ \ \ |y_2^z(x)|\leq\int_{-\infty}^{\infty}\e^{-\delta|x-s|}|\psi(s)|\,\d s=(E\ast|\psi|)(x),\ \ x\in\mathbb{R},
\end{equation*}
with $(f\ast g)(x):=\int_{-\infty}^{\infty}f(x-y)g(y)\,\d y$ and where $E(x):=\e^{-\delta|x|}$. Hence, by the $L^1(\mathbb{R})$ Fourier convolution theorem, cf. \cite[Chapter VI, $1.3$]{Kat}, for any $\xi\in\mathbb{R}$,
\begin{equation*}
	\widehat{E\ast|\phi|}(\xi)=\widehat{E}(\xi)\widehat{|\phi|}(\xi)=\frac{2\delta}{\delta^2+\xi^2}\widehat{|\phi|}(\xi),\ \ \ \widehat{E\ast|\psi|}(\xi)=\widehat{E}(\xi)\widehat{|\psi|}(\xi)=\frac{2\delta}{\delta^2+\xi^2}\widehat{|\psi|}(\xi),
\end{equation*}
where the Fourier transform $\hat{f}$ of $f\in L^1(\mathbb{R})$ is defined by
\begin{equation*}
	\hat{f}(\xi):=\int_{-\infty}^{\infty}f(x)\e^{-\im\xi x}\,\d x,\ \ \ \ \ \xi\in\mathbb{R}.
\end{equation*}
Consequently, by Plancherel's theorem, using that $\phi,\psi\in L^1(\mathbb{R})\cap L^2(\mathbb{R})$ and $E\ast|\phi|,E\ast|\psi|\in L^1(\mathbb{R})\cap L^2(\mathbb{R})$,
\begin{equation*}
	 \|y_1^z\|_{L^2(\mathbb{R})}\leq\|E\ast|\phi|\|_{L^2(\mathbb{R})}\leq c\|\widehat{E\ast|\phi|}\|_{L^2(\mathbb{R})}\leq c\|\widehat{|\phi|}\|_{L^2(\mathbb{R})}\leq c\|\phi\|_{L^2(\mathbb{R})}<\infty,\ \ c=c(\delta)>0.
\end{equation*}
An analogous bound holds for $\|y_2^z\|_{L^2(\mathbb{R})}$, so indeed $y_k^z\in L^2(\mathbb{R})$ for every $z\in\mathbb{C}\setminus\mathbb{R}$, as claimed in (a). Next, $\mathbb{C}\setminus\mathbb{R}\ni z\mapsto y_k^z(x)$ are continuous for any $x\in\mathbb{R}$ by the dominated convergence theorem using $\phi,\psi\in L^1(\mathbb{R})$. Moreover, for any triangle $T$ strictly contained in the upper or lower half-plane, by Fubini's theorem,
\begin{equation*}
	\oint_Ty_k^z(x)\,\d z=0\ \ \ \ \forall\,x\in\mathbb{R},\ \ \ k=1,2,
\end{equation*}
so the same functions $z\mapsto y_k^z(x)$ are analytic on the separate half-planes by Morera's theorem, for all $x\in\mathbb{R}$. Additionally, the existence of their continuous extensions down to the real axis and in turn the relations \eqref{sq12} are a direct consequence of the dominated convergence theorem, of the assumption $\phi,\psi\in L^1(\mathbb{R})$ and formula \eqref{sq11}. This establishes (b). To get to \eqref{j13}, we integrate by parts,
\begin{equation*}
	y_1^z(x)=\frac{\im}{z}\phi(x)+m_1^z(x),\ \ \ \ m_1^z(x):=\frac{1}{\im z}\begin{cases}\displaystyle\int_{-\infty}^x\e^{\im z(x-s)}(D\phi)(s)\,\d s,&\Im z>0\bigskip\\
	\displaystyle-\int_x^{\infty}\e^{\im z(x-s)}(D\phi)(s)\,\d s,&\Im z<0
	\end{cases},
\end{equation*}
and
\begin{equation*}
	y_2^z(x)=-\frac{\im}{z}\psi(x)+m_2^z(x),\ \ \ \ m_2^z(x):=\frac{1}{\im z}\begin{cases}\displaystyle\int_x^{\infty}\e^{-\im z(x-s)}(D\psi)(s)\,\d s,&\Im z>0\bigskip\\
	\displaystyle-\int_{-\infty}^x\e^{-\im z(x-s)}(D\psi)(s)\,\d s,&\Im z<0\end{cases}.
\end{equation*}
Here, $|zm_1^z(x)|\leq\|D\phi\|_{L^1(\mathbb{R})}$ and $|zm^z_2(x)|\leq\| D\psi\|_{L^1(\mathbb{R})}$ for all $(x,z)\in\mathbb{R}\times(\mathbb{C}\setminus\mathbb{R})$, and
\begin{equation*}
	|m_1^z(x)|\leq\frac{1}{|z|}(E\ast|D\phi|)(x),\ \ \ \ \ \ \ |m_2^z(x)|\leq\frac{1}{|z|}(E\ast |D\psi|)(x),\ \ \ \ \ \ (x,z)\in\mathbb{R}\times(\mathbb{C}\setminus\mathbb{R}).
\end{equation*}
Thus, by Plancherel's theorem, using, $D\phi,D\psi,E\ast|D\phi|,E\ast|D\psi|\in L^1(\mathbb{R})\cap L^2(\mathbb{R}_+)$,
\begin{equation*}
	\|m_1^z\|_{L^2(\mathbb{R}_+)}\leq\frac{1}{|z|}\|\chi_+(E\ast|D\phi|)\|_{L^2(\mathbb{R})}\leq\frac{c}{|z|\sqrt{|\Im z|}}\|D\phi\|_{L^1(\mathbb{R})},\ \ \ \ c>0,
\end{equation*}
where $\chi_+$ is the characteristic function on $\mathbb{R}_+$ and where we used the explicit estimate
\begin{equation*}
	\left|\int_{-\infty}^{\infty}\big(\chi_+(E\ast|D\phi|)\big)(x)\e^{-\im\xi x}\,\d x\right|\leq2\left[\frac{1}{\sqrt{\delta^2+\xi^2}}+\frac{\delta}{\delta^2+\xi^2}\right]\|D\phi\|_{L^1(\mathbb{R})},\ \ \ \xi\in\mathbb{R},\ \ \ \ \delta=|\Im z|>0.
\end{equation*}
A similar estimate holds for $\|m_2^z\|_{L^2(\mathbb{R}_+)}$ and concludes our proof of Proposition \ref{crucial1}.
\end{proof}
Equipped with Proposition \ref{crucial1} we can now solve RHP \ref{master2} and thus prove Theorem \ref{theo1}.
\begin{proof}[Proof of Theorem \ref{theo1}] Fix $t\in\mathbb{R}$ and consider the $2\times 2$ matrix-valued function
\begin{equation}\label{j14}
	{\bf X}(z):=\mathbb{I}+\begin{bmatrix}\displaystyle-\int_0^{\infty}((I-K_t^{\ast})^{-1}\tau_t\psi)(x)(\tau_ty_1^z)(x)\,\d x & \im\big((I-K_t^{\ast})^{-1}\tau_ty_2^z\big)(0)\smallskip\\
	\displaystyle-\im\big((I-K_t)^{-1}\tau_ty_1^z\big)(0) & \displaystyle-\int_0^{\infty}\big((I-K_t)^{-1}\tau_t\phi\big)(x)(\tau_ty_2^z)(x)\,\d x\end{bmatrix},
\end{equation}
defined for $z\in\mathbb{C}\setminus\mathbb{R}$ with $y_k^z,k=1,2$ as in \eqref{sq11} and $(\tau_t y_k^z)(x)=y_k^z(x+t)$. We first show that the first column entries of \eqref{j14} are analytic in $\mathbb{C}\setminus\mathbb{R}$, the second column can afterwards be treated in the same way. Since $\psi\in L^1(\mathbb{R})\cap L^2(\mathbb{R})$ and $y_1^z\in L^2(\mathbb{R})$ for $z\in\mathbb{C}\setminus\mathbb{R}$ by Proposition \ref{crucial1},
\begin{equation}\label{j15}
	\mathbb{C}\setminus\mathbb{R}\ni z\mapsto X^{11}(z)\stackrel{\eqref{j14}}{=}1-\int_0^{\infty}\big((I-K_t^{\ast})^{-1}\tau_t\psi\big)(x)(\tau_ty_1^z)(x)\,\d x
\end{equation}
is well-defined. But by \eqref{i:14} and Cauchy-Schwarz inequality, $(I-K_t^{\ast})^{-1}\tau_t\psi\in L^1(\mathbb{R}_+)$ since $I-K_t^{\ast}$ is invertible on $L^2(\mathbb{R}_+)$ and since $\psi\in L^1(\mathbb{R})\cap L^2(\mathbb{R})$. Hence, using the dominated convergence theorem, the estimate $|y_1^z(x)|\leq\|\phi\|_{L^1(\mathbb{R})}$ and Proposition \ref{crucial1}, we conclude that $\mathbb{C}\setminus\mathbb{R}\ni z\mapsto X^{11}(z)$ in \eqref{j15} is continuous. Moreover, for any triangle $T$ strictly contained in the upper or lower half-plane, by Fubini's theorem,
\begin{equation*}
	\oint_TX^{11}(z)\,\d z=0,
\end{equation*}
so $\mathbb{C}\setminus\mathbb{R}\ni z\mapsto X^{11}(z)$ is analytic by Morera's theorem. Next, look at 
\begin{equation}\label{j16}
	\mathbb{C}\setminus\mathbb{R}\ni z\mapsto X^{21}(z)\stackrel{\eqref{j14}}{=}-\im\big((I-K_t)^{-1}\tau_ty_1^z\big)(0).
\end{equation}
Since $\mathbb{R}_{\geq 0}\ni x\mapsto((I-K_t)^{-1}\tau_t y_1^z)(x)$ is locally $\frac{1}{2}$-H\"older continuous for any $z\in\mathbb{C}\setminus\mathbb{R}$ by \eqref{i:14}, by the Cauchy-Schwarz inequality and Proposition \ref{crucial1}, the function in \eqref{j16} is well-defined. In order to establish its analyticity in the separate half-planes we consider the auxiliary function
\begin{equation}\label{j17}
	\mathbb{C}\setminus\mathbb{R}\ni z\mapsto f_n(z):=\int_0^{\infty}\delta_n(x)\big((I-K_t)^{-1}\tau_ty_1^z\big)(x)\,\d x;\ \ \ \ \delta_n(x):=\frac{2n}{\sqrt{\pi}}\e^{-(nx)^2}.
\end{equation}
Here, $\delta_n\in L^1(\mathbb{R})\cap L^2(\mathbb{R}),n\in\mathbb{Z}_{\geq 1}$ with $\int_0^{\infty}\delta_n(x)\,\d x=1$ and $z\mapsto f_n(z)$ is well-defined by invertibility of $I-K_t$ on $L^2(\mathbb{R}_+)$ and by Proposition \ref{crucial1}. More is true, $(f_n)_{n=1}^{\infty}$ is a sequence of analytic functions, analytic in the separate half-planes: indeed we have, for any $n\in\mathbb{Z}_{\geq 1}$,
\begin{equation}\label{j18}
	f_n(z)\stackrel{\eqref{j17}}{=}\int_0^{\infty}\big((I-K_t^{\ast})^{-1}\delta_n\big)(x)(\tau_ty_1^z)(x)\,\d x,\ \ z\in\mathbb{C}\setminus\mathbb{R},
\end{equation}
where $(I-K_t^{\ast})^{-1}\delta_n\in L^1(\mathbb{R}_+)$ by \eqref{i:14} and invertibility of $I-K_t$ on $L^2(\mathbb{R}_+)$, and where $|y_1^z(x)|\leq\|\phi\|_{L^1(\mathbb{R})}$. So, $\mathbb{C}\setminus\mathbb{R}\ni z\mapsto f_n(z)$ is continuous by the dominated convergence theorem and by Proposition \ref{crucial1}. Moreover, by Fubini's theorem
\begin{equation*}
	\oint_Tf_n(z)\,\d z\stackrel{\eqref{j18}}{=}0
\end{equation*}
for any triangle $T$ strictly contained in the upper or lower half-plane. This implies analyticity of $\mathbb{C}\setminus\mathbb{R}\ni z\mapsto f_n(z)$ by Morera's theorem. Next, using $(I-K_t)^{-1}=I+K_t(I-K_t)^{-1}$, we have
\begin{align}
	&f_n(z)-\im X^{21}(z)=\int_0^{\infty}\delta_n(x)\big((\tau_ty_1^z)(x)-(\tau_ty_1^z)(0)\big)\,\d x\nonumber\\
	&+\int_0^{\infty}\delta_n(x)\left[\int_0^{\infty}\big(K_t(x,y)-K_t(0,y)\big)\big((I-K_t)^{-1}\tau_t y_1^z\big)(y)\,\d y\right]\d x\equiv f_{1n}(z)+f_{2n}(z),\ \ z\in\mathbb{C}\setminus\mathbb{R}.\label{j19}
\end{align}
For any $a,b\in\mathbb{R}$,
\begin{equation}\label{j20}
	\big|(\tau_t y_1^z)(a)-(\tau_ty_1^z)(b)\big|\leq|a-b|\underbrace{\big(|z|\,\|\phi\|_{L^1(\mathbb{R})}+\|\phi\|_{L^{\infty}(\mathbb{R})}\big)}_{=:L_{z,\phi}},\ \ z\in\mathbb{C}\setminus\mathbb{R},
\end{equation}
followed by, for any $(a,b,z)\in\mathbb{R}_{\geq 0}^3$, now using \eqref{n00} and \eqref{i:14},
\begin{equation}\label{j21}
	\big|K_t(a,z)-K_t(b,z)\big|\leq c_t\sqrt{|a-b|}\sqrt{\int_0^{\infty}|(\tau_t\psi)(\lambda+z)|^2\,\d\lambda},\ \ c_t>0,\ \ \ \ z\in\mathbb{C}\setminus\mathbb{R},
\end{equation}
so that all together, for any $z\in\mathbb{C}\setminus\mathbb{R}$,
\begin{equation}\label{j22}
	\big|f_{1n}(z)\big|\stackrel{\eqref{j20}}{\leq}\frac{L_{z,\phi}}{n\sqrt{\pi}},\ \ \ \ \ \ \ \ \ \big|f_{2n}(z)\big|\stackrel{\eqref{j21}}{\leq}\frac{d_t}{\sqrt{n}}\|(I-K_t)^{-1}\tau_ty_1^z\|_{L^2(\mathbb{R}_+)},\ \ \ d_t>0.
\end{equation}
Here, $d_t>0$ is $(z,n)$-independent and we have by the workings in Proposition \ref{crucial1}, $\|(I-K_t)^{-1}\tau_ty_1^z\|_{L^2(\mathbb{R}_+)}\leq 2\|(I-K_t)^{-1}\|\|\phi\|_{L^2(\mathbb{R})}/|\Im z|$. Thus, for any compact $K\subset\mathbb{C}\setminus\mathbb{R}$ in the upper half-plane, say, the above \eqref{j19} and \eqref{j22} yield
\begin{equation*}
	\sup_{z\in K}\big|f_n(z)-\im X^{21}(z)\big|\rightarrow 0\ \ \ \textnormal{as}\ n\rightarrow\infty,
\end{equation*}
i.e. $z\mapsto X^{21}(z)$ is analytic in the upper half-plane, say, by \cite[Theorem $5.2$]{SS}. This proves that \eqref{j14} satisfies property $(1)$ of RHP \ref{master2}. Moving to property $(2)$, we compute via \eqref{sq11}, the dominated convergence and Fubini's theorem, for any $z\in\mathbb{R}$,
\begin{align*}
	\lim_{\epsilon\downarrow 0}X^{11}(z+\im\epsilon)=1\,-&\,r_1(z)\im\e^{\im zt}\int_0^{\infty}\big((I-K_t^{\ast})^{-1}\tau_t\psi\big)(x)\e^{\im zx}\,\d x\\
	&\hspace{1.5cm}+\int_0^{\infty}\e^{-\im zu}\left[\int_0^{\infty}(\tau_{t+u}\phi)(x)\big((I-K_t^{\ast})^{-1}\tau_t\psi\big)(x)\,\d x\right]\d u,\\
	\lim_{\epsilon\downarrow 0}X^{11}(z-\im\epsilon)=1\,+&\,\int_0^{\infty}\e^{-\im zu}\left[\int_0^{\infty}(\tau_{t+u}\phi)(x)\big((I-K_t^{\ast})^{-1}\tau_t\psi\big)(x)\,\d x\right]\d u.
\end{align*}
These limiting values are well-defined and continuous in $z\in\mathbb{R}$ given that $(I-K_t^{\ast})^{-1}\tau_t\psi\in L^1(\mathbb{R}_+)$ and $\phi\in L^1(\mathbb{R})$. In fact, using the algebraic Lemma \ref{new4}, we can simplify them further to obtain for any $z\in\mathbb{R}$, with $X^{11}_{\pm}(z)=\lim_{\epsilon\downarrow 0}X^{11}(z\pm\im\epsilon)$,
\begin{align}
	X_+^{11}(z)=&\,1+\big((I-K_t^{\ast})^{-1}g_z\big)(0)-r_1(z)\im\e^{\im zt}\int_0^{\infty}\big((I-K_t^{\ast})^{-1}\tau_t\psi\big)(x)\e^{\im zx}\,\d x\nonumber\\
	X_-^{11}(z)=&\,1+\big((I-K_t^{\ast})^{-1}g_z\big)(0);\ \ \ \ \ \ \ \ \textnormal{with}\ \ \ \ \ g_z(x):=\int_0^{\infty}K_t^{\ast}(x,u)\e^{-\im zu}\,\d u,\ \ x\in\mathbb{R}_{\geq 0}.\label{j23}
\end{align}
Here, $\mathbb{R}_{\geq 0}\ni x\mapsto g_z(x)$ is well-defined and continuous for all $z\in\mathbb{R}$ by \eqref{i:14} and since $\psi\in L^1(\mathbb{R})\cap L^2(\mathbb{R})$. Moreover, $g_z\in L^2(\mathbb{R}_+)$ by \eqref{n00} and \eqref{i:14}, so $\mathbb{R}_{\geq 0}\ni x\mapsto ((I-K_t^{\ast})^{-1}g_z)(x)$ is continuous and thus $((I-K_t^{\ast})^{-1}g_z)(0)$ in \eqref{j23} bona fide. By similar logic, using this time \eqref{app3} instead of \eqref{app4}, we also find
\begin{align}
	X_+^{22}(z)=&\,1+\big((I-K_t)^{-1}h_z\big)(0);\ \ \ \ \ \ \ \ \textnormal{with}\ \ \ \ \ h_z(x):=\int_0^{\infty}K_t(x,u)\e^{\im zu}\,\d u,\ \ x\in\mathbb{R}_{\geq 0},\nonumber\\
	X_-^{22}(z)=&\,1+\big((I-K_t)^{-1}h_z\big)(0)+r_2(z)\im\e^{-\im zt}\int_0^{\infty}\big((I-K_t)^{-1}\tau_t\phi\big)(x)\e^{-\im zx}\,\d x.\label{j24}
\end{align}
Next, for the off-diagonal entries, $X^{21}(z)$, say, we write $(I-K_t)^{-1}=I+(I-K_t)^{-1}K_t$ and note that
\begin{equation*}
	\mathbb{R}_{\geq 0}\ni x\mapsto (K_t\tau_ty_1^z)(x)=\int_0^{\infty}K_t(x,u)(\tau_ty_1^z)(u)\,\d u
\end{equation*}
is in $L^1(\mathbb{R}_+)\cap L^2(\mathbb{R}_+)$, for all $z\in\mathbb{C}\setminus\mathbb{R}$, by \eqref{n00},\eqref{i:14} and since $\phi\in L^1(\mathbb{R})$. Moreover, by the dominated convergence theorem, with $z\in\mathbb{R}$,
\begin{equation*}
	\lim_{\epsilon\downarrow 0}(K_t\tau_ty_1^{z+\im\epsilon})(x)=r_1(z)\im\e^{\im zt}h_z(x)-\int_0^{\infty}\e^{-\im zv}(K_t\tau_{t+v}\phi)(x)\,\d v\ \ \ \forall\,x\in\mathbb{R}_{\geq 0},
\end{equation*}
noting that $\mathbb{R}_+\ni v\mapsto (K_t\tau_{t+v}\phi)(x)=(K_t\tau_{t+x}\phi)(v)$ is integrable almost everywhere $x\in\mathbb{R}_{\geq 0}$ and that 
\begin{equation*}
	\mathbb{R}_+\ni x\mapsto\int_0^{\infty}\e^{-\im zv}(K_t\tau_{t+v}\phi)(x)\,\d v
\end{equation*}
is integrable and bounded, thus square integrable. Hence, by the dominated convergence theorem, for all $z\in\mathbb{R}$, using $h_z\in L^2(\mathbb{R}_+)$ and Lemma \ref{new3},
\begin{equation}\label{j25}
	X_+^{21}(z)=\im\int_0^{\infty}\big((I-K_t)^{-1}\tau_t\phi\big)(x)\e^{-\im zx}\,\d x+r_1(z)\e^{\im zt}+r_1(z)\e^{\im zt}\big((I-K_t)^{-1}h_z\big)(0),
\end{equation}
which is well-defined and continuous in $z\in\mathbb{R}$. By similar logic,
\begin{equation}\label{j26}
	X_-^{21}(z)=\im\int_0^{\infty}\big((I-K_t)^{-1}\tau_t\phi\big)(x)\e^{-\im zx}\,\d x,
\end{equation}
followed by
\begin{align}
	X_+^{12}(z)=&\,-\im\int_0^{\infty}\big((I-K_t^{\ast})^{-1}\tau_t\psi\big)(x)\e^{\im zx}\,\d x,\nonumber\\
	X_-^{12}(z)=&\,-\im\int_0^{\infty}\big((I-K_t^{\ast})^{-1}\tau_t\psi\big)(x)\e^{\im zx}\,\d x+r_2(z)\e^{-\im zt}+r_2(z)\e^{-\im zt}\big((I-K_t^{\ast})^{-1}g_z\big)(0)\label{j27}
\end{align}
which are all well-defined and continuous functions in $z\in\mathbb{R}$. Using now \eqref{j23},\eqref{j24},\eqref{j25},\eqref{j26} and \eqref{j27} we directly verify the identity
\begin{equation*}
	{\bf X}_-(z)\begin{bmatrix}1-r_1(z)r_2(z) & -r_2(z)\e^{-\im tz}\smallskip\\
	r_1(z)\e^{\im tz} & 1\end{bmatrix}={\bf X}_+(z),\ \ \ z\in\mathbb{R},
\end{equation*}
i.e. \eqref{j14} also satisfies property $(2)$ of RHP \ref{master2}. Finally, by Proposition \ref{crucial1} and Cauchy-Schwarz inequality,
\begin{equation}\label{j28}
	\left\|{\bf X}(z)-\mathbb{I}-\frac{1}{z}\begin{bmatrix}-\im p_0 & q_0^{\ast}\\ q_0 & \im p_0^{\ast}\end{bmatrix}\right\|\leq \frac{c_t}{|z|\sqrt{|\Im z|}},\ \ \ c_t>0,
\end{equation}
which shows that ${\bf X}(z)$ in \eqref{j14} satisfies \eqref{s10} and in turn \eqref{i:14a} as $z\rightarrow\infty$ along any non-horizontal direction. 
Our proof of Theorem \ref{theo1} is now complete.
\end{proof}
\begin{proof}[Proof of Corollary \ref{impcor}] This is immediate from \eqref{i:14a} since $\im X_1^{11}(t,\phi,\psi)=p_0(t)$ and $X_1^{12}(t,\phi,\psi)X_1^{21}(t,\phi,\psi)=q_0(t)q_0^{\ast}(t)$ for all $t\in J$ under the given assumptions. Now use the last two equations in \eqref{r1} and in \eqref{n2} which says $\frac{\d}{\d t}\ln F(t)=p_0(t),t\in J$, subject to the necessary assumptions.
\end{proof}
\begin{proof}[Proof of Corollary \ref{deeper}] Under the assumptions made, repetitive integration by parts yields for $z\notin\mathbb{R}$,
\begin{align*}
	(\tau_ty_1^z)(x)=&\,-\sum_{k=1}^N(-\im)^k(D^{k-1}\tau_t\phi)(x)z^{-k}+m_{1N}^z(x;t),\\
	(\tau_ty_2^z)(x)=&\,-\sum_{k=1}^N\im^k(D^{k-1}\tau_t\psi)(x)z^{-k}+m_{2N}^z(x;t),
\end{align*}
where, with $c_k=c_k(t)>0$,
\begin{equation*}
	\|m_{1N}^z(\cdot;t)\|_{L^2(\mathbb{R}_+)}\leq\frac{c_1}{|z|^N\sqrt{|\Im z|}}\|D^N\phi\|_{L^1(\mathbb{R})},\ \ \ \ \ \ \ \|m_{2N}^z(\cdot;t)\|_{L^2(\mathbb{R}_+)}\leq\frac{c_2}{|z|^N\sqrt{|\Im z|}}\|D^N\psi\|_{L^1(\mathbb{R})}.
\end{equation*}
As in the proof of Theorem \ref{theo1}, the above is sufficient to conclude the validity of \eqref{s17}, after inserting the stated decompositions for $\tau_ty_1^z$ and $\tau_ty_2^z$ into \eqref{j14} and using the Cauchy-Schwarz inequality as well as \eqref{n8}, \eqref{n9}. This completes our proof of the Corollary.
\end{proof}
We now move on to the more specialized Hankel composition operators of Section \ref{seci22}.


\section{Proof of Theorem \ref{theo2}}\label{sec5}
In order to derive \eqref{i:19},\eqref{i:20} and \eqref{i:21} we require two auxiliary results of independent interest.
\begin{prop} Suppose $H_t\in\mathcal{L}(L^2(\mathbb{R}_+))$ defined in \eqref{sp:0} satisfies the conditions in Assumption \ref{ass1} with $\phi:\mathbb{R}\rightarrow\mathbb{C}$ continuously differentiable such that
\begin{equation}\label{c1}
	\lim_{x\rightarrow+\infty}\phi(x)=0,
\end{equation}
and for every $t\in\mathbb{R}$,
\begin{equation}\label{c2}
	\int_0^{\infty}x|(\tau_t\phi)(x)|^2\d x <\infty,\ \ \ \int_0^{\infty}x|(D\tau_t\phi)(x)|^2\d x<\infty,\ \ \ \int_0^{\infty}\sqrt{\int_0^{\infty}|(D\tau_t\phi)(x+y)|^2\d x}\,\d y<\infty.
\end{equation}
Then we have, provided the families $\{\tau_t\phi\}_{t\in\mathbb{R}}$ and $\{D\tau_t\phi\}_{t\in\mathbb{R}}$ are $L^2(\mathbb{R}_+)$ dominated and provided there exist $c,t_0>0$ such that $|q(t,\gamma)|\leq ct^{-1-\epsilon}$ for all $t\geq t_0,\gamma\in[0,1]$ with some $\epsilon>0$,
\begin{equation}\label{c3}
	\ln G(t,\sqrt{\gamma})=\ln G(t,-\sqrt{\gamma})-\int_t^{\infty}q(s,\gamma)\,\d s,\ \ \ \ (t,\gamma)\in\mathbb{R}\times[0,1],
\end{equation}
with branches for the logarithms that are analytically continued along the orbits of the corresponding functions. Here, $q(t,\gamma)=\sqrt{\gamma}((I-\gamma K_t)^{-1}\tau_t\phi)(0)$ and $G(t,\pm\sqrt{\gamma})$ as in \eqref{i:15}.
\end{prop}
\begin{proof} By assumption, $I-\gamma K_t=I-\gamma H_t^2=(I-\sqrt{\gamma}H_t)(I+\sqrt{\gamma}H_t)$ is invertible on $L^2(\mathbb{R}_+)$, so $I\mp\sqrt{\gamma} H_t$ are also invertible and thus $G(t,\pm\sqrt{\gamma})\neq 0$ for all $(t,\gamma)\in\mathbb{R}\times[0,1]$. But $\mathbb{R}\ni t\mapsto H_t$ is differentiable by assumption with $\frac{\d}{\d t}H_t$ trace class, so by Jacobi's formula, after simplification,
\begin{align}
	\frac{\d}{\d t}\ln G(t,\sqrt{\gamma})-\frac{\d}{\d t}\ln G(t,-\sqrt{\gamma})=&\,-2\sqrt{\gamma}\tr_{L^2(\mathbb{R}_+)}\left((I-\gamma K_t)^{-1}\frac{\d}{\d t}H_t\right)\nonumber\\
	=&-2\sqrt{\gamma}\tr_{L^2(\mathbb{R}_+)}\left(\frac{\d}{\d t}H_t\right)-2\sqrt{\gamma}\tr_{L^2(\mathbb{R}_+)}\left((I-\gamma K_t)^{-1}\gamma K_t\frac{\d}{\d t}H_t\right).\label{c4}
\end{align}
Here, by \eqref{c1}, $K_t\frac{\d}{\d t}H_t=-H_t(\tau_t\phi\otimes\tau_t\phi)-H_tDK_t$ where $DK_t$ is trace class on $L^2(\mathbb{R}_+)$ by \eqref{c2}.
\begin{align}
	\textnormal{LHS}\,\eqref{c4}=-2\sqrt{\gamma}\tr_{L^2(\mathbb{R}_+)}\left(\frac{\d}{\d t}H_t\right)&\,+2\sqrt{\gamma}\tr_{L^2(\mathbb{R}_+)}\big((I-\gamma K_t)^{-1}\gamma H_t(\tau_t\phi\otimes\tau_t\phi)\big)\nonumber\\
	&\,+2\sqrt{\gamma}\tr_{L^2(\mathbb{R}_+)}\big((I-\gamma K_t)^{-1}\gamma H_tDK_t\big).\label{c5}
\end{align}
Observe that the kernels of the trace class operators $\frac{\d}{\d t}H_t$ and $(I-\gamma K_t)^{-1}\gamma H_t(\tau_t\phi\otimes\tau_t\phi)$ are continuous on $\mathbb{R}_+^2$ as $\phi$ is continuously differentiable and bounded, and as \eqref{c2} is in place. Thus by \cite[$(3.4)$]{Bris}, \cite[Theorem V.3.1.1]{Du} and \eqref{c1},
\begin{align}
	-2\sqrt{\gamma}\tr_{L^2(\mathbb{R}_+)}&\,\left(\frac{\d}{\d t}H_t\right)+2\sqrt{\gamma}\tr_{L^2(\mathbb{R}_+)}\big((I-\gamma K_t)^{-1}\gamma H_t(\tau_t\phi\otimes\tau_t\phi)\big)=\sqrt{\gamma}(\tau_t\phi)(0)\nonumber\\
	&\hspace{0.5cm}+2\sqrt{\gamma}\big(\gamma K_t(I-\gamma K_t)^{-1}\tau_t\phi\big)(0)=-\sqrt{\gamma}(\tau_t\phi)(0)+2q(t,\gamma),\ \ (t,\gamma)\in\mathbb{R}\times[0,1],\label{c6}
\end{align}
where we also used $K_t=K_t^{\ast}$ and
\begin{equation*}
	\big((I-\gamma K_t)^{-1}\gamma H_t\tau_t\phi\big)(x)=\big(\gamma K_t(I-\gamma K_t)^{-1}\big)(0,x),\ \ x\in\mathbb{R}_+.
\end{equation*}
Next, $DK_t\in\mathcal{C}_1(L^2(\mathbb{R}_+))$ by \eqref{c2} and $(I-\gamma K_t)^{-1}\gamma H_t\in\mathcal{L}(L^2(\mathbb{R}_+))$. Hence by \cite[Corollary $3.8$]{S},
\begin{align}
	&\,\tr_{L^2(\mathbb{R}_+)}\big((I-\gamma K_t)^{-1}\gamma H_tDK_t\big)=\tr_{L^2(\mathbb{R}_+)}\big(D\gamma K_t(I-\gamma K_t)^{-1}H_t\big)\nonumber\\
	=&\,\tr_{L^2(\mathbb{R}_+)}\big(D(I-\gamma K_t)^{-1}H_t\big)-\tr_{L^2(\mathbb{R}_+)}(DH_t)
	=\tr_{L^2(\mathbb{R}_+)}\big(DH_t(I-\gamma K_t)^{-1}\big)-\tr_{L^2(\mathbb{R}_+)}\left(\frac{\d}{\d t}H_t\right),\label{c7}
\end{align}
since $(I-\gamma K_t)^{-1}H_t=(I-\gamma H_t^2)^{-1}H_t=H_t(I-\gamma K_t)^{-1}$ and $DH_t=\frac{\d}{\d t}H_t$. Combining \eqref{c5},\eqref{c6},\eqref{c7} we thus find
\begin{equation*}
	\textnormal{LHS}\,\eqref{c4}=2q(t,\gamma)+2\sqrt{\gamma}\tr_{L^2(\mathbb{R}_+)}\left(\left[\frac{\d}{\d t}H_t\right](I-\gamma K_t)^{-1}\right),
\end{equation*}
and so all together, using \cite[Corollary $3.8$]{S} one more time,
\begin{equation*}
	\frac{\d}{\d t}\ln G(t,\sqrt{\gamma})-\frac{\d}{\d t}\ln G(t,-\sqrt{\gamma})=q(t,\gamma),\ \ \ (t,\gamma)\in\mathbb{R}\times[0,1].
\end{equation*}
Now simply integrate
\begin{equation*}
	-\ln G(t,\sqrt{\gamma})+\ln G(t,-\sqrt{\gamma})\stackrel{\eqref{i:22}}{=}\int_t^{\infty}\left\{\frac{\d}{\d s}\ln G(s,\sqrt{\gamma})-\frac{\d}{\d s}\ln G(s,-\sqrt{\gamma})\right\}\,\d s=\int_t^{\infty}q(s,\gamma)\,\d s
\end{equation*}
where we use the assumption $\|H_t\|_1\rightarrow 0$ as $t\rightarrow+\infty$ in the first equality and the integrability of $t\mapsto q(t,\gamma)$ on $\mathbb{R}_+$ in the second. Recall that $q(t,\gamma)$ is at least continuous in $t\in\mathbb{R}$, see the workings in Corollary \ref{gencoo}. This completes our proof of \eqref{c3}.
\end{proof}
Observe that \eqref{c3} relates $q(t,\gamma)$ to a ratio of two Fredholm determinants. Below we state an alternative representation for the same function.
\begin{prop} Suppose $H_t\in\mathcal{L}(L^2(\mathbb{R}_+))$ defined in \eqref{sp:0} satisfies the conditions in Assumption \ref{ass1} with $\phi:\mathbb{R}\rightarrow\mathbb{C}$ continuously differentiable and $\phi\in L^1(\mathbb{R}_+)$. Assume further that
\begin{equation}\label{c8}
	\lim_{x\rightarrow+\infty}\phi(x)=0,
\end{equation}
and that for every $t\in\mathbb{R}$,
\begin{equation}\label{c9}
	\int_0^{\infty}x|(\tau_t\phi)(x)|^2\,\d x<\infty,\ \ \ \int_0^{\infty}x|(D\tau_t\phi)(x)|^2\,\d x<\infty,\ \ \ \int_0^{\infty}\sqrt{\int_0^{\infty}|(\tau_t\phi)(x+y)|^2\,\d y}\,\d x<\infty.
\end{equation}
Then for all $(t,\gamma)\in\mathbb{R}\times[0,1]$, 
\begin{equation}\label{c10}
	\frac{\d}{\d t}\left[1\mp\sqrt{\gamma}\int_0^{\infty}\big((I\pm\sqrt{\gamma}H_t)^{-1}\tau_t\phi\big)(x)\,\d x\right]=\pm q(t,\gamma)\left[1\mp\sqrt{\gamma}\int_0^{\infty}\big((I\pm\sqrt{\gamma}H_t)^{-1}\tau_t\phi\big)(x)\,\d x\right],
\end{equation}
provided $\{(I-\gamma K_t)^{-1}\tau_t\phi\}_{t\in\mathbb{R}}$,$\{(I\pm\sqrt{\gamma}H_t)^{-1}\tau_tf\}_{t\in\mathbb{R}}$ with $f\in\{\phi,D\phi\}$ and $\{(I-\gamma K_t)^{-1}H_tD\tau_t\phi\}_{t\in\mathbb{R}}$ are $L^1(\mathbb{R}_+)$ dominated and $\{\tau_t\phi\}_{t\in\mathbb{R}},\{D\tau_t\phi\}_{t\in\mathbb{R}}$ are $L^2(\mathbb{R}_+)$ dominated.
\end{prop}
\begin{proof} Since $\phi\in L^1(\mathbb{R}_+)\cap L^2(\mathbb{R}_+)$ and since $I-\gamma K_t$ is invertible, the map $\mathbb{R}_{\geq 0}\ni x\mapsto ((I\pm\sqrt{\gamma}H_t)^{-1}\tau_t\phi)(x)$ is well-defined and by \eqref{c9} in $L^1(\mathbb{R}_+)\cap L^2(\mathbb{R}_+)$. Thus
\begin{equation}\label{c11}
	\mathbb{R}\ni t\mapsto \int_0^{\infty}\big((I\pm\sqrt{\gamma}H_t)^{-1}\tau_t\phi\big)(x)\,\d x
\end{equation}
exists and is moreover differentiable: indeed, by \eqref{c9} and $L^2(\mathbb{R}_+)$ dominance, $DH_t\in\mathcal{L}(L^2(\mathbb{R}_+))$ as well as  $(I\pm\sqrt{\gamma}H_t)^{-1}\tau_t\phi\in W^{1,2}(\mathbb{R}_+)$. Then, using \eqref{c8},
\begin{equation*}
	\big(DH_t(I\pm\sqrt{\gamma}H_t)^{-1}\tau_t\phi\big)(x)=-(\tau_t\phi)(x)\big((I\pm\sqrt{\gamma}H_t)^{-1}\tau_t\phi\big)(0)-\big(H_tD(I\pm\sqrt{\gamma}H_t)^{-1}\tau_t\phi\big)(x),\ \ x\in\mathbb{R}_+,
\end{equation*}
or equivalently,
\begin{align}\label{c12}
	-\big((I\mp\sqrt{\gamma}H_t)DH_t(I\pm\sqrt{\gamma}H_t)^{-1}\tau_t\phi\big)(x)=\big(H_tD\tau_t\phi\big)(x)+(\tau_t\phi)(x)\big((I\pm\sqrt{\gamma}H_t)^{-1}\tau_t\phi\big)(0).
\end{align}
This yields that
\begin{equation*}
	\frac{\partial}{\partial t}\big((I\pm\sqrt{\gamma}H_t)^{-1}\tau_t\phi\big)(x)=\mp\sqrt{\gamma}\big((I\pm\sqrt{\gamma}H_t)^{-1}DH_t(I\pm\sqrt{\gamma}H_t)^{-1}\tau_t\phi\big)(x)+\big((I\pm\sqrt{\gamma}H_t)^{-1}D\tau_t\phi\big)(x)
\end{equation*}
is continuous in $x\in\mathbb{R}_+$ and $L^1(\mathbb{R}_+)$ dominated by the imposed $L^1(\mathbb{R}_+)$ dominance assumptions. Thus, \eqref{c11} is differentiable by the dominated convergence theorem and we have
\begin{equation*}
	\frac{\d}{\d t}\int_0^{\infty}\big((I\pm\sqrt{\gamma}H_t)^{-1}\tau_t\phi\big)(x)\,\d x=\int_0^{\infty}\frac{\partial}{\partial t}\big((I\pm\sqrt{\gamma}H_t)^{-1}\tau_t\phi\big)(x)\,\d x.
\end{equation*}
To evaluate the derivative in the last right hand side we note that by symmetry of $H_t$, $((I\pm\sqrt{\gamma}H_t)^{-1}\tau_t\phi)(x)=(H_t(I\pm\sqrt{\gamma}H_t)^{-1})(0,x)$ for any $x\in\mathbb{R}_+$, and so
\begin{equation*}
	\sqrt{\gamma}\int_0^{\infty}\frac{\partial}{\partial t}\big((I\pm\sqrt{\gamma}H_t)^{-1}\tau_t\phi\big)(x)\,\d x=\mp\int_0^{\infty}\frac{\partial}{\partial t}(I\pm\sqrt{\gamma}H_t)^{-1}(0,x)\,\d x.
\end{equation*}
Consequently
\begin{align*}
	\textnormal{LHS}\,\eqref{c10}=\mp\sqrt{\gamma}\int_0^{\infty}\big((I\pm\sqrt{\gamma}H_t)^{-1}DH_t(I\pm\sqrt{\gamma}H_t)^{-1}\big)(0,x)\,\d x\stackrel{\eqref{c12}}{=}\pm q(t,\gamma)\int_0^{\infty}(I\pm\sqrt{\gamma}H_t)^{-1}(0,x)\,\d x,
\end{align*}
and the remaining integral equals $1\mp\sqrt{\gamma}\int_0^{\infty}((I\pm\sqrt{\gamma}H_t)^{-1}\tau_t\phi)(x)\,\d x$ since $(I
\pm\sqrt{\gamma}H_t)^{-1}=I\mp\sqrt{\gamma}H_t(I\pm\sqrt{\gamma}H_t)^{-1}$ and since $H_t$ is symmetric. The proof of \eqref{c10} is complete.
\end{proof}
Equipped with \eqref{c3} and \eqref{c10} we now prove Theorem \ref{theo2}.
\begin{proof}[Proof of Theorem \ref{theo2}] With $\phi\in L^1(\mathbb{R}_+)$, the map $\mathbb{R}_{\geq 0}\ni x\mapsto 1-\int_x^{\infty}(\tau_t\phi)(y)\,\d y$ is bounded. Moreover, by \eqref{i:18a}, $((I-\gamma K_t)^{-1}\tau_t\phi\in L^1(\mathbb{R}_+)$, so both \eqref{i:16} and \eqref{i:17} are well-defined (note $\gamma_{\circ}\in[0,1]$ for $\gamma\in[0,1]$). To evaluate them, we use \eqref{c3} and \eqref{c10}: first by algebra and definition \eqref{sp:0},
\begin{align}
	\gamma\int_0^{\infty}\big(&(I-\gamma K_t)^{-1}\tau_t\phi\big)(x)\left[1-\int_x^{\infty}(\tau_t\phi)(y)\,\d y\right]\d x=\frac{1}{2}(\sqrt{\gamma}-1)\sqrt{\gamma}\int_0^{\infty}\big((I-\sqrt{\gamma} H_t)^{-1}\tau_t\phi\big)(z)\,\d z\nonumber\\
	&\hspace{1cm}+\frac{1}{2}(\sqrt{\gamma}+1)\sqrt{\gamma}\int_0^{\infty}\big((I+\sqrt{\gamma}H_t)^{-1}\tau_t\phi\big)(z)\,\d z\label{c13}
\end{align}
using also the identity
\begin{equation}\label{c14}
	\int_0^{\infty}(I\pm\sqrt{\gamma}H_t)^{-1}(x,z)\left[1\pm\sqrt{\gamma}\int_0^{\infty}H_t(x,y)\,\d y\right]\d x=1,\ \ z\in\mathbb{R}_+,
\end{equation}
a consequence of the symmetry of $H_t$ and $(I\pm\sqrt{\gamma}H_t)^{-1}=I\mp(I\pm\sqrt{\gamma}H_t)^{-1}\sqrt{\gamma}H_t$. Thus,
\begin{align*}
	1-\gamma&\,\int_0^{\infty}\big((I-\gamma K_t)^{-1}\tau_t\phi\big)(x)\left[1-\int_x^{\infty}(\tau_t\phi)(y)\,\d y\right]\d x\\
	&\stackrel{\eqref{c13}}{=}\frac{1}{2}(1-\sqrt{\gamma})\left[1+\sqrt{\gamma}\int_0^{\infty}\big((I-\sqrt{\gamma}H_t)^{-1}\tau_t\phi\big)(z)\,\d z\right]\\
	&\hspace{2cm}+\frac{1}{2}(1+\sqrt{\gamma})\left[1-\sqrt{\gamma}\int_0^{\infty}\big((I+\sqrt{\gamma}H_t)^{-1}\tau_t\phi\big)(z)\,\d z\right],
\end{align*}
and which yields \eqref{i:19} through \eqref{c3} and \eqref{c10} since $F(t,\gamma)=G(t,\sqrt{\gamma})G(t,-\sqrt{\gamma})$ and since $\|H_t\|\rightarrow 0$ as $t\rightarrow\infty$ in operator norm which allows one to integrate \eqref{c10}. The derivation of \eqref{i:20} is quite similar, instead of \eqref{c13} we have
\begin{align}
	\gamma\int_0^{\infty}\big(&(I-\gamma_{\circ}K_t)^{-1}\tau_t\phi\big)(x)\left[1-\int_x^{\infty}(\tau_t\phi)(y)\,\d y\right]\d x=\frac{1+\sqrt{\gamma_{\circ}}}{2(2-\gamma)}\sqrt{\gamma_{\circ}}\int_0^{\infty}\big((I+\sqrt{\gamma_{\circ}}H_t)^{-1}\tau_t\phi\big)(z)\,\d z\nonumber\\
	&\hspace{1cm}-\frac{1-\sqrt{\gamma_{\circ}}}{2(2-\gamma)}\sqrt{\gamma_{\circ}}\int_0^{\infty}\big((I-\sqrt{\gamma_{\circ}}H_t)^{-1}\tau_t\phi\big)(z)\,\d z\label{c15}
\end{align}
as \eqref{c14} is valid with the replacement $\gamma\mapsto\gamma_{\circ}$. Hence,
\begin{align*}
	1-\gamma\int_0^{\infty}\big(&(I-\gamma_{\circ}K_t)^{-1}\tau_t\phi\big)(x)\left[1-\int_x^{\infty}(\tau_t\phi)(y)\,\d y\right]\d x\\
	&\stackrel{\eqref{c15}}{=}\frac{1-\gamma}{2-\gamma}+\frac{1+\sqrt{\gamma_{\circ}}}{2(2-\gamma)}\left[1-\sqrt{\gamma_{\circ}}\int_0^{\infty}\big((I+\sqrt{\gamma_{\circ}}H_t)^{-1}\tau_t\phi\big)(z)\,\d z\right]\\
	&\hspace{3cm}+\frac{1-\sqrt{\gamma_{\circ}}}{2(2-\gamma)}\left[1+\sqrt{\gamma_{\circ}}\int_0^{\infty}\big((I-\sqrt{\gamma_{\circ}}H_t)^{-1}\tau_t\phi\big)(z)\,\d z\right],
\end{align*}
which yields \eqref{i:20} through \eqref{c3} and \eqref{c10} (simply replace $\gamma\mapsto\gamma_{\circ}$ in the last two identities). Finally, $\mathbb{R}_{\geq 0}\ni x\mapsto\int_x^{\infty}(\tau_t\phi)(y)\,\d y$ is also bounded since $\phi\in L^1(\mathbb{R}_+)$, so \eqref{i:18} is well-defined given $(I-\gamma K_t)^{-1}\tau_t\phi\in L^1(\mathbb{R}_+)$. Next,
\begin{align}
	\gamma\int_0^{\infty}\big(&(I-\gamma K_t)^{-1}\tau_t\phi\big)(x)\left[\frac{1}{2}\int_x^{\infty}(\tau_t\phi)(y)\,\d y\right]\d x\stackrel{\eqref{c14}}{=}\frac{1}{4}\sqrt{\gamma}\int_0^{\infty}\big((I-\sqrt{\gamma}H_t)^{-1}\tau_t\phi\big)(z)\,\d z\nonumber\\
	&\hspace{1cm}-\frac{1}{4}\sqrt{\gamma}\int_0^{\infty}\big((I+\sqrt{\gamma}H_t)^{-1}\tau_t\phi\big)(z)\,\d z\label{c16}
\end{align}
and so
\begin{align*}
	1+\gamma&\int_0^{\infty}\big((I-\gamma K_t)^{-1}\tau_t\phi\big)(x)\left[\frac{1}{2}\int_x^{\infty}(\tau_t\phi)(y)\,\d y\right]\d x\\
	&\stackrel{\eqref{c16}}{=}\frac{1}{2}+\frac{1}{4}\left[1+\sqrt{\gamma}\int_0^{\infty}\big((I-\sqrt{\gamma}H_t)^{-1}\tau_t\phi\big)(x)\,\d x\right]+\frac{1}{4}\left[1-\sqrt{\gamma}\int_0^{\infty}\big((I+\sqrt{\gamma}H_t)^{-1}\tau_t\phi\big)(x)\,\d x\right],
\end{align*}
which equals \eqref{i:21} by \eqref{c3} and \eqref{c10}. The proof of Theorem \ref{theo2} is complete.
\end{proof}

We are now left to prove our last result about Fredholm determinants of additive Hankel composition operators, namely Theorem \ref{theo2a}.
\begin{proof}[Proof of Theorem \ref{theo2a}] By assumption \eqref{i:24}, the coefficients $r_j(z)$ in \eqref{s9} admit analytic extensions to the closed horizontal strip $|\Im z|\leq\epsilon$, by Morera's and Fubini's theorem. Moreover, for $0\leq|\Im z|\leq\epsilon$ we have
\begin{equation}\label{c17}
	|r_1(z)|\leq\int_{-\infty}^{\infty}|\phi(y)|\e^{(\Im z)y}\,\d y\stackrel{\eqref{i:24}}{\leq}\frac{2a}{(a-\epsilon)(a+\epsilon)}<1,
\end{equation}
since $a\geq 2+\epsilon$. A bound analogous to \eqref{c17} holds for $|r_2(z)|$ as well. Next, integrating by parts, for any $z\neq 0$ with $0\leq|\Im z|\leq\epsilon$,
\begin{equation*}
	r_1(z)\stackrel{\eqref{i:24}}{=}-\frac{1}{z}\int_{-\infty}^{\infty}(D\phi)(y)\e^{-\im zy}\,\d y,\ \ \ \ \ r_2(z)\stackrel{\eqref{i:24}}{=}-\frac{1}{z}\int_{-\infty}^{\infty}(D\psi)(y)\e^{\im zy}\,\d y,
\end{equation*}
and thus in the same closed horizontal strip, by \eqref{i:24}, $|zr_j(z)|<1$. Based on these initial observations we now commence the asymptotic analysis of $F(t)$ as $t\rightarrow-\infty$. First, we $\gamma$-modify our initial setup \eqref{n1} and consider the trace class family $K_{t,\gamma}\in\mathcal{L}(L^2(\mathbb{R}_+))$ with kernel
\begin{equation*}
	K_{t,\gamma}(x,y):=\gamma K_t(x,y)=\gamma\int_0^{\infty}\phi(x+z+t)\psi(z+y+t)\,\d z,\ \ \ \gamma\in[0,1],
\end{equation*}
and associated Fredholm determinant $D(t,\gamma)$. Clearly, $D(t,1)=F(t), D(t,0)=1$, and $D(t,\gamma)$ is encoded in the following $\gamma$-modification of RHP \ref{master2}.
\begin{problem}\label{gam1} Fix $(t,\gamma)\in\mathbb{R}\times[0,1]$ and $\phi,\psi$ as in \eqref{i:24}. Find ${\bf X}(z)={\bf X}(z;t,\gamma,\phi,\psi)\in\mathbb{C}^{2\times 2}$ such that
\begin{enumerate}
	\item[(1)] ${\bf X}(z)$ is analytic for $z\in\mathbb{C}\setminus\mathbb{R}$.
	\item[(2)] ${\bf X}(z)$ admits continuous pointwise limits ${\bf X}_{\pm}(z):=\lim_{\epsilon'\downarrow 0}{\bf X}(z\pm\im\epsilon'),z\in\mathbb{R}$ which obey
	\begin{equation*}
		{\bf X}_+(z)={\bf X}_-(z)\begin{bmatrix}1-\gamma r_1(z)r_2(z) & -\sqrt{\gamma}\,r_2(z)\e^{-\im tz}\smallskip\\
		\sqrt{\gamma}\,r_1(z)\e^{\im tz} & 1\end{bmatrix},\ \ z\in\mathbb{R},
	\end{equation*}
	with $r_j(z)$ as in \eqref{s9}.
	\item[(3)] As $z\rightarrow\infty$,
	\begin{equation*}
		{\bf X}(z)=\mathbb{I}+o(1);\ \ \ \ \ \ {\bf X}_1(t,\gamma,\phi,\psi)=\big[X_1^{mn}(t,\gamma,\phi,\psi)\big]_{m,n=1}^2:=\lim_{\substack{z\rightarrow\infty\\ \Im z\not\equiv\textnormal{const.}}}z\big({\bf X}(z)-\mathbb{I}\big).
	\end{equation*}
\end{enumerate}
\end{problem}
Indeed, by assumption and Theorem \ref{theo1}, RHP \ref{gam1} is uniquely solvable for all $(t,\gamma)\in\mathbb{R}\times[0,1]$ and
\begin{equation}\label{c18}
	\frac{\partial}{\partial t}\ln D(t,\gamma)=\im X_1^{11}(t,\gamma,\phi,\psi).
\end{equation}
Second, as outlined in Remark \ref{IIKScon1}, the transformation
\begin{equation}\label{c19}
	{\bf T}(z;t,\gamma,\phi,\psi):={\bf X}(z;t,\gamma,\phi,\psi)\begin{cases}\begin{bmatrix}1 & 0\\ -\sqrt{\gamma}\,r_1(z)\e^{\im tz} & 1\end{bmatrix},&\Im z\in(0,\epsilon)\smallskip\\
	\mathbb{I},&\textnormal{else}
	\end{cases}
\end{equation}
maps RHP \ref{gam1} to an integrable operator RHP of the following type.
\begin{problem}\label{gam2} Let $(t,\gamma)\in\mathbb{R}\times[0,1]$ and $\phi,\psi$ as in the statement of Theorem \ref{theo2a}. The function ${\bf T}(z)={\bf T}(z;t,\gamma,\phi,\psi)\in\mathbb{C}^{2\times 2}$ defined in \eqref{c19} is uniquely determined by the following properties:
\begin{enumerate}
	\item[(1)] ${\bf T}(z)$ is analytic for $z\in\mathbb{C}\setminus\Sigma_{\bf T}$ with $\Sigma_{\bf T}:=\mathbb{R}\cup(\mathbb{R}+\im\epsilon)$ and extends continuously on the closure of $\mathbb{C}\setminus\Sigma_{\bf T}$.
	\item[(2)] The limiting values ${\bf T}_{\pm}(z)=\lim_{\epsilon'\downarrow 0}{\bf T}(z\pm\im\epsilon')$ on $\Sigma_{\bf T}\ni z$ satisfy the jump condition ${\bf T}_+(z)={\bf T}_-(z){\bf G}_{\bf T}(z;t,\gamma,\phi,\psi)$ where the jump matrix ${\bf G}_{\bf T}(z;t,\gamma,\phi,\psi)$ is given by
	\begin{align*}
		{\bf G}_{\bf T}(z;t,\gamma,\phi,\psi)=&\,\begin{bmatrix}1 & 0\\ \sqrt{\gamma}\,r_1(z)\e^{\im tz} & 1\end{bmatrix},\ \ \Im z=\epsilon,\\
		{\bf G}_{\bf T}(z;t,\gamma,\phi,\psi)=&\,\begin{bmatrix}1 & -\sqrt{\gamma}\,r_2(z)\e^{-\im tz}\\ 0 & 1\end{bmatrix},\ \ \Im z=0.
	\end{align*}
	\item[(3)] Since $|zr_j(z)|\rightarrow 0$ as $|z|\rightarrow\infty$ in $|\Im z|\leq\epsilon$ by the Riemann-Lebesgue lemma and since $|\e^{\pm\im tz}|\leq\e^{\epsilon|t|}$ in $|\Im z|\leq\epsilon$, we have uniformly as $z\rightarrow\infty$ in $\mathbb{C}\setminus\Sigma_{\bf T}$,
	\begin{equation}\label{c20}
		{\bf T}(z)=\mathbb{I}+o(1);\ \ \ \ \ {\bf T}_1(t,\gamma,\phi,\psi):=\lim_{\substack{z\rightarrow\infty\\ \Im z\not\equiv\textnormal{const.}}}z\big({\bf T}(z)-\mathbb{I}\big)={\bf X}_1(t,\gamma,\phi,\psi).
	\end{equation}
\end{enumerate}
\end{problem}
Observe that RHP \ref{gam2} is the RHP associated with the trace class integrable operator $M_{t,\gamma}:L^2(\Sigma_{\bf T})\rightarrow L^2(\Sigma_{\bf T})$ with kernel
\begin{equation}\label{c20a}
	M_{t,\gamma}(z,w)=\frac{f_1(z)g_1(w)+f_2(z)g_2(w)}{z-w},\ \ \ \ (z,w)\in\Sigma_{\bf T}\times\Sigma_{\bf T},
\end{equation}
cf. \cite{IIKS}, where the bounded functions
\begin{align}
	f_1(z):=&\,\frac{1}{2\pi\im}\chi_{\mathbb{R}}(z),\hspace{1.8cm} f_2(z):=-\frac{1}{2\pi\im}\e^{\im tz}\chi_{\mathbb{R}+\im\epsilon}(z),\nonumber\\
	g_1(w):=&\,\sqrt{\gamma}r_1(w)\chi_{\mathbb{R}+\im\epsilon}(w),\ \ \ \
	 g_2(w):=\sqrt{\gamma}\,\e^{-\im tw}r_2(w)\chi_{\mathbb{R}}(w),\label{rev1}
\end{align}
are expressed in terms of the characteristic functions $\chi_{\mathbb{R}+\im\epsilon},\chi_{\mathbb{R}}$ on $\Im z=\epsilon,\Im z=0$. Given that RHP \ref{gam1} is uniquely solvable for all $(t,\gamma)\in\mathbb{R}\times[0,1]$, the same applies to RHP \ref{gam2} because of \eqref{c19} and so we can compute ${\bf T}(z)$ by the general formula, cf. \cite{IIKS},
\begin{equation}\label{c21}
	{\bf T}(z)=\mathbb{I}-\int_{\Sigma_{\bf T}}\begin{bmatrix}F_1(\lambda)g_1(\lambda) & F_1(\lambda)g_2(\lambda)\\
	F_2(\lambda)g_1(\lambda) & F_2(\lambda)g_2(\lambda)\end{bmatrix}\frac{\d\lambda}{\lambda-z},\ \ \ z\notin\Sigma_{\bf T},
\end{equation}
where $g_j$ are as above and $F_j$ given on $\Sigma_{\bf T}$ by, independent of the choice of limiting values,
\begin{equation}\label{c22}
	\begin{bmatrix}F_1\\ F_2\end{bmatrix}=(I-M_{t,\gamma})^{-1}\begin{bmatrix}f_1\\ f_2\end{bmatrix}={\bf T}_{\pm}\begin{bmatrix}f_1\\ f_2\end{bmatrix}.
\end{equation}
Next, $\frac{\partial}{\partial t}M_{t,\gamma}=\im f_2\otimes g_2\in\mathcal{C}_1(L^2(\Sigma_{\bf T}))$ and so, using the differentiability of $t\mapsto M_{t,\gamma}\in\mathcal{L}(L^2(\Sigma_{\bf T}))$ and the trace identity $\tr{\bf T}_1=\tr{\bf X}_1=0$, we obtain for the Fredholm determinant $H(t,\gamma)$ of $M_{t,\gamma}$ on $L^2(\Sigma_{\bf T})$,
\begin{equation*}
	\frac{\partial}{\partial t}\ln H(t,\gamma)=-\im\tr_{L^2(\Sigma_{\bf T})}\big((I-M_{t,\gamma})^{-1}f_2\otimes g_2\big)\stackrel{\eqref{c22}}{=}-\im\tr_{L^2(\Sigma_{\bf T})}(F_2\otimes g_2)\stackrel[\eqref{c21}]{\eqref{c20}}{=}\im X_1^{11}(t,\gamma,\phi,\psi).
\end{equation*}
Consequently, comparing the last equality to \eqref{c18}, we have found
\begin{equation*}
	\ln D(t,\gamma)=\ln H(t,\gamma)+\eta(\gamma),\ \ \ \ \ \ (t,\gamma)\in\mathbb{R}\times[0,1],
\end{equation*}
with a $t$-independent term $\eta(\gamma)$. However, $D(t,\gamma)\rightarrow1$ as $t\rightarrow+\infty$ by \eqref{i:23} and since RHP \ref{gam2} is directly related to a small norm problem as $t\rightarrow+\infty$ (simply move ${\bf G}_{\bf T}$ from $\Im z=0$ to $\Im z=-\epsilon$ and keep it on $\Im z=\epsilon$ as is), we have, by general theory \cite{DZ}, $T_1^{22}(t,\gamma,\phi,\psi)\rightarrow 0$ exponentially fast as $t\rightarrow+\infty$ for all $\gamma\in[0,1]$ and so $H(t,\gamma)\rightarrow 1$ as well for $t\rightarrow+\infty$. In short, 
\begin{equation}\label{c23}
	\ln D(t,\gamma)=\ln H(t,\gamma),\ \ \ \ \ \ (t,\gamma)\in\mathbb{R}\times[0,1],
\end{equation}
which identifies our $\gamma$-modified determinant $D(t,\gamma)$ as the Fredholm determinant of the integrable operator \eqref{c20a}. In particular we can now compute $F(t)$ in terms of the solution of RHP \ref{gam2},
\begin{align}
	\ln F(t)=\int_0^1\frac{\partial}{\partial\gamma}&\ln D(t,\gamma)\,\d\gamma\stackrel{\eqref{c23}}{=}\int_0^1\frac{\partial}{\partial\gamma}\ln H(t,\gamma)\,\d\gamma=-\frac{1}{2}\int_0^1\tr_{L^2(\Sigma_{\bf T})}\big((I-M_{t,\gamma})^{-1}M_{t,\gamma}\big)\frac{\d\gamma}{\gamma}\nonumber\\
	=&\,-\frac{1}{2}\int_0^1\left[\int_{\Sigma_{\bf T}}R_{t,\gamma}(\lambda,\lambda)\,\d\lambda\right]\frac{\d\gamma}{\gamma},\ \ \ \ R_{t,\gamma}:=(I-M_{t,\gamma})^{-1}-I\in\mathcal{L}(L^2(\Sigma_{\bf T})),\label{c24}
\end{align}
where, by general theory \cite[$(5.5)$]{IIKS},
\begin{equation*}
	R_{t,\gamma}(\lambda,\lambda)=F_1'(\lambda)G_1(\lambda)+F_2'(\lambda)G_2(\lambda),\ \ \ \lambda\in\Sigma_{\bf T},
\end{equation*}
with $F_j'=DF_j$ as above in \eqref{c22} and $G_j$ given on $\Sigma_{\bf T}$ by, independent of the choice of limiting values,
\begin{equation}\label{c24a}
	\begin{bmatrix}G_1\\ G_2\end{bmatrix}=\big({\bf T}_{\pm}^{-1}\big)^{\top}\begin{bmatrix}g_1\\ g_2\end{bmatrix}.
\end{equation}
The double integral representation \eqref{c24} will serve as starting point for our asymptotic analysis of $F(t)$, and we now commence the nonlinear steepest descent analysis of RHP \ref{gam2}, equivalently of RHP \ref{gam1}, see \eqref{c19}. First define
\begin{equation}\label{c25}
	g(z;\gamma):=-\frac{1}{2\pi\im}\int_{-\infty}^{\infty}\ln\big(1-\gamma r_1(\lambda)r_2(\lambda)\big)\frac{\d\lambda}{\lambda-z},\ \ \ z\in\mathbb{C}\setminus\mathbb{R},
\end{equation}
seeing that the improper integral is absolutely convergent since $|r_j(\lambda)|<1$ and $|\lambda r_j(\lambda)|<1$ on $\mathbb{R}$ with $\gamma\in[0,1]$. Moreover, $\mathbb{R}\ni\lambda\mapsto \ln(1-\gamma r_1(\lambda)r_2(\lambda))$ is locally H\"older continuous by the continuous differentiability of $\mathbb{R}\ni \lambda\mapsto r_j(\lambda)$ and since $|r_j(\lambda)|<1$ on $\mathbb{R}$. Thus,
\begin{align*}
	g_+(z;\gamma)-g_-(z;\gamma)=&\,-\ln\big(1-\gamma r_1(z)r_2(z)\big),\ \ z\in\mathbb{R};\ \ \ \ \ \ \ \ \ \ \ \ g_{\pm}(z;\gamma):=\lim_{\epsilon'\downarrow 0}g(z\pm\im\epsilon';\gamma), \\
	g(z;\gamma)=&\,\frac{1}{2\pi\im z}\int_{-\infty}^{\infty}\ln\big(1-\gamma r_1(\lambda)r_2(\lambda)\big)\,\d\lambda+o\big(z^{-1}\big),\ \ z\rightarrow\infty,\ z\notin\mathbb{R}.
\end{align*}
Here, the asymptotics at infinity are proven by writing $\frac{1}{\lambda-z}=-\frac{1}{z}+\frac{\lambda}{z(\lambda-z)},\lambda\in\mathbb{R},z\notin\mathbb{R}$ and by exploiting the bounds $|r_j(\lambda)|,|\lambda r_j(\lambda)|<1$ on $\mathbb{R}$ in combination with the inequalities \cite[$4.5.2,4.5.6$]{NIST}. Moving ahead we use the $g$-function \eqref{c25} and transform RHP \ref{gam1} as follows: introduce
\begin{equation}\label{c26}
	{\bf Y}(z;t,\gamma,\phi,\psi):={\bf X}(z;t,\gamma,\phi,\psi)\e^{g(z;\gamma)\sigma_3},\ \ \ z\in\mathbb{C}\setminus\mathbb{R},
\end{equation}
and obtain the following problem.
\begin{problem}\label{traf1} Let $(t,\gamma)\in\mathbb{R}\times[0,1]$ and $\phi,\psi$ as in the statement of Theorem \ref{theo2a}. The function ${\bf Y}(z)={\bf Y}(z;t,\gamma,\phi,\psi)\in\mathbb{C}^{2\times 2}$ defined in \eqref{c26} is uniquely determined by the following properties:
\begin{enumerate}
	\item[(1)] ${\bf Y}(z)$ is analytic for $z\in\mathbb{C}\setminus\mathbb{R}$ and extends continuously on the closure of $\mathbb{C}\setminus\mathbb{R}$.
	\item[(2)] The continuous limiting values ${\bf Y}_{\pm}(z)=\lim_{\epsilon'\downarrow 0}{\bf Y}(z\pm\im\epsilon')$ on $\mathbb{R}\ni z$ satisfy the jump condition ${\bf Y}_+(z)={\bf Y}_-(z){\bf G}_{\bf Y}(z;t,\gamma,\phi,\psi)$ where the jump matrix ${\bf G}_{\bf Y}(z;t,\gamma,\phi,\psi)$ is given by
	\begin{equation*}
		{\bf G}_{\bf Y}(z;t,\gamma,\phi,\psi)=\begin{bmatrix}1 & -\eta_2(z;\gamma)\e^{-\im tz-2g_+(z;\gamma)}\smallskip\\ 
		\eta_1(z;\gamma)\e^{\im tz+2g_-(z;\gamma)} & 1-\gamma r_1(z)r_2(z)\end{bmatrix},\ z\in\mathbb{R}, 
	\end{equation*}
	using
	\begin{equation*}
		\eta_k(z;\gamma):=\frac{\sqrt{\gamma}\,r_k(z)}{1-\gamma r_1(z)r_2(z)},\ \ \ |\Im z|\leq\epsilon.
	\end{equation*}
	\item[(3)] Uniformly as $z\rightarrow\infty$,
	\begin{equation}\label{c27}
		{\bf Y}(z)=\mathbb{I}+o(1);
	\end{equation}
	where
	\begin{equation*}
		 {\bf Y}_1(t,\gamma,\phi,\psi):=\lim_{\substack{z\rightarrow\infty\\ \Im z\not\equiv\textnormal{const.}}}z\big({\bf Y}(z)-\mathbb{I}\big)={\bf X}_1(t,\gamma,\phi,\psi)+\frac{\sigma_3}{2\pi\im}\int_{-\infty}^{\infty}\ln\big(1-\gamma r_1(\lambda)r_2(\lambda)\big)\,\d\lambda.
	\end{equation*}
\end{enumerate}
\end{problem}
Next, we use that $\eta_k(z;\gamma)$ extends analytically to the closed horizontal strip $|\Im z|\leq\epsilon$ and that $|z\eta_k(z;\gamma)|$ tends to zero as $|z|\rightarrow\infty$ in the same strip. Moreover, ${\bf G}_{\bf Y}$ admits the factorization
\begin{equation*}
	{\bf G}_{\bf Y}(z;t,\gamma,\phi,\psi)=\begin{bmatrix}1 & 0\\ \eta_1(z;\gamma)\e^{\im tz+2g_-(z;\gamma)} & 1\end{bmatrix}\begin{bmatrix}1 & -\eta_2(z;\gamma)\e^{-\im tz-2g_+(z;\gamma)}\\ 0 & 1\end{bmatrix},\ \ z\in\mathbb{R},
\end{equation*}
so we can open lenses and transform RHP \ref{traf1} as follows. Introduce
\begin{equation}\label{c28}
	{\bf S}(z;t,\gamma,\phi,\psi):={\bf Y}(z;t,\gamma,\phi,\psi)\begin{cases}\begin{bmatrix}1 & \eta_2(z;\gamma)\e^{-\im tz-2g(z;\gamma)}\\ 0 & 1\end{bmatrix},&\Im z\in(0,\epsilon)\smallskip\\
	\begin{bmatrix}1 & 0\\ \eta_1(z;\gamma)\e^{\im tz+2g(z;\gamma)} & 1\end{bmatrix},&\Im z\in(-\epsilon,0)\smallskip\\
	\mathbb{I},&\textnormal{else}
	\end{cases}
\end{equation}
and note that RHP \ref{traf1} is transformed to the problem below
\begin{problem}\label{traf2} Let $t<0,\gamma\in[0,1]$ and $\phi,\psi$ as in the statement of Theorem \ref{theo2a}. The function ${\bf S}(z)={\bf S}(z;t,\gamma,\phi,\psi)\in\mathbb{C}^{2\times 2}$ defined in \eqref{c28} is uniquely determined by the following properties:
\begin{enumerate}
	\item[(1)] ${\bf S}(z)$ is analytic for $z\in\mathbb{C}\setminus\{\Im z=\pm\epsilon\}$ and extends continuously on the closure of $\mathbb{C}\setminus\{\Im z=\pm\epsilon\}$.
	\item[(2)] The continuous limiting values ${\bf S}_{\pm}(z):=\lim_{\epsilon'\downarrow 0}{\bf S}(z\pm\im\epsilon')$ on $\Sigma_{\bf S}=\{\Im z=\pm\epsilon\}\ni z$ satisfy the jump condition ${\bf S}_+(z)={\bf S}_-(z){\bf G}_{\bf S}(z;t,\gamma,\phi,\psi)$ where the jump matrix ${\bf G}_{\bf S}(z;t,\gamma,\phi,\psi)$ is given by
	\begin{align*}
		{\bf G}_{\bf S}(z;t,\gamma,\phi,\psi)=&\,\begin{bmatrix}1 & -\eta_2(z;\gamma)\e^{-\im tz-2g(z;\gamma)}\\ 0 & 1\end{bmatrix},\ \ \Im z=\epsilon,\\
		{\bf G}_{\bf S}(z;t,\gamma,\phi,\psi)=&\,\begin{bmatrix}1 & 0\\ \eta_1(z;\gamma)\e^{\im tz+2g(z;\gamma)} & 1\end{bmatrix},\ \ \Im z=-\epsilon.
	\end{align*}
	\item[(3)] Uniformly as $z\rightarrow\infty$ in $\mathbb{C}\setminus\Sigma_{\bf S}$, using the Riemann-Lebesgue lemma for $\eta_k$ in $|\Im z|\leq\epsilon$,
	\begin{equation}\label{c29}
		{\bf S}(z)=\mathbb{I}+o(1);\ \ \ \ \ {\bf S}_1(t,\gamma,\phi,\psi):=\lim_{\substack{z\rightarrow\infty\\ \Im z\not\equiv\textnormal{const.}}}z\big({\bf S}(z)-\mathbb{I}\big)={\bf Y}_1(t,\gamma,\phi,\psi).
	\end{equation}
\end{enumerate}
\end{problem}
Since $|\e^{\mp\im tz}|=\e^{t\epsilon}$ for $\Im z=\pm\epsilon$, and since the $t$-independent functions $\eta_k(z;\gamma)\e^{\mp 2g(z)}$ are bounded on $\Sigma_{\bf S}$ for all $\gamma\in[0,1]$, we immediately arrive at the below small norm estimate for ${\bf G}_{\bf S}(z;t,\gamma,\phi,\psi)$: subject to \eqref{i:24}, for any $\epsilon>0$ there exist $t_0=t_0(\epsilon),c=c(\epsilon)>0$ such that for all $(-t)\geq t_0$ and all $\gamma\in[0,1]$,
\begin{equation}\label{c30}
	\|{\bf G}_{\bf S}(\cdot;t,\gamma,\phi,\psi)-\mathbb{I}\|_{L^{\infty}(\Sigma_{\bf S})}\leq c\sqrt{\gamma}\,\e^{-\epsilon|t|},\ \ \|{\bf G}_{\bf S}(\cdot;t,\gamma,\phi,\psi)-\mathbb{I}\|_{L^2(\Sigma_{\bf S})}\leq c\sqrt{\gamma}\,\e^{-\epsilon|t|}.
\end{equation}
In turn, by general theory \cite{DZ}, RHP \ref{traf2} is asymptotically solvable, namely there exist $t_0=t_0(\epsilon)>0$ such that the RHP for ${\bf S}(z)$ is uniquely solvable in $L^2(\Sigma_{\bf S})$ for all $(-t)\geq t_0$ and all $\gamma\in[0,1]$. We can compute the solution of the same problem iteratively through the integral representation
\begin{equation}\label{c31}
	{\bf S}(z;t,\gamma,\phi,\psi)=\mathbb{I}+\frac{1}{2\pi\im}\int_{\Sigma_{\bf S}}{\bf S}_-(\lambda;t,\gamma,\phi,\psi)\big({\bf G}_{\bf S}(\lambda;t,\gamma,\phi,\psi)-\mathbb{I}\big)\frac{\d\lambda}{\lambda-z},\ \ \ \ z\in\mathbb{C}\setminus\Sigma_{\bf S},
\end{equation}
using that, for all $(-t)\geq t_0$ and $\gamma\in[0,1]$, with $c=c(\epsilon)>0$,
\begin{equation}\label{c32}
	\|{\bf S}_-(\cdot;t,\gamma,\phi,\psi)-\mathbb{I}\|_{L^2(\Sigma_{\bf S})}\leq c\sqrt{\gamma}\,\e^{-\epsilon|t|}.
\end{equation}
The above completes the nonlinear steepest descent analysis of RHP \ref{gam2} and we are now left to extract all relevant information. First, choosing right-sided limits for definiteness, we have for $\lambda\in\Sigma_{\bf T}$ by \eqref{c22}, \eqref{c24a} and \eqref{c19},\eqref{c26},\eqref{c28},
\begin{equation*}
	\begin{bmatrix}F_1(\lambda)\\ F_2(\lambda)\end{bmatrix}={\bf T}_-(\lambda)\begin{bmatrix}f_1(\lambda)\\ f_2(\lambda)\end{bmatrix}={\bf S}_-(\lambda)\big\{{\bf A}(\lambda)\chi_{\mathbb{R}}(\lambda)+{\bf B}(\lambda)\chi_{\mathbb{R}+\im\epsilon}(\lambda)\big\}\begin{bmatrix}f_1(\lambda)\\ f_2(\lambda)\end{bmatrix},
\end{equation*}
followed by
\begin{equation*}
	\begin{bmatrix}G_1(\lambda)\\ G_2(\lambda)\end{bmatrix}=({\bf T}_-^{-1})^{\top}(\lambda)\begin{bmatrix}g_1(\lambda)\\ g_2(\lambda)\end{bmatrix}=({\bf S}_-^{-1})^{\top}(\lambda)\big\{({\bf A}^{-1}(\lambda))^{\top}\chi_{\mathbb{R}}(\lambda)+({\bf B}^{-1}(\lambda))^{\top}\chi_{\mathbb{R}+\im\epsilon}(\lambda)\big\}\begin{bmatrix}g_1(\lambda)\\ g_2(\lambda)\end{bmatrix},
\end{equation*}
where we suppress all $(t,\gamma,\phi,\psi)$-dependencies from our notation. Here,
\begin{equation*}
	{\bf A}(\lambda):=\e^{-g_-(\lambda)\sigma_3}\begin{bmatrix}1 & 0\\ -\eta_1(\lambda)\e^{\im t\lambda} & 1\end{bmatrix},\ \ {\bf B}(\lambda):=\e^{-g(\lambda)\sigma_3}\begin{bmatrix}(1-\gamma r_1(\lambda)r_2(\lambda))^{-1} & -\eta_2(\lambda)\e^{-\im t\lambda}\smallskip\\
	-\sqrt{\gamma}\,r_1(\lambda)\e^{\im t\lambda} & 1\end{bmatrix},\ \ \lambda\in\Sigma_{\bf T}.
\end{equation*}
Hence, the innermost integral in \eqref{c24} becomes
\begin{align}
	\int_{\Sigma_{\bf T}}&R_{t,\gamma}(\lambda,\lambda)\,\d\lambda=\int_{\Sigma_{\bf T}}\begin{bmatrix}F_1'(\lambda)\\ F_2'(\lambda)\end{bmatrix}^{\top}\begin{bmatrix}G_1(\lambda)\\ G_2(\lambda)\end{bmatrix}\d\lambda=\int_{\Sigma_{\bf T}}\begin{bmatrix}f_1'(\lambda)\\ f_2'(\lambda)\end{bmatrix}^{\top}\begin{bmatrix}g_1(\lambda)\\ g_2(\lambda)\end{bmatrix}\d\lambda+\int_{\Sigma_{\bf T}}\begin{bmatrix}f_1(\lambda)\\ f_2(\lambda)\end{bmatrix}^{\top}{\bf E}(\lambda)\begin{bmatrix}g_1(\lambda)\\ g_2(\lambda)\end{bmatrix}\d\lambda\nonumber\\
	&\,\,\,+\int_{\mathbb{R}}\begin{bmatrix}f_1(\lambda)\\ f_2(\lambda)\end{bmatrix}^{\top}\big({\bf A}^{\top}(\lambda)\big)'\big({\bf A}^{-1}(\lambda)\big)^{\top}\begin{bmatrix}g_1(\lambda)\\ g_2(\lambda)\end{bmatrix}\d\lambda+\int_{\mathbb{R}+\im\epsilon}\begin{bmatrix}f_1(\lambda)\\ f_2(\lambda)\end{bmatrix}^{\top}\big({\bf B}^{\top}(\lambda)\big)'\big({\bf B}^{-1}(\lambda)\big)^{\top}\begin{bmatrix}g_1(\lambda)\\ g_2(\lambda)\end{bmatrix}\d\lambda,\label{c33}
\end{align}
with ${\bf E}(\lambda)$ as shorthand for
\begin{equation*}
	{\bf E}(\lambda):=\big\{{\bf A}^{\top}(\lambda)\chi_{\mathbb{R}}(\lambda)+{\bf B}^{\top}(\lambda)\chi_{\mathbb{R}+\im\epsilon}(\lambda)\big\}\big({\bf S}_-^{\top}(\lambda)\big)'\big({\bf S}_-^{-1}(\lambda)\big)^{\top}\big\{\big({\bf A}^{-1}(\lambda)\big)^{\top}\chi_{\mathbb{R}}(\lambda)+\big({\bf B}^{-1}(\lambda)\big)^{\top}\chi_{\mathbb{R}+\im\epsilon}(\lambda)\big\}.
\end{equation*}
Out of the four integrals in \eqref{c33} the first one equates to zero by the aforementioned choices for $f_j$ and $g_j$, compare \eqref{rev1}. For the second integral in \eqref{c33} we instead record the following rough estimate, a consequence of \eqref{c31},\eqref{c32} and the particular triangular shape of ${\bf G}_{\bf S}$: for any $\epsilon>0$ there exist $t_0=t_0(\epsilon)>0$ and $c=c(\epsilon)>0$ so that
\begin{equation}\label{c34}
	\left|\int_{\Sigma_{\bf T}}\begin{bmatrix}f_1(\lambda)\\ f_2(\lambda)\end{bmatrix}^{\top}{\bf E}(\lambda)\begin{bmatrix}g_1(\lambda)\\ g_2(\lambda)\end{bmatrix}\d\lambda\right|\leq c\gamma|t|^{-1}\ \ \ \forall\,(-t)\geq t_0,\ \ \gamma\in[0,1].
\end{equation}
The last two integrals in \eqref{c33} yield non-trivial contributions: by definition of all underlying quantities, for the fourth integral in \eqref{c33},
\begin{align}
	\int_{\mathbb{R}+\im\epsilon}&\begin{bmatrix}f_1(\lambda)\\ f_2(\lambda)\end{bmatrix}^{\top}\big({\bf B}^t(\lambda)\big)'\big({\bf B}^{-1}(\lambda)\big)^{\top}\begin{bmatrix}g_1(\lambda)\\ g_2(\lambda)\end{bmatrix}\d\lambda=\frac{\gamma}{2\pi\im}\int_{\mathbb{R}+\im\epsilon}r_1(\lambda)\e^{\im t\lambda}\frac{\d}{\d\lambda}\left[\frac{r_2(\lambda)\e^{-\im t\lambda}}{1-\gamma r_1(\lambda)r_2(\lambda)}\right]\d\lambda\nonumber\\
	&-\frac{\gamma}{2\pi^2}\int_{\mathbb{R}+\im\epsilon}\frac{r_1(\lambda)r_2(\lambda)}{1-\gamma r_1(\lambda)r_2(\lambda)}\left[\int_{-\infty}^{\infty}\ln\big(1-\gamma r_1(z)r_2(z)\big)\frac{\d z}{(z-\lambda)^2}\right]\d\lambda,\label{c35}
\end{align}
where we integrate by parts in the first term, using $|zr_j(z)|<1$ for $|\Im z|\leq\epsilon$, 
\begin{align*}
	\frac{\gamma}{2\pi\im}&\int_{\mathbb{R}+\im\epsilon}r_1(\lambda)\e^{\im t\lambda}\frac{\d}{\d\lambda}\left[\frac{r_2(\lambda)\e^{-\im t\lambda}}{1-\gamma r_1(\lambda)r_2(\lambda)}\right]\d\lambda=-\frac{\gamma}{2\pi\im}\int_{\mathbb{R}+\im\epsilon}\frac{r_1'(\lambda)r_2(\lambda)}{1-\gamma r_1(\lambda)r_2(\lambda)}\,\d\lambda\\
	&+\frac{\gamma t}{2\pi}\frac{\d}{\d\gamma}\left[\int_{-\infty}^{\infty}\ln\big(1-\gamma r_1(\lambda)r_2(\lambda)\big)\d\lambda\right]=-\gamma t\frac{\d}{\d\gamma}s(0,\gamma)+\frac{\gamma}{2\pi\im}\frac{\d}{\d\gamma}\int_{-\infty}^{\infty}\frac{r_1'(\lambda)}{r_1(\lambda)}\ln\big(1-\gamma r_1(\lambda)r_2(\lambda)\big)\,\d\lambda,
\end{align*}
with $s(x,\gamma):=-\frac{1}{2\pi}\int_{-\infty}^{\infty}\ln(1-\gamma r_1(\lambda)r_2(\lambda))\e^{\im\lambda x}\d\lambda,x\in\mathbb{R}$, after collapsing $\mathbb{R}+\im\epsilon$ to $\mathbb{R}$. The second term in \eqref{c35} we rewrite with the help of the integral identity $(z-\lambda)^{-2}=-\int_0^{\infty}\e^{\im(\lambda-z)x}x\,\d x$, for $z\in\mathbb{R},\lambda\in\mathbb{R}+\im\epsilon$,
\begin{equation*}
	-\frac{\gamma}{2\pi^2}\int_{\mathbb{R}+\im\epsilon}\frac{r_1(\lambda)r_2(\lambda)}{1-\gamma r_1(\lambda)r_2(\lambda)}\left[\int_{-\infty}^{\infty}\ln\big(1-\gamma r_1(z)r_2(z)\big)\frac{\d z}{(z-\lambda)^2}\right]\d\lambda=-2\gamma\int_0^{\infty}x\left[\frac{\partial}{\partial\gamma}s(x,\gamma)\right]s(-x,\gamma)\,\d x,
\end{equation*}
so that all together for \eqref{c35},
\begin{align}
	\int_{\mathbb{R}+\im\epsilon}\begin{bmatrix}f_1(\lambda)\\ f_2(\lambda)\end{bmatrix}^{\top}\big({\bf B}^t(\lambda)\big)'\big({\bf B}^{-1}(\lambda)\big)^{\top}\begin{bmatrix}g_1(\lambda)\\ g_2(\lambda)\end{bmatrix}\d\lambda=-\gamma t&\frac{\d}{\d\gamma}s(0,\gamma)-2\gamma\int_0^{\infty}z\left[\frac{\partial}{\partial\gamma}s(x,\gamma)\right]s(-x,\gamma)\,\d x\nonumber\\
	&+\frac{\gamma}{2\pi\im}\frac{\d}{\d\gamma}\int_{-\infty}^{\infty}\frac{r_1'(\lambda)}{r_1(\lambda)}\ln\big(1-\gamma r_1(\lambda)r_2(\lambda)\big)\,\d\lambda\label{c36}.
\end{align}
Finally, for the third integral in \eqref{c33}, by definition of all underlying quantities,
\begin{align}
	\int_{\mathbb{R}}\begin{bmatrix}f_1(\lambda)\\ f_2(\lambda)\end{bmatrix}^{\top}&\big({\bf A}^{\top}(\lambda)\big)'\big({\bf A}^{-1}(\lambda)\big)^{\top}\begin{bmatrix}g_1(\lambda)\\ g_2(\lambda)\end{bmatrix}\d\lambda=
	-\frac{\gamma}{2\pi\im}\int_{\mathbb{R}}\frac{r_2(\lambda)\e^{-\im t\lambda}}{1-\gamma r_1(\lambda)r_2(\lambda)}\frac{\d}{\d\lambda}\Big[r_1(\lambda)\e^{\im t\lambda}\Big]\d\lambda\nonumber\\
	&\hspace{1.75cm}-\frac{\gamma}{2\pi^2}\int_{\mathbb{R}}\frac{r_1(\lambda)r_2(\lambda)}{1-\gamma r_1(\lambda)r_2(\lambda)}\frac{\d}{\d\lambda}\left[\textnormal{pv}\int_{\mathbb{R}}\ln\big(1-\gamma r_1(z)r_2(z)\big)\frac{\d z}{z-\lambda}\right]\d\lambda,\label{c37}
\end{align}
where we evaluate the derivative in the first term and symmetrize,
\begin{align*}
	-\frac{\gamma}{2\pi\im}\int_{\mathbb{R}}\frac{r_2(\lambda)\e^{-\im t\lambda}}{1-\gamma r_1(\lambda)r_2(\lambda)}\frac{\d}{\d\lambda}\Big[r_1(\lambda)\e^{\im t\lambda}\Big]\d\lambda=-\gamma t\frac{\d}{\d\gamma}s(0,\gamma)-\frac{\gamma}{2\pi\im}\frac{\d}{\d\gamma}\int_{-\infty}^{\infty}\frac{r_2'(\lambda)}{r_2(\lambda)}\ln\big(1-\gamma r_1(\lambda)r_2(\lambda)\big)\,\d\lambda.
\end{align*}
For the second term in \eqref{c37} we recall Plemelj's formula,
\begin{equation*}
	\lim_{\substack{\mu\rightarrow\lambda\in\mathbb{R}\\ \Im\mu<0}}\int_{-\infty}^{\infty}\ln\big(1-\gamma r_1(z)r_2(z)\big)\frac{\d z}{z-\mu}=-\im\pi\ln\big(1-\gamma r_1(\lambda)r_2(\lambda)\big)+\textnormal{pv}\int_{-\infty}^{\infty}\ln\big(1-\gamma r_1(z)r_2(z)\big)\frac{\d z}{z-\lambda},
\end{equation*}
and notice that by the properties of $r_j$ (differentiability on $\mathbb{R}$ together with $|r_j(z)|,|zr_j(z)|<1$ on $\mathbb{R}$) and integration by parts,
\begin{equation*}
	\int_{\mathbb{R}}\frac{r_1(\lambda)r_2(\lambda)}{1-\gamma r_1(\lambda)r_2(\lambda)}\frac{\d}{\d\lambda}\Big[\ln\big(1-\gamma r_1(\lambda)r_2(\lambda)\big)\Big]\,\d\lambda=0,\ \ \ \gamma\in[0,1].
\end{equation*}
Hence,
\begin{align}
	-\frac{\gamma}{2\pi^2}&\int_{\mathbb{R}}\frac{r_1(\lambda)r_2(\lambda)}{1-\gamma r_1(\lambda)r_2(\lambda)}\frac{\d}{\d\lambda}\left[\textnormal{pv}\int_{\mathbb{R}}\ln\big(1-\gamma r_1(z)r_2(z)\big)\frac{\d z}{z-\lambda}\right]\d\lambda\nonumber\\
	=&\,-\frac{\gamma}{2\pi^2}\int_{\mathbb{R}}\frac{r_1(\lambda)r_2(\lambda)}{1-\gamma r_1(\lambda)r_2(\lambda)}\frac{\d}{\d\lambda}\left[-\im\pi\ln\big(1-\gamma r_1(\lambda)r_2(\lambda)\big)+\textnormal{pv}\int_{-\infty}^{\infty}\ln\big(1-\gamma r_1(z)r_2(z)\big)\frac{\d z}{z-\lambda}\right]\d\lambda\nonumber\\
	=&\,-\frac{\gamma}{2\pi^2}\int_{\mathbb{R}-\im\epsilon}\frac{r_1(\lambda)r_2(\lambda)}{1-\gamma r_1(\lambda)r_2(\lambda)}\frac{\d}{\d\lambda}\left[\int_{\mathbb{R}}\ln\big(1-\gamma r_1(z)r_2(z)\big)\frac{\d z}{z-\lambda}\right]\d\lambda,\label{c38}
\end{align}
where we deformed the outer integration contour to $\mathbb{R}-\im\epsilon$ in the second equality while using Plemelj's formula and analytic continuation in the integrand. Consequently, with $(z-\lambda)^{-2}=-\int_0^{\infty}\e^{\im(z-\lambda)x}x\,\d x$ for $z\in\mathbb{R},\lambda\in\mathbb{R}-\im\epsilon$, after collapsing $\mathbb{R}-\im\epsilon$ to $\mathbb{R}$,
\begin{equation*}
	-\frac{\gamma}{2\pi^2}\int_{\mathbb{R}}\frac{r_1(\lambda)r_2(\lambda)}{1-\gamma r_1(\lambda)r_2(\lambda)}\frac{\d}{\d\lambda}\left[\textnormal{pv}\int_{\mathbb{R}}\ln\big(1-\gamma r_1(z)r_2(z)\big)\frac{\d z}{z-\lambda}\right]\d\lambda\stackrel{\eqref{c38}}{=}-2\gamma\int_0^{\infty}xs(x,\gamma)\left[\frac{\partial}{\partial\gamma}s(-x,\gamma)\right]\d x,
\end{equation*}
so that all together for \eqref{c37},
\begin{align}
	\int_{\mathbb{R}}\begin{bmatrix}f_1(\lambda)\\ f_2(\lambda)\end{bmatrix}\big({\bf A}^{\top}(\lambda)\big)'\big({\bf A}^{-1}(\lambda)\big)^{\top}\begin{bmatrix}g_1(\lambda)\\ g_2(\lambda)\end{bmatrix}\d\lambda=-\gamma t&\frac{\d}{\d\gamma}s(0,\gamma)-2\gamma\int_0^{\infty}xs(x,\gamma)\left[\frac{\partial}{\partial\gamma}s(-x,\gamma)\right]\d x\nonumber\\
	&-\frac{\gamma}{2\pi\im}\frac{\d}{\d\gamma}\int_{-\infty}^{\infty}\frac{r_2'(\lambda)}{r_2(\lambda)}\ln\big(1-\gamma r_1(\lambda)r_2(\lambda)\big)\,\d\lambda.\label{c39}
\end{align}
Returning now to \eqref{c24}, identity \eqref{c33} combined with \eqref{c34},\eqref{c36} and \eqref{c39} yields
\begin{align}
	-\frac{1}{2\gamma}\int_{\Sigma_{\bf T}}R_{t,\gamma}(\lambda,\lambda)\,\d\lambda=\frac{\partial}{\partial\gamma}\bigg[t&s(0,\gamma)+\int_0^{\infty}xs(x,\gamma)s(-x,\gamma)\,\d x\nonumber\\
	&\,-\frac{1}{4\pi\im}\int_{-\infty}^{\infty}\left\{\frac{r_1'(\lambda)}{r_1(\lambda)}-\frac{r_2'(\lambda)}{r_2(\lambda)}\right\}\ln\big(1-\gamma r_1(\lambda)r_2(\lambda)\big)\,\d\lambda\bigg]+\tilde{r}(t,\gamma),\label{c40}
\end{align}
for all $(-t)\geq t_0$ and $\gamma\in[0,1]$ with $|\tilde{r}(t,\gamma)|\leq c|t|^{-1}$. In turn, performing the definite $\gamma$-integration in \eqref{c24} we finally arrive at, as $t\rightarrow-\infty$,
\begin{eqnarray}
	\ln F(t)\!\!\!\!\!&\stackrel{\eqref{c24}}{=}&\!\!\!\!\!-\frac{1}{2}\int_0^1\left[\int_{\Sigma_{\bf T}}R_{t,\gamma}(\lambda,\lambda)\,\d\lambda\right]\frac{\d\gamma}{\gamma}\label{c41}\\
	&=&\!\!\!\!\!ts(0)+\int_0^{\infty}xs(x)s(-x)\,\d x-\frac{1}{4\pi\im}\int_{-\infty}^{\infty}\left\{\frac{r_1'(\lambda)}{r_1(\lambda)}-\frac{r_2'(\lambda)}{r_2(\lambda)}\right\}\ln\big(1-r_1(\lambda)r_2(\lambda)\big)\,\d\lambda+\mathcal{O}\big(t^{-1}\big)\nonumber
\end{eqnarray}
since $s(x,0)=0$ for all $x\in\mathbb{R}$. This proves \eqref{i:25}, except for the error term. However we can easily improve the $\mathcal{O}(t^{-1})$ in \eqref{c41} by using \eqref{c18}. Indeed, by \eqref{c27},\eqref{c29} and \eqref{c31},
\begin{align*}
	\frac{\d}{\d t}\ln F(t)=&\,\im X_1^{11}(t,1,\phi,\psi)\stackrel{\eqref{c27}}{=}\im Y_1^{11}(t,1,\phi,\psi)+s(0)\stackrel{\eqref{c29}}{=}\im S_1^{11}(t,1,\phi,\psi)+s(0)\\
	&\hspace{1.175cm}\stackrel{\eqref{c31}}{=}-\frac{1}{2\pi}\int_{\Sigma_{\bf S}}\Big({\bf S}_-(\lambda;t,1,\phi,\psi)\big({\bf G}_{\bf S}(\lambda;t,1,\phi,\psi)-\mathbb{I}\big)\Big)_{11}\,\d\lambda+s(0)\ \ \ \forall\,(-t)\geq t_0
\end{align*}
and so with \eqref{c30} and \eqref{c32}, after indefinite $t$-integration, for all $(-t)\geq t_0$ with $c=c(\epsilon)>0$,
\begin{equation}\label{c42}
	\ln F(t)=ts(0)+\eta+r(t),\ \ \ \ \ |r(t)|\leq c\e^{-\epsilon|t|},
\end{equation}
with a $t$-independent term $\eta$, the integration constant. But the same term was already computed in \eqref{c40}, hence consistency between \eqref{c41} and \eqref{c42} imposes the error term to be exponentially small, as claimed in \eqref{i:26}. This completes our proof of Theorem \ref{theo2a}.
\end{proof}

\section{Proof of Lemma \ref{lem2}, Corollary \ref{corb1} and Corollary \ref{corb2}}\label{sec6}
We now prove our three algebraic results about Fredholm determinants \eqref{b:3} of \textit{multiplicative} Hankel composition operators with kernels \eqref{b:2}.

\begin{proof}[Proof of Lemma \ref{lem2}] $I-K_t,t\in J$ is invertible on $L^2(0,1)$ and $K_t,t\in J$ trace class on the same space, so $F(t)\neq 0$ on $J$ and thus $J\ni t\mapsto\ln F(t)$ is well-defined, with a branch for the logarithm that is analytically continued along the orbit of $F(t)$. Next, by the $L^2(0,1)$ dominance of $\tau_t\phi,\tau_t\psi,MD\tau_t\phi$ and $MD\tau_t\psi$ and the continuous differentiability of $\phi,\psi$, $J\ni t\mapsto K_t$ is differentiable with $\frac{\d}{\d t}K_t=\tau_t\phi\otimes\tau_t\psi,t\in J$, using \eqref{b:4}. In fact, $\frac{\d}{\d t}K_t\in\mathcal{C}_0(L^2(0,1))\subset\mathcal{C}_1(L^2(0,1))$ and so by Jacobi's formula
\begin{equation}\label{bp1}
	\frac{\d}{\d t}\ln F(t)=-\tr_{L^2(0,1)}\left((I-K_t)^{-1}\frac{\d K_t}{\d t}\right)=-\tr_{L^2(0,1)}\big((I-K_t)^{-1}\tau_t\phi\otimes\tau_t\psi\big),\ \ t\in J.
\end{equation}
Moving ahead, $\tau_t\phi,\tau_t\psi\in H^{1,2}(0,1)$ by dominance, so $J\ni t\mapsto t(\tau_t\phi\otimes\tau_t\psi)$ is differentiable,
\begin{equation}\label{bp2}
	t\frac{\d}{\d t}(\tau_t\phi\otimes\tau_t\psi)=MD\tau_t\phi\otimes\tau_t\psi+\tau_t\phi\otimes MD\tau_t\psi\in\mathcal{C}_0(L^2(0,1))\subset\mathcal{C}_1(L^2(0,1)).
\end{equation}
Likewise, $J\ni t\mapsto(I-K_t)^{-1}\in\mathcal{L}(L^2(0,1))$ is differentiable,
\begin{equation}\label{bp3}
	\frac{\d}{\d t}(I-K_t)^{-1}=(I-K_t)^{-1}\frac{\d K_t}{\d t}(I-K_t)^{-1}=(I-K_t)^{-1}(\tau_t\phi\otimes\tau_t\psi)(I-K_t)^{-1}\in\mathcal{C}_0(L^2(0,1)),
\end{equation}
since $I-K_t,t\in J$ is invertible and $J\ni t\mapsto K_t$ differentiable. So, \eqref{bp2} and \eqref{bp3} together imply that $J\ni t\mapsto(I-K_t)^{-1}\tau_t\phi\otimes\tau_t\psi\in\mathcal{C}_1(L^2(0,1))$ is differentiable and
\begin{align}
	t\frac{\d}{\d t}\big(&(I-K_t)^{-1}\tau_t\phi\otimes\tau_t\psi\big)=(I-K_t)^{-1}\tau_t\phi\otimes MD\tau_t\psi\nonumber\\
	&+\Big[(I-K_t)^{-1}MD\tau_t\phi+t(I-K_t)^{-1}(\tau_t\phi\otimes\tau_t\psi)(I-K_t)^{-1}\tau_t\phi\Big]\otimes\tau_t\psi\in\mathcal{C}_1(L^2(0,1)).\label{bp4}
\end{align}
Consequently, using that by \cite[$(3.4)$]{Bris}
\begin{equation}\label{bp5}
	\tr_{L^2(0,1)}\big((I-K_t)^{-1}\tau_t\phi\otimes\tau_t\psi\big)=\int_0^1\big((I-K_t)^{-1}\tau_t\phi\big)(x)(\tau_t\psi)(x)\,\d x,
\end{equation}
since $(0,1]^2\ni (x,y)\mapsto\big((I-K_t)^{-1}\tau_t\phi\big)(x)(\tau_t\psi)(y)$ is continuous\footnote{The map $(0,1]\ni x\mapsto((I-K_t)^{-1}\tau_t\phi)(x)$ is locally $\frac{1}{2}$-H\"older continuous for $t\in J$ by invertibility of $I-K_t,t\in J$ and by \eqref{b:1},\eqref{b:5}. This shows, in particular, that $J\ni t\mapsto q_0(t)$ is well-defined.}, the dominated convergence theorem yields
\begin{align}
	\frac{\d}{\d t}&\left\{t\frac{\d}{\d t}\ln F(t)\right\}\stackrel[\eqref{bp5}]{\eqref{bp1}}{=}-\tr_{L^2(0,1)}\big((I-K_t)^{-1}\tau_t\phi\otimes\tau_t\psi\big)-t\frac{\d}{\d t}\int_0^1\big((I-K_t)^{-1}\tau_t\phi\big)(x)(\tau_t\psi)(x)\,\d x\nonumber\\
	&\,\,\,\,=-\tr_{L^2(0,1)}\big((I-K_t)^{-1}\tau_t\phi\otimes\tau_t\psi\big)-\tr_{L^2(0,1)}\Big(t\frac{\d}{\d t}\big((I-K_t)^{-1}\tau_t\phi\otimes\tau_t\psi\big)\Big)\nonumber\\
	&\,\stackrel{\eqref{bp4}}{=}-\tr_{L^2(0,1)}\big((I-K_t)^{-1}\tau_t\phi\otimes\tau_t\psi\big)-\tr_{L^2(0,1)}\big((I-K_t)^{-1}\tau_t\phi\otimes MD\tau_t\psi\big)\nonumber\\
	&\hspace{1cm}-\tr_{L^2(0,1)}\Big(\Big[(I-K_t)^{-1}MD\tau_t\phi+t(I-K_t)^{-1}(\tau_t\phi\otimes\tau_t\psi)(I-K_t)^{-1}\tau_t\phi\Big]\otimes\tau_t\psi\Big),\ \ t\in J.\label{bp6}
\end{align}
This shows that $J\ni t\mapsto\ln F(t)$ is twice-differentiable. In order to simplify \eqref{bp6}, that all involved integral operators have continuous kernels, so in particular, integrating by parts
\begin{align}
	\tr_{L^2(0,1)}\big(&(I-K_t)^{-1}\tau_t\phi\otimes MD\tau_t\psi\big)=\int_0^1\big((I-K_t)^{-1}\tau_t\phi\big)(x)(MD\tau_t\psi)(x)\,\d x\nonumber\\
	&\,=q_0(t)(\tau_t\psi)(1)-\tr_{L^2(0,1)}\big((I-K_t)^{-1}\tau_t\phi\otimes\tau_t\psi\big)-\int_0^1\big(MD(I-K_t)^{-1}\tau_t\phi\big)(x)(\tau_t\psi)(x)\,\d x.\label{bp7}
\end{align}
Note that $(0,1]\ni x\mapsto(MD(I-K_t)^{-1}\tau_t\phi)(x)$ is continuous by dominance, \eqref{b:1} and by invertibility of $I-K_t,t\in J$ on $L^2(0,1)$. Moreover, by \eqref{b:1},\eqref{b:5} and dominance, $(I-K_t)^{-1}\tau_t\phi\in H^{1,2}(0,1)$ and $\sqrt{x}((I-K_t)^{-1}\tau_t\phi)(x)\rightarrow 0$ as $x\downarrow 0$ by \eqref{b:4}. Thus, \eqref{bp7} is justified and we can rewrite it as, cf. \cite[$(3.4)$]{Bris},
\begin{align*}
	\tr_{L^2(0,1)}\big((I-K_t)^{-1}\tau_t\phi\otimes MD\tau_t\psi\big)=q_0(t)(\tau_t\psi)&(1)-\tr_{L^2(0,1)}\big((I-K_t)^{-1}\tau_t\phi\otimes\tau_t\psi\big)\\
	&-\tr_{L^2(0,1)}\big(MD(I-K_t)^{-1}\tau_t\phi\otimes\tau_t\psi\big),\ \ t\in J.
\end{align*}
We now return to \eqref{bp6}, insert the last identity,
\begin{align}
	\frac{\d}{\d t}\left\{t\frac{\d}{\d t}\ln F(t)\right\}=-q_0(t)(\tau_t\psi)&(1)+\tr_{L^2(0,1)}\bigg(\Big[MD(I-K_t)^{-1}\tau_t\phi-(I-K_t)^{-1}MD\tau_t\phi\nonumber\\
	&-t(I-K_t)^{-1}(\tau_t\phi\otimes\tau_t\psi)(I-K_t)^{-1}\tau_t\phi\Big]\otimes\tau_t\psi\bigg),\ t\in J,\label{bp8}
\end{align}
afterwards perform algebra,
\begin{equation*}
	MD(I-K_t)^{-1}\tau_t\phi-(I-K_t)^{-1}MD\tau_t\phi=(I-K_t)^{-1}(MDK_t-K_tMD)(I-K_t)^{-1}\tau_t\phi\in L^2(0,1),
\end{equation*}
and then simplify with the help of the identities
\begin{equation}\label{bp9}
	x\frac{\partial}{\partial x}K_t(x,y)+y\frac{\partial}{\partial y}K_t(x,y)+K_t(x,y)=t\big(\tau_t\phi\otimes\tau_t\psi\big)(x,y),\ \ \ \ \ \ \lim_{y\downarrow 0}K_t(x,y)\sqrt{y}=0\ \ \forall\, x\in(0,1).
\end{equation}
The intermediate result equals
\begin{equation*}
	\big((MDK_t-K_tMD)(I-K_t)^{-1}\tau_t\phi\big)(x)=-K_t(x,1)q_0(t)+t(\tau_t\phi)(x)\int_0^1(\tau_t\psi)(y)\big((I-K_t)^{-1}\tau_t\phi\big)(y)\,\d y,
\end{equation*}
where $K_t(x,1)$ is shorthand for $K_t(x,1):=\lim_{y\uparrow 1}K_t(x,y)$, and so with the last identity back in \eqref{bp8}, using that all involved trace class operators have continuous kernels,
\begin{equation*}
	\frac{\d}{\d t}\left\{t\frac{\d}{\d t}\ln F(t)\right\}=-q_0(t)\left[(\tau_t\psi)(1)+\int_0^1\big((I-K_t)^{-1}K_t\big)(x,1)(\tau_t\psi)(x)\,\d x\right],\ \ \ t\in J.
\end{equation*}
Lastly, it remains to use $(I-K_t)^{-1}=I+(I-K_t)^{-1}K_t$ on $L^2(0,1)$ and $(I-K_t)^{-1}(x,1)=(I-K_t^{\ast})^{-1}(1,x)$ for $x\in(0,1)$,
\begin{equation*}
	\frac{\d}{\d t}\left\{t\frac{\d}{\d t}\ln F(t)\right\}=-q_0(t)\int_0^1(I-K_t)^{-1}(x,1)(\tau_t\psi)(x)\,\d x=-q_0(t)q_0^{\ast}(t)\frac{1}{t},\ \ \ t\in J,
\end{equation*}
as claimed in \eqref{b:6}. Note that $q_0^{\ast}$ is well-defined by \eqref{b:1},\eqref{b:5} and by invertibility of $I-K_t$. This completes our proof.
\end{proof}
\begin{rem} The reader will notice that the multiplicative  composition structure is used precisely in the identity $t\frac{\d}{\d t}\tau_tf=MD\tau_tf$. As such we use the multiplicative composition structure to get $\frac{\d}{\d t}K_t=\tau_t\phi\otimes\tau_t\psi$ and then again in the derivation of \eqref{bp4} and \eqref{bp9}.
\end{rem}
\begin{proof}[Proof of Corollary \ref{corb1}] We have, see Lemma \ref{lem2},
\begin{equation*}
	t\frac{\d}{\d t}\left\{t\frac{\d}{\d t}\ln F(t)\right\}=-q_0(t)q_0^{\ast}(t),\ \ t\in J.
\end{equation*}
By our workings, $J\ni t\mapsto (I-K_t)^{-1}\in\mathcal{L}(H^{1,2}(0,1))$ is differentiable with $\frac{\d}{\d t}(I-K_t)^{-1}\in\mathcal{C}_0(L^2(0,1))$, compare \eqref{bp3}. But $t\mapsto\tau_t\in\mathcal{L}(H^{1,2}(0,1))$ is also differentiable with
\begin{equation*}
	t\frac{\d}{\d t}\tau_t=MD\tau_t\in\mathcal{L}(H^{1,2}(0,1),L^2(0,1)).
\end{equation*}
Hence, $J\ni t\mapsto (I-K_t)^{-1}\tau_t\in\mathcal{L}(H^{1,2}(0,1))$ is differentiable and so $J\ni t\mapsto q_0(t)=((I-K_t)^{-1}\tau_t\phi)(1)$ in particular continuous. Repeating the same argument with $\phi\leftrightarrow\psi$ and an additional factor of $t$ yields the corresponding result for $J\ni t\mapsto q_0^{\ast}(t)$. Now let $J=\mathbb{R}_+$, then
\begin{equation}\label{bp10}
	\frac{tF'(t)}{F(t)}-\lim_{s\downarrow 0}\frac{sF'(s)}{F(s)}=\int_0^t\frac{\d}{\d s}\left\{s\frac{\d}{\d s}\ln F(s)\right\}\,\d s\stackrel{\eqref{b:6}}{=}-\int_0^tq_0(s)q_0^{\ast}(s)\frac{\d s}{s},\ \ \ t\in\mathbb{R}_+.
\end{equation}
Here, $\lim_{s\downarrow 0}F(s)=1$ by \eqref{b:28} and the right hand side of \eqref{bp10} after division by $t$, viewed as function of $t\in\mathbb{R}_+$, is integrable near zero by assumption. Consequently we must have $\lim_{s\downarrow 0}sF'(s)=0$ and so
\begin{equation*}
	\ln F(t)=\int_0^t\left[u\frac{\d}{\d u}\ln F(u)\right]\frac{\d u}{u}\stackrel{\eqref{bp10}}{=}\int_0^t\left[-\int_0^uq_0(s)q_0^{\ast}(s)\frac{\d s}{s}\right]\frac{\d u}{u}=-\int_0^t\ln\Big(\frac{t}{s}\Big)q_0(s)q_0^{\ast}(s)\frac{\d s}{s},\ \ t\in\mathbb{R}_+.
\end{equation*}
This concludes our proof of Corollary \ref{corb1}.
\end{proof}
\begin{proof}[Proof of Corollary \ref{corb2}] Although system \eqref{b:11} is highly symmetric, the functions $q_n,q_n^{\ast},p_n,p_n^{\ast}$ are not, so we need to proceed step by step and cannot rely on a $\phi\leftrightarrow\psi$ swap argument. First we note that $J\ni t\mapsto q_n(t)$ and $J\ni t\mapsto q_n^{\ast}(t)$ are differentiable for $n=0,1,\ldots,N-1$ since $(MD)^n\tau_t\phi,(DM)^n\tau_t\psi\in H^{1,2}(0,1)$ for the same $n$ and since $t\mapsto(I-K_t)^{-1}\tau_t\in\mathcal{L}(H^{1,2}(0,1))$ is differentiable with
\begin{equation}\label{bp11}
	t\frac{\d}{\d t}(I-K_t)^{-1}\tau_t=(I-K_t)^{-1}MD\tau_t+t(I-K_t)^{-1}(\tau_t\phi\otimes\tau_t\psi)(I-K_t)^{-1}\in\mathcal{L}\big(H^{1,2}(0,1),L^2(0,1)\big).
\end{equation}
Second, $J\ni t\mapsto q_N(t)$ and $J\ni t\mapsto q_N^{\ast}(t)$ are also well-defined, and by the workings in Lemma \ref{lem2}, $J\ni t\mapsto p_n(t)$ is differentiable for $n=0,1,\ldots,N-1$. This is because $J\ni t\mapsto (I-K_t)^{-1}(MD)^n\tau_t\phi\otimes\tau_t\psi\in\mathcal{C}_1(L^2(0,1))$ is differentiable,
\begin{align}
	t&\,\frac{\d}{\d t}\big((I-K_t)^{-1}(MD)^n\tau_t\phi\otimes\tau_t\psi\big)=(I-K_t)^{-1}(MD)^n\tau_t\phi\otimes MD\tau_t\psi\nonumber\\
	&\,+\Big[(I-K_t)^{-1}(MD)^{n+1}\tau_t\phi+t(I-K_t)^{-1}(\tau_t\phi\otimes\tau_t\psi)(I-K_t)^{-1}(MD)^n\tau_t\phi\Big]\otimes\tau_t\psi\in\mathcal{C}_1(L^2(0,1)),\label{bp12}
\end{align}
for $n=0,1,\ldots,N-1$, and since by the dominated convergence theorem\footnote{If $\{(MD)^n\tau_tf\}_{t\in\mathbb{R}_+}$ is $L^2(0,1)$ dominated for all $n=0,1,\ldots,N$, then so is $\{(DM)^n\tau_tf\}_{t\in\mathbb{R}_+}$, and vice versa.} and \cite[$(3.4)$]{Bris},
\begin{align}
	t\frac{\d p_n}{\d t}(t)=&\,\,p_n(t)+t^2\frac{\d}{\d t}\int_0^1\big((I-K_t)^{-1}(MD)^n\tau_t\phi\big)(x)(\tau_t\psi)(x)\,\d x\nonumber\\
	=&\,\,p_n(t)+t\tr_{L^2(0,1)}\Big(t\frac{\d}{\d t}\big((I-K_t)^{-1}(MD)^n\tau_t\phi\otimes\tau_t\psi\big)\Big).\label{bp13}
\end{align}
Likewise, $J\ni t\mapsto p_n^{\ast}(t)$ is differentiable for $n=0,1,\ldots,N-1$ since $J\ni t\mapsto (I-K_t^{\ast})^{-1}(DM)^n\tau_t\psi\otimes\tau_t\phi\in\mathcal{C}_1(L^2(0,1))$ is with
\begin{align}
	&t\frac{\d}{\d t}\big((I-K_t^{\ast})^{-1}(DM)^n\tau_t\psi\otimes\tau_t\phi\big)=(I-K_t^{\ast})^{-1}(DM)^n\tau_t\psi\otimes MD\tau_t\phi-(I-K_t^{\ast})^{-1}(DM)^n\tau_t\psi\otimes\tau_t\phi\nonumber\\
	&\,+\Big[(I-K_t^{\ast})^{-1}(DM)^{n+1}\tau_t\psi+t(I-K_t^{\ast})^{-1}(\tau_t\psi\otimes\tau_t\phi)(I-K_t^{\ast})^{-1}(DM)^n\tau_t\psi\Big]\otimes\tau_t\phi\in\mathcal{C}_1(L^2(0,1)),\label{bp14}
\end{align}
for $n=0,1,\ldots,N$, and since by the dominated convergence theorem and \cite[$(3.4)$]{Bris},
\begin{align}
	t\frac{\d p_n^{\ast}}{\d t}(t)=&\,\,p_n^{\ast}(t)+t^2\frac{\d}{\d t}\int_0^1\big((I-K_t^{\ast})^{-1}(DM)^n\tau_t\psi\big)(x)(\tau_t\phi)(x)\,\d x\nonumber\\
	=&\,\,p_n^{\ast}(t)+t\tr_{L^2(0,1)}\Big(t\frac{\d}{\d t}\big((I-K_t^{\ast})^{-1}(DM)^n\tau_t\psi\otimes\tau_t\phi\big)\Big).\label{bp14a}
\end{align}
At this point we compute the derivatives of $q_n$ and $p_n$. By \eqref{bp11} and \cite[$(3.4)$]{Bris}, for $n=0,1,\ldots,N-1$,
\begin{equation*}
	t\frac{\d q_n}{\d t}(t)=t\frac{\d}{\d t}\big((I-K_t)^{-1}(MD)^n\tau_t\phi\big)(1)\stackrel{\eqref{bp11}}{=}q_{n+1}(t)+q_0(t)p_n(t),\ \ t\in J,
\end{equation*}
since $t\frac{\d}{\d t}(MD)^n\tau_tf=(MD)^{n+1}\tau_tf$. Next, by \eqref{bp13} and \eqref{bp12}, for $n=0,1,\ldots,N-1$,
\begin{align}
	t&\frac{\d p_n}{\d t}(t)=p_n(t)+t\tr_{L^2(0,1)}\big((I-K_t)^{-1}(MD)^n\tau_t\phi\otimes MD\tau_t\psi\big)\label{bp15}\\
	&\,+t\tr_{L^2(0,1)}\Big(\Big[(I-K_t)^{-1}(MD)^{n+1}\tau_t\phi+t(I-K_t)^{-1}(\tau_t\phi\otimes\tau_t\psi)(I-K_t)^{-1}(MD)^n\tau_t\phi\Big]\otimes\tau_t\psi\Big),\ \ t\in J,\nonumber
\end{align}
where we can simplify the first operator trace as in \eqref{bp7}, relying on \eqref{b:8} and dominance,
\begin{align*}
	&\tr_{L^2(0,1)}\big((I-K_t)^{-1}(MD)^n\tau_t\phi\otimes MD\tau_t\psi\big)=\int_0^1\big((I-K_t)^{-1}(MD)^n\tau_t\phi\big)(x)(MD\tau_t\psi)(x)\,\d x\\
	&\,=q_n(t)(\tau_t\psi)(1)-\tr_{L^2(0,1)}\big((I-K_t)^{-1}(MD)^n\tau_t\phi\otimes\tau_t\psi\big)-\tr_{L^2(0,1)}\big(MD(I-K_t)^{-1}(MD)^n\tau_t\phi\otimes\tau_t\psi\big).
\end{align*}
Hence back in \eqref{bp15},
\begin{align}
	t\frac{\d p_n}{\d t}(t)=tq_n(t)(\tau_t\psi)&(1)-t\tr_{L^2(0,1)}\bigg(\Big[MD(I-K_t)^{-1}(MD)^n\tau_t\phi-(I-K_t)^{-1}(MD)^{n+1}\tau_t\phi\label{bp16}\\
	&\hspace{2.5cm}-t(I-K_t)^{-1}(\tau_t\phi\otimes\tau_t\psi)(I-K_t)^{-1}(MD)^n\tau_t\phi\Big]\otimes\tau_t\psi\bigg),\ \ t\in J,\nonumber
\end{align}
which generalizes \eqref{bp8}. We now simplify \eqref{bp16} just as we did it for \eqref{bp8} in the proof of Lemma \ref{lem2},	
\begin{align*}
	\big((MDK_t-K_tMD)(I-K_t)^{-1}(MD)^n\tau_t\phi\big)&(x)=-K_t(x,1)q_n(t)\\
	&\,+t(\tau_t\phi)(x)\int_0^1(\tau_t\psi)(y)\big((I-K_t)^{-1}(MD)^n\tau_t\phi\big)(y)\,\d y
\end{align*}
and obtain in \eqref{bp16},
\begin{equation*}
	t\frac{\d p_n}{\d t}(t)=tq_n(t)\left[(\tau_t\psi)(1)+\int_0^1\big((I-K_t)^{-1}K_t\big)(x,1)(\tau_t\psi)(x)\,\d x\right],\ \ t\in J,\ \ n=0,1,\ldots,N-1.
\end{equation*}
It remains to use $(I-K_t)^{-1}=I+(I-K_t)^{-1}K_t$ and $(I-K_t)^{-1}(x,1)=(I-K_t^{\ast})^{-1}(1,x)$, i.e. we have
\begin{equation*}
	t\frac{\d p_n}{\d t}(t)=tq_n(t)\int_0^1(I-K_t^{\ast})^{-1}(1,x)(\tau_t\psi)(x)\,\d x=q_n(t)q_0^{\ast}(t),\ \ t\in J,
\end{equation*}
as claimed in \eqref{b:11}. We now compute the derivatives of $q_n^{\ast}$ and $p_n^{\ast}$. By \eqref{bp11}, \cite[$(3.4)$]{Bris} and $t\frac{\d}{\d t}(DM)^n\tau_tf=(DM)^{n+1}\tau_tf-(DM)^n\tau_tf$, for any $n\in\{0,1,\ldots,N-1\}$,
\begin{equation*}
	t\frac{\d q_n^{\ast}}{\d t}(t)=q_n^{\ast}(t)+t^2\frac{\d}{\d t}\big((I-K_t^{\ast})^{-1}(DM)^n\tau_t\psi\big)(1)=q_{n+1}^{\ast}(t)+q_0^{\ast}(t)p_n^{\ast}(t),\ \ t\in J.
\end{equation*}
Next, by \eqref{bp14} and \eqref{bp14a}, for $n=0,1,\ldots,N-1$,
\begin{align}
	t&\frac{\d p_n^{\ast}}{\d t}(t)=t\tr_{L^2(0,1)}\big((I-K_t^{\ast})^{-1}(DM)^n\tau_t\psi\otimes MD\tau_t\phi\big)\label{bp17}\\
	&\,+t\tr_{L^2(0,1)}\bigg(\Big[(I-K_t^{\ast})^{-1}(DM)^{n+1}\tau_t\psi+t(I-K_t^{\ast})^{-1}(\tau_t\psi\otimes\tau_t\phi)(I-K_t^{\ast})^{-1}(DM)^n\tau_t\psi\Big]\otimes\tau_t\phi\bigg),\ t\in J,\nonumber
\end{align}
where by \eqref{b:8} and dominance,
\begin{equation*}
	\tr_{L^2(0,1)}\big((I-K_t^{\ast})^{-1}(DM)^n\tau_t\psi\otimes MD\tau_t\phi\big)=q_n^{\ast}(t)t^{-1}(\tau_t\phi)(1)-\tr_{L^2(0,1)}\big(DM(I-K_t^{\ast})^{-1}(DM)^n\tau_t\psi\otimes\tau_t\phi\big).
\end{equation*}
Hence back in \eqref{bp17},
\begin{align}
	t\frac{\d p_n^{\ast}}{\d t}(t)=q_n^{\ast}(t)(\tau_t\phi)(1)-t\tr_{L^2(0,1)}&\bigg(\Big[DM(I-K_t^{\ast})^{-1}(DM)^n\tau_t\psi-(I-K_t^{\ast})^{-1}(DM)^{n+1}\tau_t\psi\label{bp18}\\
	&\,-t(I-K_t^{\ast})^{-1}(\tau_t\psi\otimes\tau_t\phi)(I-K_t^{\ast})^{-1}(DM)^n\tau_t\psi\Big]\otimes\tau_t\phi\bigg),\ t\in J.\nonumber
\end{align}
But by algebra, $DM(I-K_t^{\ast})^{-1}(DM)^n\tau_t\psi-(I-K_t^{\ast})^{-1}(DM)^{n+1}\tau_t\psi=(I-K_t^{\ast})^{-1}(DMK_t^{\ast}-K_t^{\ast}DM)(I-K_t^{\ast})^{-1}(DM)^n\tau_t\psi$ and
\begin{align*}
	\big((DMK_t^{\ast}-K_t^{\ast}DM)(I-K_t^{\ast})^{-1}(DM)^n\tau_t\psi\big)&(x)=-K_t^{\ast}(x,1)q_n^{\ast}(t)t^{-1}\\
	&\,+t(\tau_t\psi)(x)\int_0^1(\tau_t\phi)(y)\big((I-K_t^{\ast})^{-1}(DM)^n\tau_t\psi\big)(y)\,\d y.
\end{align*}
Consequently \eqref{bp18} yields, for any $n\in\{0,1,\ldots,N-1\}$,
\begin{equation*}
	t\frac{\d p_n^{\ast}}{\d t}(t)=q_n^{\ast}(t)\left[(\tau_t\phi)(1)+\int_0^1\big((I-K_t^{\ast})^{-1}K_t^{\ast}\big)(x,1)(\tau_t\phi)(x)\,\d x\right]=q_n^{\ast}(t)q_0(t),\ \ t\in J,
\end{equation*}
which concludes our proof of \eqref{b:11} and thus Corollary \ref{corb2}.
\end{proof}
The last proof concludes the content of this section, we now move to the RHP \ref{masterb}.

\section{Proof of Theorem \ref{theo3} and Corollaries \ref{impcor2} and \ref{deeper2}}\label{sec7}

As before, we shall first argue that RHP \ref{masterb} admits at most one solution.
\begin{lem}\label{unique3} The RHP \ref{masterb}, if solvable, is uniquely solvable.
\end{lem}
\begin{proof} As in the proof of Lemma \ref{unique2}, the pointwise convergence requirement in condition (2) of RHP \ref{masterb} implies ${\bf X}(\frac{1}{2}\mp\epsilon+\im z)\rightarrow {\bf X}_{\pm}(\frac{1}{2}+\im z)$ locally in $L^2(\mathbb{R})$ as $\epsilon\downarrow 0$, so $\det{\bf X}(\frac{1}{2}\mp\epsilon+\im z)\rightarrow(\det{\bf X})_{\pm}(\frac{1}{2}+\im z)$ locally in $L^1(\mathbb{R})$ as $\epsilon\downarrow 0$. But from \eqref{b:12}, for every $z\in\frac{1}{2}+\im\mathbb{R}$,
\begin{equation*}
	(\det{\bf X})_+(z)=\det{\bf X}_+(z)=\det\left({\bf X}_-(z)\begin{bmatrix}1-r_1(z)r_2(z) & -r_2(z)t^z\,\smallskip\\
	r_1(z)t^{-z} & 1\end{bmatrix}\right)=\det{\bf X}_-(z)=(\det{\bf X})_-(z),
\end{equation*}
thus the proof working in \cite[Theorem $7.18$]{D} yields analyticity of $\det{\bf X}(z)$ across $\frac{1}{2}+\im\mathbb{R}$ and therefore analyticity in all of $\mathbb{C}$ by requirement $(1)$. On the other hand, by condition $(3)$, $\det{\bf X}(z)=1+\mathcal{O}(z^{-1})$ uniformly as $z\rightarrow\infty$ in $\mathbb{C}\setminus(\frac{1}{2}+\im\mathbb{R})$. Hence, by Liouville's theorem,
\begin{equation*}
	\det{\bf X}(z)\equiv 1\ \ \ \ \forall\,z,
\end{equation*}
which yields invertibility of every solution of RHP \ref{masterb}. Next, if ${\bf X}_1(z)$ and ${\bf X}_2(z)$ are two such solutions, we set ${\bf R}(z)={\bf X}_1(z){\bf X}_2(z)^{-1},z\in\mathbb{C}\setminus(\frac{1}{2}+\im\mathbb{R})$ and deduce as in the proof of Lemma \ref{unique2} that ${\bf R}(z)\equiv 1$, i.e. ${\bf X}_1(z)={\bf X}_2(z)$, implying uniqueness.
\end{proof}
Moving ahead, we require two auxiliary functions, namely specific solutions of
\begin{equation*}
		x\frac{\d y}{\d x}=-zy+\phi\ \ \ \ \ \ \textnormal{and}\ \ \ \ \ \ x\frac{\d y}{\d x}=(z-1)y+\psi,\ \ \ \ \ \ \ y=y(x,z):\mathbb{R}_+\times\mathbb{C}\rightarrow\mathbb{C}.
\end{equation*}
These will play an important role in the proof of Theorem \ref{theo3} and so we summarize their key properties below.
\begin{prop}\label{crucial2} Suppose $\phi,\psi\in H_{\circ}^{1,1}(\mathbb{R}_+)$ are continuously differentiable on $\mathbb{R}_+$ with $\sqrt{(\cdot)}\phi,\sqrt{(\cdot)}\psi\in L^{\infty}(\mathbb{R}_+)$ and $\sqrt{(\cdot)}MD\phi,\sqrt{(\cdot)}MD\psi\in L^{\infty}(0,1)$. Define, for all $x\in\mathbb{R}_+$,
\begin{equation}\label{bp19}
	y_1(x,z):=\begin{cases}\displaystyle-\int_x^{\infty}\Big(\frac{x}{s}\Big)^{-z}\phi(s)\frac{\d s}{s},&\Re z<\frac{1}{2}\bigskip\\
	\displaystyle \int_0^x\Big(\frac{x}{s}\Big)^{-z}\phi(s)\frac{\d s}{s},&\Re z>\frac{1}{2}
	\end{cases},\ \ \ \ \ y_2(x,z):=\begin{cases}\displaystyle\int_0^x\Big(\frac{x}{s}\Big)^{z-1}\psi(s)\frac{\d s}{s},&\Re z<\frac{1}{2}\bigskip\\
	\displaystyle-\int_x^{\infty}\Big(\frac{x}{s}\Big)^{z-1}\psi(s)\frac{\d s}{s},&\Re z>\frac{1}{2}
	\end{cases},
\end{equation}
and let $y_k^z$ denote the function $\mathbb{R}_+\ni x\mapsto y_k^z(x):=y_k(x,z),k=1,2$. Then with $\Sigma:=\frac{1}{2}+\im\mathbb{R}$,
\begin{enumerate}
	\item[(a)] $y_k^z\in L^2(\mathbb{R}_+),k=1,2$ for every $z\in\mathbb{C}\setminus\Sigma$. 
	\item[(b)] $\mathbb{C}\setminus\Sigma\ni z\mapsto y_k^z(x),k=1,2$ are analytic for all $x\in\mathbb{R}_+$ and extend continuously to the closed left and right, with respect to $\Sigma$, half-planes. For $z\in\Sigma$ we record the, pointwise in $x\in\mathbb{R}_+$, limits
	\begin{equation}\label{bp20}
		\lim_{\epsilon\downarrow 0}y_1^{z-\epsilon}(x)-\lim_{\epsilon\downarrow 0}y_1^{z+\epsilon}(x)=-x^{-z}r_1(z),\ \ \ \lim_{\epsilon\downarrow 0}y_2^{z-\epsilon}(x)-\lim_{\epsilon\downarrow 0}y_2^{z+\epsilon}(x)=x^{z-1}r_2(z),
	\end{equation}
	with $r_k(z),k=1,2,$ defined in \eqref{b:13}.
	\item[(c)] We have
	\begin{equation}\label{bp21}
		y_1^z(x)=\frac{1}{z}\phi(x)+m_1^z(x),\ \ \ \ \ y_2^z(x)=\frac{1}{1-z}\psi(x)+m_2^z(x),\ \ \ \ \ \ \ \ \ (x,z)\in\mathbb{R}_+\times(\mathbb{C}\setminus\Sigma),
	\end{equation}
	where $m_k^z\in L^2(0,1),k=1,2$ for every $z\in\mathbb{C}\setminus\Sigma$. In fact, there exist $c_k>0$ so that for all $z\in\mathbb{C}\setminus\Sigma$
	\begin{equation*}
		\|m_1^z\|_{L^2(0,1)}\leq\frac{c_1}{|z|\sqrt{|\Re z-\frac{1}{2}|}}\|MD\phi\|_{L_{\circ}^1(\mathbb{R}_+)},\ \ \ \ \ \|m_2^z\|_{L^2(0,1)}\leq\frac{c_2}{|1-z|\sqrt{|\Re z-\frac{1}{2}|}}\|MD\psi\|_{L_{\circ}^1(\mathbb{R}_+)}.
	\end{equation*}
\end{enumerate}

\end{prop}
\begin{proof} By construction \eqref{bp19}, for all $(x,z)\in\mathbb{R}_+\times(\mathbb{C}\setminus\Sigma)$,
\begin{equation*}
	|y_1(x,z)|\leq\frac{\|\phi\|_{L_{\circ}^1(\mathbb{R}_+)}}{\sqrt{x}}\ \ \ \ \ \ \textnormal{and}\ \ \ \ \ \ |y_2(x,z)|\leq\frac{\|\psi\|_{L_{\circ}^1(\mathbb{R}_+)}}{\sqrt{x}},
\end{equation*}
so the functions $y_k^z$ are well-defined for $z\in\mathbb{C}\setminus\Sigma$. More is true, for fixed $z\in\mathbb{C}\setminus\Sigma$ with $\delta:=|\Re z-\frac{1}{2}|>0$,
\begin{equation*}
	|y_1^z(x)|\leq\int_0^{\infty}\e^{-\delta|\ln s|}|\phi(sx)|\frac{\d s}{\sqrt{s}}=(E\ast|\phi|)(x),\ \ |y_2^z(x)|\leq\int_0^{\infty}\e^{-\delta|\ln s|}|\psi(sx)|\frac{\d s}{\sqrt{s}}=(E\ast|\psi|)(x),\ \ x\in\mathbb{R}_+,
\end{equation*}
with $(f\ast g)(x):=\int_0^{\infty}f(\frac{x}{y})g(y)\frac{\d y}{y}$ and where $E(x):=\e^{-\delta|\ln x|}\frac{1}{\sqrt{x}}$. Hence, by the $L_{\circ}^1(\mathbb{R}_+)$ Mellin convolution theorem, cf. \cite[Theorem $3$]{BJ}, for any $\xi\in\Sigma$,
\begin{equation*}
	\widehat{E\ast|\phi|}(\xi)=\widehat{E}(\xi)\widehat{|\phi|}(\xi)=\frac{2\delta}{\delta^2-(\xi-\frac{1}{2})^2}\widehat{|\phi|}(\xi),\ \ \ \ \ \ \widehat{E\ast|\psi|}(\xi)=\widehat{E}(\xi)\widehat{|\psi|}(\xi)=\frac{2\delta}{\delta^2-(\xi-\frac{1}{2})^2}\widehat{|\psi|}(\xi),
\end{equation*}
where the Mellin transform $\hat{f}$ of $f\in L_{\circ}^1(\mathbb{R}_+)$ is defined by
\begin{equation*}
	\hat{f}(\xi):=\int_0^{\infty}f(x)x^{\xi-1}\,\d x,\ \ \ \ \xi\in\Sigma.
\end{equation*}
Consequently, by Plancherel's theorem for the Mellin transform, cf. \cite[Theorem $71$]{Tit}, using that $\phi,\psi\in L_{\circ}^1(\mathbb{R}_+)\cap L^2(\mathbb{R}_+)$ and that $E\ast|\phi|,E\ast|\psi|\in L_{\circ}^1(\mathbb{R}_+)\cap L^2(\mathbb{R}_+)$,
\begin{equation*}
	\|y_1^z\|_{L^2(\mathbb{R}_+)}\leq\|E\ast|\phi|\|_{L^2(\mathbb{R}_+)}\leq c\|\widehat{E\ast|\phi|}\|_{L^2(\Sigma)}\leq c\|\widehat{|\phi|}\|_{L^2(\Sigma)}\leq c\|\phi\|_{L^2(\mathbb{R}_+)}<\infty,\ \ \ c=c(\delta)>0.
\end{equation*}
An analogous bound holds for $\|y_2^z\|_{L^2(\mathbb{R}_+)}$, so indeed $y_k^z\in L^2(\mathbb{R}_+)$ for every $z\in\mathbb{C}\setminus\Sigma$, as claimed in (a). Next, $\mathbb{C}\setminus\Sigma\ni z\mapsto y_k^z(x)$ is continuous for any $x\in\mathbb{R}_+$ by the dominated convergence theorem using $\phi,\psi\in L_{\circ}^1(\mathbb{R}_+)$. Moreover, for any triangle $T$ strictly contained in the left or right half-plane with respect to $\Sigma$, by Fubini's theorem,
\begin{equation*}
	\oint_Ty_k^z(x)\,\d z=0\ \ \ \ \forall\,x\in\mathbb{R}_+,\ \ k=1,2,
\end{equation*}
so the same functions $z\mapsto y_k^z(x)$ are analytic on the separate half-planes by Morera's theorem, for all $x\in\mathbb{R}_+$. Additionally, the existence of their continuous extensions up to the critical line $\Sigma$ and in turn the relations \eqref{bp20} are a direct consequence of the dominated convergence theorem, of the assumption $\phi,\psi\in L_{\circ}^1(\mathbb{R}_+)$ and formula \eqref{bp19}. This establishes (b). To get to \eqref{bp21}, we integrate by parts,
\begin{equation*}
	y_1^z(x)=\frac{1}{z}\phi(x)+m_1^z(x),\ \ \ \ \ m_1^z(x):=\frac{1}{z}\begin{cases}\displaystyle\int_x^{\infty}\Big(\frac{x}{s}\Big)^{-z}(MD\phi)(s)\frac{\d s}{s},&\Re z<\frac{1}{2}\bigskip\\
	\displaystyle -\int_0^x\Big(\frac{x}{s}\Big)^{-z}(MD\phi)(s)\frac{\d s}{s},&\Re z>\frac{1}{2}
	\end{cases},
\end{equation*}
and
\begin{equation*}
	y_2^z(x)=\frac{1}{1-z}\psi(x)+m_2^z(x),\ \ \ \ \ m_2^z(x):=\frac{1}{1-z}\begin{cases}\displaystyle-\int_0^x\Big(\frac{x}{s}\Big)^{z-1}(MD\psi)(s)\frac{\d s}{s},&\Re z<\frac{1}{2}\bigskip\\
	\displaystyle\int_x^{\infty}\Big(\frac{x}{s}\Big)^{z-1}(MD\psi)(s)\frac{\d s}{s},&\Re z>\frac{1}{2}
	\end{cases}.
\end{equation*}
Here, $|zm_1^z(x)|\leq\|MD\phi\|_{L_{\circ}^1(\mathbb{R}_+)}/\sqrt{x},|(1-z)m_2^z(x)|\leq\|MD\psi\|_{L_{\circ}^1(\mathbb{R}_+)}/\sqrt{x}$ for all $(x,z)\in\mathbb{R}_+\times(\mathbb{C}\setminus\Sigma)$ and
\begin{equation*}
	|m_1^z(x)|\leq\frac{1}{|z|}\big(E\ast|MD\phi|\big)(x),\ \ \ \ \ |m_2^z(x)|\leq\frac{1}{|1-z|}\big(E\ast|MD\psi|\big)(x),\ \ \ \ (x,z)\in\mathbb{R}_+\times(\mathbb{C}\setminus\Sigma).
\end{equation*}
Thus, by Plancherel's theorem for the Mellin transform, using $MD\phi,MD\psi,E\ast|MD\phi|,E\ast|MD\psi|\in L_{\circ}^1(\mathbb{R}_+)\cap L^2(0,1)$,
\begin{equation*}
	\|m_1^z\|_{L^2(0,1)}\leq\frac{1}{|z|}\|\chi_{(0,1)}(E\ast|MD\phi|)\|_{L^2(\mathbb{R}_+)}\leq\frac{c}{|z|\sqrt{|\Re z-\frac{1}{2}|}}\|MD\phi\|_{L_{\circ}^1(\mathbb{R}_+)},\ \ c>0,
\end{equation*}
where $\chi_{(0,1)}$ is the characteristic function on $(0,1)\subset\mathbb{R}_+$ and where we use the explicit estimate
\begin{equation*}
	\left|\int_0^{\infty}\big(\chi_{(0,1)}(E\ast|MD\phi|)\big)(x)x^{\xi-1}\,\d x\right|\leq 2\left[\frac{1}{\sqrt{\delta^2+|\xi-\frac{1}{2}|^2}}+\frac{\delta}{\delta^2+|\xi-\frac{1}{2}|^2}\right]\|MD\phi\|_{L_{\circ}^1(\mathbb{R}_+)},\ \ \ \xi\in\Sigma,
\end{equation*}
with $\delta=|\Re z-\frac{1}{2}|>0$. A similar estimate holds for $\|m_2^z\|_{L^2(0,1)}$ and this concludes our proof.
\end{proof}
Equipped with Proposition \ref{crucial2} we can now solve RHP \ref{masterb} and thus prove Theorem \ref{theo3}.
\begin{proof}[Proof of Theorem \ref{theo3}] Fix $t\in\mathbb{R}_+$, abbreviate $\Sigma:=\frac{1}{2}+\im\mathbb{R}$ and consider the $2\times 2$ matrix-valued function
\begin{equation}\label{bp22}
	{\bf X}(z):=\mathbb{I}+\begin{bmatrix}\displaystyle t\int_0^1\big((I-K_t^{\ast})^{-1}\tau_t\psi\big)(x)(\tau_ty_1^z)(x)\,\d x & -t\big((I-K_t^{\ast})^{-1}\tau_ty_2^z\big)(1)\smallskip\\
	\displaystyle -\big((I-K_t)^{-1}\tau_ty_1^z\big)(1) & \displaystyle t\int_0^1\big((I-K_t)^{-1}\tau_t\phi\big)(x)(\tau_ty_2^z)(x)\,\d x\end{bmatrix},
\end{equation}
defined for $z\in\mathbb{C}\setminus\Sigma$ with $y_k^z,k=1,2$ as in \eqref{bp19} and $(\tau_ty_k^z)(x)=y_k^z(xt)$. We first show that the first column entries of \eqref{bp22} are analytic in $\mathbb{C}\setminus\Sigma$, the second column can afterwards be treated in the same way. Since $\psi\in L_{\circ}^1(\mathbb{R}_+)\cap L^2(\mathbb{R}_+)$ and $y_1^z\in L^2(\mathbb{R}_+)$ for $z\in\mathbb{C}\setminus\Sigma$ by Proposition \ref{crucial2},
\begin{equation}\label{bp23}
	\mathbb{C}\setminus\Sigma\ni z\mapsto X^{11}(z)\stackrel{\eqref{bp22}}{=}1+t\int_0^1\big((I-K_t^{\ast})^{-1}\tau_t\psi\big)(x)(\tau_ty_1^z)(x)\,\d x
\end{equation}
is well-defined. But by \eqref{b:1},\eqref{b:15}, $(I-K_t^{\ast})^{-1}\tau_t\psi\in L^1_{\circ}(0,1)$ since $I-K_t^{\ast}$ is invertible on $L^2(0,1)$ and since $\psi\in L_{\circ}^1(\mathbb{R}_+)\cap L^2(\mathbb{R}_+)$. Hence, using the dominated convergence theorem, the estimate $|\sqrt{x}y_1^z(x)|\leq\|\phi\|_{L_{\circ}^1(\mathbb{R}_+)}$ and Proposition \ref{crucial2}, we conclude that $\mathbb{C}\setminus\Sigma\ni z\mapsto X^{11}(z)$ in \eqref{bp23} is continuous. Moreover, for any triangle $T$ strictly contained in the left or right half-plane with respect to $\Sigma$, by Fubini's theorem,
\begin{equation*}
	\oint_TX^{11}(z)\,\d z=0,
\end{equation*}
so $\mathbb{C}\setminus\Sigma\ni z\mapsto X^{11}(z)$ is analytic by Morera's theorem. Next, look at
\begin{equation}\label{bp24}
	\mathbb{C}\setminus\Sigma\ni z\mapsto X^{21}(z)\stackrel{\eqref{bp22}}{=}-\big((I-K_t)^{-1}\tau_ty_1^z\big)(1).
\end{equation}
Since $(0,1]\ni x\mapsto ((I-K_t)^{-1}\tau_ty_1^z)(x)$ is locally $\frac{1}{2}$-H\"older continuous for any $z\in\mathbb{C}\setminus\Sigma$ by \eqref{b:15} and Proposition \ref{crucial2}, the function \eqref{bp24} is well-defined. In order to establish its analyticity in the separate half-planes with respect to $\Sigma$ we consider the auxiliary function
\begin{equation}\label{bp25}
	\mathbb{C}\setminus\Sigma\ni z\mapsto f_n(z):=\int_0^1\delta_n(x)\big((I-K_t)^{-1}\tau_ty_1^z\big)(x)\,\d x;\ \ \ \ \delta_n(x):=\frac{2n}{\sqrt{\pi}}\e^{-n^2(1-x)^2}.
\end{equation}
Here, $\delta_n\in L_{\circ}^1(\mathbb{R}_+)\cap L^2(\mathbb{R}_+),n\in\mathbb{Z}_{\geq 1}$ with $\int_0^1\delta_n(x)\,\d x=1-2\int_n^{\infty}\e^{-x^2}\,\d x/\sqrt{\pi}$ and $z\mapsto f_n(z)$ is well-defined by invertibility of $I-K_t$ on $L^2(0,1)$ and by Proposition \ref{crucial2}. More is true, $(f_n)_{n=1}^{\infty}$ is a sequence of analytic functions, analytic in the separate half-planes with respect to $\Sigma$: indeed we have, for any $n\in\mathbb{Z}_{\geq 1}$,
\begin{equation}\label{bp26}
	f_n(z)\stackrel{\eqref{bp25}}{=}\int_0^1\big((I-K_t^{\ast})^{-1}\delta_n\big)(x)(\tau_t y_1^z)(x)\,\d x,\ \ \ \ z\in\mathbb{C}\setminus\Sigma,
\end{equation}
where $(I-K_t^{\ast})^{-1}\delta_n\in L_{\circ}^1(\mathbb{R}_+)$, by \eqref{b:15} and invertibility of $I-K_t$ on $L^2(0,1)$, and where $|\sqrt{x}y_1^z(x)|\leq\|\phi\|_{L_{\circ}^1(\mathbb{R}_+)}$. So, $\mathbb{C}\setminus\Sigma$ is continuous by the dominated convergence theorem and by Proposition \ref{crucial2}. Moreover, by Fubini's theorem
\begin{equation*}
	\oint_Tf_n(z)\,\d z\stackrel{\eqref{bp26}}{=}0
\end{equation*}
for any triangle $T$ strictly contained in the left or right half-plane. This implies analyticity of $\mathbb{C}\setminus\Sigma\ni z\mapsto f_n(z)$ by Morera's theorem. Next, using $(I-K_t)^{-1}=I+K_t(I-K_t)^{-1}$, we have for any $z\in\mathbb{C}\setminus\Sigma$,
\begin{align}
	&f_n(z)+X^{21}(z)=\int_0^1\delta_n(x)\big((\tau_ty_1^z)(x)-(\tau_ty_1^z)(1)\big)\,\d x-\big((I-K_t)^{-1}\tau_ty_1^z\big)(1)\left[1-\int_0^1\delta_n(x)\,\d x\right]\nonumber\\
	&+\int_0^1\delta_n(x)\left[\int_0^1\big(K_t(x,y)-K_t(1,y)\big)\big((I-K_t)^{-1}\tau_ty_1^z\big)(y)\,\d y\right]\d x\equiv f_{1n}(z)+f_{2n}(z)+f_{3n}(z).\label{bp27}
\end{align}
But for any $a,b\in\mathbb{R}_+$,
\begin{equation}\label{bp28}
	\big|(\tau_ty_1^z)(a)-(\tau_ty_1^z)(b)\big|\leq\frac{|a-b|}{\sqrt{ab(a+b)}}\underbrace{\frac{2}{\sqrt{t}}\big(|z|\,\|\phi\|_{L_{\circ}^1(\mathbb{R}_+)}+\|\sqrt{(\cdot)}\phi\|_{L^{\infty}(\mathbb{R}_+)}\big)}_{=:L_{z,\phi}},\ \ \ z\in\mathbb{C}\setminus\Sigma,
\end{equation}
followed by, for any $(a,b,z)\in(0,1]^3$, now using \eqref{b:15},
\begin{equation}\label{bp29}
	\big|K_t(a,z)-K_t(b,z)\big|\leq c_t\sqrt{\frac{|a-b|}{ab}}\sqrt{\int_0^1|(\tau_t\psi)(\lambda z)|^2\,\d\lambda},\ \ \ c_t>0,\ \ \ \ \ \ z\in\mathbb{C}\setminus\Sigma,
\end{equation}
so that all together, for any $z\in\mathbb{C}\setminus\Sigma$ and $n\geq n_0$ sufficiently large,
\begin{equation*}
	|f_{1n}(z)|\stackrel{\eqref{bp28}}{\leq}\frac{4L_{z,\phi}}{n\sqrt{\pi t}},\ \ \  |f_{2n}(z)|\leq\frac{2}{n\sqrt{\pi}}\e^{-n^2}\big|\big((I-K_t)^{-1}\tau_ty_1^z\big)(1)\big|,\ \ \ |f_{3n}(z)|\leq\frac{d_t}{\sqrt{n}}\|(I-K_t)^{-1}\tau_ty_1^z\|_{L^2(0,1)}.
\end{equation*}
Here, $d_t>0$ is $(z,n)$-independent and we have by the workings in Proposition \ref{crucial2}, $\|(I-K_t)^{-1}\tau_ty_1^z\|_{L^2(0,1)}\leq 2\|(I-K_t)^{-1}\|\|\phi\|_{L^2(\mathbb{R}_+)}/|\Re z-\frac{1}{2}|$. Thus, for any compact $K\subset\mathbb{C}\setminus\Sigma$ in the left half-plane, say, the above \eqref{bp27} together with the estimates for $f_{kn}$ yield
\begin{equation*}
	\sup_{z\in K}\big|f_n(z)+X^{21}(z)\big|\rightarrow 0\ \ \textnormal{as}\ n\rightarrow\infty,
\end{equation*}
i.e. $z\mapsto X^{21}(z)$ is analytic in the left half-plane, say, by \cite[Theorem $5.2$]{SS}. This proves that \eqref{bp22} satisfies property $(1)$ of RHP \ref{masterb}. Moving to property $(2)$, we compute via \eqref{bp19}, the dominated convergence theorem and Fubini's theorem, for any $z\in\Sigma$,
\begin{align*}
	\lim_{\epsilon\downarrow 0}X^{11}(z-\epsilon)=&\,1-r_1(z)t^{-z}\int_0^1\big((I-K_t^{\ast})^{-1}\tau_t\psi\big)(x)x^{-z}\,\d x\\
	&\hspace{1.5cm}+t\int_0^1u^{z-1}\left[\int_0^1(\tau_{tu}\phi)(x)\big((I-K_t^{\ast})^{-1}\tau_t\psi\big)(x)\,\d x\right]\d u,\\
	\lim_{\epsilon\downarrow 0}X^{11}(z+\epsilon)=&\,1+t\int_0^1u^{z-1}\left[\int_0^1(\tau_{tu}\phi)(x)\big((I-K_t^{\ast})^{-1}\tau_t\psi\big)(x)\,\d x\right]\d u.
\end{align*}
These limiting values are well-defined and continuous in $z\in\Sigma$ given that $(I-K_t^{\ast})^{-1}\tau_t\psi\in L_{\circ}^1(0,1)$ and $\phi\in L_{\circ}^1(\mathbb{R}_+)$. In fact, using the algebraic Lemma \ref{new6}, we can simplify them further to obtain that for any $z\in\Sigma$,
\begin{align}
	X_+^{11}(z)=&\,1+\big((I-K_t^{\ast})^{-1}g_z\big)(1)-r_1(z)t^{1-z}\int_0^1\big((I-K_t^{\ast})^{-1}\tau_t\psi\big)(x)x^{-z}\,\d x\nonumber\\
	X_-^{11}(z)=&\,1+\big((I-K_t^{\ast})^{-1}g_z\big)(1);\ \ \ \ \ \ \textnormal{with}\ \ \ \ \ g_z(x):=\int_0^1K_t^{\ast}(x,u)u^{z-1}\,\d u,\ \ \ x\in(0,1].\label{bp30}
\end{align}
Here, $(0,1]\ni x\mapsto g_z(x)$ is well-defined and continuous for all $z\in\Sigma$ by \eqref{b:15} and since $\psi\in L_{\circ}^1(\mathbb{R}_+)\cap L^2(\mathbb{R}_+)$. Also, $g_z\in L^2(0,1)$ by \eqref{b:1} and \eqref{b:15}, so $(0,1]\ni x\mapsto((I-K_t^{\ast})^{-1}g_z)(x)$ is continuous and thus $((I-K_t^{\ast})^{-1}g_z)(1)$ in \eqref{bp30} bona fide. By similar logic, using this time \eqref{app8} instead of \eqref{app9}, we also find
\begin{align}
	X_+^{22}(z)=&\,1+\big((I-K_t)^{-1}h_z\big)(1);\ \ \ \ \ \ \textnormal{with}\ \ \ \ \ h_z(x):=\int_0^1K_t(x,u)u^{-z}\,\d u,\ \ \ x\in(0,1],\nonumber\\
	X_-^{22}(z)=&\,1+\big((I-K_t)^{-1}h_z\big)(1)-r_2(z)t^z\int_0^1\big((I-K_t)^{-1}\tau_t\phi\big)(x)x^{z-1}\,\d x.\label{bp31}
\end{align}
Next, for the off-diagonal entries, $X^{21}(z)$, say, we write $(I-K_t)^{-1}=I+(I-K_t)^{-1}K_t$ and note that
\begin{equation*}
	(0,1]\ni x\mapsto (K_t\tau_ty_1^z)(x)=\int_0^1K_t(x,u)(\tau_ty_1^z)(u)\,\d u
\end{equation*}
is in $L_{\circ}^1(0,1)\cap L^2(0,1)$, for all $z\in\mathbb{C}\setminus\Sigma$, by \eqref{b:1},\eqref{b:15} and since $\phi\in L_{\circ}^1(\mathbb{R}_+)$. Moreover, by the dominated convergence theorem, with $z\in\Sigma$,
\begin{equation*}
	\lim_{\epsilon\downarrow 0}(K_t\tau_ty_1^{z-\epsilon})(x)=-t^{-z}r_1(z)h_z(x)+\int_0^1v^{z-1}(K_t\tau_{tv}\phi)(x)\,\d v\ \ \ \ \forall\,x\in(0,1],
\end{equation*}
noting that $(0,1]\ni v\mapsto (K_t\tau_{tv}\phi)(x)=(K_t\tau_{tx}\phi)(v)$ is in $L_{\circ}^1(0,1)$ almost everywhere $x\in(0,1]$ and that
\begin{equation*}
	(0,1]\ni x\mapsto \int_0^1v^{z-1}(K_t\tau_{tv}\phi)(x)\,\d v=:f(x)
\end{equation*}
is in $L_{\circ}^1(0,1)$ and $\sqrt{(\cdot)}f\in L^{\infty}(\mathbb{R}_+)$, thus $f\in L^2(0,1)$. Hence, by the dominated convergence theorem, for all $z\in\Sigma$, using $h_z\in L^2(0,1)$ and Lemma \ref{new5},
\begin{equation}\label{bp32}
	X_+^{21}(z)=-\int_0^1\big((I-K_t)^{-1}\tau_t\phi\big)(x)x^{z-1}\,\d x+r_1(z)t^{-z}+r_1(z)t^{-z}\big((I-K_t)^{-1}h_z\big)(1),
\end{equation}
which is well-defined and continuous in $z\in\Sigma$. By similar logic,
\begin{equation}\label{bp33}
	X_-^{21}(z)=-\int_0^1\big((I-K_t)^{-1}\tau_t\phi\big)(x)x^{z-1}\,\d x,
\end{equation}
followed by
\begin{align}
	X_+^{12}(z)=&\,-t\int_0^1\big((I-K_t^{\ast})^{-1}\tau_t\psi\big)(x)x^{-z}\,\d x\nonumber\\
	X_-^{12}(z)=&\,-t\int_0^1\big((I-K_t^{\ast})^{-1}\tau_t\psi\big)(x)x^{-z}\,\d x+r_2(z)t^z+r_2(z)t^z\big((I-K_t^{\ast})^{-1}g_z\big)(1),\label{bp34}
\end{align}
which are all well-defined and continuous functions in $z\in\Sigma$. Using now \eqref{bp30},\eqref{bp31},\eqref{bp32},\eqref{bp33} and \eqref{bp34} we directly verify the identity
\begin{equation*}
	{\bf X}_-(z)\begin{bmatrix}1-r_1(z)r_2(z) & -r_2(z)t^z\smallskip\\ r_1(z)t^{-z} & 1\end{bmatrix}={\bf X}_+(z),\ \ \ \ z\in\Sigma,
\end{equation*}
i.e. \eqref{bp22} also satisfies property $(2)$ of RHP \ref{masterb}. Finally, by Proposition \ref{crucial2} and Cauchy-Schwarz inequality,
\begin{equation}\label{bp35}
	\left\|{\bf X}(z)-\mathbb{I}-\frac{1}{z}\begin{bmatrix}p_0 & q_0^{\ast}\\ - q_0 & -p_0^{\ast}\end{bmatrix}\right\|\leq\frac{c_t}{|z|\sqrt{|\Re z-\frac{1}{2}|}},\ \ c_t>0,\ \ \ z\neq 1,
\end{equation}
which shows that ${\bf X}(z)$ in \eqref{bp22} satisfies \eqref{b:14} and in turn \eqref{b:16} as $z\rightarrow\infty$ along any non-vertical direction. 
Our proof of Theorem \ref{theo3} is now complete.
\end{proof}
\begin{proof}[Proof of Corollary \ref{impcor2}] This is immediate from \eqref{b:16} since $X_1^{11}(t,\phi,\psi)=p_0(t)$ and $X_1^{12}(t,\phi,\psi)X_1^{21}(t,\phi,\psi)=q_0(t)q_0^{\ast}(t)$ for all $t\in J$ under the given assumptions. Now use the last two equations in \eqref{b:6} and \eqref{bp1} which says $t\frac{\d}{\d t}\ln F(t)=-p_0(t),t\in J$, subject to the necessary assumptions.
\end{proof}
\begin{proof}[Proof of Corollary \ref{deeper2}] Under the assumptions made, repetitive integration by parts yields for $z\notin(\frac{1}{2}+\im\mathbb{R})$,
\begin{align}
	(\tau_ty_1^z)(x)=&\,\sum_{k=1}^N(-1)^{k-1}\big((MD)^{k-1}\tau_t\phi\big)(x)z^{-k}+m_{1N}^z(x;t),\nonumber\\
	(\tau_ty_2^z)(x)=&\,\sum_{k=1}^N(-1)^{k-1}\big((MD)^{k-1}\tau_t\psi\big)(x)(1-z)^{-k}+m_{2N}^z(x;t),\label{bp36}
\end{align}
where, with $c_k=c_k(t)>0$,
\begin{align}
	\|m_{1N}^z(\cdot;t)\|_{L^2(0,1)}\leq&\,\frac{c_1}{|z|^N\sqrt{|\Re z-\frac{1}{2}|}}\|(MD)^N\phi\|_{L_{\circ}^1(\mathbb{R}_+)},\label{bp37}\\
	\|m_{2N}^z(\cdot;t)\|_{L^2(0,1)}\leq&\,\frac{c_2}{|z|^N\sqrt{|\Re z-\frac{1}{2}|}}\|(MD)^N\psi\|_{L_{\circ}^1(\mathbb{R}_+)}.\label{bp38}
\end{align}
But Taylor expanding $(1-z)^{-k}$ at $z=\infty$ we can convert $(MD)^{k-1}$ to $(DM)^{k-1}$ and so the leading orders in the expansion of $\tau_ty_2^z$ in \eqref{bp36} become
\begin{equation*}
	-\sum_{k=1}^N\big((DM)^{k-1}\tau_t\psi\big)(x)z^{-k}.
\end{equation*}
As in the proof of Theorem \ref{theo3}, the above is sufficient to conclude the validity of \eqref{b:17}, after inserting the stated decompositions for $\tau_ty_1^z$ and $\tau_ty_2^z$ into \eqref{bp22} and estimating the remainder with the Cauchy-Schwarz inequality using \eqref{bp37} and \eqref{bp38}. This completes our proof.
\end{proof}
We now move on to the more specialized Hankel composition operators of Section \ref{seci25}.
\section{Proof of Theorem \ref{theo4}}\label{sec8}
Once more, in order to derive \eqref{b:25},\eqref{b:26} and \eqref{b:27}, we require two auxiliary results of independent interest.
\begin{prop} Suppose $H_t\in\mathcal{L}(L^2(0,1))$ defined in \eqref{b:18} satisfies the conditions in Assumption \ref{ass2} with $\phi:\mathbb{R}_+\rightarrow\mathbb{C}$ continuously differentiable on $\mathbb{R}_+$ such that 
\begin{equation}\label{bp39}
	\lim_{x\downarrow 0}\sqrt{x}\,\phi(x)=0,
\end{equation}
and for every $t\in\mathbb{R}_+$, 
\begin{equation*}
		\int_0^1|(\tau_t\phi)|^2\ln\Big(\frac{1}{x}\Big)\,\d x<\infty,\ \ \ \int_0^1|(MD\tau_t\phi)(x)|^2\ln\Big(\frac{1}{x}\Big)\,\d x<\infty,\ \ \ \int_0^1\sqrt{\int_0^1|(\tau_t\phi)(xy)|^2\,\d y}\frac{\d x}{\sqrt{x}}<\infty,
\end{equation*}	
\begin{equation}\label{bp40}
	\int_0^1\sqrt{\int_0^1|(MD\tau_t\phi)(xy)|^2\,\d y}\frac{\d x}{\sqrt{x}}<\infty.
\end{equation}
Then we have, provided the families $\{\tau_t\phi\}_{t\in\mathbb{R}_+}$ and $\{MD\tau_t\phi\}_{t\in\mathbb{R}_+}$ are $L^2(0,1)$ dominated and provided there exist $c,t_0>0$ such that $|q(t,\gamma)|\leq ct^{-\frac{1}{2}+\epsilon}$ for all $0<t\leq t_0^{-1},\gamma\in[0,1]$ with some $\epsilon>0$,
\begin{equation}\label{bp41a}
	\ln G(t,\sqrt{\gamma})=\ln G(t,-\sqrt{\gamma})-\int_0^tq(s,\gamma)\frac{\d s}{\sqrt{s}},\ \ \ \ \ (t,\gamma)\in\mathbb{R}_+\times[0,1],
\end{equation}
with branches for the logarithms that are analytically continued along the orbits of the corresponding functions.\end{prop}
\begin{proof} By assumption, $I-\gamma K_t=I-\gamma H_t^2=(I-\sqrt{\gamma}H_t)(I+\sqrt{\gamma}H_t)$ is invertible on $L^2(0,1)$, so $I\mp\sqrt{\gamma}H_t$ are likewise invertible and thus $G(t,\pm\sqrt{\gamma})\neq 0$ for all $(t,\gamma)\in\mathbb{R}_+\times[0,1]$. But $\mathbb{R}_+\ni t\mapsto H_t$ is differentiable by assumption with $\frac{\d}{\d t}H_t$ trace class, so by Jacobi's formula, after simplification
\begin{align}
	t\frac{\d}{\d t}\ln G(t,\sqrt{\gamma})-t&\,\frac{\d}{\d t}\ln G(t,-\sqrt{\gamma})=-2\sqrt{\gamma}\tr_{L^2(0,1)}\bigg((I-\gamma K_t)^{-1}\bigg[t\frac{\d}{\d t}H_t\bigg]\bigg)\nonumber\\
	=&\,-2\sqrt{\gamma}\tr_{L^2(0,1)}\bigg(t\frac{\d}{\d t}H_t\bigg)-2\sqrt{\gamma}\tr_{L^2(0,1)}\bigg((I-\gamma K_t)^{-1}\gamma K_t\bigg[t\frac{\d}{\d t}H_t\bigg]\bigg).\label{bp41}
\end{align}
Here, by \eqref{bp39}, $K_t\big[t\frac{\d}{\d t}H_t\big]=-\frac{1}{2}H_tK_t+tH_t(\tau_t\phi\otimes\tau_t\phi)-H_tMDK_t$ where $MDK_t$ is trace class on $L^2(0,1)$ by \eqref{bp40}. Thus,
\begin{align}	
	\textnormal{LHS}\ &\eqref{bp41}=-2\sqrt{\gamma}\tr_{L^2(0,1)}\bigg(t\frac{\d}{\d t}H_t\bigg)+\sqrt{\gamma}\tr_{L^2(0,1)}\big((I-\gamma K_t)^{-1}\gamma H_tK_t\big)\nonumber\\
	&\,-2\sqrt{\gamma }\tr_{L^2(0,1)}\big((I-\gamma K_t)^{-1}\gamma tH_t(\tau_t\phi\otimes\tau_t\phi)\big)+2\sqrt{\gamma}\tr_{L^2(0,1)}\big((I-\gamma K_t)^{-1}\gamma  H_tMDK_t\big).\label{bp42}
\end{align}
Observe that the kernels of the trace class operators $t\frac{\d}{\d t}H_t$ and $(I-\gamma K_t)^{-1}\gamma tH_t(\tau_t\phi\otimes\tau_t\phi)$ are continuous on $(0,1)\times(0,1)$ as $\phi$ is continuously differentiable and integrable on $(0,1)$, and as \eqref{bp40} is in place. Thus by \cite[$(3.4)$]{Bris}, \cite[Theorem V.3.1.1]{Du} and \eqref{bp39},
\begin{align}
	-2\sqrt{\gamma}\tr_{L^2(0,1)}&\bigg(t\frac{\d}{\d t}H_t\bigg)-2\sqrt{\gamma t}\tr_{L^2(0,1)}\big((I-\gamma K_t)^{-1}\gamma tH_t(\tau_t\phi\otimes\tau_t\phi)\big)=\sqrt{\gamma t}(\tau_t\phi)(1)\nonumber\\
	&\,\,\,-2\sqrt{\gamma t}\big((I-\gamma K_t)^{-1}\tau_t\phi\big)(1)=\sqrt{\gamma t}\,(\tau_t\phi)(1)-2\sqrt{t}\,q(t,\gamma),\ \ \ \ (t,\gamma)\in\mathbb{R}_+\times[0,1],\label{bp43}
\end{align}
where we also used $K_t=K_t^{\ast}$ and 
\begin{equation*}
	\big((I-\gamma K_t)^{-1}\gamma tH_t\tau_t\phi\big)(x)=\sqrt{t}\,\big(\gamma K_t(I-\gamma K_t)^{-1}\big)(1,x),\ \ \ x\in(0,1).
\end{equation*}
Next, $MDK_t\in\mathcal{C}_1(L^2(0,1))$ by \eqref{bp40} and $(I-\gamma K_t)^{-1}\gamma H_t\in\mathcal{L}(L^2(0,1))$. Hence by \cite[Corollary $3.8$]{S},
\begin{align}
	&\tr_{L^2(0,1)}\big((I-\gamma K_t)^{-1}\gamma H_tMDK_t\big)=\tr_{L^2(0,1)}\big(MD\gamma K_t(I-\gamma K_t)^{-1}H_t\big)\nonumber\\
	=&\,\tr_{L^2(0,1)}\big(MD(I-\gamma K_t)^{-1}H_t\big)-\tr_{L^2(0,1)}(MDH_t)=\tr_{L^2(0,1)}\big(MDH_t(I-\gamma K_t)^{-1}\big)-\tr_{L^2(0,1)}(MDH_t),\label{bp44}
\end{align}
since $(I-\gamma K_t)^{-1}H_t=H_t(I-\gamma K_t)^{-1}$ and where $MDH_t=t\frac{\d}{\d t}H_t-\frac{1}{2}H_t$. Also, by \eqref{bp39},
\begin{equation}\label{bp45}
	\tr_{L^2(0,1)}(MDH_t)=\frac{1}{2}\sqrt{t}(\tau_t\phi)(1)-\frac{1}{2}\tr_{L^2(0,1)}H_t,
\end{equation}
and since $(I-\gamma K_t)^{-1}\gamma H_tK_t-H_t(I-\gamma K_t)^{-1}+H_t$ is traceless on $L^2(0,1)$ we can combine \eqref{bp42},\eqref{bp44} and \eqref{bp45} to obtain
\begin{equation*}
	\textnormal{LHS}\,\eqref{bp41}=-2\sqrt{t}\,q(t,\gamma)+2\sqrt{\gamma}\tr_{L^2(0,1)}\bigg(\bigg[t\frac{\d}{\d t} H_t\bigg](I-\gamma K_t)^{-1}\bigg),
\end{equation*}
and so all together, using \cite[Corollary $3.8$]{S} one more time,
\begin{equation*}
	t\frac{\d}{\d t}\ln G(t,\sqrt{\gamma})-t\frac{\d}{\d t}\ln G(t,-\sqrt{\gamma})=-\sqrt{t}\,q(t,\gamma),\ \ \ (t,\gamma)\in\mathbb{R}_+\times[0,1].
\end{equation*}
Now simply integrate
\begin{equation*}
	\ln G(t,\sqrt{\gamma})-\ln G(t,-\sqrt{\gamma})\stackrel{\eqref{i:22}}{=}\int_0^t\left\{s\frac{\d}{\d s}\ln G(s,\sqrt{\gamma})-s\frac{\d}{\d s}\ln G(s,-\sqrt{\gamma})\right\}\frac{\d s}{s}=-\int_0^tq(s,\gamma)\frac{\d s}{\sqrt{s}},
\end{equation*}
where we use the assumption $\|H_t\|_1\rightarrow 0$ as $t\downarrow 0$ in the first equality and the integrability of $t\mapsto q(t,\gamma)/\sqrt{t}$ near zero in the second. Recall that $q(t,\gamma)$ is at least continuous in $t\in\mathbb{R}_+$, see the workings in Corollary \ref{corb2}. This completes our proof of \eqref{bp41}.
\end{proof}
Observe that \eqref{bp41} relates $q(t,\gamma)/\sqrt{t}$ to a ratio of two Fredholm determinants. Below we state an alternative representation for the same function.
\begin{prop} Suppose $H_t\in\mathcal{L}(L^2(0,1))$ defined in \eqref{b:18} satisfies the conditions in Assumption \ref{ass2} with $\phi:\mathbb{R}_+\rightarrow\mathbb{C}$ continuously differentiable and $\phi\in L_{\circ}^1(0,1)$. Assume further that
\begin{equation}\label{bp46}
	\lim_{x\downarrow 0}\sqrt{x}\phi(x)=0,
\end{equation}
and that for every $t\in\mathbb{R}_+$,
\begin{equation*}
		\int_0^1|(\tau_t\phi)|^2\ln\Big(\frac{1}{x}\Big)\,\d x<\infty,\ \ \ \int_0^1|(MD\tau_t\phi)(x)|^2\ln\Big(\frac{1}{x}\Big)\,\d x<\infty,\ \ \ \int_0^1\sqrt{\int_0^1|(\tau_t\phi)(xy)|^2\,\d y}\frac{\d x}{\sqrt{x}}<\infty,
\end{equation*}	
\begin{equation}\label{bp47}
	\int_0^1\sqrt{\int_0^1|(MD\tau_t\phi)(xy)|^2\,\d y}\frac{\d x}{\sqrt{x}}<\infty.
\end{equation}
Then for all $(t,\gamma)\in\mathbb{R}_+\times[0,1]$,
\begin{equation}\label{bp48}
	\frac{\d}{\d t}\left[1\mp\sqrt{\gamma t}\int_0^1\big((I\pm\sqrt{\gamma}H_t)^{-1}\tau_t\phi\big)(x)\frac{\d x}{\sqrt{x}}\right]=\mp\frac{q(t,\gamma)}{\sqrt{t}}\left[1\mp\sqrt{\gamma t}\int_0^1\big((I\pm\sqrt{\gamma}H_t)^{-1}\tau_t\phi\big)(x)\frac{\d x}{\sqrt{x}}\right]
\end{equation}
provided $\{(\cdot)^{-1/2}(I-\gamma K_t)^{-1}\tau_t\phi\}_{t\in\mathbb{R}_+},\{(\cdot)^{-1/2}(I-\gamma K_t)^{-1}H_tMD\tau_t\phi\}_{t\in\mathbb{R}_+},\{(\cdot)^{-1/2}(I\pm\sqrt{\gamma}H_t)^{-1}\tau_t\phi\}_{t\in\mathbb{R}_+}$, $\{(\cdot)^{-1/2}(I\pm\sqrt{\gamma}H_t)^{-1}MD\tau_t\phi\}_{t\in\mathbb{R}_+}$ are $L^1(0,1)$ dominated and $\{\tau_t\phi\}_{t\in\mathbb{R}_+},\{MD\tau_t\phi\}_{t\in\mathbb{R}_+}$ are $L^2(0,1)$ dominated.
\end{prop}
\begin{proof} Since $\phi\in L_{\circ}^1(0,1)\cap L^2(0,1)$, and since $I-\gamma K_t$ is invertible on $L^2(0,1)$, the map $(0,1]\ni x\mapsto((I\pm\sqrt{\gamma}H_t)^{-1}\tau_t\phi)(x)$ is well-defined and by \eqref{bp47} in $L_{\circ}^1(0,1)\cap L^2(0,1)$. Thus
\begin{equation}\label{bp49}
	\mathbb{R}_+\ni t\mapsto\int_0^1\big((I\pm\sqrt{\gamma}H_t)^{-1}\tau_t\phi\big)(x)\frac{\d x}{\sqrt{x}}
\end{equation} 
exists and is moreover differentiable: indeed, by \eqref{bp47} and $L^2(0,1)$ dominance, $MDH_t\in\mathcal{L}(L^2(0,1))$ as well as $(I\pm\sqrt{\gamma}H_t)^{-1}\tau_t\phi\in H^{1,2}(0,1)$. Then, using \eqref{bp46},
\begin{align*}
	\big(MDH_t(I\pm\sqrt{\gamma}H_t)^{-1}\tau_t\phi\big)(x)=\sqrt{t}\,&(\tau_t\phi)(x)\big((I\pm\sqrt{\gamma}H_t)^{-1}\tau_t\phi\big)(1)-\big(H_t(I\pm\sqrt{\gamma}H_t)^{-1}\tau_t\phi\big)(x)\\
	&\,-\big(H_tMD(I\pm\sqrt{\gamma}H_t)^{-1}\tau_t\phi\big)(x),\ \ \ x\in(0,1],
\end{align*}
or equivalently, 
\begin{align}
	\big((I\mp\sqrt{\gamma}H_t)MDH_t(I\pm\sqrt{\gamma}H_t)^{-1}\tau_t\phi\big)(x)=-(H_tMD\tau_t\phi)(x)+&\sqrt{t}\,(\tau_t\phi)(x)\big((I\pm\sqrt{\gamma}H_t)^{-1}\tau_t\phi\big)(1)\nonumber\\
	&\,-\big(H_t(I\pm\sqrt{\gamma}H_t)^{-1}\tau_t\phi\big)(x).\label{bp50}
\end{align}
This yields that
\begin{align*}
	t\frac{\partial}{\partial t}\big((I\pm\sqrt{\gamma}H_t)^{-1}\tau_t\phi\big)(x)\frac{1}{\sqrt{x}}=\mp\sqrt{\gamma}\big((I\pm\sqrt{\gamma}H_t)^{-1}\left[t\frac{\d}{\d t}H_t\right](I\pm\sqrt{\gamma}&H_t)^{-1}\tau_t\phi\big)(x)\frac{1}{\sqrt{x}}\\
	&\,+\big((I\pm\sqrt{\gamma}H_t)^{-1}MD\tau_t\phi\big)(x)\frac{1}{\sqrt{x}}
\end{align*}
is continuous in $x\in(0,1]$ and $L^1(0,1)$ dominated by the imposed $L^1(0,1)$ dominance assumptions. Thus, \eqref{bp49} is differentiable by the dominated convergence theorem and we have
\begin{equation}\label{bp51}
	t\frac{\d}{\d t}\int_0^1\big((I\pm\sqrt{\gamma}H_t)^{-1}\tau_t\phi\big)(x)\frac{\d x}{\sqrt{x}}=\int_0^1t\frac{\partial}{\partial t}\big((I\pm\sqrt{\gamma}H_t)^{-1}\tau_t\phi\big)(x)\frac{\d x}{\sqrt{x}}.
\end{equation}
To evaluate the derivative in the left hand side of \eqref{bp48} we use \eqref{bp51} and note that by symmetry of $H_t$, $((I\pm\sqrt{\gamma}H_t)^{-1}\tau_t\phi)(x)=\frac{1}{\sqrt{t}}(H_t(I\pm\sqrt{\gamma}H_t)^{-1})(1,x)$ for any $x\in(0,1]$, so
\begin{equation*}
	t\frac{\d}{\d t}\sqrt{\gamma t}\int_0^1\big((I\pm\sqrt{\gamma}H_t)^{-1}\tau_t\phi\big)(x)\frac{\d x}{\sqrt{x}}=\mp\int_0^1t\frac{\partial}{\partial t}(I\pm\sqrt{\gamma}H_t)^{-1}(1,x)\frac{\d x}{\sqrt{x}}.
\end{equation*}
Consequently
\begin{align*}
	\textnormal{LHS}\ \eqref{bp48}=\mp\frac{\sqrt{\gamma}}{t}&\int_0^1\big((I\pm\sqrt{\gamma}H_t)^{-1}\left[\frac{1}{2}H_t+MDH_t\right](I\pm\sqrt{\gamma}H_t)^{-1}\big)(1,x)\frac{\d x}{\sqrt{x}}\\
	&\stackrel{\eqref{bp50}}{=}\mp\frac{q(t,\gamma)}{\sqrt{t}}\int_0^1(I\pm\sqrt{\gamma}H_t)^{-1}(1,x)\frac{\d x}{\sqrt{x}},
\end{align*}
and the remaining integral equals $1\mp\sqrt{\gamma t}\int_0^1((I\pm\sqrt{\gamma}H_t)^{-1}\tau_t\phi)(x)\frac{\d x}{\sqrt{x}}$ since $(I\pm\sqrt{\gamma}H_t)^{-1}=I\mp\sqrt{\gamma}H_t(I\pm\sqrt{\gamma}H_t)^{-1}$ and since $H_t$ is symmetric. The proof of \eqref{bp48} is complete.
\end{proof}
Equipped with \eqref{bp41a} and \eqref{bp48} we now prove Theorem \ref{theo4}.
\begin{proof}[Proof of Theorem \ref{theo4}] With $\phi\in L_{\circ}^1(0,1)$, the map $[0,1]\ni x\mapsto 1-\sqrt{t}\int_0^x(\tau_t\phi)(y)\frac{\d y}{\sqrt{y}}$ is bounded. Moreover, by \eqref{b:24}, $(I-\gamma K_t)^{-1}\tau_t\phi\in L_{\circ}^1(0,1)$, so both \eqref{b:21} and \eqref{b:22} are well-defined (note $\gamma_{\circ}\in[0,1]$ for $\gamma\in[0,1]$). To evaluate them, we use \eqref{bp41} and \eqref{bp48}: first by algebra and definition \eqref{b:18},
\begin{align*}
	&\gamma\sqrt{t}\int_0^1\big((I-\gamma K_t)^{-1}\tau_t\phi\big)(x)\left[1-\sqrt{t}\int_0^x(\tau_t\phi)(y)\frac{\d y}{\sqrt{y}}\right]\frac{\d x}{\sqrt{x}}=\frac{1}{2}(\sqrt{\gamma}-1)\sqrt{\gamma t}\int_0^1\big((I-\gamma K_t)^{-1}\tau_t\phi\big)(x)\\
	&\times\left[\frac{1}{\sqrt{x}}+\sqrt{\gamma}\int_0^1H_t(x,y)\frac{\d y}{\sqrt{y}}\right]\d x+\frac{1}{2}(\sqrt{\gamma}+1)\sqrt{\gamma t}\int_0^1\big((I-\gamma K_t)^{-1}\tau_t\phi\big)(x)\left[\frac{1}{\sqrt{x}}-\sqrt{\gamma}\int_0^1H_t(x,y)\frac{\d y}{\sqrt{y}}\right]\d x,
\end{align*}
and thus
\begin{align}
	\gamma\sqrt{t}\int_0^1\big((I-\gamma &K_t)^{-1}\tau_t\phi\big)(x)\left[1-\sqrt{t}\int_0^x(\tau_t\phi)(y)\frac{\d y}{\sqrt{y}}\right]\frac{\d x}{\sqrt{x}}=\frac{1}{2}(\sqrt{\gamma}-1)\sqrt{\gamma t}\int_0^1\big((I-\sqrt{\gamma}H_t)^{-1}\tau_t\phi\big)(z)\frac{\d z}{\sqrt{z}}\nonumber\\
	&+\frac{1}{2}(\sqrt{\gamma}+1)\sqrt{\gamma t}\int_0^1\big((I+\sqrt{\gamma}H_t)^{-1}\tau_t\phi\big)(z)\frac{\d z}{\sqrt{z}},\label{bp52}
\end{align}
using also the identity
\begin{equation}\label{bp53}
	\int_0^1(I\pm\sqrt{\gamma}H_t)^{-1}(x,z)\left[\frac{1}{\sqrt{x}}\pm\sqrt{\gamma}\int_0^1H_t(x,y)\frac{\d y}{\sqrt{y}}\right]\d x=\frac{1}{\sqrt{z}},\ \ \ z\in(0,1],
\end{equation}
which is a consequence of the symmetry of $H_t$ and $(I\pm\sqrt{\gamma}H_t)^{-1}=I\mp(I\pm\sqrt{\gamma}H_t)^{-1}\sqrt{\gamma}H_t$. Thus,
\begin{align*}
	1-\gamma\sqrt{t}&\int_0^1\big((I-\gamma K_t)^{-1}\tau_t\phi\big)(x)\left[1-\sqrt{t}\int_0^x(\tau_t\phi)(y)\frac{\d y}{\sqrt{y}}\right]\frac{\d x}{\sqrt{x}}\\
	&\,\stackrel{\eqref{bp52}}{=}\frac{1}{2}(1-\sqrt{\gamma})\left[1+\sqrt{\gamma t}\int_0^1\big((I-\sqrt{\gamma}H_t)^{-1}\tau_t\phi\big)(z)\frac{\d z}{\sqrt{z}}\right]\\
	&\hspace{2cm}+\frac{1}{2}(1+\sqrt{\gamma})\left[1-\sqrt{\gamma t}\int_0^1\big((I+\sqrt{\gamma}H_t)^{-1}\tau_t\phi\big)(z)\frac{\d z}{\sqrt{z}}\right],
\end{align*}
and which yields \eqref{b:25} through \eqref{bp41a} and \eqref{bp48} since $F(t,\gamma)=G(t,\sqrt{\gamma})G(t,-\sqrt{\gamma})$ and since $\|H_t\|\rightarrow 0$ as $t\downarrow 0$ in operator norm which allows one to integrate \eqref{bp48}. The derivation of \eqref{b:26} is quite similar, instead of \eqref{bp52} we have
\begin{align}
	\gamma\sqrt{t}\int_0^1\big((I-\gamma_{\circ}&K_t)^{-1}\tau_t\phi\big)(x)\left[1-\sqrt{t}\int_0^x(\tau_t\phi)(y)\frac{\d y}{\sqrt{y}}\right]\frac{\d x}{\sqrt{x}}=\frac{1+\sqrt{\gamma_{\circ}}}{2(2-\gamma)}\sqrt{\gamma_{\circ}t}\int_0^1\big((I+\sqrt{\gamma_{\circ}}H_t)^{-1}\tau_t\phi\big)(z)\frac{\d z}{\sqrt{z}}\nonumber\\
	&-\frac{1-\sqrt{\gamma_{\circ}}}{2(2-\gamma)}\sqrt{\gamma_{\circ}t}\int_0^1\big((I-\sqrt{\gamma_{\circ}}H_t)^{-1}\tau_t\phi\big)(z)\frac{\d z}{\sqrt{z}}\label{bp54}
\end{align}
as \eqref{bp53} is valid with the replacement $\gamma\mapsto\gamma_{\circ}$. Hence,
\begin{align*}
	1-\gamma\sqrt{t}&\int_0^1\big((I-\gamma_{\circ}K_t)^{-1}\tau_t\phi\big)(x)\left[1-\sqrt{t}\int_0^x(\tau_t\phi)(y)\frac{\d y}{\sqrt{y}}\right]\frac{\d x}{\sqrt{x}}\\
	&\,=\frac{1-\gamma}{2-\gamma}+\frac{1+\sqrt{\gamma_{\circ}}}{2(2-\gamma)}\left[1-\sqrt{\gamma_{\circ}t}\int_0^1\big((I+\sqrt{\gamma_{\circ}}H_t)^{-1}\tau_t\phi\big)(z)\frac{\d z}{\sqrt{z}}\right]\\
	&\hspace{3cm}+\frac{1-\sqrt{\gamma_{\circ}}}{2(2-\gamma)}\left[1+\sqrt{\gamma_{\circ}t}\int_0^1\big((I-\sqrt{\gamma_{\circ}}H_t)^{-1}\tau_t\phi\big)(z)\frac{\d z}{\sqrt{z}}\right],
\end{align*}
and which yields \eqref{b:26} through \eqref{bp41a} and \eqref{bp48} (simply replace $\gamma\mapsto\gamma_{\circ}$ in the last two identities). Lastly, $[0,1]\ni x\mapsto\int_0^x(\tau_t\phi)(y)\frac{\d y}{\sqrt{y}}$ is bounded since $\phi\in L_{\circ}^1(0,1)$, so \eqref{b:23} is well-defined given $(I-\gamma K_t)^{-1}\tau_t\phi\in L_{\circ}^1(0,1)$. Next,
\begin{align*}
	\gamma\sqrt{t}\int_0^1\big((I-\gamma K_t)^{-1}&\tau_t\phi\big)(x)\left[\frac{1}{2}\sqrt{t}\int_0^x(\tau_t\phi)(y)\frac{\d y}{\sqrt{y}}\right]\frac{\d x}{\sqrt{x}}\stackrel{\eqref{bp53}}{=}\frac{1}{4}\sqrt{\gamma t}\int_0^1\big((I-\sqrt{\gamma}H_t)^{-1}\tau_t\phi\big)(z)\frac{\d z}{\sqrt{z}}\\
	&-\frac{1}{4}\sqrt{\gamma t}\int_0^1\big((I+\sqrt{\gamma}H_t)^{-1}\tau_t\phi\big)(z)\frac{\d z}{\sqrt{z}},
\end{align*}
and so
\begin{align*}
	1+&\,\gamma\sqrt{t}\int_0^1\big((I-\gamma K_t)^{-1}\tau_t\phi\big)(x)\left[\frac{1}{2}\sqrt{t}\int_0^x(\tau_t\phi)(y)\frac{\d y}{\sqrt{y}}\right]\frac{\d x}{\sqrt{x}}\\
	&\,=\frac{1}{2}+\frac{1}{4}\left[1+\sqrt{\gamma t}\int_0^1\big((I-\sqrt{\gamma}H_t)^{-1}\tau_t\phi\big)(z)\frac{\d z}{\sqrt{z}}\right]+\frac{1}{4}\left[1-\sqrt{\gamma t}\int_0^1\big((I+\sqrt{\gamma}H_t)^{-1}\tau_t\phi\big)(z)\frac{\d z}{\sqrt{z}}\right],
\end{align*}
which equals \eqref{b:27} by \eqref{bp41a} and \eqref{bp48}. The proof of Theorem \ref{theo4} is now complete.
\end{proof}
The last result to be proven concerns the $t\rightarrow+\infty$ asymptotic behavior of $F(t)$ in \eqref{b:3}, i.e. we now prove Theorem \ref{theo4a}.
\begin{proof}[Proof of Theorem \ref{theo4a}] By assumption \eqref{b:28a}, the coefficients $r_j(z)$ in \eqref{b:13} admit analytic extensions to the closed vertical strip $|\Re z-\frac{1}{2}|\leq\frac{\epsilon}{2}$, by Morera's and Fubini's theorem. Moreover, for $0\leq|\Re z-\frac{1}{2}|\leq\frac{\epsilon}{2}$ we have
\begin{equation}\label{bp55}
	|r_1(z)|\leq\int_0^1y^{\Re z-\frac{1}{2}}|\phi(y)|\frac{\d y}{\sqrt{y}}+\int_1^{\infty}y^{\Re z-\frac{1}{2}}|\phi(y)|\frac{\d y}{\sqrt{y}}\stackrel{\eqref{b:28a}}{\leq}\frac{2}{a-\frac{\epsilon}{2}}<1
\end{equation}
since $a\geq 2+\epsilon$ and a similar bound holds for $|r_2(z)|$, too. Also, integrating by parts, for any $z\notin\{0,1\}$ with $|\Re z-\frac{1}{2}|\leq\frac{\epsilon}{2}$,
\begin{equation*}
	r_1(z)\stackrel{\eqref{b:13}}{=}-\frac{1}{z}\int_0^{\infty}y^{z-1}(MD\phi)(y)\,\d y,\ \ \ \ r_2(z)\stackrel{\eqref{b:13}}{=}-\frac{1}{1-z}\int_0^{\infty}y^{-z}(MD\psi)(y)\,\d y
\end{equation*}
and so in the same closed vertical strip, by \eqref{b:28a}, $|zr_1(z)|<1$ and $|(1-z)r_2(z)|<1$. We now commence the asymptotic analysis of $F(t)$ as $t\rightarrow+\infty$, keeping in mind these initial observations. First, as in the proof of Theorem \ref{theo2a}, we $\gamma$-modify RHP \ref{masterb} and consider the trace class family $K_{t,\gamma}\in\mathcal{L}(L^2(0,1))$ with kernel
\begin{equation*}
	K_{t,\gamma}(x,y):=\gamma K_t(x,y)\stackrel{\eqref{b:2}}{=}\gamma t\int_0^1\phi(xzt)\psi(zyt)\,\d z,\ \ \ \ \gamma\in[0,1],
\end{equation*}
and associated Fredholm determinant $D(t,\gamma)$ which is well-defined for all $t\in\mathbb{R}_+$ by \eqref{b:28a}. Note that $D(t,1)=F(t)$ and $D(t,0)=1$ for all $t\in\mathbb{R}_+$, moreover $D(t,\gamma)$ is encoded in the following $\gamma$-modification of RHP \ref{masterb}.
\begin{problem}\label{ggam1} Fix $(t,\gamma)\in\mathbb{R}_+\times[0,1]$ and $\phi,\psi$ as in \eqref{b:28a}. Find ${\bf X}(z)={\bf X}(z;t,\gamma,\phi,\psi)\in\mathbb{C}^{2\times 2}$ such that
\begin{enumerate}
	\item[(1)] ${\bf X}(z)$ is analytic for $z\in\mathbb{C}\setminus\Sigma$ with $\Sigma:=\frac{1}{2}+\im\mathbb{R}$.
	\item[(2)] ${\bf X}(z)$ admits continuous pointwise limits ${\bf X}_{\pm}(z):=\lim_{\epsilon'\downarrow 0}{\bf X}(z)\mp\epsilon'),z\in\Sigma$ which obey
	\begin{equation*}
		{\bf X}_+(z)={\bf X}_-(z)\begin{bmatrix}1-\gamma r_1(z)r_2(z) & -\sqrt{\gamma}\,r_2(z)t^z\\ \sqrt{\gamma}\,r_1(z)t^{-z} & 1\end{bmatrix},\ \ z\in\Sigma,
	\end{equation*}
	with $r_j(z)$ as in \eqref{b:13}.
	\item[(3)] As $z\rightarrow\infty$,
	\begin{equation*}
		{\bf X}(z)=\mathbb{I}+o(1);\ \ \ \ \ {\bf X}_1={\bf X}_1(t,\gamma,\phi,\psi)=\big[X_1^{mn}(t,\gamma,\phi,\psi)\big]_{m,n=1}^2:=\lim_{\substack{z\rightarrow\infty\\ \Re z\not\equiv\textnormal{const.}}}z\big({\bf X}(z)-\mathbb{I}\big)
	\end{equation*}
\end{enumerate}
\end{problem}
Indeed, by assumption and Theorem \ref{theo3}, RHP \ref{ggam1} is uniquely solvable for all $(t,\gamma)\in\mathbb{R}_+\times[0,1]$ and
\begin{equation}\label{bp56}
	t\frac{\partial}{\partial t}\ln D(t,\gamma)=-X_1^{11}(t,\gamma,\phi,\psi).
\end{equation}
Second, as outlined in Remark \ref{IIKScon2}, the transformation
\begin{equation}\label{bp57}
	{\bf T}(z;t,\gamma,\phi,\psi):={\bf X}(z;t,\gamma,\phi,\psi)\begin{cases}\begin{bmatrix}1 & 0\\ -\sqrt{\gamma}\,r_1(z)t^{-z} & 1\end{bmatrix},&\Re z\in(\frac{1}{2}-\frac{\epsilon}{2},\frac{1}{2})\\
	\mathbb{I},&\textnormal{else}
	\end{cases}
\end{equation}
maps RHP \ref{ggam1} to an integrable operator RHP of the following type.
\begin{problem}\label{ggam2} Let $(t,\gamma)\in\mathbb{R}_+\times[0,1]$ and $\phi,\psi$ as in the statement of Theorem \ref{theo4a}. The function ${\bf T}(z)={\bf T}(z;t,\gamma,\phi,\psi)\in\mathbb{C}^{2\times 2}$ defined in \eqref{bp57} is uniquely determined by the following properties:
\begin{enumerate}
	\item[(1)] ${\bf T}(z)$ is analytic for $z\in\mathbb{C}\setminus\Sigma_{\bf T}$ with $\Sigma_{\bf T}:=\Sigma\cup(\frac{1}{2}-\frac{\epsilon}{2}+\im\mathbb{R})$ and extends continuously on the closure of $\mathbb{C}\setminus\Sigma_{\bf T}$.
	\item[(2)] The limiting values ${\bf T}_{\pm}(z)=\lim_{\epsilon'\downarrow 0}{\bf T}(z\mp\epsilon')$ on $\Sigma_{\bf T}\ni z$ satisfy the jump condition ${\bf T}_+(z)={\bf T}_-(z){\bf G}_{\bf T}(z;t,\gamma,\phi,\psi)$ where the jump matrix ${\bf G}_{\bf T}(z;t,\gamma,\phi,\psi)$ is given by
	\begin{align*}
		{\bf G}_{\bf T}(z;t,\gamma,\phi,\psi)=&\,\begin{bmatrix}1 & 0\\ \sqrt{\gamma}\,r_1(z)t^{-z} & 1\end{bmatrix},\ \ \ \Re z=\frac{1}{2}-\frac{\epsilon}{2},\\
		{\bf G}_{\bf T}(z;t,\gamma,\phi,\psi)=&\,\begin{bmatrix}1&-\sqrt{\gamma}\,r_2(z)t^z\\ 0 & 1\end{bmatrix},\ \ \ \Re z=\frac{1}{2}.
	\end{align*}
	\item[(3)] Since $|zr_1(z)|\rightarrow 0$ as $|z|\rightarrow\infty$ in $|\Re z-\frac{1}{2}|\leq\frac{\epsilon}{2}$ by the Riemann-Lebesgue lemma and since $|t^{\pm z}|\leq t^{\pm\Re z}$ in $|\Re z-\frac{1}{2}|\leq\epsilon$, we have uniformly as $z\rightarrow\infty$ in $\mathbb{C}\setminus\Sigma_{\bf T}$,
	\begin{equation}\label{bp58}
		{\bf T}(z)=\mathbb{I}+o(1);\ \ \ \ \ \ {\bf T}_1(t,\gamma,\phi,\psi):=\lim_{\substack{z\rightarrow\infty\\ \Re z\not\equiv\textnormal{const.}}}z\big({\bf T}(z)-\mathbb{I}\big)={\bf X}_1(t,\gamma,\phi,\psi).
	\end{equation}
\end{enumerate}
\end{problem}
Observe that RHP \ref{ggam2} is the RHP associated with the trace class integrable operator $M_{t,\gamma}:L^2(\Sigma_{\bf T})\rightarrow L^2(\Sigma_{\bf T})$ with kernel
\begin{equation}\label{bp58a}
	M_{t,\gamma}(z,w)=\frac{f_1(z)g_1(w)+f_2(z)g_2(w)}{z-w},\ \ \ (z,w)\in\Sigma_{\bf T}\times\Sigma_{\bf T},
\end{equation}
cf. \cite{IIKS}, where the bounded functions
\begin{equation*}
	f_1(z):=\frac{1}{2\pi\im}\chi_{\Sigma}(z),\ f_2(z):=-\frac{1}{2\pi\im}t^{-z}\chi_{\Sigma-\frac{\epsilon}{2}}(z),\ g_1(w):=\sqrt{\gamma}\,r_1(w)\chi_{\Sigma-\frac{\epsilon}{2}}(w),\ g_2(w):=\sqrt{\gamma}\,t^wr_2(w)\chi_{\Sigma}(w),
\end{equation*}
are expressed in terms of the characteristic functions $\chi_{\Sigma},\chi_{\Sigma-\frac{\epsilon}{2}}$ on $\Re z=\frac{1}{2},\Re z=\frac{1}{2}-\frac{\epsilon}{2}$. Given that RHP \ref{ggam1} is uniquely solvable for all $(t,\gamma)\in\mathbb{R}_+\times[0,1]$, the same applies to RHP \ref{ggam2} because of \eqref{bp57} and so we can compute ${\bf T}(z)$ by the general formula, cf. \cite{IIKS},
\begin{equation}\label{bp59}
	{\bf T}(z)=\mathbb{I}-\int_{\Sigma_{\bf T}}\begin{bmatrix}F_1(\lambda)g_1(\lambda) & F_1(\lambda)g_2(\lambda)\\ F_2(\lambda)g_1(\lambda) & F_2(\lambda)g_2(\lambda)\end{bmatrix}\frac{\d\lambda}{\lambda-z},\ \ \ z\notin\Sigma_{\bf T},
\end{equation}
where $g_j$ are as above and $F_j$ given on $\Sigma_{\bf T}$ by, independent of the choice of limiting values,
\begin{equation}\label{bp60}
	\begin{bmatrix}F_1\\ F_2\end{bmatrix}=(I-M_{t,\gamma})^{-1}\begin{bmatrix}f_1\\ f_2\end{bmatrix}={\bf T}_{\pm}\begin{bmatrix}f_1\\ f_2\end{bmatrix}.
\end{equation}
Next, $t\frac{\partial}{\partial t}M_{t,\gamma}=-f_2\otimes g_2\in\mathcal{C}_1(L^2(\Sigma_{\bf T}))$ and so, using the differentiability of $t\mapsto M_{t,\gamma}\in\mathcal{L}(L^2(\Sigma_{\bf T}))$ and the trace identity $\tr{\bf T}_1=\tr{\bf X}_1=0$, we obtain for the Fredholm determinant $H(t,\gamma)$ of $M_{t,\gamma}$ on $L^2(\Sigma_{\bf T})$,
\begin{equation*}
	t\frac{\partial}{\partial t}\ln H(t,\gamma)=\tr_{L^2(\Sigma_{\bf T})}\big((I-M_{t,\gamma})^{-1}f_2\otimes g_2\big)\stackrel{\eqref{bp60}}{=}\tr_{L^2(\Sigma_{\bf T})}(F_2\otimes g_2)\stackrel[\eqref{bp59}]{\eqref{bp58}}{=}-X_1^{11}(t,\gamma,\phi,\psi).
\end{equation*}
Consequently, comparing the latter to \eqref{bp56}, we have found
\begin{equation*}
	\ln D(t,\gamma)=\ln H(t,\gamma)+\eta(\gamma),\ \ \ \ \ (t,\gamma)\in\mathbb{R}_+\times[0,1],
\end{equation*}
with a $t$-independent term $\eta(\gamma)$. However, $D(t,\gamma)\rightarrow 1$ as $t\downarrow 0$ by \eqref{b:28} and since RHP \ref{ggam2} is directly related to a small norm problem as $t\downarrow 0$, we have, by general theory \cite{DZ}, $T_1^{22}(t,\gamma,\phi,\psi)\rightarrow 0$ power-like as $t\downarrow 0$ for all $\gamma\in[0,1]$ and so $H(t,\gamma)\rightarrow 1$ as well for $t\downarrow 0$. In short,
\begin{equation}\label{bp61}
	\ln D(t,\gamma)=\ln H(t,\gamma),\ \ \ \ \ (t,\gamma)\in\mathbb{R}_+\times[0,1],
\end{equation}
which identifies our $\gamma$-modified determinant $D(t,\gamma)$ as the Fredholm determinant of the integrable operator \eqref{bp58a}. In particular we can now compute $F(t)$ in terms of the solution of RHP \ref{ggam2},
\begin{align}
	\ln F(t)=\int_0^1&\frac{\partial}{\partial\gamma}\ln D(t,\gamma)\,\d\gamma\stackrel{\eqref{bp61}}{=}\int_0^1\frac{\partial}{\partial\gamma}\ln H(t,\gamma)\,\d\gamma=-\frac{1}{2}\int_0^1\tr_{L^2(\Sigma_{\bf T})}\big((I-M_{t,\gamma})^{-1}M_{t,\gamma}\big)\frac{\d\gamma}{\gamma}\nonumber\\
	&\,=-\frac{1}{2}\int_0^1\left[\int_{\Sigma_{\bf T}}R_{t,\gamma}(\lambda,\lambda)\,\d\lambda\right]\frac{\d\gamma}{\gamma},\ \ \ \ \ \ R_{t,\gamma}:=(I-M_{t,\gamma})^{-1}-I\in\mathcal{L}(L^2(\Sigma_{\bf T})),\label{bp62}
\end{align}
where, by general theory \cite[$(5.5)$]{IIKS},
\begin{equation*}
	R_{t,\gamma}(\lambda,\lambda)=F_1'(\lambda)G_1(\lambda)+F_2'(\lambda)G_2(\lambda),\ \ \ \lambda\in\Sigma_{\bf T},
\end{equation*}
with $F_j'=DF_j$ as in \eqref{bp60} and $G_j$ given on $\Sigma_{\bf T}$ by, independent of the choice of limiting values,
\begin{equation}\label{bp63}
	\begin{bmatrix}G_1\\ G_2\end{bmatrix}=\big({\bf T}_{\pm}^{-1}\big)^{\top}\begin{bmatrix}g_1\\ g_2\end{bmatrix}.
\end{equation}
As in the proof of Theorem \ref{theo2a} we shall use the double integral representation \eqref{bp62} as starting point for our asymptotic analysis of $F(t)$. We now commence the necessary nonlinear steepest descent analysis of RHP \ref{ggam2}, equivalently of RHP \ref{ggam1}, see \eqref{bp57}. First define,
\begin{equation}\label{bp64}
	g(z;\gamma):=-\frac{1}{2\pi\im}\int_{\Sigma}\ln\big(1-\gamma r_1(\lambda)r_2(\lambda)\big)\frac{\d\lambda}{\lambda-z},\ \ \ \ z\in\mathbb{C}\setminus\Sigma,
\end{equation}
and note that the improper integral is absolutely convergent since $|r_j(\lambda)|<1$ as well as $|\lambda r_1(\lambda)|<1$ and $|(1-\lambda)r_2(\lambda)|<1$ on $\Sigma$, in addition to $\gamma\in[0,1]$. Moreover, $\Sigma\ni\lambda\mapsto\ln(1-\gamma r_1(\lambda)r_2(\lambda))$ is locally H\"older continuous by the continuous differentiability of $\Sigma\ni\lambda\mapsto r_j(\lambda)$ and since $|r_j(\lambda)|<1$ on $\Sigma$. Thus,
\begin{align*}
	g_+(z;\gamma)-g_-(z;\gamma)=&\,-\ln\big(1-\gamma r_1(z)r_2(z)\big),\ \ \ z\in\Sigma;\ \ \ \ \ \ \ \ g_{\pm}(z;\gamma)=\lim_{\epsilon\downarrow 0}g(z\mp\epsilon;\gamma),\\
	g(z;\gamma)=&\,\frac{1}{2\pi\im z}\int_{\Sigma}\ln\big(1-\gamma r_1(\lambda)r_2(\lambda)\big)\,\d\lambda+o\big(z^{-1}\big),\ \ \ z\rightarrow\infty,\ z\notin\Sigma.
\end{align*}
Now, using the $g$-function \eqref{bp64}, we transform RHP \ref{ggam1} as follows: define
\begin{equation}\label{bp65}
	{\bf Y}(z;t,\gamma,\phi,\psi):=t^{-\frac{1}{4}\sigma_3}{\bf X}(z;t,\gamma,\phi,\psi)\e^{g(z;\gamma)\sigma_3}t^{\frac{1}{4}\sigma_3},\ \ z\in\mathbb{C}\setminus\Sigma,
\end{equation}
and obtain the following problem.
\begin{problem}\label{ggam3} Let $(t,\gamma)\in\mathbb{R}_+\times[0,1]$ and $\phi,\psi$ as in the statement of Theorem  \ref{theo4a}. The function ${\bf Y}(z)={\bf Y}(z;t,\gamma,\phi,\psi)\in\mathbb{C}^{2\times 2}$ defined in \eqref{bp65} is uniquely determined by the following properties:
\begin{enumerate}
	\item[(1)] ${\bf Y}(z)$ is analytic for $z\in\mathbb{C}\setminus\Sigma$ and extends continuously on the closure of $\mathbb{C}\setminus\Sigma$.
	\item[(2)] The continuous limiting values ${\bf Y}_{\pm}(z)=\lim_{\epsilon'\downarrow 0}{\bf Y}(z\mp\epsilon')$ on $\Sigma\ni z$ satisfy the jump condition ${\bf Y}_+(z)={\bf Y}_-(z){\bf G}_{\bf Y}(z;t,\gamma,\phi,\psi)$ where the jump matrix ${\bf G}_{\bf Y}(z;t,\gamma,\phi,\psi)$ is given by
	\begin{equation*}
		{\bf G}_{\bf Y}(z;t,\gamma,\phi,\psi)=\begin{bmatrix}1 & -\eta_2(z;\gamma)t^{z-\frac{1}{2}}\e^{-2g_+(z;\gamma)}\\ \eta_1(z;\gamma)t^{\frac{1}{2}-z}\e^{2g_-(z;\gamma)} & 1-\gamma r_1(z)r_2(z)\end{bmatrix},\ \ z\in\Sigma,
	\end{equation*}
	using
	\begin{equation*}
		\eta_k(z;\gamma):=\frac{\sqrt{\gamma}\,r_k(z)}{1-\gamma r_1(z)r_2(z)},\ \ \ \ \left|\Re z-\frac{1}{2}\right|\leq\frac{\epsilon}{2}.
	\end{equation*}
	\item[(3)] As $z\rightarrow\infty$,
	\begin{equation}\label{bp66}
		{\bf Y}(z)=\mathbb{I}+o(1);
	\end{equation}
	where
	\begin{equation*}
		{\bf Y}_1(t,\gamma,\phi,\psi):=\lim_{\substack{z\rightarrow\infty\\ \Re z\not\equiv\textnormal{const.}}}z\big({\bf Y}(z)-\mathbb{I}\big)=t^{-\frac{1}{4}\sigma_3}\left[{\bf X}_1(t,\gamma,\phi,\psi)+\frac{\sigma_3}{2\pi\im}\int_{\Sigma}\ln\big(1-\gamma r_1(\lambda)r_2(\lambda)\big)\,\d\lambda\right]t^{\frac{1}{4}\sigma_3}.
	\end{equation*}
\end{enumerate}
\end{problem}
Next, we use that $\eta_k(z;\gamma)$ extends analytically to the closed vertical strip $|\Re z-\frac{1}{2}|\leq\frac{\epsilon}{2}$ and that $|z\eta_1(z;\gamma)|$ and $|(1-z)\eta_2(z;\gamma)|$ tend to zero as $|z|\rightarrow\infty$ in the same strip. Moreover, ${\bf G}_{\bf Y}$ admits the factorization
\begin{equation*}
	{\bf G}_{\bf Y}(z;t,\gamma,\phi,\psi)=\begin{bmatrix}1 & 0\\ \eta_1(z;\gamma)t^{\frac{1}{2}-z}\e^{2g_-(z;\gamma)} & 1\end{bmatrix}\begin{bmatrix}1 & -\eta_2(z;\gamma)t^{z-\frac{1}{2}}\e^{-2g_+(z;\gamma)}\\ 0 & 1\end{bmatrix},\ \ z\in\Sigma,
\end{equation*}
so we can open lenses and transform RHP \ref{ggam3} as follows. Introduce
\begin{equation}\label{bp67}
	{\bf S}(z;t,\gamma,\phi,\psi):={\bf Y}(z;t,\gamma,\phi,\psi)\begin{cases}\begin{bmatrix}1 & \eta_2(z;\gamma)t^{z-\frac{1}{2}}\e^{-2g(z;\gamma)}\\ 0 & 1\end{bmatrix},&\Re z\in(\frac{1}{2}-\frac{\epsilon}{2},\frac{1}{2})\smallskip\\
	\begin{bmatrix}1 & 0\\ \eta_1(z;\gamma)t^{\frac{1}{2}-z}\e^{2g(z;\gamma)} & 1\end{bmatrix},&\Re z\in(\frac{1}{2},\frac{1}{2}+\frac{\epsilon}{2})\\
	\mathbb{I},&\textnormal{else}
	\end{cases}
\end{equation}
and note that RHP \ref{ggam3} is transformed to the following problem.
\begin{problem}\label{ggam4} Let $t>1,\gamma\in[0,1]$ and $\phi,\psi$ as in the statement of Theorem \ref{theo4a}. The function ${\bf S}(z)={\bf S}(z;t,\gamma,\phi,\psi)\in\mathbb{C}^{2\times 2}$ defined in \eqref{bp67} is uniquely determined by the following properties:
\begin{enumerate}
	\item[(1)] ${\bf S}(z)$ is analytic for $z\in\mathbb{C}\setminus\{\Re z-\frac{1}{2}=\pm\frac{\epsilon}{2}\}$ and extends continuously on the closure of $\mathbb{C}\setminus\{\Re z-\frac{1}{2}=\pm\frac{\epsilon}{2}\}$.
	\item[(2)] The continuous limiting values ${\bf S}_{\pm}(z):=\lim_{\epsilon'\downarrow 0}{\bf S}(z\mp\epsilon')$ on $\Sigma_{\bf S}:=\{\Re z-\frac{1}{2}=\pm\frac{\epsilon}{2}\}\ni z$ satisfy the jump condition ${\bf S}_+(z)={\bf S}_-(z){\bf G}_{\bf S}(z;t,\gamma,\phi,\psi)$ where the jump matrix ${\bf G}_{\bf S}(z;t,\gamma,\phi,\psi)$ is given by
	\begin{align*}
		{\bf G}_{\bf S}(z;t,\gamma,\phi,\psi)=&\,\begin{bmatrix}1 & -\eta_2(z;\gamma)t^{z-\frac{1}{2}}\e^{-2g(z;\gamma)}\\ 0 & 1\end{bmatrix},\ \ \Re z-\frac{1}{2}=-\frac{\epsilon}{2},\\
		{\bf G}_{\bf S}(z;t,\gamma,\phi,\psi)=&\,\begin{bmatrix}1&0\\ \eta_1(z;\gamma)t^{\frac{1}{2}-z}\e^{2g(z;\gamma)} & 1\end{bmatrix},\ \ \Re z-\frac{1}{2}=\frac{\epsilon}{2}.
	\end{align*}
	\item[(3)] Uniformly as $z\rightarrow\infty$ in $\mathbb{C}\setminus\Sigma_{\bf S}$, using the Riemann-Lebesgue lemma for $\eta_k$ in $|\Re z-\frac{1}{2}|\leq\frac{\epsilon}{2}$,
	\begin{equation}\label{bp68}
		{\bf S}(z)=\mathbb{I}+o(1);\ \ \ \ \ {\bf S}_1(t,\gamma,\phi,\psi):=\lim_{\substack{z\rightarrow\infty\\ \Re z\not\equiv\textnormal{const.}}}z\big({\bf S}(z)-\mathbb{I}\big)={\bf Y}_1(t,\gamma,\phi,\psi).
	\end{equation}
\end{enumerate}
\end{problem}
Since $|t^{\pm(z-\frac{1}{2})}|=t^{-\frac{1}{2}\epsilon}$ for $\Re z-\frac{1}{2}=\mp\frac{\epsilon}{2}$, and since the $t$-independent functions $\eta_k(z;\gamma)\e^{\mp 2g(z;\gamma)}$ are bounded on $\Sigma_{\bf S}$ for all $\gamma\in[0,1]$, we arrive at the below small norm estimate for ${\bf G}_{\bf S}(z;t,\gamma,\phi,\psi)$: subject to \eqref{b:28a}, for any $\epsilon>0$ there exist $t_0=t_0(\epsilon),c=c(\epsilon)>0$ such that for all $t\geq t_0$ and all $\gamma\in[0,1]$,
\begin{equation}\label{bp69}
	\|{\bf G}_{\bf S}(\cdot;t,\gamma,\phi,\psi)-\mathbb{I}\|_{L^{\infty}(\Sigma_{\bf S})}\leq c\sqrt{\gamma}\,t^{-\frac{1}{2}\epsilon},\ \ \|{\bf G}_{\bf S}(\cdot;t,\gamma,\phi,\psi)-\mathbb{I}\|_{L^2(\Sigma_{\bf S})}\leq c\sqrt{\gamma}\,t^{-\frac{1}{2}\epsilon}.
\end{equation}
Thus, by general theory \cite{DZ}, RHP \ref{ggam4} is asymptotically solvable, namely there exist $t_0=t_0(\epsilon)>0$ such that the RHP for ${\bf S}(z)$ is uniquely solvable in $L^2(\Sigma_{\bf S})$ for all $t\geq t_0$ and all $\gamma\in[0,1]$. We can compute the solution of the same problem iteratively through the integral representation
\begin{equation}\label{bp70}
	{\bf S}(z;t,\gamma,\phi,\psi)=\mathbb{I}+\frac{1}{2\pi\im}\int_{\Sigma_{\bf S}}{\bf S}_-(\lambda;t,\gamma,\phi,\psi)\big({\bf G}_{\bf S}(\lambda;t,\gamma,\phi,\psi)-\mathbb{I}\big)\frac{\d\lambda}{\lambda-z},\ \ \ \ z\in\mathbb{C}\setminus\Sigma_{\bf S},
\end{equation}
using that, for all $t\geq t_0$ and $\gamma\in[0,1]$, with $c=c(\epsilon)>0$,
\begin{equation}\label{bp71}
	\|{\bf S}_-(\cdot;t,\gamma,\phi,\psi)-\mathbb{I}\|_{L^2(\Sigma_{\bf S})}\leq c\sqrt{\gamma}\,t^{-\frac{1}{2}\epsilon}.
\end{equation}
The last two equations complete the nonlinear steepest descent analysis of RHP \ref{ggam2} and we are now left to extract all relevant information. First, choosing right-sided limits for definiteness, we have for $\lambda\in\Sigma_{\bf T}$ by \eqref{bp60},\eqref{bp63} and \eqref{bp57},\eqref{bp65},\eqref{bp67},
\begin{equation*}
	\begin{bmatrix}F_1(\lambda)\\ F_2(\lambda)\end{bmatrix}={\bf T}_-(\lambda)\begin{bmatrix}f_1(\lambda)\\ f_2(\lambda)\end{bmatrix}=t^{\frac{1}{4}\sigma_3}{\bf S}_-(\lambda)t^{-\frac{1}{4}\sigma_3}\big\{{\bf A}(\lambda)\chi_{\Sigma}(\lambda)+{\bf B}(\lambda)\chi_{\Sigma-\frac{\epsilon}{2}}(\lambda)\big\}\begin{bmatrix}f_1(\lambda)\\ f_2(\lambda)\end{bmatrix},
\end{equation*}
followed by
\begin{equation*}
	\begin{bmatrix}G_1(\lambda)\\ G_2(\lambda)\end{bmatrix}=\big({\bf T}_-^{-1}(\lambda)\big)^{\top}\begin{bmatrix}g_1(\lambda)\\ g_2(\lambda)\end{bmatrix}=t^{-\frac{1}{4}\sigma_3}\big({\bf S}_-^{-1}(\lambda)\big)^{\top}t^{\frac{1}{4}\sigma_3}\big\{\big({\bf A}^{-1}(\lambda)\big)^{\top}\chi_{\Sigma}(\lambda)+\big({\bf B}^{-1}(\lambda)\big)^{\top}\chi_{\Sigma-\frac{\epsilon}{2}}(\lambda)\big\}\begin{bmatrix}g_1(\lambda)\\ g_2(\lambda)\end{bmatrix},
\end{equation*}
where we suppress all $(t,\gamma,\phi,\psi)$-dependencies from our notation. Here,
\begin{equation*}
	{\bf A}(\lambda):=\e^{-g_-(\lambda)\sigma_3}\begin{bmatrix}1 & 0\\ -\eta_1(\lambda)t^{-\lambda}&1\end{bmatrix},\ \ \ {\bf B}(\lambda):=\e^{-g(\lambda)\sigma_3}\begin{bmatrix}(1-\gamma r_1(\lambda)r_2(\lambda))^{-1} & -\eta_2(\lambda)t^{\lambda}\smallskip\\ -\sqrt{\gamma}\,r_1(\lambda)t^{-\lambda} & 1\end{bmatrix},\ \ \lambda\in\Sigma_{\bf T}.
\end{equation*}
Consequently the innermost integral in \eqref{bp62} becomes
\begin{align}
	\int_{\Sigma_{\bf T}}&R_{t,\gamma}(\lambda,\lambda)\,\d\lambda=\int_{\Sigma_{\bf T}}\begin{bmatrix}f_1'(\lambda)\\ f_2'(\lambda)\end{bmatrix}^{\top}\begin{bmatrix}g_1(\lambda)\\ g_2(\lambda)\end{bmatrix}\d\lambda+\int_{\Sigma_{\bf T}}\begin{bmatrix}f_1(\lambda)\\ f_2(\lambda)\end{bmatrix}^{\top}{\bf E}(\lambda)\begin{bmatrix}g_1(\lambda)\\ g_2(\lambda)\end{bmatrix}\d\lambda\nonumber\\
	&\,+\int_{\Sigma}\begin{bmatrix}f_1(\lambda)\\ f_2(\lambda)\end{bmatrix}^{\top}\big({\bf A}^{\top}(\lambda)\big)'\big({\bf A}^{-1}(\lambda)\big)^{\top}\begin{bmatrix}g_1(\lambda)\\ g_2(\lambda)\end{bmatrix}\d\lambda+\int_{\Sigma-\frac{\epsilon}{2}}\begin{bmatrix}f_1(\lambda)\\ f_2(\lambda)\end{bmatrix}^{\top}\big({\bf B}^{\top}(\lambda)\big)'\big({\bf B}^{-1}(\lambda)\big)^{\top}\begin{bmatrix}g_1(\lambda)\\ g_2(\lambda)\end{bmatrix}\d\lambda,\label{bp72}
\end{align}
with ${\bf E}(\lambda)$ being shorthand for
\begin{align*}
	{\bf E}(\lambda):=\big\{{\bf A}^{\top}(\lambda)\chi_{\Sigma}(\lambda)+&\,\,{\bf B}^{\top}(\lambda)\chi_{\Sigma-\frac{\epsilon}{2}}(\lambda)\big\}t^{-\frac{1}{4}\sigma_3}\big({\bf S}_-^{\top}(\lambda)\big)'\\
	&\times\,\big({\bf S}_-^{-1}(\lambda)\big)^{\top}t^{\frac{1}{4}\sigma_3}\big\{\big({\bf A}^{-1}(\lambda)\big)^{\top}\chi_{\Sigma}(\lambda)+\big({\bf B}^{-1}(\lambda)\big)^{\top}\chi_{\Sigma-\frac{\epsilon}{2}}(\lambda)\big\}.
\end{align*}
Out of the four integrals in \eqref{bp72} the first one equates to zero by the choices for $f_j$ and $g_j$, see \eqref{bp58a}. The second integral in \eqref{bp72} can be estimated as follows, compare \eqref{bp70},\eqref{bp71} and the particular structure of ${\bf G}_{\bf S}$: for any $\epsilon>0$ there exist $t_0=t_0(\epsilon)>0$ and $c=c(\epsilon)>0$ so that
\begin{equation}\label{bp73}
	\left|\int_{\Sigma_{\bf T}}\begin{bmatrix}f_1(\lambda)\\ f_2(\lambda)\end{bmatrix}^{\top}{\bf E}(\lambda)\begin{bmatrix}g_1(\lambda)\\ g_2(\lambda)\end{bmatrix}\d\lambda\right|\leq\frac{c\gamma}{\ln t}\ \ \ \ \forall\,t\geq t_0,\ \ \gamma\in[0,1].
\end{equation}
Lastly, integral three and four in \eqref{bp72} yield leading order contributions: by definition of all underlying quantities, for the fourth integral in \eqref{bp72},
\begin{align}
	\int_{\Sigma-\frac{\epsilon}{2}}&\begin{bmatrix}f_1(\lambda)\\ f_2(\lambda)\end{bmatrix}^{\top}\big({\bf B}^{\top}(\lambda)\big)'\big({\bf B}^{-1}(\lambda)\big)^{\top}\begin{bmatrix}g_1(\lambda)\\ g_2(\lambda)\end{bmatrix}\d\lambda=\frac{\gamma}{2\pi\im}\int_{\Sigma-\frac{\epsilon}{2}}r_1(\lambda)t^{-\lambda}\frac{\d}{\d\lambda}\left[\frac{r_2(\lambda)t^{\lambda}}{1-\gamma r_1(\lambda)r_2(\lambda)}\right]\d\lambda\nonumber\\
	&\,-\frac{\gamma}{2\pi^2}\int_{\Sigma-\frac{\epsilon}{2}}\frac{r_1(\lambda)r_2(\lambda)}{1-\gamma r_1(\lambda)r_2(\lambda)}\left[\int_{\Sigma}\ln\big(1-\gamma r_1(z)r_2(z)\big)\frac{\d z}{(z-\lambda)^2}\right]\d\lambda,\label{bp74}
\end{align}
where we integrate by parts in the first term, using $|zr_1(z)|<1$ and $|(1-z)r_2(z)|<1$ for $|\Re z-\frac{1}{2}|\leq\frac{\epsilon}{2}$,
\begin{align*}
\frac{\gamma}{2\pi\im}&\int_{\Sigma-\frac{\epsilon}{2}}r_1(\lambda)t^{-\lambda}\frac{\d}{\d\lambda}\left[\frac{r_2(\lambda)t^{\lambda}}{1-\gamma r_1(\lambda)r_2(\lambda)}\right]\d\lambda=-\frac{\gamma}{2\pi\im}\int_{\Sigma-\frac{\epsilon}{2}}\frac{r_1'(\lambda)r_2(\lambda)}{1-\gamma r_1(\lambda)r_2(\lambda)}\,\d\lambda\\
&\,-\frac{\gamma\ln t}{2\pi\im}\frac{\d}{\d\gamma}\left[\int_{\Sigma}\ln\big(1-\gamma r_1(\lambda)r_2(\lambda)\big)\d\lambda\right]=\gamma\ln t\left[\frac{\d}{\d\gamma}s(1,\gamma)\right]+\frac{\gamma}{2\pi\im}\frac{\d}{\d\gamma}\int_{\Sigma}\frac{r_1'(\lambda)}{r_1(\lambda)}\ln\big(1-\gamma r_1(\lambda)r_2(\lambda)\big)\d\lambda,
\end{align*}
with $s(x,\gamma):=-\frac{1}{2\pi\im}\int_{\Sigma}\ln(1-\gamma r_1(\lambda)r_2(\lambda))x^{-\lambda}\d\lambda,x\in\mathbb{R}_{\geq 1}$, after collapsing $\Sigma-\frac{\epsilon}{2}$ to $\Sigma$. The second term in \eqref{bp74} we rewrite with the help of the integral identity $(z-\lambda)^{-2}=\int_1^{\infty}x^{\lambda-z-1}\ln x\,\d x$, for $z\in\Sigma,\lambda\in\Sigma-\frac{\epsilon}{2}$,
\begin{equation*}
	-\frac{\gamma}{2\pi^2}\int_{\Sigma-\frac{\epsilon}{2}}\frac{r_1(\lambda)r_2(\lambda)}{1-\gamma r_1(\lambda)r_2(\lambda)}\left[\int_{\Sigma}\ln\big(1-\gamma r_1(z)r_2(z)\big)\frac{\d z}{(z-\lambda)^2}\right]\d\lambda=-2\gamma\int_1^{\infty}\ln x\left[\frac{\partial}{\partial\gamma}\hat{s}(x,\gamma)\right]s(x,\gamma)\,\d x,
\end{equation*}
with $\hat{s}(x,\gamma):=-\frac{1}{2\pi\im}\int_{\Sigma}\ln(1-\gamma r_1(\lambda)r_2(\lambda))x^{\lambda-1}\d\lambda$. All together, for \eqref{bp74},
\begin{align}
	\int_{\Sigma-\frac{\epsilon}{2}}\begin{bmatrix}f_1(\lambda)\\ f_2(\lambda)\end{bmatrix}^{\top}\big({\bf B}^{\top}(\lambda)\big)'\big({\bf B}^{-1}(\lambda)\big)^{\top}\begin{bmatrix}g_1(\lambda)\\ g_2(\lambda)\end{bmatrix}\d\lambda=&\,\gamma\ln t\left[\frac{\d}{\d\gamma}s(1,\gamma)\right]-2\gamma\int_1^{\infty}\ln x\left[\frac{\partial}{\partial\gamma}\hat{s}(x,\gamma)\right]s(x,\gamma)\,\d x\nonumber\\
	&\,+\frac{\gamma}{2\pi\im}\frac{\d}{\d\gamma}\int_{\Sigma}\frac{r_1'(\lambda)}{r_1(\lambda)}\ln\big(1-\gamma r_1(\lambda)r_2(\lambda)\big)\d\lambda.\label{bp75}
\end{align}
Finally, for the third integral in \eqref{bp72}, by definition of all underlying quantities,
\begin{align}
	\int_{\Sigma}\begin{bmatrix}f_1(\lambda)\\ f_2(\lambda)\end{bmatrix}^t\big({\bf A}^{\top}(\lambda)\big)'&\big({\bf A}^{-1}(\lambda)\big)^{\top}\begin{bmatrix}g_1(\lambda)\\ g_2(\lambda)\end{bmatrix}\d\lambda=-\frac{\gamma}{2\pi\im}\int_{\Sigma}\frac{r_2(\lambda)t^{\lambda}}{1-\gamma r_1(\lambda)r_2(\lambda)}\frac{\d}{\d\lambda}\Big[r_1(\lambda)t^{-\lambda}\Big]\d\lambda\nonumber\\
	&\,-\frac{\gamma}{2\pi^2}\int_{\Sigma}\frac{r_1(\lambda)r_2(\lambda)}{1-\gamma r_1(\lambda)r_2(\lambda)}\frac{\d}{\d\lambda}\left[\textnormal{pv}\int_{\Sigma}\ln\big(1-\gamma r_1(z)r_2(z)\big)\frac{\d z}{z-\lambda}\right]\d\lambda,\label{bp76}
\end{align}
where we evaluate the derivative in the first term and symmetrize,
\begin{equation*}
	-\frac{\gamma}{2\pi\im}\int_{\Sigma}\frac{r_2(\lambda)t^{\lambda}}{1-\gamma r_1(\lambda)r_2(\lambda)}\frac{\d}{\d\lambda}\Big[r_1(\lambda)t^{-\lambda}\Big]\d\lambda=\gamma\ln t\left[\frac{\d}{\d\gamma}s(1,\gamma)\right]-\frac{\gamma}{2\pi\im}\frac{\d}{\d\gamma}\int_{\Sigma}\frac{r_2'(\lambda)}{r_2(\lambda)}\ln\big(1-\gamma r_1(\lambda)r_2(\lambda)\big)\d\lambda.
\end{equation*}
For the second term in \eqref{bp76} we recall Plemelj's formula
\begin{equation*}
	\lim_{\substack{\mu\rightarrow\lambda\in\Sigma\\ \Re\mu>\frac{1}{2}}}\int_{\Sigma}\ln\big(1-\gamma r_1(z)r_2(z)\big)\frac{\d z}{z-\mu}=-\im\pi\ln\big(1-\gamma r_1(\lambda)r_2(\lambda)\big)+\textnormal{pv}\int_{\Sigma}\ln\big(1-\gamma r_1(z)r_2(z)\big)\frac{\d z}{z-\lambda},
\end{equation*}
and notice that by the properties of $r_j$,
\begin{equation*}
	\int_{\Sigma}\frac{r_1(\lambda)r_2(\lambda)}{1-\gamma r_1(\lambda)r_2(\lambda)}\frac{\d}{\d\lambda}\Big[\ln\big(1-\gamma r_1(\lambda)r_2(\lambda)\big)\Big]\,\d\lambda=0,\ \ \ \gamma\in[0,1].
\end{equation*}
Hence,
\begin{align}
	-&\frac{\gamma}{2\pi^2}\int_{\Sigma}\frac{r_1(\lambda)r_2(\lambda)}{1-\gamma r_1(\lambda)r_2(\lambda)}\frac{\d}{\d\lambda}\left[\textnormal{pv}\int_{\Sigma}\ln\big(1-\gamma r_1(z)r_2(z)\big)\frac{\d z}{z-\lambda}\right]\d\lambda\nonumber\\
	&\,\,=-\frac{\gamma}{2\pi^2}\int_{\Sigma}\frac{r_1(\lambda)r_2(\lambda)}{1-\gamma r_1(\lambda)r_2(\lambda)}\frac{\d}{\d\lambda}\left[-\im\pi\ln\big(1-\gamma r_1(\lambda)r_2(\lambda)\big)+\textnormal{pv}\int_{\Sigma}\ln\big(1-\gamma r_1(z)r_2(z)\big)\frac{\d z}{z-\lambda}\right]\d\lambda\nonumber\\
	&\,\,=-\frac{\gamma}{2\pi^2}\int_{\Sigma+\frac{\epsilon}{2}}\frac{r_1(\lambda)r_2(\lambda)}{1-\gamma r_1(\lambda)r_2(\lambda)}\frac{\d}{\d\lambda}\left[\int_{\Sigma}\ln\big(1-\gamma r_1(z)r_2(z)\big)\frac{\d z}{z-\lambda}\right]\d\lambda,\label{bp77}
\end{align}
where we deformed the outer integration contour $\Sigma$ to $\Sigma+\frac{\epsilon}{2}$ in the second equality while using Plemelj's formula and analytic continuation in the integrand. Consequently, with $(z-\lambda)^{-2}=\int_1^{\infty}x^{z-\lambda-1}\ln x\,\d x$, for $z\in\Sigma,\lambda\in\Sigma+\frac{\epsilon}{2}$, after collapsing $\Sigma+\frac{\epsilon}{2}$ to $\Sigma$,
\begin{equation*}
	-\frac{\gamma}{2\pi^2}\int_{\Sigma}\frac{r_1(\lambda)r_2(\lambda)}{1-\gamma r_1(\lambda)r_2(\lambda)}\frac{\d}{\d\lambda}\left[\textnormal{pv}\int_{\Sigma}\ln\big(1-\gamma r_1(z)r_2(z)\big)\frac{\d z}{z-\lambda}\right]\d\lambda\stackrel{\eqref{bp77}}{=}-2\gamma\int_1^{\infty}\ln x\left[\frac{\partial}{\partial\gamma}s(x,\gamma)\right]\hat{s}(x,\gamma)\,\d x,
\end{equation*}
so that all together for \eqref{bp76},
\begin{align}
	\int_{\Sigma}\begin{bmatrix}f_1(\lambda)\\ f_2(\lambda)\end{bmatrix}^{\top}\big({\bf A}^{\top}(\lambda)\big)'\big({\bf A}^{-1}(\lambda)\big)^{\top}\begin{bmatrix}g_1(\lambda)\\ g_2(\lambda)\end{bmatrix}\d\lambda=\gamma\ln t&\left[\frac{\d}{\d\gamma}s(1,\gamma)\right]-2\gamma\int_1^{\infty}\ln x\left[\frac{\partial}{\partial\gamma}s(x,\gamma)\right]\hat{s}(x,\gamma)\,\d x\nonumber\\
	&\,-\frac{\gamma}{2\pi\im}\frac{\d}{\d\gamma}\int_{\Sigma}\frac{r_2'(\lambda)}{r_2(\lambda)}\ln\big(1-\gamma r_1(\lambda)r_2(\lambda)\big)\d\lambda.\label{bp78}
\end{align}
Returning now to \eqref{bp62}, identity \eqref{bp72} combined with \eqref{bp73},\eqref{bp75} and \eqref{bp78} yields
\begin{align*}
	-\frac{1}{2\gamma}\int_{\Sigma_{\bf T}}R_{t,\gamma}(\lambda,\lambda)\,\d\lambda=\frac{\partial}{\partial\gamma}\bigg[-s(1,\gamma)&\ln t+\int_1^{\infty}s(x,\gamma)\hat{s}(x,\gamma)\ln x\,\d x\nonumber\\
	&\,-\frac{1}{4\pi\im}\int_{\Sigma}\left\{\frac{r_1'(\lambda)}{r_1(\lambda)}-\frac{r_2'(\lambda)}{r_2(\lambda)}\right\}\ln\big(1-\gamma r_1(\lambda)r_2(\lambda)\big)\,\d\lambda\bigg]+\tilde{r}(t,\gamma),
\end{align*}
for all $t\geq t_0$ and $\gamma\in[0,1]$ with $|\tilde{r}(t,\gamma)|\leq\frac{c}{\ln t}$. In turn, performing the definite $\gamma$-integration in \eqref{bp62} we finally arrive at, as $t\rightarrow+\infty$,
\begin{eqnarray}
	\ln F(t)\!\!\!\!\!&\stackrel{\eqref{bp62}}{=}&\!\!\!\!\!-\frac{1}{2}\int_0^1\left[\int_{\Sigma_{\bf T}}R_{t,\gamma}(\lambda,\lambda)\d\lambda\right]\frac{\d\gamma}{\gamma}\label{bp79}\\
	&=&\!\!\!\!\!-s(1)\ln t+\int_1^{\infty}s(x)\hat{s}(x)\ln x\,\d x-\frac{1}{4\pi\im}\int_{\Sigma}\left\{\frac{r_1'(\lambda)}{r_1(\lambda)}-\frac{r_2'(\lambda)}{r_2(\lambda)}\right\}\ln\big(1-r_1(\lambda)r_2(\lambda)\big)\,\d\lambda+\mathcal{O}\Big(\frac{1}{\ln t}\Big),\nonumber
\end{eqnarray}
since $s(x,0)=\hat{s}(x,0)=0$ for all $x\geq 1$. This proves \eqref{b:29a} except for the error term. However we can easily improve the $\mathcal{O}(\frac{1}{\ln t})$ in \eqref{bp79} by using \eqref{bp56}. Indeed, by \eqref{bp66},\eqref{bp68} and \eqref{bp70},
\begin{align*}
	t\frac{\d}{\d t}\ln F(t)=-X_1^{11}&(t,1,\phi,\psi)\stackrel{\eqref{bp66}}{=}-Y_1^{11}(t,1,\phi,\psi)-s(1)\\
	&\,\stackrel{\eqref{bp70}}{=}\frac{1}{2\pi\im}\int_{\Sigma_{\bf S}}\Big({\bf S}_-(\lambda;t,1,\phi,\psi)\big({\bf G}_{\bf S}(\lambda;t,1,\phi,\psi)-\mathbb{I}\big)\Big)_{11}\d\lambda-s(1)\ \ \ \ \forall\,t\geq t_0
\end{align*}
and so with \eqref{bp69} and \eqref{bp71}, after indefinite $t$-integration, for all $t\geq t_0$ with $c=c(\epsilon)>0$,
\begin{equation}\label{bp80}
	\ln F(t)=-s(1)\ln t+\eta+r(t),\ \ \ |r(t)|\leq ct^{-\frac{1}{2}\epsilon},
\end{equation}
with a $t$-independent term $\eta$, the integration constant. But the same term was already computed in \eqref{bp79}, hence consistency between \eqref{bp79} and \eqref{bp80} imposes the error term to be power-like small, as claimed in \eqref{b:29b}. This completes our proof of Theorem \ref{theo4a}.
\end{proof}


\begin{appendix}
\section{Useful identities}\label{appA} 
Below we summarize a few algebraic identities used in the proofs of Theorem \ref{theo1} and \ref{theo3}.

\begin{lem}\label{new3} Let $\phi,\psi:\mathbb{R}\rightarrow\mathbb{C}$ satisfy all assumptions of Theorem \ref{theo1}. Then, for any $a,x\in\mathbb{R}_{\geq 0},t\in\mathbb{R}$, with $\tau_t$ as translation,
\begin{equation*}
	\big((I-K_t)^{-1}\tau_{t+a}\phi\big)(x)=\big((I-K_t)^{-1}\tau_{t+x}\phi\big)(a),\ \ \ \ \ \big((I-K_t^{\ast})^{-1}\tau_{t+a}\psi\big)(x)=\big((I-K_t^{\ast})^{-1}\tau_{t+x}\psi\big)(a).
\end{equation*}
\end{lem}
\begin{proof} By symmetry $\phi\leftrightarrow\psi$, which swaps $K_t\leftrightarrow K_t^{\ast}$, it suffices to derive one of the two stated identities. So, fix $a,t$ and note that $\mathbb{R}_{\geq 0}\ni x\mapsto ((I-K_t)^{-1}\tau_{t+x}\phi)(a)$ is in $L^2(\mathbb{R}_+)$ for almost all $a$. Indeed, we need to look at $\mathbb{R}_{\geq 0}\ni x\mapsto(K_t(I-K_t)^{-1}\tau_{t+x}\phi)(a)$ only, given that $\phi\in L^1(\mathbb{R})\cap L^2(\mathbb{R})$ by assumption and since $(I-K_t)^{-1}=I+K_t(I-K_t)^{-1}$. However, by Cauchy-Schwarz inequality,
\begin{equation*}
	\int_0^{\infty}\big|\big(K_t(I-K_t)^{-1}\tau_{t+x}\phi\big)(a)\big|^2\,\d x\leq\int_0^{\infty}\big|(K_t(I-K_t)^{-1}\big)(a,y)\big|^2\,\d y\int_0^{\infty}\int_0^{\infty}\big|(\tau_t\phi)(x+y)\big|^2\,\d x\,\d y,
\end{equation*}
where the second factor is finite by \eqref{n00} and the first is finite almost everywhere in $a$ by Fubini's theorem since $K_t(I-K_t)^{-1}\in\mathcal{C}_2(L^2(\mathbb{R}_+))$. Next, put $j_a(x):=((I-K_t)^{-1}\tau_{t+x}\phi)(a)$, so that by the previous $(I-K_t)j_a\in L^2(\mathbb{R}_+)$ almost everywhere in $a$ and, almost everywhere in $a,x\in\mathbb{R}_{\geq 0}$,
\begin{align}
	&\hspace{1cm}\big((I-K_t)j_a\big)(x)=\int_0^{\infty}(I-K_t)(x,y)\big((I-K_t)^{-1}\tau_{t+y}\phi\big)(a)\,\d y\label{app1}\\
	=&\,\int_0^{\infty}(I-K_t)(x,y)\left[\int_0^{\infty}(I-K_t)^{-1}(a,z)(\tau_{t+y}\phi)(z)\,\d z\right]\d y
	=\int_0^{\infty}(I-K_t)^{-1}(a,z)\big((I-K_t)\tau_{t+z}\phi\big)(x)\,\d z\nonumber
\end{align}
by Fubini's theorem where we used $(\tau_{t+y}\phi)(z)=(\tau_{t+z}\phi)(y)$. However, by \eqref{n1} and again Fubini's theorem,
\begin{equation*}
	\big((I-K_t)\tau_{t+z}\phi\big)(x)=\big((I-K_t)\tau_{t+x}\phi\big)(z)\ \ \ \forall\,x,z\in\mathbb{R}_{\geq 0},\ t\in\mathbb{R},
\end{equation*}
and thus back in \eqref{app1},
\begin{equation}\label{app2}
	\big((I-K_t)j_a\big)(x)=\int_0^{\infty}(I-K_t)^{-1}(a,z)\big((I-K_t)\tau_{t+x}\phi\big)(z)\,\d z=(\tau_{t+x}\phi)(a)=(\tau_{t+a}\phi)(x).
\end{equation}
Consequently, for almost all $a,x\in\mathbb{R}_{\geq 0}$,
\begin{equation*}
	\big((I-K_t)^{-1}\tau_{t+x}\phi\big)(a)=j_a(x)=\big((I-K_t)^{-1}(I-K_t)j_a\big)(x)\stackrel{\eqref{app2}}{=}\big((I-K_t)^{-1}\tau_{t+a}\phi\big)(x).
\end{equation*}
However, by the properties of $\phi,\psi$, both sides in the last equality are continuous functions in $a,x\in\mathbb{R}_{\geq 0}$, thus the same identity holds for all $a,x\in\mathbb{R}_{\geq 0}$. This concludes our proof.
\end{proof}

\begin{lem}\label{new4} Let $\phi,\psi:\mathbb{R}\rightarrow\mathbb{C}$ satisfy all assumptions of Theorem \ref{theo1}. Then for all $a\in\mathbb{R}_{\geq 0},t\in\mathbb{R},z\in\mathbb{R}$, with $\tau_t$ as translation,
\begin{equation}\label{app3}
	\,\,\int_0^{\infty}\e^{\im zu}\left[\int_0^{\infty}(\tau_{t+u}\psi)(x)\big((I-K_t)^{-1}\tau_{t+a}\phi\big)(x)\,\d x\right]\d u=\big((I-K_t)^{-1}h_z\big)(a),
\end{equation}
\begin{equation}\label{app4}
	\int_0^{\infty}\e^{-\im zu}\left[\int_0^{\infty}(\tau_{t+u}\phi)(x)\big((I-K_t^{\ast})^{-1}\tau_{t+a}\psi\big)(x)\,\d x\right]\d u=\big((I-K_t^{\ast})^{-1}g_z\big)(a),
\end{equation}
with 
\begin{equation*}
	g_z(x):=\int_0^{\infty}K_t^{\ast}(x,u)\e^{-\im zu}\,\d u,\ \ \ \ \ h_z(x):=\int_0^{\infty}K_t(x,u)\e^{\im zu}\,\d u,\ \ x\in\mathbb{R}_{\geq 0}.
\end{equation*}
\end{lem}
\begin{proof} As argued in the proof of Theorem \ref{theo1}, the left and right hand sides in \eqref{app4} are well-defined for all $a\geq 0$ and $z\in\mathbb{R}$. To establish the identity, we use Lemma \ref{new3} in the left hand side of \eqref{app4} and then simplify by Fubini's theorem,
\begin{align}
	\int_0^{\infty}&\,(\tau_{t+u}\phi)(x)\big((I-K_t^{\ast})^{-1}\tau_{t+a}\psi\big)(x)\,\d x=\int_0^{\infty}(\tau_{t+u}\phi)(x)\big((I-K_t^{\ast})^{-1}\tau_{t+x}\psi\big)(a)\,\d x\label{app5}\\
	=&\,\int_0^{\infty}(I-K_t^{\ast})^{-1}(a,y)\left[\int_0^{\infty}(\tau_t\psi)(y+x)(\tau_t\phi)(x+u)\,\d x\right]\d y=\int_0^{\infty}(I-K_t^{\ast})^{-1}(a,y)K_t^{\ast}(y,u)\,\d y,\nonumber
\end{align}
so one more time with Fubini,
\begin{align*}
	\int_0^{\infty}\e^{-\im zu}&\,\left[\int_0^{\infty}(\tau_{t+u}\phi)(x)\big((I-K_t^{\ast})^{-1}\tau_{t+a}\psi\big)(x)\,\d x\right]\d u\stackrel{\eqref{app5}}{=}\int_0^{\infty}(I-K_t^{\ast})^{-1}(a,y)\left[\int_0^{\infty}K_t^{\ast}(y,u)\e^{-\im zu}\,\d u\right]\d y\\
	&\hspace{1cm}=\big((I-K_t^{\ast})^{-1}g_z\big)(a),\ \ \ a\geq 0,\ \ z\in\mathbb{R}.
\end{align*}
This concludes our proof of Lemma \ref{new4}, using once more the symmetry $\phi\leftrightarrow\psi$ which swaps \eqref{app4} with \eqref{app3}, modulo the replacement $z\mapsto -z$.
\end{proof}
\begin{lem}\label{new5} Let $\phi,\psi:\mathbb{R}_+\rightarrow\mathbb{C}$ satisfy all assumptions of Theorem \ref{theo3}. Then, for any $a,x\in(0,1],t\in\mathbb{R}_+$, with $\tau_t$ as dilation,
\begin{equation*}
	\big((I-K_t)^{-1}\tau_{ta}\phi\big)(x)=\big((I-K_t)^{-1}\tau_{tx}\phi\big)(a),\ \ \ \ \big((I-K_t^{\ast})^{-1}\tau_{ta}\psi\big)(x)=\big((I-K_t^{\ast})^{-1}\tau_{tx}\psi\big)(a).
\end{equation*}
\end{lem}
\begin{proof} By symmetry $\phi\leftrightarrow\psi$, which swaps $K_t\leftrightarrow K_t^{\ast}$, it suffices to derive one of the two stated identities. So, fix $a,t$ and note that $(0,1]\ni x\mapsto ((I-K_t)^{-1}\tau_{tx}\phi)(a)$ is in $L^2(0,1)$ for almost all $a$: Indeed, we need to look at $(0,1]\ni x\mapsto (K_t(I-K_t)^{-1}\tau_{tx}\phi)(a)$ only, given that $\phi\in L_{\circ}^1(\mathbb{R}_+)\cap L^2(\mathbb{R}_+)$ by assumption and since $(I-K_t)^{-1}=I+K_t(I-K_t)^{-1}$. However, by Cauchy-Schwarz inequality,
\begin{equation*}
	\int_0^1\big|\big(K_t(I-K_t)^{-1}\tau_{tx}\phi\big)(a)\big|^2\,\d x\leq\int_0^1\big|K_t(I-K_t)^{-1}(a,y)\big|^2\,\d y\int_0^1\int_0^1\big|(\tau_t\phi)(xy)\big|^2\,\d x\,\d y,
\end{equation*}
where the second factor is finite by \eqref{b:1} and the first is finite almost everywhere in $a$ by Fubini's theorem since $K_t(I-K_t)^{-1}\in\mathcal{C}_2(L^2(0,1))$. Next, put $j_a(x):=((I-K_t)^{-1}\tau_{tx}\phi)(a)$, so that by the previous $(I-K_t)j_a\in L^2(0,1)$ almost everywhere in $a$ and, almost everywhere in $a,x\in(0,1]$,
\begin{align}
	&\hspace{1.5cm}\big((I-K_t)j_a\big)(x)=\int_0^1(I-K_t)(x,y)\big((I-K_t)^{-1}\tau_{ty}\phi\big)(a)\,\d y\label{app6}\\
	&=\int_0^1(I-K_t)(x,y)\left[\int_0^1(I-K_t)^{-1}(a,z)(\tau_{ty}\phi)(z)\,\d z\right]\d y=\int_0^1(I-K_t)^{-1}(a,z)\big((I-K_t)\tau_{tz}\phi\big)(x)\,\d z\nonumber
\end{align}
by Fubini's theorem where we used $(\tau_{ty}\phi)(z)=(\tau_{tz}\phi)(y)$. However, by \eqref{b:2} and again Fubini's theorem,
\begin{equation*}
	\big((I-K_t)\tau_{tz}\phi\big)(x)=\big((I-K_t)\tau_{tx}\phi\big)(z)\ \ \ \forall\,x,z\in(0,1],\ t\in\mathbb{R}_+,
\end{equation*}
and thus back in \eqref{app6},
\begin{equation}\label{app7}
	\big((I-K_t)j_a\big)(x)=\int_0^1(I-K_t)^{-1}(a,z)\big((I-K_t)\tau_{tx}\phi\big)(z)\,\d z=(\tau_{tx}\phi)(a)=(\tau_{ta}\phi)(x).
\end{equation}
Consequently, for almost all $a,x\in(0,1]$,
\begin{equation*}
	\big((I-K_t)^{-1}\tau_{tx}\phi\big)(a)=j_a(x)=\big((I-K_t)^{-1}(I-K_t)j_a\big)(x)\stackrel{\eqref{app7}}{=}\big((I-K_t)^{-1}\tau_{ta}\phi\big)(x).
\end{equation*}
However, by the properties of $\phi,\psi$, both sides in the last equality are continuous functions in $a,x\in(0,1]$, thus the same identity holds for all $a,x\in(0,1]$. This concludes our proof.
\end{proof}
\begin{lem}\label{new6} Let $\phi,\psi:\mathbb{R}_+\rightarrow\mathbb{C}$ satisfy all assumptions of Theorem \ref{theo3}. Then, for any $a\in(0,1],t\in\mathbb{R}_+,z\in\Sigma=\frac{1}{2}+\im\mathbb{R}$, with $\tau_t$ as dilation,
\begin{equation}\label{app8}
	\int_0^1u^{-z}\left[\int_0^1(\tau_{tu}\psi)(x)\big((I-K_t)^{-1}\tau_{ta}\phi\big)(x)\,\d x\right]\d u=\frac{1}{t}\big((I-K_t)^{-1}h_z\big)(a),
\end{equation}
\begin{equation}\label{app9}
	\int_0^1u^{z-1}\left[\int_0^1(\tau_{tu}\phi)(x)\big((I-K_t^{\ast})^{-1}\tau_{ta}\psi\big)(x)\,\d x\right]\d u=\frac{1}{t}\big((I-K_t^{\ast})^{-1}g_z\big)(a),
\end{equation}
with
\begin{equation*}
	g_z(x):=\int_0^1K_t^{\ast}(x,u)u^{z-1}\,\d u,\ \ \ \ h_z(x):=\int_0^1K_t(x,u)u^{-z}\,\d u,\ \ \ \ x\in(0,1].
\end{equation*}
\end{lem}
\begin{proof} As argued in the proof of Theorem \ref{theo3}, the left and right hand sides in \eqref{app9} are well-defined for all $a\in(0,1]$ and $z\in\Sigma$. To establish the identity, we use Lemma \ref{new5} in the left hand side of \eqref{app9} and then simplify by Fubini's theorem,
\begin{align}
	\int_0^1&(\tau_{tu}\phi)(x)\big((I-K_t^{\ast})^{-1}\tau_{ta}\psi\big)(x)\,\d x=\int_0^1(\tau_{tu}\phi)(x)\big((I-K_t^{\ast})^{-1}\tau_{tx}\psi\big)(a)\,\d x\label{app10}\\
	&=\int_0^1(I-K_t^{\ast})^{-1}(a,y)\left[\int_0^1(\tau_t\psi)(yx)(\tau_t\phi)(xu)\,\d x\right]\d y=\frac{1}{t}\int_0^1(I-K_t^{\ast})^{-1}(a,y)K_t^{\ast}(y,u)\,\d y\nonumber,
\end{align}
so one more time with Fubini,
\begin{align*}
	\int_0^1u^{z-1}&\left[\int_0^1(\tau_{tu}\phi)(x)\big((I-K_t^{\ast})^{-1}\tau_{ta}\psi\big)(x)\,\d x\right]\d u\stackrel{\eqref{app10}}{=}\frac{1}{t}\int_0^1(I-K_t^{\ast})^{-1}(a,y)\left[\int_0^1K_t^{\ast}(y,u)u^{z-1}\,\d u\right]\d y\\
	&\hspace{0.5cm}=\frac{1}{t}\big((I-K_t^{\ast})^{-1}g_z\big)(a),\ \ \ a\in(0,1],\ \ \ z\in\Sigma.
\end{align*}
This concludes our proof of Lemma \ref{new6}, using once more the symmetry $\phi\leftrightarrow\psi$ which swaps \eqref{app9} with \eqref{app8}, modulo the replacement $z\mapsto 1-z$.
\end{proof}

\end{appendix}


\end{document}